\documentclass[a4paper, 11pt]{article}
\usepackage[dvipsnames]{xcolor}
\usepackage{amsmath}
\usepackage{bbm}
\usepackage{amsthm}
\usepackage{amssymb}
\usepackage{dsfont}
\usepackage{comment}
\usepackage{authblk}
\usepackage{cancel}
\usepackage{soul}
\usepackage[utf8]{inputenc}

\usepackage[active]{srcltx}
\usepackage{stackengine}
\usepackage[bookmarksnumbered=true]{hyperref}

\usepackage{geometry}
\newtheorem{definition}{Definition}[section]
\newtheorem{lemma}[definition]{Lemma}
\newtheorem{proposition}[definition]{Proposition}
\newtheorem{theorem}[definition]{Theorem}

\newtheorem{assumption}[definition]{Assumption}
\theoremstyle{definition}
\newtheorem{remark}[definition]{Remark}

\numberwithin{equation}{section}

\renewcommand{\epsilon}{\varepsilon}

\newcommand{\cB}{\mathcal{B}}

\renewcommand{\phi}{\varphi}
\newcommand{\R}{\mathbb{R}}

\newcommand{\Z}{\mathbb{Z}}
\newcommand{\cN}{\mathcal{N}}
\newcommand\eps\varepsilon

\newcommand\bQ{\mathbb{Q}}
\newcommand\bH{\mathbb{H}}
\newcommand\cH{\mathcal{H}}
\newcommand\nn{\nonumber}

\renewcommand\rho\varrho

\newcommand\ind{\mathds{1}}
\newcommand\1{\mathds{1}}

\DeclareMathOperator{\supp}{supp}

\allowdisplaybreaks

\begin{document}

\title{The Huang--Yang conjecture for the low-density Fermi gas}

\author[1]{Emanuela L. Giacomelli}
\affil[1]{\small University of Milan, Department of Mathematics, Via Cesare Saldini 50, 20133 Milan, Italy}
\author[2]{Christian Hainzl}
\author[2]{Phan Thành Nam}
\affil[2]{\small LMU Munich, Department of Mathematics, Theresienstr. 39, 80333 Munich, Germany}
\author[3]{Robert Seiringer}
\affil[3]{\small Institute of Science and Technology Austria, Am Campus 1, 3400 Klosterneuburg, Austria {\footnotesize{emanuela.giacomelli@unimi.it, hainzl@math.lmu.de, nam@math.lmu.de, robert.seiringer@ist.ac.at}}}

\maketitle

\abstract{Our work establishes  a three-term asymptotic expansion of the ground state energy of a dilute gas of spin $1/2$ fermions with repulsive short-range interactions, validating a formula predicted by Huang and Yang in 1957. The formula is universal in the sense that it holds for a large class  of interaction potentials and depends on those  only via their scattering length. 
We have recently proved an upper bound on the ground state energy of the desired form, and the present work completes the program by proving the matching lower bound.} 

\tableofcontents

\section{Introduction and main result}

In 1957, Huang and Yang \cite{HY} utilized a pseudopotential method to derive an asymptotic expansion of the ground state energy of dilute Fermi gases. They predicted that the ground state energy per unit volume in the thermodynamic limit of a spin $1/2$ Fermi gas of hard spheres of diameter $a>0$ with  
equal spin densities $\rho_\uparrow=\rho_\downarrow = \rho/2$ is given by 
\begin{equation}\label{eq:HY}
  e(\rho/2, \rho/2) = \frac{3}{5}(3\pi^2)^{\frac{2}{3}}\rho^{\frac{5}{3}} + 2\pi a \rho^2 + \frac{4}{35}(11 -2\log2) (9\pi)^{\frac{2}{3}} a^2 \rho^{7/3} + o(\rho^{\frac{7}{3}})_{\rho\to 0}. 
\end{equation}
The first term can be understood as  arising from the kinetic energy of a non-interacting Fermi gas, the second term captures the leading order contribution of the interaction from opposite-spin components, and the third term emerges from a renormalization of the scattering equation, where the scattering process is constrained by Pauli blocking in the Fermi sea.

Proving the validity of the formula \eqref{eq:HY} from first principles of quantum mechanics is a major challenge in mathematical physics. It is expected that the formula holds for a large class of repulsive and short-range interactions, thus exhibiting a remarkable universality: the first three terms of the energy depend on the interaction potential only via its scattering length $a$. Motivated by this, the validity of the first two terms in \eqref{eq:HY} was proved in \cite{LSS-05} for a large class of repulsive interactions, including hard spheres. More generally, \cite[Theorem 1]{LSS-05} states that for two arbitrary densities $\rho_\uparrow, \rho_\downarrow$ of the two spin components,  the corresponding ground state energy density satisfies 
\begin{align} \label{eq:LSS}
e(\rho_\uparrow, \rho_\downarrow)= \frac{3}{5}(6\pi^2)^{\frac{2}{3}} \left(  \rho_\uparrow^{5/3} + \rho_\downarrow^{5/3}\right)  +8\pi a \rho_\uparrow\rho_\downarrow + o(\rho^2)_{\rho\to 0}. 
\end{align}
We refer to \cite{FGHP} for an alternative proof of \eqref{eq:LSS}, and to \cite{Gia1,Lau,Gia3} for improved error estimates. Recently in \cite{GHNS}, we considered a system of spin-$1/2$ fermions 
interacting via a square-integrable repulsive potential, and proved the upper bound
\begin{equation}\label{eq: main thm 1-upper}
  e(\rho_\uparrow, \rho_\downarrow) \leq  \frac{3}{5}(6\pi^2)^{\frac{2}{3}} \left(\rho_\uparrow^{\frac{5}{3}} + \rho_\downarrow^{\frac{5}{3}}\right) + 8\pi a \rho_\uparrow\rho_\downarrow + a^2 \rho_\uparrow^{\frac{7}{3}} F\left(\frac{\rho_\downarrow}{\rho_\uparrow}\right) + O(\rho^{\frac{7}{3} +\frac{1}{9}})_{\rho \to 0},
\end{equation}
where the third order contribution $a^2 \rho_\uparrow^{\frac{7}{3}} F({\rho_\downarrow}/{\rho_\uparrow})$ comes from the expression 
\begin{equation}\label{inte:F}
 \frac{ a^2}{8\pi^7}\int_{\mathbb{R}^3} dp \int_{|r| \le k_F^\uparrow }\hspace{-0.4cm}dr \int_{|r^\prime| \le k_F^\downarrow}\hspace{-0.45cm} dr^\prime 
 \left(\frac{1}{2|p|^2} - \frac{\ind_{\{|r+p|>k_{F}^\uparrow\}} \ind_{\{|r^\prime-p|>k_{F}^\downarrow\}}}{|r+p|^2 - |r|^2 + |r^\prime -p|^2 - |r^\prime|^2} \right)
\end{equation}
with $k_{F}^\sigma= (6\pi^2)^{\frac{1}{3}}\rho_\sigma^{1/3}$ the radii of the Fermi balls of spin $\sigma\in \{\uparrow,\downarrow\}$. 
As proved in \cite{GHNS}, the explicit form of the function $F:\R_+\to \R_+$ is 
\begin{align}\label{eq:def-F}
F(x) =  \frac {(6\pi^2)^{1/3}}{35} & \Biggl(   {16 x^{7/3}}  \ln x  -   48 \left(x^{7/3} + 1\right)  \ln (1+x^{1/3}) \nonumber \\
& +  6 \left( 15  x^{1/3} - 4x^{2/3} + 33 x + 33 x^{4/3}  -4 x^{5/3} + 15 x^{2}\right) \nonumber\\
 &+   {21} \left( 1 - 6   x^{2/3} +5  x +5   x^{4/3} -6 x^{5/3}
 + x^{7/3} \right) \ln\frac {|1- x^{1/3}|}{1+x^{1/3}} 
  \Biggl),
\end{align}
which can also be found in the physics literature \cite{Kanno,ChaWoj,Pera}. Note that 
$$F(1) = \lim_{x\rightarrow 1}F(x) = \frac{48}{35}(11 -2\log2) (6 \pi^2)^{\frac{1}{3}},$$
and hence \eqref{eq: main thm 1-upper} agrees with the prediction \eqref{eq:HY} when $\rho_\uparrow = \rho_\downarrow = \rho/2$.

In the present paper, we will prove a lower bound matching \eqref{eq: main thm 1-upper}, thus justifying the Huang--Yang conjecture \eqref{eq:HY} for a general class of interaction potentials specified in Assumption 1.1 below. Let us give a precise description of the model before stating our main result.

We consider a system of $N$ fermions with spin $1/2$  
in a box $\Lambda:= [-L/2,L/2]^3\subset\mathbb{R}^3$ with periodic boundary conditions. The Hamiltonian of the system is given by 
\begin{equation}\label{eq: hamiltonian}
  H_N = -\sum_{j=1}^N\Delta_{x_j} + \sum_{1\leq i< j\leq N} V(x_i - x_j), 
\end{equation}
acting on $\bigwedge^N L^2(\Lambda,\mathbb{C}^2)$, the  Hilbert space of antisymmetric square-integrable functions of $N$ space-spin variables $(x_i,\sigma_i)$, with $x_i\in \Lambda$ and  $\sigma_i \in \{\uparrow,\downarrow\}$. Here the standard Laplacian $-\Delta$ on $L^2(\Lambda)$ stands for the kinetic energy of each particle, and the interaction between particles is described by the  function $V: \Lambda\to \mathbb{R}$ of the form
\begin{align}\label{eq:def-V}V(x) = \sum_{z\in\Z^3} V_\infty(x+Lz).
\end{align}

We impose the following conditions on the interaction potential $V_\infty: \R^3\to \R$. 
\begin{assumption}\label{asu: potential V}
  The interaction potential $V_\infty\in L^2(\mathbb{R}^3)$ is nonnegative, radial and compactly supported.
  \end{assumption}
Under this condition, the Hamiltonian $H_N$ is bounded from below and can be defined as a self-adjoint operator. We are interested in the restriction of $H_N$ to  the subspace $\mathfrak{h}(N_\uparrow, N_\downarrow) \subset \bigwedge^N L^2(\Lambda,\mathbb{C}^2)$ consisting of $N$-body wave functions with exactly $N_{\sigma}$ particles of spin $\sigma\in \{\uparrow,\downarrow\}$, where $N_\uparrow + N_\downarrow=N$. The corresponding ground state energy in the subspace $\mathfrak{h}(N_\uparrow, N_\downarrow)$ is given by 
\begin{equation}\label{eq: def gs energy}
  E_L(N_\uparrow, N_\downarrow) = \inf_{\psi\in \mathfrak{h}(N_\uparrow, N_\downarrow)}\frac{\langle \psi, H_N\psi\rangle}{\langle \psi, \psi\rangle}.
\end{equation}
Our result concerns the ground state energy density in the thermodynamic limit
\begin{equation}\label{eq: def gs energy density}
  e(\rho_\uparrow, \rho_\downarrow) = \lim_{\substack{L \rightarrow \infty \\ N_\sigma/L^3 \rightarrow \rho_\sigma, \,\, \sigma \in \{\uparrow, \downarrow\}}}\frac{E_L(N_\uparrow, N_\downarrow)}{L^3},
\end{equation}
whose existence is well-known \cite{Ruelle-99,Robinson-71}. In fact, the ground state energy  $e(\rho_\uparrow, \rho_\downarrow)$ is  independent  of the boundary conditions imposed on $\Lambda= [-\frac{L} 2,\frac{L} 2]^3\subset\mathbb{R}^3$ (see \cite{Ruelle-99,Robinson-71}), but we fix the periodic boundary conditions to have a concrete model.  

In the low density limit, the asymptotic expansion of $e(\rho_\uparrow, \rho_\downarrow)$ involves the $s$-wave scattering length $a>0$ of the interaction potential $V_\infty$, which is defined by 
\begin{align}\label{eq: 8pia-first-way}
  8\pi a = \int_{\R^3} V_\infty (1-\varphi_\infty)
\end{align}
  with $\varphi_\infty:\R^3\to [0,1]$ the solution of the zero-energy scattering equation 
\begin{equation}\label{eq: zero en scatt eq}
  2\Delta\varphi_\infty + V_\infty (1-\varphi_\infty) = 0, \qquad \varphi_\infty(x)\rightarrow 0 \quad \mbox{as}\quad |x| \rightarrow \infty.
\end{equation}

Our main result is a rigorous justification of the expansion \eqref{eq: main thm 1-upper} as an equality. 
We allow for different values of the spin densities $\rho_\uparrow$ and $\rho_\downarrow$, but for simplicity assume that they are of the same order in the low-density limit $\rho\to 0$.

\begin{theorem}[Huang--Yang formula]\label{thm: main lw bd} Let $V_\infty$ be as in Assumption \ref{asu: potential V}, and let $a>0$ be its scattering length. Assume that $C^{-1}\leq \rho_\uparrow/\rho_\downarrow\leq C$ for some constant $C>0$ independent  of $\rho$.
In the low-density limit $\rho_\uparrow+ \rho_\downarrow = \rho \rightarrow 0$, the ground state energy density defined in \eqref{eq: def gs energy density} satisfies
\begin{equation}\label{eq: main thm 1}
  e(\rho_\uparrow, \rho_\downarrow) =  \frac{3}{5}(6\pi^2)^{\frac{2}{3}} \left(\rho_\uparrow^{\frac{5}{3}} + \rho_\downarrow^{\frac{5}{3}}\right) + 8\pi a \rho_\uparrow\rho_\downarrow + a^2 \rho_\uparrow^{\frac{7}{3}} F\left(\frac{\rho_\downarrow}{\rho_\uparrow}\right) + O(\rho^{\frac{7}{3} + \frac{1}{120} }), 
\end{equation}
where $F$ is defined in \eqref{eq:def-F}. 
In particular, for $\rho_\uparrow = \rho_\downarrow = \rho/2$, \eqref{eq:HY} holds true. 
\end{theorem}

As we have mentioned, the upper bound of \eqref{eq: main thm 1} was proved in \cite{GHNS}, and we will focus here on establishing the matching lower bound. Our proof of the lower bound is rather different from that of the upper bound, although both are based on the general strategy of adapting the bosonic Bogoliubov theory to fermionic systems. 

The upper bound in  \cite{GHNS} is achieved by constructing a trial state using quasi-bosonic Bogoliubov transformations, where suitable pairs of fermions are treated as bosons. The corresponding energy is extracted from the (approximate) diagonalization of the Hamiltonian by  the Bogoliubov transformations, with many error terms controlled by specific properties of the trial state. 

For the lower bound, however, the available a-priori estimates on the ground state are insufficient to implement the diagonalization method up to the Huang--Yang correction. Instead, we will use a version of the ``completing the square" method, i.e., we will decompose the Hamiltonian into several terms, in such a way that we can drop some non-negative terms of the form $A^*A$ to obtain a lower bound.   Our approach is inspired by the recent work  \cite{ChrHaiNam-24} on the fermionic correlation energy in the mean-field regime, and also earlier works \cite{FS-20,FS-23,Brooks-25} on  
dilute Bose gases. The rigorous implementation of the method for dilute Fermi gases requires key new ideas, however, which we will explain in the next section.

All the results discussed above concern the ground-state energy at zero temperature. We refer to \cite{S06} for an extension of \eqref{eq:LSS} to the thermodynamic pressure at  temperatures of the order of the Fermi temperature, i.e., $T \sim \rho^{2/3}$. A study of the free energy of the dilute Fermi gas to the accuracy of the Huang--Yang order  appeared shortly after the completion of the present work in \cite{CWZ25}, where temperatures in the range $T \lesssim \rho^{2/3 + 1/6}$ are considered.
For related results on spinless (or, equivalently, completely spin-polarized) fermions at zero and positive temperatures, we refer to \cite{LS2,LS1,LS3,LS4}; in this case, the leading correction due to the repulsive interactions is much smaller, of the order $\rho^{8/3}$.
In the last years, there has also been substantial progress on the study of the ground state energy for interacting fermionic particles in the finite-volume setting; see \cite{BNPSS-20,BNPSS-21,ChrHaiNam-23a,BPSS-21} for results on the mean-field limit with regular potentials, \cite{ChrHaiNam-23,ChrHaiNam-24} for an analogue of the Gell-Mann--Brueckner formula for the mean-field Coulomb gas, and \cite{CheWuZha} for a simplified version of \eqref{eq:HY} in an ultra-dilute regime inspired by the Gross--Pitaevskii limit for bosons \cite{LSY}.

\section{Key ingredients in the proof}\label{sec: strategy}

In this section we shall explain the overall strategy, highlighting the main ideas and key ingredients in the proof. 

\bigskip
\noindent
{\bf Notation.}  We will denote by $C$ a general finite constant independent of $L$ and $\rho$ (it may depend on $V$). The value of $C$ may change from line to line. We use $\lim_{L \to \infty}$ as shorthand for the thermodynamic limit, and we denote by $\mathfrak{e}_L$ a general expression that tends to $0$ as $L \to \infty$. For any $f\in L^2(\Lambda, \mathbb{C}^2)$, we write $\hat{f}(\cdot, \sigma)  = \hat{f}_\sigma(\cdot)$ and, similarly, $f(\cdot, \sigma) = f_\sigma(\cdot)$. We will often write $\int$ in place of $\int_\Lambda$,  $\sum_p$ in place of $\sum_{p\in\Lambda^\ast}$ and $\sum_\sigma$ in place of $\sum_{\sigma\in\{\uparrow, \downarrow\}}$, where $\Lambda^*=(2\pi/ L)\mathbb{Z}^3$.

\subsection{The correlation Hamiltonian}\label{ss:Fock}

We will first extract the energy contribution of the ground state of the non-interacting Fermi gas, and then focus on the correlation energy. For this purpose, it is convenient to work in the fermionic Fock space 
$$\mathcal{F}_{\mathrm{f}} = \bigoplus_{n\geq 0} \mathcal{F}^{(n)}_{\mathrm{f}},\quad \mathcal{F}^{(n)}_{\mathrm{f}} = \bigwedge^n L^2(\Lambda, \mathbb{C}^2)$$ 
where the number of particles is not fixed. A state $\psi=( \psi^{(n)})_{n=0}^\infty \in \mathcal{F}_{\mathrm{f}}$ has components $\psi^{(n)}= \psi^{(n)}((x_1,\sigma_1), \ldots, (x_n,\sigma_n)) \in \mathcal{F}^{(n)}_{\mathrm{f}}$. The vacuum state will be denoted by $\Omega=(1,0,0,...)$.

We will use the language of second quantization, where for any $f\in L^2(\Lambda, \mathbb{C}^2)$ we denote by $a^\ast(f)$ and $a(f)$ the standard creation and annihilation operators on Fock space $\mathcal{F}_{\mathrm{f}}$. They satisfy the canonical anticommutation relations (CAR)
\begin{align}\label{eq:CAR}
  \{a(f), a^\ast(g)\} = \langle f,g\rangle_{L^2(\Lambda;\mathbb{C}^2)}, \quad \{a(f), a(g)\} = \{a^\ast(f), a^\ast (g)\} = 0
\end{align}
for any $f,g\in L^{2}(\Lambda; \mathbb{C}^2)$, where $\{A,B\}=AB+BA$.  
For $\sigma \in \{\uparrow, \downarrow\}$ and $k\in \Lambda^*:=(2\pi/L)\mathbb{Z}^3$, we denote $(\delta_\sigma f_k)(x,\sigma') = \delta_{\sigma,\sigma'}{L^{-3/2}} e^{i k\cdot x}$ and 
\begin{equation}\label{eq: def ak a*k}
\hat{a}_{k,\sigma}  
 = a(\delta_{\sigma} f_k)   = \frac 1{L^{3/2}} \int_\Lambda dx \, a_{x,\sigma} e^{- i k x}
\end{equation}
where 
we introduced the operator-valued distributions
$a_{x,\sigma} = a(\delta_{x,\sigma})$ with $\delta_{x,\sigma}(y,\sigma^\prime) = \delta_{\sigma,\sigma^\prime}\delta(x-y)$. 

We denote by $\mathcal{N}$ the total number operator, which can be further decomposed as 
\begin{equation}\label{eq: def N}
  \mathcal{N}=\sum_{\sigma \in\{\uparrow, \downarrow\}}\mathcal{N}_\sigma,\quad 
\mathcal{N}_\sigma = \sum_{k\in\Lambda^*} \hat{a}_{k,\sigma}^\ast \hat{a}_{k,\sigma}=  \int_\Lambda dx\, a_{x,\sigma}^\ast {a}_{x,\sigma},
 \end{equation}
where $\mathcal{N}_\sigma$ is the number operator associated with particles of spin $\sigma$. 

The second quantized version of $H_N$ in \eqref{eq: hamiltonian} is the following Hamiltonian on Fock space 
\begin{align*}
  \mathcal{H} &= \sum_{\sigma\in\{\uparrow, \downarrow\}}\int_{\Lambda}dx\, \nabla_x a^\ast_{x,\sigma}\nabla_x a_{x,\sigma} + \frac{1}{2}\sum_{\sigma,\sigma^\prime \in \{\uparrow, \downarrow\}} \int_{\Lambda \times \Lambda} dxdy\, V(x-y) a^\ast_{x,\sigma}a^\ast_{y,\sigma^\prime}a_{y,\sigma^\prime}a_{x,\sigma} \\
&= \sum_{\sigma\in\{\uparrow, \downarrow\}} \sum_{k\in\Lambda^*} |k|^2 \hat{a}_{k,\sigma}^\ast \hat{a}_{k,\sigma}  + \frac{1}{2 L^3}\sum_{\sigma,\sigma^\prime \in \{\uparrow, \downarrow\}} \sum_{k,p,q\in\Lambda^*} \hat V(k) \hat{a}_{p+k,\sigma}^\ast \hat{a}_{q-k,\sigma'}^\ast  \hat{a}_{q,\sigma'} \hat{a}_{p,\sigma}, 
 \end{align*}
 which coincides with $H_N$ when restricted to $\mathcal{F}^{(N)}_{\mathrm{f}} = \bigwedge^N L^2(\Lambda,\mathbb{C}^2)$ where $\cN=N$. Here we use the convention 
\[
  \hat{f}(p) = \int_{\Lambda} dx\, f(x)e^{-ip\cdot x}\ , \qquad f(x) = \frac{1}{L^3}\sum_{p\in\Lambda^*}\hat{f}(p) e^{ip\cdot x}
\]
for the Fourier transform of $f:\Lambda\to \mathbb{C}$. 
The Fourier transform of a function $g:\mathbb{R}^3\to \mathbb{C}$ will instead be denoted by $\mathcal{F}(g)$. 

For given positive integers $N_\uparrow$, $N_\downarrow$ with $N_\uparrow+N_\downarrow=N$, we can identify $\mathfrak{h}(N_\uparrow, N_\downarrow)$ with the subspace $\mathcal{F}_{\mathrm{f}}(N_\uparrow, N_\downarrow)\subset \mathcal{F}_{\mathrm{f}}$  where $\mathcal{N}_\sigma = N_\sigma$ for $\sigma\in\{\uparrow,\downarrow\}$ (a subspace left invariant by $\mathcal{H}$), and write the ground state energy in \eqref{eq: def gs energy} as
\[
  E_{L}(N_\uparrow,N_\downarrow) = \inf_{\psi\in\mathcal{F}_{\mathrm{f}}{(N_\uparrow, N_\downarrow)}}\frac{\langle \psi, \mathcal{H}\psi\rangle}{\langle \psi, \psi \rangle}.
\]

At low density, the ground state energy $E_{L}(N_\uparrow,N_\downarrow)$ agrees with that of a non-interacting system to leading order. The ground state of the free Fermi gas is given by the Slater determinant  
$$
\psi_{\mathrm{FFG}} = \prod_{\sigma\in \{\uparrow,\downarrow\}} \prod_{k\in \mathcal{B}_F^\sigma} \hat a^*_{k,\sigma} \Omega
$$
where the relevant momenta are restricted to the Fermi balls 
$$
\mathcal{B}_F^\sigma = \{k\in \Lambda^* = (2\pi/L)\mathbb{Z}^3 \,|\, |k|\le k_F^\sigma\}
$$
where  \[k_F^\sigma= (6\pi^2)^{\frac{1}{3}}\rho_\sigma^{1/3} + o(1)_{L\to \infty},\]
where for simplicity we suppress the dependence on $L$ in the notation $k_F^\sigma$. As explained in \cite{FGHP,Gia1, GHNS}, we can assume, without loss of generality, that the Fermi balls are completely filled, i.e., 
\begin{align}\label{eq:Nsigma}
N_\sigma = |\mathcal{B}_F^\sigma|, \quad \sigma \in \{\uparrow,\downarrow\}. 
\end{align}

As in \cite{FGHP,Gia1, GHNS} we factor out the contribution of the free Fermi gas state $\psi_{\mathrm{FFG}}$ by a (unitary) particle-hole  transformation. For that purpose, let us denote the orthogonal projections $u,v$ on $L^2(\Lambda,\mathbb{C}^2)$ with integral kernels
\begin{align}\label{eq: def u,v}
v_{\sigma,\sigma^\prime}(x;y) =v_{\sigma,\sigma^\prime}(x-y) = \frac{\delta_{\sigma,\sigma^\prime}}{L^3}\sum_{k\in\mathcal{B}_F^\sigma}e^{ik\cdot (x-y)}, \nn\\
  u_{\sigma, \sigma^\prime}(x;y) =   u_{\sigma, \sigma^\prime}(x-y) = \frac{\delta_{\sigma,\sigma^\prime}}{L^3} \sum_{k\notin\mathcal{B}_F^{\sigma}}e^{ik\cdot(x-y)}. 
\end{align}
Note that  $v$ is the one-particle  reduced density matrix of $ \psi_{\mathrm{FFG}}$ and $u = 1 - v$. In particular, $u v = v u = 0$. We will write $u_\sigma = u_{\sigma, \sigma}$ and $v_\sigma = v_{\sigma, \sigma}$ for simplicity. The Fourier coefficients $\hat{u}_\sigma,\hat{v}_\sigma$  of the kernels introduced in \eqref{eq: def u,v} are given by  
\begin{equation}\label{eq: def u hat, v hat}
\hat{u}_\sigma(k) = \begin{cases}
  0 &\mbox{if}\,\,\, |k| \leq k_F^\sigma, \\ 1 &\mbox{if}\,\,\, |k| > k_F^\sigma,
  \end{cases}\qquad \hat{v}_\sigma(k) = \begin{cases} 1 &\mbox{if}\,\,\, |k| \leq k_F^\sigma, \\ 0 &\mbox{if} \,\,\, |k| > k_F^\sigma. \end{cases}
\end{equation}
In the following, we will always denote $f_x(\cdot)= f(\cdot -x)$ and $a_{\sigma}(f)=a(\delta_\sigma f)$.

\begin{definition}[Particle-hole transformation]\label{def: ferm bog} The particle-hole transformation $R: \mathcal{F}_{\mathrm{f}}\rightarrow\mathcal{F}_{\mathrm{f}}$ is the unitary operator defined by $R\Omega=\Psi_{\mathrm{FFG}}$ and 
  \begin{equation}\label{eq: prop II R}
    R^\ast a_{x,\sigma}^\ast R = a_\sigma^\ast(u_x) + a_\sigma(v_x),
  \end{equation}
  where $u$,$v: L^2(\Lambda;\mathbb{C}^2)\rightarrow L^2(\Lambda;\mathbb{C}^2)$ are given in \eqref{eq: def u,v}. 
\end{definition}

If $\Psi \in \mathcal{F}_{\mathrm{f}}$ satisfies $\mathcal{N}_\sigma \Psi=N_\sigma\Psi$ for $\sigma\in \{\uparrow,\downarrow\}$, then $\psi = {R}^* \Psi \in \mathcal{F}_{\mathrm{f}}$ obeys the particle-hole relation  
\begin{align}\label{eq:particle-hole}
 \sum_{k\in \mathcal{B}_F^{\sigma}} \hat{a}^*_{k,\sigma}\hat{a}_{k,\sigma}  \psi =  \sum_{k\in (\mathcal{B}_F^{\sigma})^c } \hat{a}^*_{k,\sigma}\hat{a}_{k,\sigma}  \psi  = \frac{1}{2}\mathcal{N}_\sigma  \psi, \qquad \forall \sigma\in \{\uparrow,\downarrow\},
\end{align}
i.e., in the transformed state $\psi = R^* \Psi$ there are as many particles outside the Fermi balls as there are inside (corresponding to holes for the original state $\Psi$). In configuration space, we can alternatively write
\begin{equation}\label{eq: int au av particle hole}
 \sum_{k\in \mathcal{B}_F^{\sigma}} \hat{a}^*_{k,\sigma}\hat{a}_{k,\sigma}  =  \int dx \,  |a_{\sigma}(v_x)|^2 \ , \quad \sum_{k\in (\mathcal{B}_F^{\sigma})^c } \hat{a}^*_{k,\sigma}\hat{a}_{k,\sigma}  = \int dx\, |a_\sigma(u_x)|^2 
\end{equation}
using the standard convention $|A|^2 = A^\ast A$ for an operator $A$.

By omitting the non-negative contributions from the equal spin interactions and performing a straightforward computation using the decomposition \eqref{eq: prop II R}, we arrive at the following lower bound. 

\begin{proposition}[Conjugation of $\mathcal{H}$ by $R$]\label{pro: fermionic transf}
Let $\Psi\in \mathcal{F}_{\mathrm{f}}$ be a normalized state satisfying $\mathcal{N}_\sigma\Psi = N_\sigma\Psi$ for $\sigma\in\{\uparrow,\downarrow\}$ with $N_\sigma$ given in \eqref{eq:Nsigma}. 
  Then 
  \begin{equation} \label{eq: particle hole lw bd}
    \langle \Psi, \mathcal{H}\Psi\rangle \geq E_{\mathrm{FFG}} + \langle R^\ast\Psi,   \mathcal{H}_{\mathrm{corr}} R^\ast\Psi\rangle
      \end{equation}
      where the free Fermi gas energy $E_{\mathrm{FFG}}$ satisfies, in the thermodynamic limit, 
      \begin{equation}\label{eq: FFG energy}
  \lim_{L\to \infty}  \frac{E_{\mathrm{FFG}}}{L^3} 
    = \frac{3}{5}(6\pi^2)^{\frac{2}{3}}\left(\rho_\uparrow^{\frac{5}{3}} + \rho_{\downarrow}^{\frac{5}{3}}\right) +\hat{V}(0)\rho_\uparrow\rho_\downarrow + {O}(\rho^{\frac{8}{3}})_{\rho \to 0}
\end{equation}
and $\mathcal{H}_{\mathrm{corr}}$ is the correlation Hamiltonian
      $$ \mathcal{H}_{\mathrm{corr}} = \mathbb{H}_0 +  \mathbb{Q}_2 + \mathbb{Q}_3 + \mathbb{Q}_4 + \mathcal{E}_{\rm corr}$$
with 
\begin{align}\label{eq: def H-corr}
  \mathbb{H}_0 &= \sum_\sigma\sum_k ||k|^2 -(k_F^\sigma)^2| \hat{a}_{k,\sigma}^\ast \hat{a}_{k,\sigma} \nonumber\\
  \mathbb{Q}_2 &= \frac{1}{2}\sum_{\sigma \neq \sigma^\prime} \int_{\Lambda^2} dxdy\, V(x-y) a^\ast_\sigma(u_x)a^\ast_{\sigma^\prime}(u_y)a^\ast_{\sigma^\prime}(v_y)a^\ast_{\sigma}(v_x) + \mathrm{h.c.}\nonumber
  \\
   \mathbb{Q}_3 &= \sum_{\sigma \neq \sigma^\prime} \int_{\Lambda^2} dxdy\, V(x-y) a^\ast_{\sigma'}(u_y)a^\ast_{\sigma}(u_x) a^\ast_{\sigma}(v_x) a_{\sigma'}(u_y) + \mathrm{h.c.}\nonumber 
  \\
    \mathbb{Q}_4 &= \frac{1}{2}\sum_{\sigma\neq \sigma^\prime}\int_{\Lambda^2} dxdy \, V(x-y)  {a}^\ast_\sigma(u_x){a}^\ast_{\sigma^\prime}(u_y){a}_{\sigma^\prime}(u_y){a}_\sigma(u_x) \nonumber
    \\
  \mathcal{E}_{\rm corr} &=\frac{1}{2}\sum_{\sigma\neq\sigma^\prime}\int_{\Lambda^2} dxdy\, V(x-y) a^\ast_\sigma(v_x)a^\ast_{\sigma^\prime}(v_y)a_{\sigma^\prime}(v_y)a_\sigma(v_x)\nonumber 
 \\
&\quad  +  \sum_{\sigma\neq \sigma^\prime}\int_{\Lambda^2} dxdy\, V(x-y) a^\ast_\sigma(u_x)a^\ast_{\sigma}(v_x)a_{\sigma^\prime}(v_y)a_{\sigma^\prime}(u_y)\nn\\
  &\quad - \sum_{\sigma\neq \sigma^\prime}\int_{\Lambda^2} dxdy\, V(x-y) a^\ast_\sigma(u_x)a^\ast_{\sigma^\prime}(v_y)a_{\sigma^\prime}(v_y)a_\sigma(u_x)\nonumber
\\
  &\quad  + \sum_{\sigma \neq \sigma^\prime} \int_{\Lambda^2} dxdy\, V(x-y)(a^\ast_\sigma(u_x)a^\ast_{\sigma^\prime}(v_y) a^\ast_{\sigma}(v_x)  a_{\sigma^\prime}(v_y) + \mathrm{h.c.}).
\end{align}
Moreover, for any $\psi \in \mathcal{F}_{\mathrm{f}}$ and $0\le \beta<1$, we have 
\begin{equation}\label{eq: eps-1}
  |\langle \psi, \mathcal{E}_{\rm corr} \psi\rangle| \leq C\rho^{1+\frac{\beta}{6}} \langle \psi, \mathcal{N}\psi\rangle + C\rho \|\mathcal{N}^{\frac{1}{2}} \psi\|\|\mathcal{N}_\beta^{\frac{1}{2}}\psi\|
\end{equation}
with $\mathcal{N}$ defined in \eqref{eq: def N} and 
\begin{equation}\label{eq: def Nalpha >}
\mathcal{N}_\beta = \sum_{\sigma\in \{\uparrow, \downarrow\}}\sum_{ ||k|- k_F^\sigma| > (k_F^\sigma)^{1+\beta}}\hat{a}_{k,\sigma}^\ast \hat{a}_{k,\sigma}.
\end{equation}
\end{proposition}

In \eqref{eq: particle hole lw bd}, we dropped the (non-negative) equal spin interactions for a lower bound, which explains the absence of an exchange term in the correlation Hamiltonian $\mathcal{H}_{\mathrm{corr}}$ (compare with  \cite[Proposition 2.3]{GHNS}). Consequently, $E_{\mathrm{FFG}}$ in \eqref{eq: particle hole lw bd} does not include equal-spin interactions as in the previous works \cite{FGHP,Gia1, GHNS}, but the difference is negligible (in fact, $O(\rho^{8/3})$) and \eqref{eq: FFG energy} remains valid.  
Additionally, our operators $\bQ_j$ slightly differ from those in \cite{GHNS} as all interaction terms in the present paper are restricted to particles with opposite spins. The proof of \eqref{eq: eps-1}  will be given at the end of Section~\ref{sec:pre}. 

From \cite[Lemma 3.5, Corollary 3.7, Lemma 3.9]{FGHP}, we have the following basic a priori estimates. 

\begin{lemma}[Basic a priori estimates for $\mathbb{H}_0$, $\mathbb{Q}_4$, $\mathcal{N}$]\label{lem: a priori} Let $V_\infty$ be as in Assumption \ref{asu: potential V}.
Let $\Psi\in \mathcal{F}_{\mathrm{f}}$ be a normalized state satisfying $\mathcal{N}_\sigma\Psi = N_\sigma \Psi$ for $\sigma\in \{\uparrow,\downarrow\}$ and  
\[
  \left|\frac{\langle \Psi, \mathcal{H}\Psi\rangle}{L^3} - \frac{1}{L^3}\sum_{\sigma \in \{\uparrow,\downarrow\} }\sum_{k\in\mathcal{B}_F^\sigma}|k|^2 \right|\leq C\rho^{2}.
\]
Then the state $\psi=R^* \Psi$ satisfies 
\begin{equation*} 
  \langle \psi, \mathbb{H}_0 \psi\rangle \leq CL^3\rho^2,  \quad \langle  \psi,\mathbb{Q}_4  \psi\rangle \leq CL^3\rho^2, \quad \langle   \psi, \mathcal{N}  \psi\rangle \leq CL^3\rho^{\frac{7}{6}}. 
\end{equation*}
\end{lemma}

The bounds on the expectation values of $\mathbb{H}_0$ and $\bQ_4$ in Lemma \ref{lem: a priori} are optimal, even for the  ground state of $\cH$. Furthermore, we expect that the bound on $\cN$ cannot be improved under the sole assumption that  $\langle R^\ast\psi, \mathbb{H}_0 R^\ast\psi\rangle \leq CL^3\rho^2$. However, we are able to derive a sharper bound for the number operator in the case of the true ground state. 

\begin{proposition}[Improved a priori bounds]\label{pro: a priori} 
If  $\Psi$ is a ground state of $H_N$ (or any normalized state satisfying $\mathcal{N}_\sigma\Psi = N_\sigma \Psi$ for $\sigma\in \{\uparrow,\downarrow\}$ and having an energy at most $O(L^3 \rho^{7/3})$ above the ground state energy), then $\psi=R^* \Psi$ satisfies 
\begin{equation}
   \langle   \psi, \mathcal{N}   \psi \rangle \leq CL^3\rho^{\frac{4}{3}}, \quad   \langle  \psi, \mathcal{N}_\beta   \psi \rangle \leq CL^3 \rho^{\frac{5-\beta}{3}}
\end{equation}
for all $0\le \beta<1$, with $\mathcal{N}$ and $\mathcal{N}_\beta$ defined in \eqref{eq: def N} and \eqref{eq: def Nalpha >}, respectively. 
\end{proposition}
The proof of Proposition \ref{pro: a priori} will be given in Section \ref{sec:apriori}. A slightly weaker result, with $\mathcal{N}_\beta$ replaced by an operator involving only momenta $|k| > k_F^\sigma + (k_F^\sigma)^{1+\beta}$, can be found in \cite[Theorem 2.2]{Gia3}.


\medskip
To extract the correlation energy, we will make use of the solution of the scattering equation \eqref{eq: zero en scatt eq}. 
To be precise, we use the following periodization  of $\varphi_\infty$.
\begin{definition}[Periodic scattering solution]\label{def:scattering-phi} We define the  function $\varphi: \Lambda\rightarrow \mathbb{R}$ by $\hat{\varphi}(0)= 0$ and 
$$\hat{\varphi}(p) = \mathcal{F}(\varphi_\infty)(p), \quad 0\neq p\in\Lambda^\ast,$$
where $\varphi_\infty: \mathbb{R}^3\rightarrow \mathbb{R}$ is the zero-energy scattering solution in \eqref{eq: zero en scatt eq}. Also, we set \[f = 1-\varphi.\]
\end{definition}


\subsection{Heuristic ideas}\label{sec: main def and heuristics}

We shall now discuss the proof strategy of the present work in relation to the approach  in \cite{GHNS}. In \cite{GHNS}, the energy upper bound \eqref{eq: main thm 1-upper} is obtained by a bosonization technique, where suitable pairs of fermions are interpreted as bosons via the quasi-bosonic (annihilation) operators 
\begin{equation}\label{eq: def bp bpr}
  b_{p,k,\sigma} 
  =\hat{u}_\sigma(p+k)\hat{v}_\sigma(k)\hat{a}_{p+k,\sigma}\hat{a}_{-k,\sigma},
\end{equation}
and the correlation structure is imposed by the quasi-bosonic Bogoliubov transformation 
\begin{equation}\label{eq: def T}
  e^{\cB} = \exp \left(\frac{1}{L^3}\sum_{p,r,r^\prime} \hat\varphi_{r,r^\prime}(p) b_{p,r,\uparrow} b_{-p,r^\prime, \downarrow} -\mathrm{h.c.}\right).  
  \end{equation}
Here the relevant scattering equation defining $\varphi_{r,r'}$ is influenced by the filled Fermi sea and reads 
\begin{equation} \label{BetheGold1}
   ( |p+r|^2 - |r|^2 + |p-r'|^2 - |r'|^2) \hat\phi_{r,r^\prime}(p)  = \mathcal{F}(V_\infty (1- \phi_{r,r^\prime}))(p),
\end{equation}
which is the {\em Bethe--Goldstone} equation, see \cite[Eq. (2.2)]{BG} and \cite[Eq. (2.23)]{GHNS}. Roughly speaking,  we expect that   
\begin{align}\label{eq:upper-bound-transformation}
&  e^{-\cB}  (\bH_0 + \bQ_2 + \bQ_4)   e^{\cB}  \approx \bH_0 + \bQ_4 -\rho_\uparrow\rho_\downarrow L^3 \int_\Lambda V \varphi   \\
&\quad + \frac{(8\pi a)^2}{L^{6}}\sum_{p,r,r^\prime \in\Lambda^*} \left(  \frac{1}{2|p|^2} - \frac{\hat{u}_\uparrow(r+p) \hat{u}_\downarrow(r^\prime - p) \hat{v}_\uparrow(r) \hat{v}_\downarrow(r^\prime)}{|r+p|^2 - |r|^2 + |r^\prime - p|^2 - |r^\prime|^2} \right) + O(\rho \cN) + o(L^3 \rho^{\frac 7 3}) \nn
\end{align}
when the expectation is taken against a state with few excitations, such as the vacuum $\Omega$. Note that the constant $-\rho_\uparrow\rho_\downarrow L^3 \int_\Lambda V \varphi$ in \eqref{eq:upper-bound-transformation} combines with the energy $E_{\mathrm{FFG}}$ of the free Fermi gas state in \eqref{eq: FFG energy} to yield the first two terms on the right-hand side of \eqref{eq: main thm 1-upper}, while the remaining constant reproduces the third-order Huang--Yang correction. 
 
Therefore, $R   e^{\cB} \Omega$ seems a natural trial state. However, for technical reasons, the rigorous implementation in  \cite{GHNS} relies on the replacement { {$e^{\cB}\approx   e^{\cB_1}   e^{\cB_2}$, where the operators $\cB_1$ and $\cB_2$}}  correspond to the following explicit approximation of the Bethe--Goldstone solution in \eqref{BetheGold1}: 
\begin{align}\label{BetheGold1-app}
\hat\phi_{r,r^\prime}(p) \approx  \begin{cases} \hat\varphi(p) ,\quad & |p| \gg \rho^{1/3}\\
\dfrac{8\pi a}{|p+r|^2-|r|^2 + |p-r'|^2 - |r'|^2},\quad & |p| \ll 1/a.
\end{cases}
\end{align}
The two lines describe approximations in overlapping momentum regimes; the first arises from approximating \eqref{BetheGold1} by the zero-energy scattering equation \eqref{eq: zero en scatt eq} for $|p| \gg |r|,|r'|$, while the second is the observation that $ \mathcal{F}(V_\infty(1-\varphi))(p) \approx 8\pi a$ for $|p|\ll 1/a$.

\begin{remark}\label{rmk:8pia-V0}
 Effectively, the transformation $e^{\cB_1}$ replaces $\bQ_4+\bQ_2$ by $\bQ_2\big\vert_{Vf}$ (see \cite[Eq. (2.32)]{GHNS}), and $e^{\cB_2}$ diagonalizes $\bH_0+\bQ_2\big\vert_{Vf}$, where 
\begin{align}\label{eq:Q2|Vf}
   \mathbb{Q}_2|_{Vf} = \frac{1}{2}\sum_{\sigma \neq \sigma^\prime} \int_{\Lambda^2} dxdy\, (Vf)(x-y) a^\ast_\sigma(u_x)a^\ast_{\sigma^\prime}(u_y)a^\ast_{\sigma^\prime}(v_y)a^\ast_{\sigma}(v_x) + \mathrm{h.c.}
\end{align}
Note that $\bQ_2\big\vert_{Vf}$ is defined similarly as $\bQ_2$, but with  $V$ replaced by $Vf$. We emphasize that replacing $V$ by $Vf$ is essential for recovering the correct scattering length $a$ (rather than the  larger quantity $(8\pi)^{-1} \hat V(0)$). Similar renormalization techniques were employed in the bosonic case, see e.g. \cite{BBCS,FS-20, FS-23, HHNST-23,HHST-24}. 
\end{remark}

{ {The above diagonalization method does not seem to be feasible for the lower bound, however.  Establishing a matching lower bound is delicate due to several extra difficulties, most notably:}} 
\begin{itemize}
\item First, implementing the quasi-bosonic Bogoliubov transformations as in \cite{GHNS} leads to several error terms of the form $\rho \cN$, which we can only bound by $O(L^3\rho^{7/3})$ using the a priori estimate in Proposition \ref{pro: a priori}. This is insufficient to resolve the third-order Huang--Yang correction. For the trial state in  \cite{GHNS}, the expectation value of $\rho\cN$ is of order $L^3\rho^{8/3}$, significantly smaller than what we can prove for the true ground state.

\item Second, the cubic term $\mathbb{Q}_3$ cannot be directly treated as an error term for the true ground state. A standard Cauchy--Schwarz estimate (see the proof of Proposition \ref{pro: a priori} for instance) shows that its expectation value is bounded by $O(L^3\rho^{7/3})$, which is insufficient for our purpose. In contrast, for the trial state in  \cite{GHNS}, the expectation value of $\bQ_3$ is exactly zero. 
\end{itemize}

For these reasons, proving the energy lower bound is conceptually more challenging than obtaining the upper bound. Rather than attempting to diagonalize the Hamiltonian, we employ a version of the ``completing the square" method, expressing the correlation Hamiltonian $\cH$ as the sum of non-negative terms of the form $|A|^2=A^*A$, plus the relevant constant term and  controllable error terms.  { {This general method has proved suitable for deriving lower bounds in many instances; for example, the energy of dilute Bose gases in \cite{FS-20,FS-23}, and more recently, the energy of the fermionic mean-field Coulomb system in \cite{ChrHaiNam-24}. In our context, we will further develop this method to effectively address the two difficulties outlined above.  }}

To make the idea transparent, let us first focus on the main terms $\bH_0 + \bQ_2 + \bQ_4$. With 
$$e^{\cB}\hat{a}_{k,\sigma}  e^{-\cB}\approx \hat{a}_{k,\sigma}+T_{\sigma}(k), \quad e^{\cB}  {a}_{\sigma^\prime}(u_y){a}_\sigma(u_x) e^{-\cB} \approx  {a}_{\sigma^\prime}(u_y){a}_\sigma(u_x) + \mathcal{T}_{\sigma,\sigma^\prime}(x,y)$$
we can rewrite \eqref{eq:upper-bound-transformation} as 
\begin{align}\label{eq:upper-bound-transformation-2}
&  \bH_0 + \bQ_2 + \bQ_4  
\approx  \sum_\sigma\sum_k ||k|^2 -(k_F^\sigma)^2|  | \hat{a}_{k,\sigma} + T_{\sigma}(k)|^2  \\
 &\quad + \frac{1}{2}\sum_{\sigma\neq \sigma^\prime}\int_{\Lambda^2} dxdy \, V(x-y)   | {a}_{\sigma^\prime}(u_y){a}_\sigma(u_x) +   \mathcal{T}_{\sigma,\sigma^\prime}(x,y) |^2 + O(\rho \cN) + {\rm constant} . \nn
\end{align}
Here $O(\rho\cN)$ is short for an operator proportional to $\rho_\uparrow\cN_\downarrow+\rho_\downarrow\cN_\uparrow$.
This expression suggests an alternative approach to extract the correlation energy from $\mathcal{H}$ without employing unitary transformations: Find suitable operators $T_\sigma(k)$ and $\mathcal{T}_{\sigma,\sigma'}(x,y)$ such that a completion of two squares yields an identity as in \eqref{eq:upper-bound-transformation-2} with the constant term reproducing the Huang–Yang correction. According to Remark \ref{rmk:8pia-V0}, we expect that the resulting terms correspond to a splitting of the interaction potential into two parts, \[V = Vf + V\varphi\,,\] 
using that $f+\varphi=1$ according to Definition~\ref{def:scattering-phi}. 
Mathematically, \eqref{eq:upper-bound-transformation-2} is easier to use than unitary transformations since it allows independent choices of $T_{\sigma}(k)$ and $\mathcal{T}_{\sigma,\sigma^\prime}(x,y)$. 
Moreover, this formulation is particularly suitable for obtaining lower bounds: all non-negative terms $|A|^2=A^*A$ can be dropped for a lower bound (but not for an upper bound). 

In the present paper, we use \eqref{eq:upper-bound-transformation-2} with the explicit choices of $T_{\sigma}(k)$ and $\mathcal{T}_{\sigma,\sigma'}(x,y)$ given in Definition \ref{def:ren-Q2} below. Before going into details, let us explain the motivation behind these choices. We choose $\mathcal{T}_{\sigma,\sigma'}(x,y)$  such that 
\begin{align}
 \frac{1}{2}\sum_{\sigma\neq \sigma^\prime}\int dxdy\, V(x-y)\mathcal{T}_{\sigma,\sigma^\prime}^\ast (x,y)a_{\sigma^\prime}(u_y)a_\sigma(u_x) + {\rm h.c.} = \mathbb{Q}_2\big\vert_{V\varphi} . \label{eq:Q4-square-dec-1}
\end{align}
Here and in the following, $\bQ_j|_{Vg}$ is defined similarly as $\bQ_j$ but with $V$ replaced by $Vg$ (as in \eqref{eq:Q2|Vf} above).
 This allows 
to complete the square as 
\begin{align}\label{eq:tQ4}
  \widetilde{\bQ}_4 &:= \frac{1}{2}\sum_{\sigma\neq \sigma^\prime}\int_{\Lambda^2} dxdy \, V(x-y)   | {a}_{\sigma^\prime}(u_y){a}_\sigma(u_x) +   \mathcal{T}_{\sigma,\sigma^\prime}(x,y) |^2  \\
  &\approx \bQ_4 + \mathbb{Q}_2\big\vert_{V\varphi} + \rho_\uparrow\rho_\downarrow L^3\int_{\Lambda} V\varphi^2 
  -\sum_{\sigma\neq\sigma^\prime} \rho_\sigma  \left( \int_{\Lambda} V\varphi^2 \right)  \sum_{r} \hat{v}_{\sigma'}(r) \hat{a}_{r,\sigma^\prime}^\ast \hat{a}_{r,\sigma^\prime} . \nn 
\end{align}
 In \eqref{eq:tQ4}, the constant and the quadratic contribution are extracted by writing $|\mathcal{T}_{\sigma, \sigma'}(x,y)|^2$ in normal order.  Similarly, we choose $T_{\sigma}(r)$ such that 
 \begin{equation}\label{eq: T*a-intro}
 \sum_\sigma\sum_r ||r|^2 - (k_F^\sigma)^2|T^\ast_{\sigma}(r)\hat{a}_{r,\sigma} + \mathrm{h.c.}  \approx \mathbb{Q}_2\vert_{Vf},
\end{equation}
which then yields 
\begin{align}  \label{eq:tH0}
  \widetilde{\bH}_0 &:= \sum_\sigma\sum_r ||r|^2 -(k_F^\sigma)^2|  | \hat{a}_{r,\sigma} + T_{\sigma}(r)|^2 \\
  &\approx \bH_0 + \mathbb{Q}_2\vert_{Vf} - \sum_{\sigma\neq \sigma^\prime} \rho_\sigma \left( \int_{\Lambda} 2|\nabla \varphi|^2 \right) \sum_{r} \hat{v}_{\sigma'}(r) \hat{a}_{r,\sigma^\prime}^\ast \hat{a}_{r,\sigma^\prime} 
 \nn \\
&\quad   + \frac{1}{L^6}\sum_{p,r,r^\prime}\frac{(2|p|^2 \hat{\varphi}(p))^2\hat{u}_{\uparrow}(r+p)\hat{u}_{\downarrow}(r^\prime - p)\hat{v}_\uparrow(r)\hat{v}_{\downarrow}(r^\prime)}{|r+p|^2 - |r|^2 + |r^\prime - p|^2 - |r^\prime|^2 }.\nn
\end{align}
Again the constant and the quadratic contribution  arise from writing $|T_{\sigma}(r)|^2$ in normal order. 
Combining \eqref{eq:tQ4} and \eqref{eq:tH0} we obtain
\begin{align}\label{eq:completion-square-simple}
\bH_0 +  \bQ_2 + \bQ_4 & \approx   \widetilde{\bH}_0 +   \widetilde{\bQ}_4  -  \rho_\uparrow\rho_\downarrow L^3  \int_{\Lambda} V\varphi^2  \\
&\quad - \frac{1}{L^6}\sum_{p,r,r^\prime}\frac{(2|p|^2 \hat{\varphi}(p))^2 \hat{u}_{\uparrow}(r+p)\hat{u}_{\downarrow}(r^\prime - p)\hat{v}_\uparrow(r)\hat{v}_{\downarrow}(r^\prime) }{|r+p|^2 - |r|^2 + |r^\prime - p|^2 - |r^\prime|^2}  + O(\rho \cN). \nn
\end{align}
We will prove that the constant in \eqref{eq:completion-square-simple}, when combined with the free Fermi gas energy in \eqref{eq: FFG energy}, exactly reproduces the Huang--Yang formula; see \eqref{eq:corr-energy-1} below. 

Unfortunately, the expression \eqref{eq:completion-square-simple} is insufficient to conclude the lower bound in Theorem \ref{thm: main lw bd}, since we can only control $\rho \mathcal{N}$ by $O(L^3\rho^{7/3})$ via Proposition \ref{pro: a priori}. Moreover, it does not explain how to handle the  term $\mathbb{Q}_3$. Including $\mathbb{Q}_3$ in the analysis, and along the way canceling the term  of order $\rho \mathcal{N}$, are  in fact the main challenges in our proof of the lower bound.

We resolve these  issues by a modification of the above argument, which will allow to obtain the following refined version of \eqref{eq:completion-square-simple}:
\begin{align}\label{eq:completion-square-refine}
&\bH_0 + \bQ_2 + \bQ_3 + \bQ_4  \approx   \sum_{\sigma}\sum_{r} ||r|^2 - (k_F^\sigma)^2| \left| \hat{a}_{r,\sigma} + T_{\sigma}(r) + S_{\sigma}(r)\right|^2 \\
&\qquad + \frac{1}{2}\sum_{\sigma\neq\sigma^\prime}\int_{\Lambda^2} dxdy\, V(x-y) \left| a_{\sigma^\prime}(u_y)a_\sigma(u_x) + \mathcal{T}_{\sigma,\sigma^\prime}(x,y) + \mathcal{S}_{\sigma,\sigma^\prime}(x,y)\right|^2  \nn\\
& \qquad - \rho_\uparrow \rho_\downarrow  L^3  \int_{\Lambda} V\varphi^2   - \frac{1}{L^6}\sum_{p,r,r^\prime}\frac{(2|p|^2 \hat{\varphi}(p))^2 \hat{u}_{\uparrow}(r+p)\hat{u}_{\downarrow}(r^\prime - p)\hat{v}_\uparrow(r)\hat{v}_{\downarrow}(r^\prime) }{|r+p|^2 - |r|^2 + |r^\prime - p|^2 - |r^\prime|^2} . \nn
\end{align}
Here the new renormalization terms $\mathcal{S}_{\sigma,\sigma^\prime}(x,y)$ and $S_{\sigma}(r)$, which are defined in Definition \ref{def:ren-Q3} below, serve to  take into account the contribution of $\bQ_3$. In fact, we choose them in such a way that 
\begin{align}
\sum_{\sigma\neq \sigma^\prime}\int dxdy\, V (x-y)\mathcal{S}_{\sigma, \sigma^\prime}^\ast(x,y) a_{\sigma^\prime}(u_y)a_\sigma(u_x) + {\rm h.c.} = \mathbb{Q}_3\big\vert_{V\varphi}
\label{eq:Q4-square-dec-2}
\end{align}
and 
 \begin{equation}\label{eq: S*a-intro}
 \sum_\sigma\sum_r ||r|^2 - (k_F^\sigma)^2|S^\ast_{\sigma}(r)\hat{a}_{r,\sigma} + \mathrm{h.c.} \approx \mathbb{Q}_3\vert_{Vf}. 
\end{equation}
A crucial point in our analysis is that, with the additional terms $\mathcal{S}_{\sigma,\sigma^\prime}(x,y)$ and $S_{\sigma}(r)$, we achieve the cancellation of all the quadratic terms of order $\rho \mathcal{N}$. As we will see, when we expand the first square on the right-hand side of \eqref{eq:completion-square-refine}, the terms involving $|S_{\sigma}(r)|^2$ lead to a contribution of the order $\rho \mathcal{N}$ that exactly cancels the corresponding contribution arising from $|T_{\sigma}(r)|^2$. The mixed terms involving $T^*_{\sigma}(r) S_{\sigma}(r)$ are negligible. Similarly, for the second square in \eqref{eq:completion-square-refine}, the $\rho \mathcal{N}$-type contribution from the terms involving $| \mathcal{S}_{\sigma,\sigma^\prime}(x,y)|^2$ is exactly canceled by a similar term involving $| \mathcal{T}_{\sigma,\sigma^\prime}(x,y)|^2$, and the mixed terms involving $\mathcal{T}^*_{\sigma,\sigma^\prime}(x,y) \mathcal{S}_{\sigma,\sigma^\prime}(x,y)$ are negligible. In this cancellation of the $\rho \mathcal{N}$ contributions, the particle-hole relation \eqref{eq:particle-hole} plays a crucial role.

From \eqref{eq:completion-square-refine}, we simply drop the two square terms to obtain a lower bound, and thereby conclude the Huang--Yang formula. The precise formulation of the above estimates will be given in the next subsection.


\subsection{Main propositions and the proof of Theorem \ref{thm: main lw bd}}\label{sec:proof-outline}

Now we are ready to present the main estimates that lead to the conclusion of Theorem \ref{thm: main lw bd}. We will introduce the renormalization terms $T_{\sigma}(r)$, $S_{\sigma}(r)$, $\mathcal{T}_{\sigma,\sigma^\prime}(x,y)$, and $\mathcal{S}_{\sigma,\sigma^\prime}(x,y)$, and provide a precise version of \eqref{eq:completion-square-refine}. First, we make the following explicit choices of $T_{\sigma}(r)$ and $\mathcal{T}_{\sigma,\sigma^\prime}(x,y)$ to obtain \eqref{eq:Q4-square-dec-1} and \eqref{eq: T*a-intro}. 

\begin{definition}[Renormalization terms for $\bQ_2$]\label{def:ren-Q2} For $\sigma \ne \sigma'\in\{\uparrow, \downarrow\}$, we define 
\begin{equation*}  
  \mathcal{T}_{\sigma,\sigma^\prime}^\ast(x,y) = \varphi(x-y)a_{\sigma}(v_x)a_{\sigma^\prime}(v_y) 
\end{equation*}
with ${v}_\sigma$, ${u}_\sigma$ in \eqref{eq: def u,v} and  $\varphi$ in Definition \ref{def:scattering-phi}. Moreover, for a parameter $\delta>1/3$, we define  
\begin{equation} \label{eq: def Tk}
T^\ast_{\sigma} (r)= \frac{1}{L^3}\sum_{p,r^\prime \in  \Lambda^*}  \Big( \widehat{\omega}^\varepsilon_{-r,r^\prime}(p)\hat{u}_\sigma(p-r)\hat{v}_\sigma(r) - \widehat{\omega}^\varepsilon_{r-p,r^\prime}(p)\hat{v}_\sigma(r-p)\hat{u}_\sigma(r)\Big) \hat{a}_{p-r, \sigma} b_{-p,r^\prime,\sigma^\prime}
\end{equation}
with $\sigma'\neq \sigma$,  $b_{-p,r',\sigma'}$ defined in \eqref{eq: def bp bpr}  and 
\begin{equation}\label{eq: def omega eps}
  \widehat{\omega}^\varepsilon_{r,r^\prime}(p) = \frac{2|p|^2 \hat{\varphi}(p)}{\lambda_{p,r}+ \lambda_{-p,r^\prime} + 2\varepsilon}, \qquad \lambda_{p,r}=|r+p|^2 -|r|^2,  \qquad \varepsilon = \rho^{\frac{2}{3} +\delta}. 
\end{equation}
\end{definition}

The kernel $\widehat{\omega}^\varepsilon_{r,r^\prime}(p)$ in \eqref{eq: def omega eps} serves as a replacement of the Bethe--Goldstone solution \eqref{BetheGold1}. In fact, it is straightforward to verify that $\widehat{\omega}^\varepsilon_{r,r^\prime}(p)$ satisfies the approximation in \eqref{BetheGold1-app}.

Next, we will choose $\mathcal{S}^*_{\sigma,\sigma^\prime}(x,y)$ and $S^*_{\sigma}(r)$ in order to fulfill \eqref{eq:Q4-square-dec-2} and \eqref{eq: S*a-intro}. For technical reasons, we will have to introduce some momentum cutoff. We fix a radial smooth function $\widehat \chi:\R^3 \to [0,1]$ such that 
$
\1_{\{|p|\le 1\}} \le \widehat \chi (p) \le \1_{\{|p|\le 5/4\}}. 
$ 
In the following, we pick two parameters \[0<\gamma<1/6<\alpha<1/3\,,\] which will be optimized  at the end,  and define the  functions $\varphi^>,   \varphi^<, u^{<\alpha}_{\sigma}: \Lambda\to \mathbb{C}$ by  
\begin{equation}\label{eq: def phi><}
 \hat{\varphi}^< (p)= \widehat \varphi(p) \widehat{\chi}_{<}(p),\quad \hat{ \varphi}^>(p)= \widehat \varphi(p) \widehat{\chi}_{>} (p), \quad \hat u^{<\alpha}_{\sigma} (p) = \hat  u_\sigma (p) \1_{\{|p| < \rho^{1/3-\alpha}\}} 
  \end{equation}
for $p\in \Lambda^*$, where
\begin{equation}\label{eq: def chi< chi>}
\widehat \chi_{<}(p) =\widehat \chi(p/ (4 \rho^{1/3-\gamma})), \quad  \widehat \chi_{>}(p) = 1-  \widehat \chi_{<}(p).
\end{equation}

\begin{definition}[Renormalization terms for $\bQ_3$]\label{def:ren-Q3} For $\sigma \ne \sigma'\in\{\uparrow, \downarrow\}$, we define 
\begin{equation*} 
  \mathcal{S}^\ast_{\sigma,\sigma^\prime}(x,y) =  \varphi(x-y)\left(a^\ast_{\sigma}(u_x)a_{\sigma^\prime}(v_y) - a^\ast_{\sigma^\prime}(u_y)a_\sigma(v_x)\right)
\end{equation*}
with $\hat{v}_\sigma$, $\hat{u}_\sigma$ as in \eqref{eq: def u,v} and  $\varphi$ as in Definition \ref{def:scattering-phi}. Moreover, for two parameters $0<\gamma<1/6<\alpha<1/3$, we define 
\begin{align}\nn
&  S^\ast_{\sigma}(r)= 
    \frac{1}{L^3}\sum_{p}\hat{\varphi}^>(p) (D^{<\alpha}_{p,\sigma^\prime})^\ast  \hat{a}_{p-r,\sigma}\left(\hat{v}_\sigma(r-p)\hat{u}_\sigma(r) -\hat{u}_\sigma(p-r)\hat{v}_\sigma(r) \right)\nonumber
    \\
      &\quad - \frac{1}{L^3}\sum_{p}\hat{\varphi}^>(p)\left(\hat{u}_\sigma(r)\hat{{u}}^{<\alpha}_\sigma(r+p)\hat{a}^\ast_{r+p,\sigma}b_{p,\sigma^\prime} -\hat{{u}}^{<\alpha}_\sigma(r) \hat{u}_\sigma(r-p)\hat{a}_{r-p,\sigma}^\ast b_{p,\sigma^\prime}^\ast\right) 
      \label{eq: def Sk}
\end{align}
with $\sigma'\neq \sigma$, $\hat{\varphi}^>$, $\hat{{u}}^{<\alpha}_\sigma$ defined in \eqref{eq: def phi><} and 
\begin{equation}\label{def:D}
D_{p,\sigma^\prime}^{<\alpha} = \sum_{r^\prime}\hat{a}_{r^\prime - p,\sigma^\prime}^\ast\hat{a}_{r^\prime,\sigma^\prime}\hat{{u}}^{<\alpha}_{\sigma^\prime}(r^\prime)\hat{u}_{\sigma^\prime}(r^\prime-p) .
\end{equation}
\end{definition}

Our proof of Theorem \ref{thm: main lw bd} is based on the following two key propositions, which correspond to the completion of the two squares in \eqref{eq:completion-square-refine}. Recall that $f=1-\varphi$ with the scattering solution $\varphi$ in Definition \ref{def:scattering-phi}, and $\bQ_j|_{Vg}$ is defined similarly as $\bQ_j$ but with $V$ replaced by $Vg$ (compare with \eqref{eq:Q2|Vf}).

\begin{proposition}[Completion of the square for $V\varphi$]\label{pro: completion square Vphi} Let $V_\infty$ be as in Assumption \ref{asu: potential V}. With $\mathcal{T}^\ast_{\sigma, \sigma^\prime}(x,y)$ and $\mathcal{S}^\ast_{\sigma, \sigma^\prime}(x,y)$ defined in Definitions \ref{def:ren-Q2} and  \ref{def:ren-Q3}, 
\begin{align}\label{eq: first id Q4}
  &\mathbb{Q}_4 + (\mathbb{Q}_2 +  \mathbb{Q}_3)\big\vert_{V\varphi} =  -L^3\rho_\uparrow\rho_\downarrow \int_{\Lambda} V \varphi^2 \\
  &+\frac{1}{2}\sum_{\sigma\neq\sigma^\prime}\int_{\Lambda^2} dxdy\, V(x-y) \big| a_{\sigma^\prime}(u_y)a_\sigma(u_x) + \mathcal{T}_{\sigma,\sigma^\prime}(x,y) + \mathcal{S}_{\sigma,\sigma^\prime}(x,y)\big|^2 \nn
 + \mathcal{E}_{V\varphi},
\end{align}
where 
\begin{equation}\label{evbound}
  |\langle \psi,\mathcal{E}_{V\varphi}\psi\rangle| \leq  C\rho^{{1}+\frac{\beta}{6}}\langle \psi,\mathcal{N}\psi\rangle + C\rho\|\mathcal{N}^{\frac{1}{2}}\psi\|\|\mathcal{N}^{\frac{1}{2}}_\beta\psi\|,
\end{equation}
for all $\psi\in\mathcal{F}_{\mathrm{f}}$ satisfying the particle-hole relation \eqref{eq:particle-hole}, and for any $0\le \beta<1$, with  $\mathcal{N}_\beta$  given in \eqref{eq: def Nalpha >}.
\end{proposition}

\begin{proposition}[Completion of the square for $Vf$]\label{pro: completion square Vf} Let $0<\gamma <1/6< \alpha< 1/3<\delta$ such that  
\begin{equation}\label{eq: mix cond alpha gamma}
  \frac{1}{3} + 2\gamma - 3\alpha >0, \qquad \frac{1}{3} + 4\gamma - 4\alpha >0.
\end{equation}
With $T_\sigma^\ast(r)$ and $S_\sigma^\ast(r)$  defined in Definitions \ref{def:ren-Q2} and \ref{def:ren-Q3}, 
\begin{align*}
  &  \mathbb{H}_0 + (\mathbb{Q}_2 +  \mathbb{Q}_3)\big\vert_{Vf} 
=  - \frac{1}{L^6}\sum_{p,r,r^\prime \in \Lambda^*}\frac{(2|p|^2 \hat{\varphi}(p))^2 \hat{u}_{\uparrow}(r+p)\hat{u}_{\downarrow}(r^\prime - p)\hat{v}_\uparrow(r)\hat{v}_{\downarrow}(r^\prime)  }{\lambda_{p,r}+ \lambda_{-p,r^\prime} + 2\varepsilon} \\
  &\qquad\qquad\qquad\qquad\qquad\qquad + \sum_{\sigma\in \{\uparrow, \downarrow\}}\sum_{r\in \Lambda^*} ||r|^2 - (k_F^\sigma)^2| \left| \hat{a}_{r,\sigma} + T_{\sigma}(r) + S_{\sigma}(r)\right|^2
- \mathcal{E}_{Vf},
\end{align*}
where 
\begin{align*}
 &  \langle \psi, \mathcal{E}_{Vf} \psi\rangle \leq   C\rho^{\frac{1}{6} + \alpha}\|\mathbb{Q}_4^{\frac{1}{2}} \psi\|\|\mathbb{H}_0^{\frac{1}{2}}  \psi\|  + \mathfrak{e}_L L^{3/2} \|\bQ_4^{1/2}\psi\| + C\rho^{1-\frac {3\gamma}2}\|\mathcal{N}^{\frac{1}{2}}_\beta\psi\|\|\mathcal{N}^{\frac{1}{2}}\psi\|
 \\  
 &\quad+  C\bigg( \rho^{1+\frac{\beta}{6}} +  \rho^{\frac{7}{6} + \frac{\gamma}{2} -\alpha} +  \rho^{\frac{4}{3}- 2\gamma - \kappa} +  \rho^{\frac{4}{3} + 2\gamma - 3\alpha}
\\
& \qquad \qquad \qquad + \rho^{\frac{4}{3} + \frac{3}{4}\gamma -\frac{3}{2}\alpha -\kappa} +\rho^{\frac{4}{3} - \alpha - \kappa} +  \rho^{1+\frac{4}{5}\gamma - \frac{3}{10}\alpha} \bigg) \langle \psi, \mathcal{N}\psi\rangle
\\
&\quad+ C\left(\rho^{\frac{1}{3} + \frac{4}{5}\gamma - \frac{3}{10}\alpha} + \rho^{\frac{2}{3} + 6\gamma - 5\alpha} + \rho^{\frac{2}{3} + \frac{11}{2}\gamma - \frac{7}{2}\alpha} + \rho^{\frac{5}{6} + \frac{13}{2}\gamma - 6\alpha}\right)\langle \psi, \mathbb{H}_0 \psi\rangle 
\\
&\quad + C \left( \rho^{\frac{3}{2} - \gamma - \kappa}  + \rho^{\frac{3}{2} + \frac{21}{10}\gamma - \frac{21}{10}\alpha - \kappa}   \right) L^{\frac{3}{2}} \Big( \|\mathbb{H}_0^{\frac{1}{2}}\psi\|  + \rho^{\frac 1 3}\|\mathcal{N}^{\frac{1}{2}}\psi\| \Big)
\\
& \quad+ C\bigg(\rho^{1+\frac{\gamma}{2} -\kappa} + \rho^{1+\frac{7}{4}\gamma - \frac{3}{2}\alpha -\kappa} + \rho^{1+\frac{9}{2}\gamma - \frac{7}{2}\alpha -\kappa} + \rho^{1-\alpha - \kappa})\|\mathbb{H}_0^{\frac{1}{2}}\psi\|\|\mathcal{N}^{\frac{1}{2}}\psi\|
\end{align*}
for all $\psi\in\mathcal{F}_{\mathrm{f}}$ with $\|\psi\|=1$ and satisfying the particle-hole relation \eqref{eq:particle-hole}, and for any $0\le \beta<1$, with  $\mathcal{N}_\beta$  given in \eqref{eq: def Nalpha >}. Here, $\kappa > 0$ can be chosen arbitrarily small, and $\mathfrak{e}_L \to 0$ as $L \to \infty$.
\end{proposition}

While the proof of Proposition \ref{pro: completion square Vphi} is relatively short since it relies on the exact identities \eqref{eq:Q4-square-dec-1} and \eqref{eq:Q4-square-dec-2}, the proof of Proposition \ref{pro: completion square Vf} is the most difficult part of the present work. The proof and further details will be given in Sections~\ref{sec: completion square Vphi} and~\ref{sec: completion square Vf}.

The a priori estimates in Proposition \ref{pro: a priori} will be obtained by suitably adapting the proofs of Propositions \ref{pro: completion square Vphi} and \ref{pro: completion square Vf}. In fact, for this purpose, we will not directly use the results in Propositions \ref{pro: completion square Vphi} and \ref{pro: completion square Vf}, since, for instance, the error term $\rho^{1-\frac{3\gamma}{2}} \|\mathcal{N}^{1/2}_\beta \psi\| \|\mathcal{N}^{1/2} \psi\|$ in Proposition \ref{pro: completion square Vf} is not obviously bounded by $O(\rho \mathcal{N})$. Instead, we will use the simpler square completion in \eqref{eq:completion-square-simple}, which is sufficient to achieve a precision up to $O(\rho \mathcal{N})$.

We conclude this section by explaining how our main result in Theorem~\ref{thm: main lw bd} readily follows from the above propositions. 
\begin{proof}[Proof of Theorem \ref{thm: main lw bd}]  Let $\Psi=R \psi$ be a ground state for $\cH$. From Lemma \ref{lem: a priori} and Proposition \ref{pro: a priori}, we have the a-priori estimates 
\begin{equation*}  
  \langle \psi, \mathbb{H}_0 \psi\rangle \leq CL^3\rho^2,  \,\, \langle  \psi,\mathbb{Q}_4  \psi\rangle \leq CL^3\rho^2, \,\, \langle   \psi, \mathcal{N}  \psi\rangle \leq CL^3\rho^{\frac{4}{3}}, \,\,   \langle   \psi, \mathcal{N}_\beta  \psi\rangle \leq CL^3\rho^{\frac{5-\beta}{3}}
\end{equation*}
for any $0\le \beta <1$. Consequently, the error term in Proposition \ref{pro: fermionic transf} can be bounded as 
\begin{equation}\label{eq: eps-1-final}
  |\langle \psi, \mathcal{E}_{\rm corr} \psi\rangle| \leq CL^3 \rho^{\frac 7 3+ \frac 1 {12}}
  \end{equation}
  where we applied \eqref{eq: eps-1} with $\beta=1/2$.  Similarly, using Proposition \ref{pro: completion square Vphi} 
    and dropping the non-negative contribution from the second term on  the right-hand side of \eqref{eq: first id Q4} we find that 
\begin{align}  \left\langle \psi,  \left[\mathbb{Q}_4 + (\mathbb{Q}_2 +  \mathbb{Q}_3)\big\vert_{V\varphi}\right] \psi\right\rangle \ge -L^3\rho_\uparrow\rho_\downarrow \int_{\Lambda} V\varphi^2  - C  L^3  \rho^{\frac 7 3+ \frac 1 {12}} . 
   \label{eq:conclusion-square-2}
\end{align} 
Next, we apply Proposition \ref{pro: completion square Vf} with the parameters 
\begin{align}\label{eq:choice-parameters}
\gamma = \frac{1}{10}, \quad \alpha = \frac{1}{6}+ \frac{1}{120}, \quad \beta = \frac{1-9\gamma}{2}=\frac{1}{20}, \quad \delta=1
\end{align}
Here, $\beta= (1-9\gamma)/2$ is chosen in order to balance the error terms 
$$
\rho^{1-\frac {3\gamma}2}\|\mathcal{N}^{\frac{1}{2}}_\beta\psi\|\|\mathcal{N}^{\frac{1}{2}}\psi\| + \rho^{1+\frac \beta 6} \langle \psi, \cN \psi\rangle
 \le C L^3 \rho^{\frac 7 3} ( \rho^{\frac {1-\beta}{6}- \frac{3\gamma}{2}}+ \rho^{\frac \beta 6}) \le C L^3 \rho^{\frac 7 3 + \frac 1 {12} - \frac{3 \gamma}{4}},
$$
and  $\gamma$ and $\alpha$ are chosen to optimize this and the most significant error terms among the remaining ones, namely 
\begin{align*}
&\rho^{\frac{1}{6} + \alpha}\|\mathbb{Q}_4^{\frac{1}{2}} \psi\|\|\mathbb{H}_0^{\frac{1}{2}}  \psi\| + \rho^{\frac{4}{3} +2\gamma-3\alpha} \langle \psi, \cN \psi\rangle  \le CL^3 \rho^{\frac 7 3}\Big( \rho^{\alpha -\frac  1 6 } + \rho^{\frac 1 3 + 2 \gamma-3\alpha }\Big).
\end{align*}
Moreover, we chose $\kappa>0$ sufficiently small so that all terms containing a factor $\rho^{-\kappa}$ are of lower order than the preceding terms.
Dropping also the non-negative contribution from the term involving $|\hat{a}_{r,\sigma} + T_{\sigma}(r) + S_{\sigma}(r)|^2$, we conclude from Proposition \ref{pro: completion square Vf} that 
  \begin{align}
 &\left\langle  \psi, \Big[ \mathbb{H}_0 + (\mathbb{Q}_2 +  \mathbb{Q}_3)\big\vert_{Vf}  \Big]  \psi \right\rangle\nn
  \\
 & \geq  - \frac{1}{L^6}\sum_{p,r,r^\prime}\frac{(2|p|^2 \hat{\varphi}(p))^2 \hat{u}_{\uparrow}(r+p)\hat{u}_{\downarrow}(r^\prime - p)\hat{v}_\uparrow(r)\hat{v}_{\downarrow}(r^\prime) }{\lambda_{p,r}+ \lambda_{-p,r^\prime} + 2\varepsilon}  - C L^3\rho^{\frac 7 3+ \frac{1}{120}}.
 \label{eq:conclusion-square-1}
\end{align}
Combining \eqref{eq: eps-1-final}, \eqref{eq:conclusion-square-2} and \eqref{eq:conclusion-square-1} with \eqref{eq: particle hole lw bd}, we find that 
\begin{align}    \label{eq:conclusion-square-3}& \left\langle \Psi,  \mathcal{H} \Psi\right\rangle \ge  E_{\mathrm{FFG}}  - L^3 \rho_\uparrow\rho_\downarrow \int_{\Lambda} V\varphi^2  \\
&\quad - \frac{1}{L^6}\sum_{p,r,r^\prime}\frac{(2|p|^2 \hat{\varphi}(p))^2 \hat{u}_{\uparrow}(r+p)\hat{u}_{\downarrow}(r^\prime - p)\hat{v}_\uparrow(r)\hat{v}_{\downarrow}(r^\prime) }{\lambda_{p,r}+ \lambda_{-p,r^\prime} + 2\varepsilon} - C L^3 \rho^{\frac 7 3+ \frac{1}{120}}. \nn
\end{align}

It remains to show that the expression on the right-hand side of \eqref{eq:conclusion-square-3} agrees with the Huang--Yang formula  \eqref{eq: main thm 1}. Since $E_{\mathrm{FFG}}$ satisfies  \eqref{eq: FFG energy}, the desired bound 
follows if we can show that 
\begin{align}\nn
& \rho_\uparrow\rho_\downarrow  \int_{\R^3} V_\infty (1-\varphi_\infty^2) \\ &  - \frac 1{(2\pi)^{9}}\int_{\R^9} dp dr dr' \frac{(2|p|^2 \widehat{\varphi}_\infty(p))^2\hat{u}_{\uparrow}(r+p)\hat{u}_{\downarrow}(r^\prime - p)\hat{v}_\uparrow(r)\hat{v}_{\downarrow}(r^\prime)}{\lambda_{p,r}+ \lambda_{-p,r^\prime} + 2\varepsilon}  \nn\\
&\ge 8\pi a \rho_\uparrow\rho_\downarrow + a^2 \rho_\uparrow^{\frac{7}{3}} F\left(\frac{\rho_\downarrow}{\rho_\uparrow}\right) +  O(\rho^{7/3+\gamma}) \label{eq:corr-energy-1}
\end{align}
with the above choice of $\gamma$ and $\delta$. 
For the second term on the left-hand side of \eqref{eq:corr-energy-1}, we split the integration over $p$ into $\{|p|\le  \rho^{1/3-\gamma}\}$ and $\{|p|\ge \rho^{1/3-\gamma}\}$. In the following we will use 
\begin{align}\label{eq:8pia-Vf}
0\le 8\pi a - 2 |p|^2 |\widehat \varphi_\infty(p)|   =  \mathcal{F}(V_\infty f_\infty) (0)-|\mathcal{F}(V_\infty f_\infty) (p)| \le C|p|^2
\end{align}
with $f_\infty=1-\varphi_\infty$. On the domain $\{|p|\le \rho^{1/3-\gamma}\}$, using the upper bound $2 |p|^2 |\widehat \varphi_\infty(p)| \le 8\pi a $ in \eqref{eq:8pia-Vf} and the computations in \cite[Lemma 6.1 and Appendix B]{GHNS} we have
\begin{align}\label{eq:corr-energy-1a}
&- \int_{|p|\le \rho^{1/3-\gamma}} dp dr dr'  \frac{(2|p|^2 \widehat{\varphi}_\infty(p))^2 \hat{u}_{\uparrow}(r+p)\hat{u}_{\downarrow}(r^\prime - p)\hat{v}_\uparrow(r)\hat{v}_{\downarrow}(r^\prime)}{\lambda_{p,r}+ \lambda_{-p,r^\prime} + 2\varepsilon}\nn\\
&\ge -  (8\pi a)^2\int_{|p|\le \rho^{1/3-\gamma}} dp dr dr'    \frac{\hat{u}_{\uparrow}(r+p)\hat{u}_{\downarrow}(r^\prime - p)\hat{v}_\uparrow(r)\hat{v}_{\downarrow}(r^\prime)}{\lambda_{p,r}+ \lambda_{-p,r^\prime} + 2\varepsilon}  \nn\\
&=   (8\pi a)^2 \int_{|p|\le \rho^{1/3-\gamma}} dp dr dr'   \Big[   \frac{1}{2|p|^2} - \frac{\hat{u}_{\uparrow}(r+p)\hat{u}_{\downarrow}(r^\prime - p)}{\lambda_{p,r}+ \lambda_{-p,r^\prime} + 2\varepsilon} \Big] \hat{v}_\uparrow(r)\hat{v}_{\downarrow}(r^\prime) \nn\\
&\qquad\qquad -  (8\pi a)^2 (2\pi)^6 \rho_\uparrow\rho_\downarrow  \int_{|p|\le  \rho^{1/3-\gamma}} dp \frac{1}{2|p|^2}\nn\\
&= (2\pi)^9 a^2 \rho_\uparrow^{\frac{7}{3}} F\left(\frac{\rho_\downarrow}{\rho_\uparrow}\right) -  (8\pi a)^2 (2\pi)^6 \rho_\uparrow\rho_\downarrow \int_{|p|\le  \rho^{1/3-\gamma}}  dp\frac{1}{2|p|^2} + O(\rho^{7/3+\gamma}) .  
\end{align}
For $|p|\ge \rho^{1/3-\gamma}$ and both $|r|\leq C\rho^{1/3}$ and  $|r'|\leq C\rho^{1/3}$ we have 
$$
\frac{1}{\lambda_{p,r}+ \lambda_{-p,r^\prime} + 2\varepsilon} \le \frac{1}{2 |p|^2 + 2 p \cdot (r-r' )} \le \frac{1}{2 |p|^2} \left[ 1 - \frac{ p \cdot (r-r' )}{|p|^2} + \frac{C |r-r'|^2}{|p|^2} \right] . 
$$
When we integrate against $\hat{v}_\uparrow(r)\hat{v}_{\downarrow}(r^\prime) d r d r'$, the term involving $p \cdot (r-r' )$ vanishes. Using $|r-r'|^2\le C\rho^{\frac 2 3}$ and also \eqref{eq:8pia-Vf} we find that 
\begin{align}\label{eq:corr-energy-1b}
& - \int_{|p|\ge \rho^{1/3-\gamma}} dp dr dr'  \frac{(2|p|^2 \widehat{\varphi}_\infty(p))^2\hat{u}_{\uparrow}(r+p)\hat{u}_{\downarrow}(r^\prime - p)\hat{v}_\uparrow(r)\hat{v}_{\downarrow}(r^\prime)}{\lambda_{p,r}+ \lambda_{-p,r^\prime} + 2\varepsilon}\nn\\
&\ge - 2 (2\pi)^6 \rho_\uparrow\rho_\downarrow \int_{|p|\ge \rho^{1/3-\gamma}} dp \, |p|^2 \widehat{\varphi}_\infty(p)^2 - C \rho^{2+2/3}  \int_{|p|\ge \rho^{1/3-\gamma}} dp\, \widehat{\varphi}_\infty(p)^2 \nn\\
&\ge  - 2 (2\pi)^6 \rho_\uparrow\rho_\downarrow \int_{\R^3} dp  \, |p|^2 \widehat{\varphi}_\infty(p)^2 + (2\pi)^6 \rho_\uparrow\rho_\downarrow \int_{|p|\le \rho^{1/3-\gamma}} dp  \frac{1}{2|p|^2 }\left( 8\pi a  -  C |p|^2 \right)^2  \nn\\
&\qquad \qquad  - C \rho^{2+2/3}  \int_{|p|\ge \rho^{1/3-\gamma}} dp \left( \frac{4\pi a}{|p|^2}\right)^2 \nn\\
&=  - 2  (2\pi)^9 \rho_\uparrow\rho_\downarrow \int_{\R^3}  |\nabla \varphi_\infty|^2+ (8\pi a)^2 (2\pi)^6 \rho_\uparrow\rho_\downarrow \int_{|p|\le \rho^{1/3-\gamma}} dp  \frac{1}{2|p|^2 }   +  O(\rho^{7/3+\gamma}). 
\end{align}
The second term on the right-hand side of \eqref{eq:corr-energy-1b} cancels the second term on the right-hand side of \eqref{eq:corr-energy-1a}. Moreover, using $1-\varphi_\infty^2  = 2f_\infty - f_\infty^2$ and the scattering equation for $\varphi_\infty$, we get 
\begin{align*}  
  \int_{\R^3} V_\infty (1-\varphi_\infty^2)  -  \int_{\R^3} 2 |\nabla {\varphi_\infty}| ^2  &= 2 \int_{\R^3} V_\infty f_\infty -  \int_{\R^3} \Big[ 2 |\nabla {\varphi_\infty}| ^2 + V_\infty |f_\infty|^2 \Big] = 8\pi a. 
\end{align*}
Combining this with \eqref{eq:corr-energy-1a} and \eqref{eq:corr-energy-1b}, we obtain \eqref{eq:corr-energy-1}. This completes the proof of Theorem \ref{thm: main lw bd}. 
\end{proof}

\bigskip
\noindent
{\bf Organization of the paper.} In Section \ref{sec:pre}, we  collect several basic but helpful preliminary results. Then we give the proof of two key results, Propositions \ref{pro: completion square Vphi} and \ref{pro: completion square Vf},  in Sections \ref{sec: completion square Vphi} and \ref{sec: completion square Vf}, respectively. Finally we prove Proposition \ref{pro: a priori} in Section \ref{sec:apriori}, which, together with the previous results,  leads to the conclusion of Theorem \ref{thm: main lw bd}, as explained above. 

\bigskip
\noindent
{\bf Acknowledgments.} We would like to thank Martin Ravn Christiansen for inspiring discussions on the lower bound, and Marcello Porta for initiating discussions   on bosonization methods for dilute Fermi gases. We thank the referees for their careful reading and helpful comments. This work was partially funded by the Deutsche Forschungsgemeinschaft (DFG,
German Research Foundation) via TRR 352 -- Project-ID 470903074.  PTN
was partially supported by the European Research Council via the ERC Consolidator Grant RAMBAS -- Project-Nr. 101044249.

\section{Preliminaries}\label{sec:pre}

In this section we collect several basic but helpful estimates.  

\subsection{Scattering solutions} 

Recall the following well-known properties of the solution of the zero-energy scattering equation \eqref{eq: zero en scatt eq} (see e.g. \cite{LSSY}). 

\begin{lemma}[Properties of $\varphi_\infty$]\label{lem: prop of phi infty} Let $V_\infty$ be as in Assumption~\ref{asu: potential V}. Let  
 $\varphi_\infty:\R^3\to \R$ be the solution of \eqref{eq: zero en scatt eq}, with  scattering length $a$. Then
  \begin{equation}\label{eq: point bd phi infty}
    0 \leq \varphi_\infty(x)\leq \min \left\{ 1 , \frac a {|x|} \right\} \quad \forall x\in\mathbb{R}^3 \ , \qquad  
        \varphi_\infty(x) = \frac{a}{|x|}\quad   \forall |x| > R,
  \end{equation}
where $R>0$ is such that $\mathrm{supp}V_\infty \subset \{|x|\leq R\}$.
\end{lemma}
In the next lemma, we collect some estimates for the periodization of $\varphi_\infty$. 

\begin{lemma}[Bounds for $\varphi$, $\varphi^>$, $\varphi^<$]\label{lem: bounds phi} 
Let $V_\infty$ be as in Assumption \ref{asu: potential V} and let $\varphi$, $\varphi^>$, $\varphi^<$ be as in Definition \ref{def:scattering-phi} and \eqref{eq: def phi><}. Then we have the pointwise bound 
  \begin{equation}\label{eq: periodic-scattering-equation}
\left| -2\Delta\varphi (x) - (Vf )(x)+ \frac{8\pi a}{L^3} \right| \le \mathfrak{e}_L V(x).  
\end{equation}
Moreover,  
$$
\|\Delta \varphi\|_{L^1(\Lambda)} + \|\nabla\varphi\|_{L^2(\Lambda)} + \|\varphi\|_{L^\infty(\Lambda)} \le C
$$
 uniformly for $L$ large, and the same holds with $\varphi$ replaced by $\varphi^>$ or $\varphi^<$. We also have
  \begin{equation*}
  \|\varphi^>\|_{L^1(\Lambda)} \leq C\rho^{-\frac{2}{3} +2\gamma}, \,\,\|\nabla \varphi^>\|_{L^1(\Lambda)}\leq C\rho^{-\frac{1}{3} +\gamma}, 
  \quad \|\varphi^>\|_{L^2(\Lambda)} \leq C\rho^{-\frac{1}{6} +\frac{\gamma}{2}}.
  \end{equation*}
\end{lemma}

\begin{proof} From the definition of $\varphi$, we have $\widehat{\Delta \varphi}(p)=\mathcal{F}(\Delta \varphi_\infty) (p)$ for all $0\ne p\in \Lambda^*$. Moreover,  $\widehat{\Delta \varphi}(0)=0$ and $\mathcal{F}(-2\Delta \varphi_\infty) (0)= \mathcal{F}(V_\infty f_\infty)(0)=8\pi a$ by the zero-scattering equation \eqref{eq: zero en scatt eq}. Therefore,  
\begin{align*}  
-2\Delta \varphi (x) - (Vf)(x)+ \frac{8\pi a}{L^3} &= \sum_{ z \in \mathbb{Z}^3 } (V_\infty f_\infty)(x+L z) - (Vf)(x)  \\
& = \sum_{ z \in \mathbb{Z}^3 }  \Big( V_\infty(\varphi-\varphi_\infty) \Big) (x+Lz),
  \end{align*}
  where $f_\infty = 1-\varphi_\infty$.
The absolute value of the above function is bounded by $V(x)$ times  $\|\varphi-\varphi_\infty\|_{L^\infty(\supp V_\infty)}$, which  tends to $0$ as $L \to \infty$ by \cite[Eq. (A.4)]{GHNS}. This proves \eqref{eq: periodic-scattering-equation}.    
  
  From \eqref{eq: periodic-scattering-equation}, we find that $\| \Delta \varphi \|_{L^1(\Lambda)} \le C$. From this we can also deduce that $\|\Delta\varphi^<\|_1 \leq C$, using that $\|\chi_<\|_1 \leq C$ \cite[Lemma~A.5]{GHNS}, where $\chi_<$ denotes the function with Fourier coefficients $\widehat \chi_<(p)$. Moreover, $\|\nabla \varphi\|_{L^2} \le C \|\nabla \varphi_\infty\|_{L^2} \le C $. For all of the remaining bounds, we refer to \cite[Lemma A.2]{GHNS}. \end{proof}

\subsection{Configuration space representation and $t$-integral bounds} \label{rem:configuration space}

In our analysis, it will be convenient to estimate many of the error terms in configuration space. This allows  to  capture the volume factors correctly, and to efficiently exploit the positivity of the interaction potential $V$. In the derivation of the correlation Hamiltonian in Proposition \ref{pro: fermionic transf}, we use the  identity for the full interaction 
\begin{multline*}
\frac{1}{ L^3}\sum_{\sigma, \sigma\in\{\uparrow, \downarrow\}}\sum_{k,p,q\in\Lambda^*} \hat V(k) \hat{a}_{p+k,\sigma}^\ast \hat{a}_{q-k,\sigma'}^\ast  \hat{a}_{q,\sigma'} \hat{a}_{p,\sigma}
\\
 = \sum_{\sigma, \sigma\in\{\uparrow, \downarrow\}}\int_{\Lambda \times \Lambda} dxdy\, V(x-y) a^\ast_{x,\sigma}a^\ast_{y,\sigma^\prime}a_{y,\sigma^\prime}a_{x,\sigma}
\end{multline*}
which can be seen straightforwardly by applying \eqref{eq: def ak a*k} and then using
$$
\frac{1}{ L^9}\sum_{k,p,q\in\Lambda^*}  \hat V(k) e^{i (p+ k)\cdot x} e^{i (q-k)\cdot y} e^{-i q \cdot z} e^{-i p \cdot z'} 
=V(x-y) \delta(y-z) \delta (x-z').
 $$ 
Many other identities can be derived in the same way. For example, for square-integrable functions $f,h_1,h_2,g_1,g_2$, using  
\begin{multline*}
  \frac{1}{L^9}\sum_{p,r,r^\prime}\hat{f}(p)\hat{g}_1(r+p)\hat{h}_1(r)\hat{g}_2(r^\prime - p)\hat{h}_2(r^\prime) 
  \\
  = \frac{1}{L^9}\sum_{p,r,r^\prime\in \Lambda^\ast }\hat{f}(p)\hat{h}_1(r)\hat{h}_2(r^\prime)\left(\int_{\Lambda} dx\, g_1(x) e^{-ix\cdot(r+p)}\right)\left(\int_{\Lambda} dy\, g_2(x) e^{-iy\cdot (r^\prime - p)}\right)
  \\ = \int_{\Lambda\times \Lambda} dxdy\, f(y-x) g_1(x)h_1(x)g_2(y)h_2(y)
\end{multline*}
we can write 
\begin{multline*}
  \frac{1}{L^3}\sum_{\sigma, \sigma\in\{\uparrow, \downarrow\}}\sum_{p,r,r^\prime}\hat{V}(p) \hat{u}_{\sigma}(r+p) \hat{u}_{\sigma'}(r'-p)  \hat u_{\sigma}(r)\hat v_{\sigma'}(r')  \hat u_{\sigma'}(r') \hat{a}^*_{r+p,\sigma}  \hat{a}^*_{r'-p,\sigma'} \hat{a}^*_{r',\sigma'} \hat{a}_{r,\sigma}
  \\
  = \sum_{\sigma, \sigma\in\{\uparrow, \downarrow\}}\int_{\Lambda\times \Lambda} dxdy\, V(x-y) a_\sigma^*(u_{x}) a_{\sigma'}^*( u_{y}) a_{\sigma'}^*( v_{y})a_{\sigma}( u_{x}), 
\end{multline*}
which appears in $ \mathbb{Q}_3$. 

To facilitate several error estimates involving the operator $T_{\sigma}(r)$ introduced in Definition \ref{def:ren-Q2}, we will frequently use the identity 
\begin{align}  \label{eq:int-t}
\ell^{-1} = \int_0^\infty dt \, e^{-t \ell}, \quad \ell>0 
\end{align}
to decompose $\omega^\eps_{r,r'}(p)$ in \eqref{eq: def omega eps} into a factorized form as 
\begin{align}\label{eq: int t conf space}
 \frac{1}{\lambda_{p,r} + \lambda_{-p,r^\prime}+2\varepsilon} 
 =\int_0^\infty dt\, e^{-2t\varepsilon} e^{-t|r+p|^2}e^{t|r|^2} e^{-t|r^\prime -p|^2} e^{t|r^\prime|^2}
\end{align}
under the conditions $r+p\notin\mathcal{B}_F^\sigma$, $r^\prime - p\notin\mathcal{B}_F^{\sigma^\prime}$, $r\in\mathcal{B}_F^\sigma$, and $r^\prime\in \mathcal{B}_F^{\sigma^\prime}$. When applying \eqref{eq: int t conf space}, it is convenient to introduce,  for $t>0$,  the functions  $v_{t,\sigma}, u_{t,\sigma}: \Lambda \to \mathbb{C}$ as 
\begin{equation}\label{eq: def vt}
  \hat v_{t,\sigma}(k)= \hat v_{\sigma} (k) e^{t |k|^2},\quad    \hat u_{t,\sigma}(k)= \hat u_{\sigma} (k) e^{-t |k|^2},\quad k \in \Lambda^*. 
 \end{equation}
We will often split further the momenta outside the Fermi ball. Accordingly, with $0<\gamma<1/6$ we define
\begin{equation}\label{eq: u< u> gamma}
\hat{u}^<_\sigma(k) = \begin{cases} 1 &\mbox{if}\,\,\, k_F^\sigma < |k| \leq 6\rho^{\frac{1}{3} - \gamma}, \\ 0 &\mbox{otherwise}   \end{cases},\quad \hat{u}^>_\sigma(k) = \begin{cases} 1 &\mbox{if}\,\,\, |k| \geq 3\rho^{\frac{1}{3} - \gamma}, \\ 0 &\mbox{otherwise}\end{cases}\end{equation}
and correspondingly
\begin{equation}\label{eq: def u<>t}
  \hat u^<_{t,\sigma}(k) = \hat{u}^<_\sigma(k) e^{-t|k|^2},\quad   \hat u^>_{t,\sigma}(k) = \hat{u}^>_\sigma(k) e^{-t|k|^2}. 
  \end{equation}
  
 The following bounds will be helpful when we use the representation in \eqref{eq: int t conf space}. 
 
\begin{lemma}[Integration over $t$]\label{lem:t} Let $0<\gamma<1/3$ and let $\hat{v}_{t,\sigma}$, $\hat{u}^<_{t,\sigma}$ and $\hat{u}^>_{t,\sigma}$ be as in \eqref{eq: def vt}, \eqref{eq: def u<>t}. Then we have
\begin{align}\label{fac1}
 \int_0^\infty dt \, e^{-2t\varepsilon} e^{2t (k_F^\sigma)^2} \|\hat{u}^<_{t,\sigma}\|^2_2 &\leq C \rho^{\frac 13-\gamma},\\
\label{fac2}
\int_0^\infty dt \, e^{-2t\varepsilon} e^{-2t (k_F^\sigma)^2} \|\hat{v}_{t,\sigma}\|^2_2  &\leq C\rho^{\frac{1}{3} -\kappa},\\
\label{fac2u>}
\int_0^\infty dt \, e^{2t (k_F^\sigma)^2} \||\cdot|^{-1}\,\hat{u}^>_{t,\sigma}\|^2_2  &\leq  C\rho^{-\frac{1}{3} + \gamma}, \\
\label{fac2u>b}
\int_0^\infty dt \, e^{2t (k_F^\sigma)^2} \|\widetilde{u}^>_{t,\sigma}\|^2_2  &\leq  C\rho^{-1 + 3\gamma},\\
  \int_{0}^{\infty} dt\, e^{t(k_F^\sigma)^2+t(k_F^{\sigma^\prime})^2} \|\widetilde{u}^>_{t,\sigma}\|_\infty &\leq C\rho^{-\frac{1}{3} + \gamma}.\label{eq: int t tilde nu infty}
\end{align}
Here \eqref{fac2} holds for every constant $\kappa>0$, and $\widetilde{u}^>_{t,\sigma}:\Lambda \to \mathbb{C}$ in \eqref{fac2u>b}--\eqref{eq: int t tilde nu infty}  is the function with Fourier coefficients $ |k|^{-2} \hat{u}^>_{t,\sigma}(k)$. 
\end{lemma}

\begin{proof} 
For \eqref{fac1} and \eqref{fac2}, we refer to the proof of \cite[Lemma 3.4]{GHNS}. 
For \eqref{fac2u>}, we obtain with the aid of  \eqref{eq:int-t} 
\begin{align*}
 \int_0^\infty dt \, e^{2t (k_F^\sigma)^2} \||\cdot|^{-1}\,{u}^>_{t,\sigma}\|^2_2 & = \frac{1}{2 L^3}\sum_k  \frac{ |\hat{u}^>_{t,\sigma}(k)|^2}{ |k|^2 (|k|^2 - (k_F^\sigma)^2 )} \leq \frac{C}{L^3}\sum_{|k| \ge 3 \rho^{\frac 1 3 -\gamma}}\frac{1}{|k|^4} \leq C\rho^{-\frac{1}{3} + \gamma}.
\end{align*}
The  bound in \eqref{fac2u>b} can be proved analogously. To prove \eqref{eq: int t tilde nu infty}, we bound $\|\widetilde{u}^>_{t,\sigma}\|_{\infty}$ by $\| |\cdot|^{-2}\hat u^>_{t,\sigma}\|_{1}$ and obtain 
\begin{align*}
   \int_{0}^{\infty} dt\, e^{t(k_F^\sigma)^2+t(k_F^{\sigma^\prime})^2} \|\widetilde{u}^>_{t,\sigma}\|_\infty &\leq  \frac{1}{L^3}\sum_{|k| \geq 3\rho^{\frac{1}{3} - \gamma}} \int_{0}^{\infty}  dt\, \frac{e^{-t(|k|^2 - (k_F^\sigma)^2 - (k_F^{\sigma^\prime})^2)}}{|k|^2} \\
    &\leq \frac{C}{L^3}\sum_{|k| \geq 3\rho^{\frac{1}{3} - \gamma}}\frac{1}{|k|^{4}} \leq C\rho^{-\frac{1}{3} + \gamma}.
    \end{align*}
\end{proof}

With the aid of the Cauchy--Schwarz inequality, Lemma~\ref{lem:t} also implies the bounds 
\begin{align}\label{eq: est int t vu< final}
  \int dt \, e^{-2t\varepsilon} \|\hat{u}^<_{t,\sigma}\|_2 \|\hat{v}_{t,\sigma}\|_2  &\leq   C\rho^{\frac{1}{3} - \frac{\gamma}{2} - \kappa}, \\
\label{eq: est int t vu> final}
\int dt \, e^{-2t\varepsilon} \||\cdot|^{-1}\,\hat{u}^>_{t,\sigma}\|_2 \|\hat{v}_{t,\sigma}\|_2 &\leq C\rho^{\frac{\gamma}{2} - \kappa},\\
\label{eq: est int t vu> final-b}
 \int dt \, e^{-2t\varepsilon} \||\cdot|^{-2}\,\hat{u}^>_{t,\sigma}\|_2 \|\hat{v}_{t,\sigma}\|_2 &\leq C\rho^{-\frac{1}{3} + \frac{3}{2}\gamma - \kappa} . 
\end{align}

The use of \eqref{eq: int t conf space} naturally leads to the appearance of the periodic heat kernel in the analysis. The following bounds will be helpful. 

\begin{lemma}[Periodic heat kernel]\label{lem: zeta t L1}
The functions $\zeta_1^t,\zeta_>^t:\Lambda\to \mathbb{C}$ defined by $\hat \zeta_1^t (k)=e^{-t|k|^2}$ and $\hat \zeta_>^t (k)=e^{-t|k|^2} \1_{\{|k| \ge 3\rho^{\frac 1 3 -\gamma}\}}$  satisfy 
$$ \|\zeta_1^t\|_{L^1(\Lambda)} =1 \ ,\quad \|\zeta_>^t\|_{L^2(\Lambda)} \leq C t^{- \frac 3 4}  e^{-  \frac{9}{2}t\rho^{\frac 2 3-2\gamma}}.$$ 
\end{lemma}
\begin{proof} 
The function $\zeta_1^t$ is the periodization of the heat kernel 
$$\zeta^t_\infty(x):=(2\pi)^{-3}\int_{\R^3} dp\, e^{-tp^2}e^{ip\cdot x}= (4\pi t)^{-3/2} e^{-x^2/(4t)}.$$
Hence it is non-negative, and  $\|\zeta_1^t\|_{L^1(\Lambda)} = \hat \zeta_1^t(0) = 1$. 
Moreover,
$$
\|\zeta_>^t\|_{L^2(\Lambda)}^2 = \frac{1}{L^3} \sum_{|k|\ge 3\rho^{\frac 1 3 -\gamma}} e^{-2t |k|^2} \le \frac{1}{L^3} e^{-9t\rho^{\frac 2 3-2\gamma}}  \sum_{|k|\ge 3\rho^{\frac 1 3 -\gamma}} e^{-t |k|^2} .
$$
For any $0\neq k \in \Lambda^*$, we can get an upper bound on $e^{-t |k|^2}$ by averaging it over a part of the cube of side length $2\pi/L$ centered at $k$ that is closer to the origin than $k$, which contains at least a volume $(\pi/L)^3$. Hence 
$$
\frac 1{L^{3}} \sum_{k\neq 0} e^{-t |k|^2} \leq  \frac 1{\pi^3}  \int_{\R^3} dk\,  e^{- t |k|^2} = C t^{-3/2}.
$$ 
\end{proof}

\subsection{Some operator bounds}\label{sec: operator bounds}
In this subsection we collect some basic operator inequalities on Fock space, which will be used several times throughout the proof. We will denote by $a^{\#}(f)$  either  $a(f)$ or $a^*(f)$. Moreover, we always use the notation $f_x(\cdot)=f(\cdot -x)$.  

First, as a consequence of the canonical anti-commutation relations \eqref{eq:CAR}, we have 
\begin{equation}\label{eq:Pauli}
\|a(f)\| =  \|a^\ast(f)\| = \|f\|_{L^2(\Lambda, \mathbb{C}^2)}.
\end{equation}
We will often use \eqref{eq:Pauli} for $\hat{v}_\sigma$, $\hat{u}^<_\sigma$, $\hat{v}_{t,\sigma}$, $\hat{u}^<_{t,\sigma}$  defined in \eqref{eq: def u,v}, \eqref{eq: def vt}, \eqref{eq: u< u> gamma}, \eqref{eq: def u<>t}. For example,  for  $n=0,1,...$ and $\ell\in \{1,2,3\}$, we have the  uniform bounds for $x\in\Lambda$: 
 \begin{align}\label{eq:Pauli-1}
\|a_\sigma^{\#}(\partial_\ell^n v_x)\| &\leq C\rho^{\frac{1}{2} + \frac{n}{3}}, \quad \|a_\sigma^{\#}(\partial_\ell^nu_x ^<)\|\leq C\rho^{\frac{1}{2}-\frac{3}{2}\gamma +n\left(\frac{1}{3} -\gamma\right)}, \nn\\
\|a_\sigma^{\#}(\partial^n_\ell  {v}_{t,x})\| &\leq C e^{t(k_F^\sigma)^2}\rho^{\frac 1 2 + \frac{n}{3}} ,  \quad \|a_\sigma^{\#}(\partial^n_\ell  {u}^<_{t,x})\|  \leq C e^{-t(k_F^\sigma)^2}\rho^{\frac 1 2 -\frac {3\gamma}{2} + n\left(\frac{1}{3} - \gamma\right)}. 
 \end{align}

Next, we discuss three lemmas for operators of the form $a^{\#} a^{\#}$. We start with some simple bounds involving the number and kinetic energy operators. 

\begin{lemma} \label{lem:1a} For  $\sigma \in \{\uparrow,\downarrow\}$ and $\hat f \in \ell^\infty(\Lambda^*)$, we have
 \begin{align}\label{eq:1a-1}
\int_{\Lambda} dx\, | a_\sigma(f_x)|^2 & \le {\|\hat f  \|^2_\infty}\cN ,\\
\label{eq:1a-2}
\int_{\Lambda} dx\, | a_\sigma (\nabla f_x)|^2 &\le  {\|\hat f  \|^2_\infty} (\bH_0 + C\rho^{\frac 2 3} \cN)\,.
\end{align}
Moreover, if $\hat f (p)$ is supported on $\{|p|\ge \rho^{1/3-\kappa}\}$ for some $\kappa >0$, then
 \begin{equation}\label{eq:1a-3}
 \int_{\Lambda} dx\, | a_\sigma(f_x)|^2  \le C \rho^{-\frac 2 3 + 2 \kappa} {\|\hat f  \|^2_\infty} \bH_0. 
\end{equation}
\end{lemma} 
\begin{proof} The bound \eqref{eq:1a-1} follows from the representation in Fourier space
$$
\int_{\Lambda} dx\, | a_\sigma(f_x)|^2 = \sum_{p\in \Lambda^*} | {\hat f}(p)|^2 \hat{a}^*_{p,\sigma} \hat{a}_{p,\sigma} \le  {\|\hat f  \|^2_\infty} \sum_{p\in \Lambda^*} \hat{a}^*_{p,\sigma} \hat{a}_{p,\sigma} \le   {\|\hat f  \|^2_\infty} \cN. 
$$
Similarly,  \eqref{eq:1a-2} follows from 
\begin{align*}
\int_{\Lambda} dx\, | a_\sigma(\nabla f_x)|^2 &= \sum_{p\in \Lambda^*} | {\hat f}(p)|^2 |p|^2 \hat{a}^*_{p,\sigma} \hat{a}_{p,\sigma}
\\
& \le  {\|\hat f  \|^2_\infty} \sum_{p\in \Lambda^*} \Big( ||p|^2-(k_F^{\sigma})^2| + (k_F^{\sigma})^2 \Big) \hat{a}^*_{p,\sigma} \hat{a}_{p,\sigma}. 
\end{align*}
For \eqref{eq:1a-3}, we use $| {\hat f}(p)|^2 \le  C \|\hat f\|_\infty^2  \rho^{-\frac 2 3+2\kappa} ||p|^2 -(k_F^{\sigma})^2|$ for $\{|p|\ge \rho^{1/3-\kappa}\}$.
\end{proof}

From the support properties of  $\hat u^>$ and $\hat u^>_t$, it readily follows from \eqref{eq:1a-3} that 
 \begin{align}
   \int_\Lambda dx\, |a_{\sigma}(u^>_x)|^2  & \leq C\rho^{-\frac{2}{3} + 2\gamma} \mathbb{H}_0 , \label{eq: est H0 sim}\\
 \label{eq: est N u>t}
  \int_\Lambda dx\,  |a_\sigma(u^>_{t,x})|^2  & \leq C e^{- 18 t \rho^{2/3-2\gamma}} \rho^{-\frac{2}{3} + 2\gamma}  \mathbb{H}_0. 
  \end{align}

While Lemma \ref{lem:1a} allows us to control many operators of the form $a^*a$ (i.e., one creation and one annihilation operator), the following result is helpful for operators of the form $aa$ or $a^*a^*$ (i.e., two annihilation operators or two creation operators).

\begin{lemma}  
\label{lem:2a} 
Let $h\in L^1(\Lambda)$ and $f,g\in L^2(\Lambda)$. For $\sigma\in \{\uparrow,\downarrow\}$, the operator 
\begin{equation}\label{eq: def b general}
b_\sigma(h,f,g) = \int dz \, h(z)a_\sigma(f_z)a_{\sigma}(g_z)
\end{equation}
satisfies 
\[
 \|b_\sigma(h,f,g)\| \leq  \| h \|_{1} \|{f}\|_2 \|{g}\|_2
 \]
and 
\begin{equation}\label{lb6}
  {\|b_\sigma(h,f,g)\| \leq C \| |\cdot|^{-2}\hat{g}\|_2\left(\|\Delta h\|_1 \|\hat{f}\|_2 + \|\nabla h\|_1 \||\cdot|\hat{f}\|_2 + \|h\|_1 \||\cdot|^2 \hat{f}\|_2\right).}
\end{equation}
\end{lemma}

\begin{proof}
The first bound follows immediately from \eqref{eq:Pauli}. For the second, we proceed similarly to  \cite[Lemma 4.8]{Gia1}. A simple calculation in Fourier space yields
  \begin{align*}
    b_\sigma(h,f,g)&= -\int dz\, \Delta h(z) a_\sigma(f_{z})a_\sigma(\widetilde{g}_{z})  - 2\sum_{j=1}^3\int dz\, \partial_j h(z) a_\sigma(\widetilde{g}_{z})a_\sigma(\partial_jf_{z})\\
&\qquad  - \int dz\, h(z) a_\sigma(\widetilde{g}_{z})a_\sigma(\Delta f_{z}),
  \end{align*}
where $\widetilde{g}:\Lambda\to \mathbb{C}$ has Fourier coefficients $|p|^{-2} \hat g(p)$. 
The claimed bound then follows again directly from  
\eqref{eq:Pauli}. 
\end{proof}

We frequently apply Lemma \ref{lem:2a} with $h = \partial_\ell^n\varphi^>_{z^\prime}$, $f= \partial^m_\ell v$ and  $\hat{g}$ supported outside the Fermi ball. For example, from Lemmas~\ref{lem:2a} and~\ref{lem: bounds phi} we obtain
\begin{equation}\label{eq: b op phi vu}
  \sup_{z'}\|b_\sigma(\varphi^>_{z^\prime}, v, u^>)\|\leq C\rho^{1/3 + \gamma/2}
\end{equation}
which can be also found in \cite[Lemma 4.8]{GHNS}. Moreover, the same technique as in Lemma \ref{lem:2a} (i.e., multiplying and dividing by $|r+p|^2$ in momentum space in order to redistribute the momentum) will be used several times in the proof, in different situations and involving more than two creation and annihilation operators.

The following estimate, involving a particular combination of $a^{\#}a^{\#}$, will also be helpful.

\begin{lemma}\label{lem:6a} Let  {$\hat{f}, \hat{g},\hat{h} \in \ell^\infty(\Lambda^\ast)$. Let $\{X(x)\}_{x\in \Lambda}$, $\{Y(y)\}_{y\in \Lambda}$  be two families of bounded operators on $\mathcal{F}_{\mathrm{f}}$.} Then for all $\sigma,\sigma'\in \{\uparrow,\downarrow\}$ and $\psi\in\mathcal{F}_{\mathrm{f}}$, we have
 \begin{align}\label{eq:6a1}
 &\left| \frac{1}{L^3} \sum_{p\in \Lambda^*} \hat h(p) \Big\langle \int dx\, e^{-ipx} X(x) a_\sigma(f_x) \psi, \int dy\, e^{-ip\cdot y} Y(y) a_{\sigma'}^{\#} (g_y) \psi \Big\rangle \right| 
\\ &\le \|\hat h\|_\infty  \sup_x \| X(x)\| \sup_y \|Y(y)\| \left( \int dx  \left\|  a_\sigma(f_x) \psi \right\|^2 \right)^{\frac 1 2} \left(   \int dy \left\| a_{\sigma'}^{\#} (g_y) \psi \right\|^2 \right)^{\frac 1 2}.  \nn
 \end{align}
\end{lemma}

\begin{proof} By the Cauchy--Schwarz inequality with respect to the summation over $p$, we have
\begin{align*}
 &\left| \frac{1}{L^3} \sum_{p\in \Lambda^*} \hat h(p) \Big\langle \int dx\, e^{-ipx} X(x) a_\sigma(f_x) \psi, \int dy\, e^{-ip\cdot y} Y(y) a_{\sigma'}^{\#} (g_y) \psi \Big\rangle \right| \nn\\
 &\le \|\hat h\|_\infty \left( \frac{1}{L^3} \sum_{p\in \Lambda^\ast }   \left\| \int dx\, e^{-ipx} X(x) a_\sigma(f_x) \psi \right\|^2 \right)^{\frac 1 2} \left( \frac{1}{L^3} \sum_{p\in \Lambda^\ast }  \left\| \int dy\, e^{-ip\cdot y} Y(y) a_{\sigma'}^{\#} (g_y) \psi \right\|^2 \right)^{\frac 1 2} \\
 &= \|\hat h\|_\infty \left( \int dx  \left\|  X(x) a_\sigma(f_x) \psi \right\|^2 \right)^{\frac 1 2} \left(   \int dy \left\| Y(y) a_{\sigma'}^{\#} (g_y) \psi \right\|^2 \right)^{\frac 1 2} 
\end{align*}
from which the claimed bound readily follows. 
\end{proof}

In applications, we can further estimate the right-hand side of \eqref{eq:6a1} by using Lemma~\ref{lem:1a}. For example, by combining Lemmas~\ref{lem:6a} and~\ref{lem:1a}, we obtain
 \begin{align}\label{eq:6a1a}
 &\left| \frac{1}{L^3} \sum_{p\in \Lambda^*} \hat h(p) \Big\langle \int dx\, e^{-ipx} X(x) a_\sigma(f_x) \psi, \int dy\, e^{-ip\cdot y} Y(y) a_{\sigma'} (g_y) \psi \Big\rangle \right| \nn\\
 &\le \|\hat h\|_\infty \|\hat f_\sigma\|_\infty  \|\hat g_{\sigma'}\|_\infty  \| \sup_x \| X(x)\| \sup_y \|Y(y)\| \langle \psi, \cN \psi\rangle.
  \end{align}

The next three lemmas give bounds for operators of the form $a^{\#}a^{\#} a^{\#} a^{\#}$. 

\begin{lemma}\label{lem:4a} Let $W\in L^1(\Lambda)$, $f,h \in L^2(\Lambda)$ and  $\hat g, \hat k\in \ell^\infty(\Lambda^*)$. For $\sigma \neq \sigma'\in \{\uparrow,\downarrow\}$, the operator
\begin{equation}\label{eq:4a}
B_{\sigma,\sigma'}(W,f,g,h,k) = \iint_{\Lambda^2} dx dy\, W(x-y) a_\sigma^{\#}  (f_x) a_{\sigma'}^*  (g_y) a_{\sigma'}^{\#}(h_y) a_\sigma(k_x)
\end{equation}
satisfies 
\begin{align}\label{eq:4a0}
&|\langle \psi, B_{\sigma,\sigma'}(W,f,g,h,k) \psi\rangle| \nn\\
&\le \|W\|_1  {\|f\|_2 \|{h}\|_2} 
\left(\int dy \| a_{\sigma'}  (g_y) \psi\|^2 \right)^{\frac 1 2}\left(\int dx \| a_{\sigma}  (k_x) \psi\|^2 \right)^{\frac 1 2}
\end{align}
for all $\psi\in \mathcal{F}_{\rm f}$. 
\end{lemma}

\begin{proof} This follows immediately from \eqref{eq:Pauli}  and the 
Cauchy--Schwarz inequality, using also that $a_\sigma^{\#}  (f_x)$ and $a_{\sigma'}^*  (g_y)$ anticommute for $\sigma\neq\sigma'$.
\end{proof}

Again, we can further estimate the right-hand side of \eqref{eq:4a0} by Lemma \ref{lem:1a}. For example, combining Lemmas~\ref{lem:4a} and~\ref{lem:1a}, 
\begin{align}\label{eq:4a1}
|\langle \psi, B_{\sigma,\sigma'}(W,f,g,h,k) \psi\rangle| &\le \|W\|_1  {\|{f}\|_2 \|{h}\|_2}\|\hat g\|_\infty \|\hat k\|_\infty \langle \psi, \cN \psi\rangle,\\
\label{eq:4a3}
|\langle \psi, B_{\sigma,\sigma'}(W,f,\partial_\ell g,h,\partial_m k) \psi\rangle| 
&\le \|W\|_1  {\|{f}\|_2 \|{h}\|_2} \|\hat g\|_\infty \|\hat k\|_\infty \! \left(\! \langle \psi,  \bH_0  \psi\rangle  + C  \rho^{\frac 2 3} \langle \psi, \cN \psi\rangle \! \right)\! .
\end{align}

\begin{lemma}\label{lem:4aa} Let  $W_1,W_2\in L^1(\Lambda)$, $W_3,W_4,g,h\in  L^2(\Lambda)$ and $\hat f, \hat k\in \ell^\infty(\Lambda^*)$. Then for all $\sigma,\sigma'\in \{\uparrow,\downarrow\}$ and $\psi\in\mathcal{F}_{\mathrm{f}}$, we have
 \begin{align}\label{eq:4aa}
 &\biggl| \biggl\langle \psi,  \int_{\Lambda^4} dx dy dz dz' W_1(x-y) W_2 (z-z') W_3 (x-z) W_4 (y-z')  \nn \\
&\qquad\qquad\qquad\qquad\qquad\qquad\qquad\qquad\qquad   \times a^*_\sigma(f_x) a_{\sigma'}^{\#} (g_y) a^{\#}_{\sigma} (h_z) a_{\sigma'} (k_{z'}) \psi  \biggl\rangle \biggl| \nn \\
 &\le \|W_1\|_1 \|W_2\|_1 \|W_3\|_2 \|W_4\|_2  {\|{g}\|_2 \|{h}\|_2} \left(\int dx \| a_{\sigma}  (f_x) \psi\|^2 \right)^{\frac 1 2}\left(\int dz' \| a_{\sigma'}  (k_{z'}) \psi\|^2 \right)^{\frac 1 2}.
 \end{align}
\end{lemma}

\begin{proof} Again this follows easily from  \eqref{eq:Pauli} and the Cauchy--Schwarz inequality.
\end{proof}

Combining Lemmas~\ref{lem:4aa} and~\ref{lem:1a}, we obtain 
 \begin{align}\label{eq:4aa-1}
 &\biggl|  \biggl\langle \psi,  \int_{\Lambda^4} dx dy dz dz' W_1(x-y) W_2 (z-z') W_3 (x-z) W_4 (y-z')   \nn \\
&\qquad\qquad\qquad\qquad\qquad\qquad\qquad\qquad\qquad   \times a^*_\sigma(f_x) a_{\sigma'}^{\#} (g_y) a^{\#}_{\sigma} (h_z) a_{\sigma'} (k_{z'}) \psi  \biggl\rangle \biggl| \nn \\
 &\le \|W_1\|_1 \|W_2\|_1 \|W_3\|_2 \|W_4\|_2 \| \hat f\|_\infty \|\hat k\|_\infty  {\|{g}\|_2 \|{h}\|_2} \langle \psi, \cN \psi\rangle. 
 \end{align}
 
The following bound will be useful to obtain the improved a priori bounds in Proposition \ref{pro: a priori}, and also to prove  \eqref{eq: eps-1}  from Proposition \ref{pro: fermionic transf}. 

\begin{lemma}\label{lem:4aaa} Let   
$W\in L^1(\Lambda)$, $f,g \in L^2(\Lambda)$ and   $\hat w, \hat h  \in \ell^\infty  (\Lambda^*)$. Then for all $\sigma \neq\sigma'\in \{\uparrow,\downarrow\}$ and $\psi\in\mathcal{F}_{\mathrm{f}}$, we have
  \begin{align}\label{eq:4aaa}
 &\left| \left\langle \psi,  \int_{\Lambda^2} dx dy\, W(x-y) a^*_\sigma(w_x) a_{\sigma'}^{*} (f_y) a^{\#}_{\sigma'} (g_{y}) a_{\sigma} (h_{x}) \psi  \right\rangle \right|  \\
 &\le  C  \|W\|_1   { \|\hat w\|_\infty \|{g}\|_2}\|\hat h\|_\infty \Big( \|\hat f\|_\infty   \rho^{\frac 1 2+\frac{\beta}{6}} \langle \psi, \mathcal{N} \psi\rangle +   \| f\|_2  \|\mathcal{N}_\beta^{\frac{1}{2}}\psi\|  \|\mathcal{N}^{\frac{1}{2}}\psi\|\Big) \nn
 \end{align}
for all $0\le \beta<1$,  where $\cN_\beta$ is defined in \eqref{eq: def Nalpha >}. The same holds true when $a^*_\sigma(w_x) a_{\sigma'}^{*} (f_y) a^{\#}_{\sigma'} (g_{y}) a_{\sigma} (h_{x})$ is replaced by $a^*_\sigma(w_x) a_{\sigma}^{*} (f_x) a^{\#}_{\sigma'} (g_{y}) a_{\sigma^\prime} (h_{y})$.
\end{lemma}

\begin{proof} Given $0\leq \beta <1$, we split the left-hand side using 
\begin{equation}\label{eq: def u<< and u>>}
 \hat{w}(k) = \hat w^{\ll}(k) + \hat w^{\gg}(k),\quad \hat{w}^{\ll}(k) = \begin{cases}  \hat{w}(k) &\mbox{if}\, \, \,  | |k|- k_F^\sigma| \le (k_F^\sigma)^{1+\beta}, 
  \\
  0 &\mbox{otherwise}.\end{cases}
\end{equation}
For the term involving $\omega^{\ll}$, we use \eqref{eq:4a3} and 
  $\|{w}^{\ll}\|_2 \leq C \|\hat w\|_\infty \rho^{1/2 + \beta/6}$, yielding the first term on the right-hand side.

  For the term involving $w^{\gg}$, we interchange the roles of $f$ and $w$. Lemma~\ref{lem:4a} together with   
$$
\int_{\Lambda} dx \|a_\sigma(w^{\gg}_x) \psi\|^2 \le \|\hat w\|^2_\infty \langle \psi, \cN_\beta \psi\rangle
$$
as a consequence of the support property of $\hat{\omega}_{\sigma}^{\gg}$ then yields the desired result  \eqref{eq:4aaa}. The proof for the case in which $a^*_\sigma(\omega_x) a_{\sigma'}^{*} (f_y) a^{\#}_{\sigma'} (g_{y}) a_{\sigma} (h_{x})$ is replaced by $a^*_\sigma(\omega_x) a_{\sigma}^{*} (f_x) a^{\#}_{\sigma'} (g_{y}) a_{\sigma^\prime} (h_{y})$ can be done in the same way.
 \end{proof}

As an application of Lemma \ref{lem:4aaa}, we can now provide the

\begin{proof}[Proof of  \eqref{eq: eps-1}] Applying Lemma \ref{lem:4aaa} we can estimate separately the four terms in  $\mathcal{E}_{\rm corr}$ in \eqref{eq: def H-corr}. This gives 
 \begin{align*}
 |\langle \psi, \mathcal{E}_{\rm corr} \psi \rangle| &\le C\|V\|_1 \Big(\rho^{\frac 1 2+\frac{\beta}{6}} \max_\sigma \| v_\sigma\|_2 \langle \psi, \mathcal{N} \psi\rangle + \max_\sigma \| v_\sigma\|_2^2  \|\mathcal{N}_\beta^{\frac{1}{2}}\psi\|  \|\mathcal{N}^{\frac{1}{2}}\psi\| \Big) \\
 &\le C \rho^{ 1 +\frac{\beta}{6}}  \langle \psi, \mathcal{N} \psi\rangle + C\rho  \|\mathcal{N}_\beta^{\frac{1}{2}}\psi\|  \|\mathcal{N}^{\frac{1}{2}}\psi\|,
\end{align*}
as claimed in \eqref{eq: eps-1}.
\end{proof}


\section{Completion of the square for $V\varphi$}\label{sec: completion square Vphi}

In this section we prove Proposition \ref{pro: completion square Vphi}. 

\begin{proof}
Using the identities \eqref{eq:Q4-square-dec-1} and \eqref{eq:Q4-square-dec-2}, we find that 
\begin{align}\nn 
& \mathbb{Q}_4 + (\mathbb{Q}_2 +  \mathbb{Q}_3)\big\vert_{V\varphi}  - \frac{1}{2}\sum_{\sigma\neq \sigma^\prime}\int dxdy\, V(x-y)  \left| a_{\sigma^\prime}(u_y)a_\sigma(u_x) + \mathcal{T}_{\sigma, \sigma^\prime}(x,y) + \mathcal{S}_{\sigma, \sigma^\prime}(x,y)\right|^2 
  \\
  & =- \frac{1}{2}\sum_{\sigma\neq\sigma^\prime}\int dxdy\, V(x-y)\left( |\mathcal{T}_{\sigma, \sigma'}(x,y)|^2 +|\mathcal{S}_{\sigma, \sigma'}(x,y)|^2 + (\mathcal{T}^\ast_{\sigma, \sigma^\prime}(x, y)\mathcal{S}_{\sigma, \sigma^\prime}(x,y) + \mathrm{h.c.}) \right).
  \label{eq:Q4-square-dec}
\end{align}
For the terms involving $|\mathcal{T}_{\sigma,\sigma^\prime} (x,y)|^2$, we write the corresponding contribution in normal order and isolate the constant term. Using 
 \begin{align*}
 a_\sigma(v_x)a_{\sigma^\prime}(v_y)a_{\sigma^\prime}^\ast (v_y)a_\sigma^\ast (v_x)  & = \rho_\sigma \rho_{\sigma^\prime} - \rho_{\sigma^\prime}a_{\sigma}^\ast(v_x)a_{\sigma}(v_x) - \rho_\sigma a_{\sigma^\prime}^\ast (v_y)a_{\sigma^\prime}(v_y) \\
 &\quad + a_\sigma^\ast(v_x)a_{\sigma^\prime}^\ast(v_y)a_{\sigma^\prime}(v_y)a_\sigma(v_x), 
 \end{align*}
we obtain
\begin{align}\label{eq:Q4-square-dec-3}
  & \frac{1}{2}\sum_{\sigma\neq \sigma^\prime}\int dxdy\, V(x-y) |\mathcal{T}_{\sigma,\sigma^\prime}(x,y)|^2  \\
  &= \rho_\uparrow\rho_\downarrow L^3\int_{\Lambda} V\varphi^2 + \mathbb{Q}_0\big\vert_{V\varphi^2}  - \sum_{\sigma\neq\sigma^\prime} \rho_\sigma \int V\varphi^2 \int dx |a_{\sigma^\prime}(v_x)|^2\nn
\end{align}
where 
\begin{equation} \label{eq:Q0-Vphi}
\mathbb{Q}_0\big\vert_{V\varphi^2} = \frac{1}{2}\sum_{\sigma\neq\sigma^\prime}\int_{\Lambda^2} dxdy\, (V\varphi^2) (x-y) a^\ast_\sigma(v_x)a^\ast_{\sigma^\prime}(v_y)a_{\sigma^\prime}(v_y)a_\sigma(v_x). 
\end{equation}
Similarly, we have
\begin{align}\label{eq:Q4-square-dec-4}
  \frac{1}{2}\sum_{\sigma\neq \sigma^\prime}\int dxdy\, V(x-y)|\mathcal{S}_{\sigma, \sigma^\prime}(x,y)|^2 =
 \sum_{\sigma\neq\sigma^\prime}\rho_{\sigma^\prime}\int  V\varphi^2 \int dx |a_\sigma(u_x)|^2 +\mathcal{E}_{\mathcal{S}} 
\end{align}
with
\begin{align*}
\mathcal{E}_{\mathcal{S}}    
&= - \sum_{\sigma\neq \sigma^\prime}\int dxdy\, (V\varphi^2) (x-y)a_\sigma^\ast(u_x)a_{\sigma^\prime}^\ast(v_y)a_{\sigma^\prime}(v_y)a_\sigma(u_x) \nn
  \\
  &\quad + \sum_{\sigma\neq \sigma^\prime}\int dxdy\, (V\varphi^2)(x-y) a^\ast_\sigma(u_x)a_\sigma^\ast(v_x)a_{\sigma^\prime}(v_y)a_{\sigma^\prime}(u_y)
\end{align*}
and also, using $u_\sigma v_\sigma=0$, 
\begin{align} \label{eq:Q4-square-dec-5}
 & \frac{1}{2}\sum_{\sigma\neq \sigma^\prime} \int dxdy\, V(x-y) \mathcal{T}_{\sigma, \sigma^\prime}^\ast(x,y)\mathcal{S}_{\sigma, \sigma^\prime}(x,y) \nn\\
&= \sum_{\sigma\neq \sigma^\prime}\int dxdy\, (V\varphi^2)(x-y) a^\ast_{\sigma^\prime}(v_y)a_\sigma(v_x)a_{\sigma^\prime}(v_y)a_{\sigma}(u_x)=: \mathcal{E}_{\mathcal{T}\mathcal{S}} 
   \end{align}
On the relevant subspace satisfying the particle-hole relation \eqref{eq:particle-hole}, the quadratic terms on the right-hand side of Eqs.~\eqref{eq:Q4-square-dec-3} and~\eqref{eq:Q4-square-dec-4} cancel if the two equations are added, see \eqref{eq: int au av particle hole}.
We thus conclude that \eqref{eq: first id Q4} holds with 
$$
\mathcal{E}_{V\varphi} = - \mathbb{Q}_0 \big\vert_{V\varphi^2}-\mathcal{E}_{\mathcal{S}} - \Big( \mathcal{E}_{\mathcal{T}\mathcal{S}} 
 + {\rm h.c.}\Big). 
$$
All these terms can be bounded by Lemma \ref{lem:4aaa} similarly to the error term $\mathcal{E}_{\rm corr}$ in  \eqref{eq: eps-1}, using $\|V\varphi^2\|_1 \le \|V\|_1 \|\varphi\|_\infty^2\le C$ from Lemma \ref{lem: bounds phi}. This yields \eqref{evbound}. 
\end{proof}


\section{Completion of the square for $Vf$} \label{sec: completion square Vf}

In this section, we prove Proposition \ref{pro: completion square Vf}. We begin by explaining the main ingredients of the proof in the first subsection, and then prove the key lemmas in the remaining subsections.

\subsection{Proof outline of Proposition \ref{pro: completion square Vf}}\label{sec: lemmas completion square Vf 1}

To prove Proposition \ref{pro: completion square Vf}, we will expand the square 
in
\begin{align}\label{eq:ttH0}
 \sum_{\sigma\in \{\uparrow, \downarrow\}}\sum_{r\in \Lambda^*} ||r|^2 - (k_F^\sigma)^2| \left| \hat{a}_{r,\sigma} + T_{\sigma}(r) + S_{\sigma}(r)\right|^2
\end{align}
and estimate each term separately. We start with two key bounds on the cross terms involving $T^*_{\sigma}(r)\hat{a}_{r,\sigma}$ and $S^*_{\sigma}(r)\hat{a}_{r,\sigma}$, which  justify  \eqref{eq: T*a-intro} and \eqref{eq: S*a-intro}. 

\begin{lemma}[Analysis of $T^*a$]\label{lem: Ta} Let $\delta>1/3$ and let $T_\sigma^\ast(r)$ be as in Definition \ref{def:ren-Q2}. We have 
\begin{equation}\label{eq: T*a}
  \sum_{\sigma\in \{\uparrow, \downarrow\}}\sum_{r\in \Lambda^*} ||r|^2 - (k_F^\sigma)^2|T^\ast_{\sigma}(r)\hat{a}_{r,\sigma} + \mathrm{h.c.} = \mathbb{Q}_2\vert_{Vf}+ \mathcal{E}_{T } 
\end{equation}
where $\mathcal{E}_{T }$ satisfies that for all  $0<\gamma<1/6$, $\kappa>0$ and $\psi\in\mathcal{F}_{\mathrm{f}}$ with $\|\psi\|=1$, 
\begin{align*}
  |\langle \psi, \mathcal{E}_{T}\psi\rangle| &\leq CL^{\frac{3}{2}}\rho^{\frac{7}{6} + \frac{3}{2}\gamma +\delta-\kappa} \|\mathbb{H}_0^{\frac{1}{2}}\psi\| + C L^{\frac{3}{2}} \rho^{\frac{3}{2} - \frac{\gamma}{2} + \delta -\kappa}\|\mathcal{N}^{\frac{1}{2}}\psi\| 
  \\
  &\quad + \mathfrak{e}_L   L^{3/2} \|\bQ_4^{1/2}\psi\|. 
\end{align*}
\end{lemma}

\begin{lemma}[Analysis of $S^*a$]\label{lem: Sa} Let $0<\gamma<1/6<\alpha<1/3$ and let $S_\sigma^\ast(r)$ be as in Definition \ref{def:ren-Q3}. We have 
\begin{equation}\label{eq: S*a}
  \sum_{\sigma\in\{\uparrow, \downarrow\}}\sum_{r\in\Lambda^\ast} ||r|^2 - (k_F^\sigma)^2|S^\ast_{\sigma}(r)\hat{a}_{r,\sigma} + \mathrm{h.c.} = \mathbb{Q}_3\vert_{Vf} + \mathcal{E}_{S},
\end{equation}
where $\mathcal{E}_{S}$ satisfies that for all $0\le \beta<1$ and $\psi\in\mathcal{F}_{\mathrm{f}}$  with  $\|\psi\|=1$,  
\begin{align*}
  |\langle \psi, \mathcal{E}_{S}\psi\rangle| &\leq C \rho^{\frac{1}{3} + \frac{4}{5}\gamma -\frac{3}{10}\alpha}   \langle \psi, \mathbb{H}_0 \psi\rangle +  C(  \rho^{1+ \frac{4}{5}\gamma - \frac{3}{10}\alpha} + \rho^{1+\frac{\beta}{6}}) \langle \psi, \mathcal{N}\psi\rangle \nn \\
  &\quad + C\rho^{\frac{1}{6} + \alpha}\|\mathbb{Q}_4^{\frac{1}{2}} \psi\|\|\mathbb{H}_0^{\frac{1}{2}}  \psi\|+ C\rho^{1-\frac{3}{2}\gamma}\|\mathcal{N}^{\frac{1}{2}}\psi\| \|\mathcal{N}^{\frac{1}{2}}_\beta\psi\|
 +  \mathfrak{e}_L L^{3/2} \|\bQ_4 ^{1/2}\psi\|.
 \end{align*}
\end{lemma}

Next, we consider the terms involving $|T_\sigma(r)|^2=T^\ast_\sigma(r)T_\sigma(r)$. We will replace $|T_\sigma(r)|^2$ with the anti-commutator $\{T^\ast_\sigma(r), T_\sigma(r)\}=T^\ast_\sigma(r)T_\sigma(r)+ T_\sigma(r) T^*_\sigma(r)$ which is larger but can be bounded more conveniently via normal ordering. Due to the canonical anticommutation relations, the anti-commutator of two operators involving an odd number of creation and annihilation operators typically results in terms of lower order in the number of creation and annihilation operators, making them easier to handle. { {The idea of using the anti-commutator goes back to Bell's work \cite{Bell}, and has been used recently in \cite{Christiansen-24,ChrHaiNam-23,ChrHaiNam-24}.

\begin{lemma}[Analysis of $\{T^*,T\}$]\label{lem: T} Let $\delta>1/3$ and let $T^\ast_{\sigma}(r)$ be as in Definition \ref{def:ren-Q2}. We have
\begin{align}\nn
&  \sum_{\sigma\in \{\uparrow, \downarrow\}}\sum_{r\in\Lambda^*} ||r|^2 - (k_F^\sigma)^2| \{T^\ast_\sigma(r), T_\sigma(r)\} \le 
- 2 \|\nabla \varphi\|_{L^2}^2  \sum_{\sigma\neq \sigma^\prime}\rho_\sigma \sum_{r\in \Lambda^* } \hat{v}_{\sigma^\prime}(r)\hat{a}_{r,\sigma^\prime}^\ast \hat{a}_{r,\sigma^\prime}  
\\
&\quad   + \frac{1}{L^6}\sum_{p,r,r^\prime \in \Lambda^*}\frac{(2|p|^2 \hat{\varphi}(p))^2 \hat{u}_{\uparrow}(r+p)\hat{u}_{\downarrow}(r^\prime - p)\hat{v}_\uparrow(r)\hat{v}_{\downarrow}(r^\prime) }{\lambda_{p,r}+ \lambda_{-p,r^\prime} + 2\varepsilon} + \mathcal{E}_{T^\ast T} , \label{c1}
\end{align}
where $\mathcal{E}_{T^\ast T}$ satisfies that for all $0<\gamma<1/6$, $0\le \beta<1$, $\kappa>0$ and $\psi\in\mathcal{F}_{\mathrm{f}}$, 
\begin{align*}
  |\langle \psi, \mathcal{E}_{T^\ast T}\psi\rangle| &\leq  C(\rho^{1+ \gamma} +\rho^{1+\frac{\beta}{6}} + \rho^{\frac{4}{3} -2\gamma - \kappa}  ) \langle \psi, \mathcal{N}\psi\rangle \\
  &\qquad  + C\rho  \|\mathcal{N}_\beta^{\frac{1}{2}}\psi\|  \|\mathcal{N}^{\frac{1}{2}}\psi\|  + C\rho^{\frac{2}{3} + \gamma}\langle \psi, \mathbb{H}_0 \psi\rangle.\end{align*}
\end{lemma}

It remains to deal with all terms involving  $|S_\sigma(r)|^2$ and $S^*_\sigma(r)T_\sigma(r)$. These terms do not contain any constant after normal ordering, but we have to isolate correctly a quadratic  contribution (of order  $\rho \cN$) that  cancels the corresponding term in Lemma~\ref{lem: T}. Unlike the analysis of $|T_\sigma(r)|^2$, we cannot bound $|S_\sigma(r)|^2$ by the anti-commutator $\{S^\ast_\sigma(r), S_\sigma(r)\}$ since the latter is too large and would give extra constant and quadratic terms. 
We shall find it convenient to split $S_\sigma(r)$ into two parts
\begin{equation}\label{eq: def S1 S2}
S_\sigma^\ast(r)  = S_{1,\sigma}^\ast(r)+S_{2,\sigma}^\ast(r)
\end{equation}
with 
\begin{align}\label{eq: defS1}
  S^\ast_{1,\sigma}(r)&= 
    \frac{1}{L^3}\sum_{p}\hat{\varphi}^>(p) (D^{<\alpha}_{p,\sigma^\prime})^\ast  \hat{a}_{p-r,\sigma}\left(\hat{v}_\sigma(r-p)\hat{u}_\sigma(r) -\hat{u}_\sigma(p-r)\hat{v}_\sigma(r) \right)\nonumber
    \\
      &\quad - \frac{1}{L^3}\sum_{p}\hat{\varphi}^>(p)\hat{u}_\sigma(r)\hat{{u}}^{<\alpha}_\sigma(r+p)\hat{a}^\ast_{r+p,\sigma}b_{p,\sigma^\prime},
\end{align}
and 
\begin{equation}\label{def:S2}
S_{2,\sigma}^\ast(r) =   \frac{1}{L^3}\sum_{p}\hat{\varphi}(p)\widehat{\chi}_>(p) \hat{{u}}^{<\alpha}_\sigma(r) \hat{u}_\sigma(r-p)\hat{a}_{r-p,\sigma}^\ast b_{p,\sigma^\prime}^\ast 
\end{equation}
corresponding to the last term in \eqref{eq: def Sk}.
{{We will estimate separately the two anti-commutators $\{S^*_{1,\sigma}(r), S_{1,\sigma}(r)\}$ and $\{T^*_{\sigma}(r), S_{1,\sigma}(r)\}$, and the product $S^\ast_{2,\sigma}(r) (T_\sigma(r) + S_\sigma(r))$. From the first anti-commutator  we extract the correct 
contribution that cancels the corresponding quantity in Lemma \ref{lem: T}. 
}}

\begin{lemma}[Analysis of $\{S_1^*,S_1\}$]\label{lem: S1} Let $0<\gamma <1/6< \alpha< 1/3$, and let $S_{1,\sigma}(r)$ be as in \eqref{eq: defS1}. We have
\begin{multline}\label{eq: S1*S1}
\sum_{\sigma\in \{\uparrow, \downarrow\}}\sum_{r\in \Lambda^*} ||r|^2 - (k_F^\sigma)^2| \{S^\ast_{1,\sigma}(r), S_{1,\sigma}(r)\}
\\
 \le 2  \|\nabla \varphi\|_2^2 \sum_{\sigma\neq \sigma^\prime}\rho_\sigma \sum_{r\in \Lambda^*} \hat{u}_{\sigma'}(r)\hat{a}_{r, \sigma^\prime}^\ast \hat{a}_{r,\sigma^\prime} + \mathcal{E}_{S^\ast_1 S_1}
\end{multline}
where $\mathcal{E}_{S^\ast_1 S_1}$ satisfies that for all $0\le \beta<1$ and $\psi\in\mathcal{F}_{\mathrm{f}}$, 
\begin{align*}
  |\langle \psi, \mathcal{E}_{S^\ast_1 S_1}\psi\rangle| &\leq C\rho \|\mathcal{N}^{\frac{1}{2}}_\beta\psi\|\|\mathcal{N}^{\frac{1}{2}}\psi\|  + C\rho^{\frac{2}{3} + 6\gamma - 5\alpha}\langle \psi, \mathbb{H}_0 \psi\rangle 
  \\
   &\quad + C\left(\rho^{1+\gamma} + \rho^{1+\frac{\beta}{6}} +  \rho^{\frac{4}{3} + 2\gamma - 3\alpha}\right)\langle \psi, \mathcal{N}\psi\rangle . \nn
  \end{align*}
\end{lemma}

Finally, we show that all contributions involving the anti-commutator $\{T^\ast_{\sigma}(r), S_{1,\sigma}(r)\}$ and the product  $S^\ast_{2,\sigma}(r) (T_\sigma(r) + S_\sigma(r))$ are negligible. 

\begin{lemma}[Analysis of $\{T^\ast,S_{1}\}$]\label{lem: TS1} Let $0<\gamma <1/6< \alpha< 1/3<\delta$. Let $T_\sigma(r)$ and $S_{1,\sigma}(r)$ be as in Definition \ref{def:ren-Q2} and \eqref{eq: def S1 S2}, respectively. For all $0\le \beta<1$, $\kappa>0$ and $\psi\in\mathcal{F}_{\mathrm{f}}$, we have 
\begin{align}\label{eq: TS1}
  &\sum_{\sigma\in\{\uparrow, \downarrow\}} \sum_{r\in \Lambda^*}  ||r|^2 - (k_F^\sigma)^2| |\langle \psi, \{T_\sigma^\ast(r),S_{1,\sigma}(r)\}\psi\rangle|
  \\
  &\quad \leq  C\left(\rho^{1+\frac{7}{4}\gamma  - \frac{3}{2}\alpha -\kappa} + \rho^{1+\frac{9}{2}\gamma - \frac{7}{2}\alpha -\kappa} + \rho^{1-\alpha - \kappa}\right)\|\mathbb{H}_0^{\frac{1}{2}}\psi\|\|\mathcal{N}^{\frac{1}{2}}\psi\|\nonumber
  \\
  &\quad\quad +C\left(\rho^{1+\gamma}  + \rho^{1+\frac{\beta}{6}} + \rho^{\frac{4}{3} + \frac{3}{4}\gamma - \frac{3}{2}\alpha - \kappa} +  \rho^{\frac{4}{3}-2\gamma -\kappa} + \rho^{\frac{4}{3} -\alpha - \kappa} \right)\langle \psi, \mathcal{N}\psi\rangle. \nn
  \\
  &\quad\quad + C\ \rho^{\frac{2}{3} +\frac{11}{2} \gamma - \frac{7}{2}\alpha}\langle \psi, \mathbb{H}_0 \psi\rangle + \rho \|\mathcal{N}^{\frac{1}{2}}\psi\|\|\mathcal{N}^{\frac{1}{2}}_\beta\psi\|\nn .
\end{align}
\end{lemma}

\begin{lemma}[Analysis of $S_2^*(T+S)$]\label{lem: S2} Let $0<\gamma <1/6< \alpha< 1/3<\delta$ such that  \eqref{eq: mix cond alpha gamma} holds.
Let $T_\sigma^\ast(r)$, $S_\sigma^\ast(r)$ be as in Definitions \ref{def:ren-Q2}, \ref{def:ren-Q3} and let $S_{2,\sigma}(r)$ be as in \eqref{def:S2}. For all $\kappa>0$ and $\psi \in \mathcal{F}_{\mathrm{f}}$ with $\|\psi\|=1$, we have
\begin{align}\label{eq: S2}
    &\sum_{\sigma\in \{\uparrow, \downarrow\}}\sum_{r\in \Lambda^*}  ||r|^2 - (k_F^\sigma)^2| \langle\psi, S^\ast_{2,\sigma}(r) (T_\sigma(r) + S_\sigma(r)) \psi\rangle|
  \nn \\
    &\leq C \left(\rho^{\frac{3}{2} -\gamma - \kappa}  
    + \rho^{\frac{3}{2} + \frac{21}{10}\gamma- \frac{21}{10}\alpha - \kappa} \right) L^{\frac{3}{2}} (\|\mathbb{H}_0^{\frac{1}{2}}\psi\| + \rho^{\frac{1}{3}}\|\mathcal{N}^{\frac{1}{2}}\psi\|) \nn
  \\
  &\quad 
 + C\rho^{\frac{5}{6} +\frac{13}{2}\gamma - 6\alpha} \langle \psi, \mathbb{H}_0 \psi\rangle  +  C \left(\rho^{\frac{7}{6} +\frac{\gamma}{2} - \alpha} + \rho^{\frac{4}{3} + \gamma - 2\alpha}\right)\langle \psi, \mathcal{N}\psi\rangle . 
\end{align}
\end{lemma}

Proposition \ref{pro: completion square Vf} follows easily from the above lemmas. 

\begin{proof}[Proof of Proposition \ref{pro: completion square Vf}] 
We start by expanding the square as 
\begin{align}\label{eq:com-Vf-basic-1}
 \left| \hat{a}_{r,\sigma} + T_{\sigma}(r) + S_{\sigma}(r)\right|^2 &= |\hat{a}_{r,\sigma}|^2 + (T^*_{\sigma}(r)\hat{a}_{r,\sigma} + {\rm h.c.}) +  (S^*_{\sigma}(r)\hat{a}_{r,\sigma} + {\rm h.c.})  \nn\\
 &\qquad+   |T_{\sigma}(r)+S_{\sigma}(r)|^2. 
\end{align}
Moreover, using \eqref{eq: def S1 S2} we have 
\begin{align}\label{eq:com-Vf-basic-2}
&|T_{\sigma}(r)+S_{\sigma}(r)|^2 =  | T_{\sigma}(r)+S_{1,\sigma}(r)|^2 + |S_{2,\sigma}(r)|^2 + \Big( S^*_{2,\sigma}(r) (T_{\sigma}(r)+S_{1,\sigma}(r) ) + {\rm h.c.} \Big)\nn\\
&= \Big\{ T^*_{\sigma}(r)+S^*_{1,\sigma}(r), T_{\sigma}(r)+S_{1,\sigma}(r)\Big\} +  \Big( S^*_{2,\sigma}(r) (T_{\sigma}(r)+S_{\sigma}(r) ) + {\rm h.c.} \Big)\nn\\
&\quad - |T^*_{\sigma}(r)+S^*_{1,\sigma}(r) |^2 -  |S_{2,\sigma}(r)|^2 .
\end{align}
The last two terms in \eqref{eq:com-Vf-basic-2} can be dropped for an upper bound. Therefore, we deduce from \eqref{eq:com-Vf-basic-1} that 
\begin{align*}  
& \left| \hat{a}_{r,\sigma} + T_{\sigma}(r) + S_{\sigma}(r)\right|^2 \le |\hat{a}_{r,\sigma}|^2 + \Big(T^*_{\sigma}(r)\hat{a}_{r,\sigma} + S^*_{\sigma}(r)\hat{a}_{r,\sigma} + {\rm h.c.}\Big) +  \Big\{ T^*_{\sigma}(r),T_{\sigma}(r)\Big\} \nn\\
 &\quad +\Big\{ S^*_{1,\sigma}(r),S_{1,\sigma}(r)\Big\}  + \left(\Big\{ T^*_{\sigma}(r),S_{1,\sigma}(r)\Big\} + S^*_{2,\sigma}(r) (T_{\sigma}(r)+S_{\sigma}(r) ) +  \mathrm{h.c.}\right).
\end{align*}
By inserting the latter bound into \eqref{eq:ttH0} and estimating the resulting terms using Lemmas~\ref{lem: Ta}--\ref{lem: S2}, we obtain Proposition \ref{pro: completion square Vf}. Here we also used the particle-hole relation \eqref{eq:particle-hole} to cancel the quadratic terms on the right-hand side of Eqs.~\eqref{c1} and~\eqref{eq: S1*S1}.
\end{proof}

It remains to prove the key lemmas above in the next subsections.

\subsection{Proof of Lemma \ref{lem: Ta}}

In this subsection we prove Lemma \ref{lem: Ta}. A key ingredient will be the  zero-energy scattering equation \eqref{eq: zero en scatt eq}. The proof of Lemma \ref{lem: Ta} would be substantially shorter if we had not introduced the $\eps$  in \eqref{eq: def omega eps} in the definition of $T_\sigma(r)$.  
Estimating the extra error term due to $\eps$ requires some work here, but including $\eps$ in the definition \eqref{eq: def omega eps} is important in order to avoid a singularity at the Fermi surface.

\begin{proof} Using the definition of $T_\sigma(r)$ and changing the variables $r \notin B_F^{\sigma}\mapsto r+p$  {and  $r \in B_F^{\sigma}\mapsto -r$}, we have
\begin{align}
 & \sum_\sigma\sum_r ||r|^2 - (k_F^\sigma)^2|T^\ast_{\sigma}(r)\hat{a}_{r,\sigma} \nn\\
  &= \frac{1}{L^3}\sum_{p,r,r^\prime} (|r+p|^2 - |r|^2 + |r^\prime -p|^2 - |r^\prime|^2)\widehat{\omega}^\varepsilon_{r,r^\prime}(p)b_{p,r,\uparrow}b_{-p,r^\prime, \downarrow}
  \nn\\
  &= \frac{1}{L^3}\sum_{p,r,r^\prime} 2|p|^2 \hat{\varphi}(p)b_{p,r,\uparrow}b_{-p,r^\prime, \downarrow} -  \frac{2\varepsilon}{L^3}\sum_{p,r,r^\prime} \widehat{\omega}^\varepsilon_{r,r^\prime}(p)b_{p,r,\uparrow}b_{-p,r^\prime, \downarrow}. \label{eq:proof-lem-5.1}
\end{align}

For the first term on the right-hand side of \eqref{eq:proof-lem-5.1}, we can write 
\begin{align}\label{eq:proof-lem-5.1-a}
  &\frac{1}{L^3}\sum_{p,r,r^\prime} 2|p|^2 \hat{\varphi}(p)b_{p,r,\uparrow}b_{-p,r^\prime, \downarrow}+ \mathrm{h.c.} \\
  &= \int dxdy\, (-2\Delta\varphi) (x-y)a_\uparrow(u_x)a_\uparrow(v_x)a_\downarrow(u_y)a_\downarrow(v_y) + \mathrm{h.c} \nn\\
  &=   \mathbb{Q}_2 \vert_{Vf} + \left( \int dxdy \left( -2\Delta \varphi - Vf +  8\pi a L^{-3}  \right) (x-y)a_\uparrow(u_x)a_\uparrow(v_x)a_\downarrow(u_y)a_\downarrow(v_y) + \mathrm{h.c} \right)\nn
  \end{align}
where we inserted the constant $8\pi a L^{-3}$ making use of  the orthogonality $u_\uparrow v_\uparrow =0$ via 
\begin{align}\label{eq: auav}
\int_{\Lambda} dx\, a_\uparrow(u_x)a_\uparrow(v_x) =0.
\end{align}
Using \eqref{eq: periodic-scattering-equation} and the Cauchy--Schwarz inequality, we have 
\begin{align}\label{eq:proof-lem-5.1-aa}
&\left| \left\langle \psi,  \int dxdy \left( -2\Delta \varphi - Vf +  8\pi a L^{-3}  \right)   (x-y)a_\uparrow(u_x)a_\uparrow(v_x)a_\downarrow(u_y)a_\downarrow(v_y) \psi\right\rangle \right| \nn \\
& \le \mathfrak{e}_L   \int dxdy\, V (x-y) \| a_\uparrow(u_x) a_\downarrow(u_y)  \psi\| \le \mathfrak{e}_L   L^{3/2} \|\bQ_4^{1/2}\psi\|. 
\end{align}
To bound the last operator on the right-hand side of \eqref{eq:proof-lem-5.1}, we distinguish between high and low momenta by splitting  $
\varphi = \varphi^>  + \varphi^<
$ as in \eqref{eq: def phi><} and writing correspondingly  
\[
 \frac{2\varepsilon}{L^3}\sum_{p,r,r^\prime}   \frac{2|p|^2 ( \hat{\varphi}^>(p)+ \hat{\varphi}^<(p)  )}{\lambda_{p,r}+ \lambda_{-p,r^\prime} + 2\varepsilon} b_{p,r,\uparrow}b_{-p,r^\prime, \downarrow} = \mathrm{I}^> + \mathrm{I}^<.
\]
We start with estimating $\mathrm{I}^>$. Recall from \eqref{eq: def bp bpr} that $b_{p,r,\uparrow}  = \hat{a}_{r+p,\uparrow}\hat{a}_{-r,\uparrow}\hat{u}_\uparrow(r+p)\hat{v}_\uparrow(r)$ and therefore in $\mathrm{I}^>$, using the constraints on $p$ and $r$, we can replace $\hat{u}_\uparrow(r+p)$ by $\hat{u}^>_\uparrow(r+p)$ defined in \eqref{eq: u< u> gamma}. We can do similarly for $b_{-p,r',\downarrow}$. Moreover, writing 
\begin{equation}\label{eq: comparison w eps with phi}
\frac{2|p|^2\hat{\varphi}^>(p)}{\lambda_{p,r} + \lambda_{-p, r^\prime} +2\varepsilon} = \hat{\varphi}^>(p) - \frac{2p\cdot (r-r^\prime)\hat{\varphi}^>(p)}{\lambda_{p,r} + \lambda_{-p, r^\prime}+ 2\varepsilon}  - \frac{2\varepsilon \hat{\varphi}^>(p)}{\lambda_{p,r} + \lambda_{-p,r^\prime} + 2\varepsilon},
\end{equation}
we have 
\begin{equation}\label{eq: def I>1 I>2 I>3}
  \mathrm{I}^> = \mathrm{I}_{a}^> +\mathrm{I}_{b}^> + \mathrm{I}_{c}^>
\end{equation}
accordingly. The term $\mathrm{I}_{a}^>$ is 
\begin{align*}
  \mathrm{I}_{a}^> &= 2\varepsilon\int dxdy\, {\varphi}^>(x-y) a_{\uparrow}(u^>_x)a_\uparrow(v_x)a_\downarrow(u^>_y)a_\downarrow(v_y) \\
  &= 2\varepsilon\int dx\, a_{\uparrow}(u^>_x)a_\uparrow(v_x)b_\downarrow(\varphi^>_x,v_\downarrow,u^>_\downarrow)
\end{align*}
where we used the notation introduced in \eqref{eq: def b general}.  
The  bound \eqref{lb6} in Lemma~\ref{lem:2a}, together with $\|\Delta\varphi^>\|_1 + \rho^{\frac{1}{3}}\|\nabla\varphi^>\|_1 + \rho^{\frac{2}{3}}\|\varphi^>\|_1 \le C$ from Lemma~\ref{lem: bounds phi}, implies that $\| b_\downarrow(\varphi^>_x,v_\downarrow,u^>_\downarrow)\| \leq C \rho^{1/3+\gamma/2}$ uniformly in $x$. In combination with $\eps=\rho^{2/3+\delta}$,  $\|v_\uparrow\|_2 \leq C \rho^{1/2}$ and Lemma \ref{lem:1a}, an application of 
 the Cauchy--Schwarz inequality yields
\begin{align}\label{eq:proof-lem-5.1-b1}
  |\langle \psi, \mathrm{I}_{a}^>\psi\rangle | &\leq 2\eps \int dx \|b_\downarrow({\varphi}_x^>,v_\downarrow,u^>_\downarrow)\|\|a_\uparrow(v_x)\| \|a_\uparrow(u^>_x)\psi\| \leq 
C L^{\frac{3}{2}} \rho^{\frac{7}{6} + \frac{3}{2}\gamma + \delta}\|\mathbb{H}_0^{\frac{1}{2}}\psi\|.
\end{align}

Next, we consider $\mathrm{I}_{b}^>$. It is convenient to decompose  
\begin{equation}\label{eq: redistribute momenta}
p\cdot (r-r^\prime) = -(r^\prime -p)\cdot (r-r^\prime)  + r^\prime \cdot (r-r^\prime).
\end{equation}
Redistributing the momenta this way, we will have fewer derivatives of $\varphi$ to consider (which would  require stronger regularity assumptions on the interaction for their treatment). Accordingly, we split $\mathrm{I}_{b}^> = \mathrm{I}_{b;1}^> + \mathrm{I}_{b;2}^>$. 
 
For $\mathrm{I}_{b;1}^>$, we can multiply and divide by $|r^\prime - p|^2$ and use that $|r^\prime - p|^2 = |p|^2 + |r^\prime|^2 - 2p\cdot r^\prime$ for the numerator. Using \eqref{eq: int t conf space} and \eqref{eq: def u<>t} we can write $\mathrm{I}_{b;1}^>$ in configuration space as
\begin{align*}  
 \mathrm{I}_{b;1}^>  &= 4\varepsilon \sum_{\ell=1}^3 \int_0^{\infty}dt\,e^{-2t\varepsilon}\int dxdy\, \bigg(\Delta{\varphi}^>(x-y)a_\downarrow(v_{t,y})  + \varphi^>(x-y)a_\downarrow(\Delta v_{t,y})\nn\\
  &\qquad \qquad \qquad\qquad +2 \sum_{m=1}^3\partial_m\varphi^>(x-y)a_\downarrow(\partial_mv_{t,y})\bigg) a_\uparrow (u^>_{t,x})a_\uparrow(\partial_\ell v_{t,x})a_\downarrow(\partial_\ell\widetilde{u}^>_{t,y})\nonumber
  \\
  &\quad -4\varepsilon \sum_{\ell=1}^3 \int_0^{\infty}\hspace{-0.2cm} dt\,  e^{-2t\varepsilon}\hspace{-0.1cm}\int dxdy\, \bigg(\Delta{\varphi}^>(x-y)a_\downarrow(\partial_\ell v_{t,y})  + \sum_{m=1}^3\varphi^>(x-y)a_\downarrow(\partial_m^2 \partial_\ell v_{t,y}) \nn
  \\
  &\qquad \qquad \qquad\qquad +2 \sum_{m=1}^3\partial_m\varphi^>(x-y)a_\downarrow(\partial_m \partial_\ell v_{t,y})\bigg)a_\uparrow (u^>_{t,x})a_\uparrow( v_{t,x})a_\downarrow(\partial_\ell\widetilde{u}^>_{t,y}),
\end{align*}
where $\widetilde{u}^>_{t,\sigma}$ is defined in Lemma \ref{lem:t}.  
With the aid of  \eqref{eq:Pauli} we find 
\begin{align*}
  |\langle \psi, \mathrm{I}_{b;1}^>\psi\rangle| &\leq C\eps   \int_{0}^{\infty} dt\, e^{-2t\varepsilon}  \|{v}_{t,\uparrow}\|_2  \|{v}_{t,\downarrow}\|_2 \||\cdot|^{-1} \hat{u}^>_{t,\downarrow}\|_2   \\
&\quad  \times \int dxdy \left(\rho^{\frac{1}{3}} |\Delta\varphi^>(x-y)| + \rho^{\frac{2}{3}}|\nabla\varphi^>(x-y)| + \rho |\varphi^>(x-y)|\right) \|a_\uparrow(u^>_{t,x})\psi\|\\
&\leq C\rho^{3 + \delta}  \ \int_0^\infty dt\, e^{-2t\varepsilon} \|{v}_{t,\downarrow}\|_2 \| |\cdot|^{-1} \hat{u}^>_{t, \downarrow}\|_2
 e^{t (k_F^{\uparrow})^2 }  \int dx\, \|a_\uparrow(u^>_{t,x})\psi\| .
 \end{align*}
In the last estimate we used again the bounds in Lemma \ref{lem: bounds phi}. Lemma  \ref{lem:1a}, the Cauchy--Schwarz inequality and \eqref{eq: est int t vu> final}  then yield 
$$
 |\langle \psi, \mathrm{I}_{b;1}^>\psi\rangle| \leq CL^{\frac{3}{2}}\rho^{\frac{7}{6} + \frac{3}{2}\gamma  +\delta-\kappa}\|\mathbb{H}_0^{\frac{1}{2}}\psi\|.
 $$

In a similar way, we can write  
\begin{align*} 
\mathrm{I}_{b;2}^> &= 4\varepsilon \sum_{\ell=1}^3 \int_0^\infty dt  \int dxdy\, \bigg(\Delta{\varphi}^>(x-y)a_\downarrow(\partial_\ell v_{t,y})  +\sum_{m=1}^3 \varphi^>(x-y)a_\downarrow(\partial_m^2\partial_\ell v_{t,y}) \\
  &\qquad\qquad\qquad\qquad +2\sum_{m=1}^3\partial_m\varphi^>(x-y)a_\downarrow(\partial_\ell \partial_m v_{t,y})\bigg) a_\uparrow (u^>_{t,x})a_\uparrow(\partial_\ell v_{t,x})a_\downarrow(\widetilde{u}^>_{t,y})\nonumber
  \\
  &\quad - 4\varepsilon \sum_{\ell=1}^3 \int_0^\infty dt   \int dxdy\, \bigg(\Delta{\varphi}^>(x-y)a_\downarrow(\partial_\ell^2 v_{t,y})  + \sum_{m=1}^3\varphi^>(x-y)a_\downarrow(\partial_m^2 \partial_\ell^2 v_{t,y}) \\
  &\qquad\qquad\qquad\qquad +2\sum_{m=1}^3\partial_m\varphi^>(x-y)a_\downarrow(\partial_m \partial_\ell^2 v_{t,y})\bigg) a_\uparrow (u^>_{t,x})a_\uparrow(v_{t,x})a_\downarrow(\widetilde{u}^>_{t,y})
\end{align*}
Proceeding as  above, but using  \eqref{eq: est int t vu> final-b} instead of \eqref{eq: est int t vu> final}, we obtain
$$
  |\langle \psi, \mathrm{I}_{b;2}^>\psi\rangle| \leq  CL^{\frac{3}{2}}\rho^{\frac{7}{6} + \frac{5}{2}\gamma  +\delta-\kappa}\|\mathbb{H}_0^{\frac{1}{2}}\psi\|.
$$

It remains to bound $\mathrm{I}_{c}^>$ in \eqref{eq: def I>1 I>2 I>3}. As for $\mathrm{I}_{b}^>$ above,   
we split it into three terms by multiplying and dividing by $|r^\prime -p|^2$. More precisely, using that $|r^\prime - p|^2 = |p|^2 +|r^\prime|^2 - 2p\cdot r^\prime$, we write 
\begin{align*}
  \mathrm{I}^>_c &= -\frac{4\varepsilon^2}{L^3}\sum_{p,r,r^\prime}   \frac{ \hat{\varphi}^>(p)}{\lambda_{p,r}+ \lambda_{-p,r^\prime} + 2\varepsilon} b_{p,r,\uparrow}b_{-p,r^\prime, \downarrow} 
  \\
  &= -\frac{4\varepsilon^2}{L^3}\sum_{p,r,r^\prime}   \frac{ \hat{\varphi}^>(p)}{\lambda_{p,r}+ \lambda_{-p,r^\prime} + 2\varepsilon} \hat{u}_\uparrow(r+p) \hat{u}_\downarrow(r^\prime - p) \hat{v}_\uparrow(r)\hat{v}_\downarrow(r^\prime) \hat{a}_{r+p,\uparrow}\hat{a}_{-r,\uparrow}\hat{a}_{r^\prime - p,\downarrow}\hat{a}_{-r^\prime, \downarrow}
  \\
  &= -\frac{4\varepsilon^2}{L^3}\sum_{p,r,r^\prime}   \frac{|p|^2 \hat{\varphi}^>(p)\hat{v}_\downarrow(r^\prime)}{\lambda_{p,r}+ \lambda_{-p,r^\prime} + 2\varepsilon} \hat{u}_\uparrow(r+p) \frac{\hat{u}_\downarrow(r^\prime - p)}{|r^\prime-p|^2} \hat{v}_\uparrow(r) \hat{a}_{r+p,\uparrow}\hat{a}_{-r,\uparrow}\hat{a}_{r^\prime - p,\downarrow}\hat{a}_{-r^\prime, \downarrow}
  \\
  &\quad -\frac{4\varepsilon^2}{L^3}\sum_{p,r,r^\prime}   \frac{|r^\prime|^2 \hat{\varphi}^>(p)\hat{v}_\downarrow(r^\prime)}{\lambda_{p,r}+ \lambda_{-p,r^\prime} + 2\varepsilon} \hat{u}_\uparrow(r+p) \frac{\hat{u}_\downarrow(r^\prime - p)}{|r^\prime -p|^2} \hat{v}_\uparrow(r) \hat{a}_{r+p,\uparrow}\hat{a}_{-r,\uparrow}\hat{a}_{r^\prime - p,\downarrow}\hat{a}_{-r^\prime, \downarrow}
  \\
  &\quad +\frac{8\varepsilon^2}{L^3}\sum_{p,r,r^\prime}   \frac{ r^\prime \cdot p \,\hat{\varphi}^>(p)\hat{v}_\downarrow(r^\prime)}{\lambda_{p,r}+ \lambda_{-p,r^\prime} + 2\varepsilon} \hat{u}_\uparrow(r+p) \frac{\hat{u}_\downarrow(r^\prime - p)}{|r^\prime - p|^2} \hat{v}_\uparrow(r) \hat{a}_{r+p,\uparrow}\hat{a}_{-r,\uparrow}\hat{a}_{r^\prime - p,\downarrow}\hat{a}_{-r^\prime, \downarrow}\,.
\end{align*}
In configuration space, this gives 
\begin{align*}  
  \mathrm{I}_{c}^> &= 4\varepsilon^2  \int_0^{\infty}\hspace{-0.2cm} dt\,e^{-2t\varepsilon}\hspace{-0.1cm}\int dxdy\, \bigg(\Delta{\varphi}^>(x-y)a_\downarrow(v_{t,y})  + \varphi^>(x-y)a_\downarrow(\Delta v_{t,y}) \\
  &\qquad \qquad \qquad +2 \sum_{m=1}^3\partial_m\varphi^>(x-y)a_\downarrow(\partial_m v_{t,y})\bigg) a_\uparrow (u^>_{t,x})a_\uparrow(v_{t,x})a_\downarrow(\widetilde{u}^>_{t,y}).
\end{align*}
Again proceeding  as above we find
$$
  |\langle \psi, \mathrm{I}_{c}^>\psi\rangle| \leq CL^{\frac{3}{2}}\rho^{\frac{7}{6} +\frac{5}{2}\gamma + 2\delta -\kappa}\|\mathbb{H}_0^{\frac{1}{2}}\psi\|.
$$

Combining  the bounds for $\mathrm{I}_{a}^>,\mathrm{I}_{b}^>,\mathrm{I}_{c}^>$, we thus find that 
\begin{align}\label{eq:proof-lem-5.1iI>}
  |\langle \psi, \mathrm{I}^> \psi\rangle | \leq CL^{\frac{3}{2}}\rho^{\frac{7}{6} + \frac{3}{2}\gamma +\delta-\kappa}\|\mathbb{H}_0^{\frac{1}{2}}\psi\|.
\end{align}

 It remains to consider $\mathrm{I}^<$. In this case, the estimate can be done as in \cite[Proposition 5.2, Eq. (5.6)]{GHNS}, using $\|\Delta\varphi^<\|_1 \le C$ from Lemma \ref{lem: bounds phi}. The result is  
\begin{align}\label{eq:proof-lem-5.1iI<}
  |\langle \psi, \mathrm{I}^< \psi\rangle | \leq CL^{\frac{3}{2}}\rho^{\frac{3}{2} - \frac{\gamma}{2} + \delta-\kappa}\|\mathcal{N}^{\frac{1}{2}}\psi\|.
\end{align}
The claim of Lemma  \ref{lem: Ta} follows from \eqref{eq:proof-lem-5.1}, \eqref{eq:proof-lem-5.1-a}, \eqref{eq:proof-lem-5.1-aa}, \eqref{eq:proof-lem-5.1iI>} and \eqref{eq:proof-lem-5.1iI<}. 
\end{proof}


\subsection{Proof of Lemma \ref{lem: Sa} } 

Similarly to Lemma \ref{lem: Ta}, the proof of Lemma \ref{lem: Sa} relies on the scattering equation \eqref{eq: zero en scatt eq}. The analysis is complicated by the presence of ${\varphi}^>$ and ${u}_\sigma^{<\alpha}$ instead of ${\varphi}$ and ${u}_\sigma$ in Definition \ref{def:ren-Q3}, but the cut-offs are important for Lemmas \ref{lem: S1}, \ref{lem: TS1}, and \ref{lem: S2} in the following subsections.

\begin{proof} Using the definition \eqref{eq: def Sk} of  $S^\ast_{\sigma}(r)$, rearranging terms and changing variables, we obtain the identity  
$$
\sum_\sigma\sum_r ||r|^2 - (k_F^\sigma)^2| S^\ast_{\sigma}(r)\hat{a}_{r,\sigma} + \mathrm{h.c.}  =  \bQ_3|_{Vf} + \mathrm{I}_a + \mathrm{I}_b + \mathrm{I}_c  + \mathrm{I}_d + \mathrm{I}_e 
$$
with
\begin{align}\nn  
\mathrm{I}_a &= - \frac{2}{L^3}\sum_{\sigma\neq \sigma^\prime}\sum_{p,r,r^\prime}p\cdot r \hat{\varphi}^>(p)\hat{u}_{\sigma^\prime}^{<\alpha}(r^\prime)\hat{u}_{\sigma^\prime}(r^\prime - p)\hat{u}_\sigma(r+p)\hat{v}_\sigma(r)  \nn \\
&\qquad \qquad\qquad \times \hat{a}_{r^\prime, \sigma^\prime}^\ast \hat{a}_{r^\prime - p, \sigma^\prime}\hat{a}_{r+p,\sigma}\hat{a}_{-r,\sigma} + \mathrm{h.c.}\nonumber
  \\
\mathrm{I}_b &=  \frac{2}{L^3}\sum_{\sigma\neq \sigma^\prime}  \sum_{p,r,r^\prime} p\cdot r^\prime\hat{\varphi}^>(p) \hat{u}_{\sigma^\prime}^{<\alpha}(r^\prime) \hat{u}_{\sigma^\prime}(r^\prime-p)\hat{u}_\sigma(r+p)\hat{v}_\sigma(r)\nn\\
&\qquad \qquad\qquad \times \hat{a}^\ast_{r^\prime,\sigma^\prime}\hat{a}_{r^\prime-p,\sigma^\prime} \hat{a}_{r+p,\sigma}\hat{a}_{-r,\sigma} + \mathrm{h.c.},\nn\\
\mathrm{I}_c  &=\frac{2}{L^3}\sum_{\sigma\neq \sigma^\prime}\sum_p|p|^2 \hat{\varphi}^<(p) \left(D^{<\alpha}_{p,\sigma^\prime}\right)^\ast b_{p,\sigma} + \mathrm{h.c.}, \nn \\
\mathrm{I}_d  &=-\frac{1}{L^3}\sum_{\sigma\neq \sigma^\prime} ( 2|p|^2 \hat{\varphi}(p) - \widehat{Vf}(p) ) \left(D^{<\alpha}_{p,\sigma^\prime}\right)^\ast b_{p,\sigma} + \mathrm{h.c.},\nn \\
\mathrm{I}_e  &=-2\sum_{\sigma\neq \sigma^\prime}\int dxdy\, (Vf)  (x-y) a^\ast_{\sigma^\prime}(u^{>\alpha}_{y})a_\sigma(v_x)a_{\sigma^\prime}(u_y) a_\sigma(u_x) + \mathrm{h.c.} \label{eq: err S*a}
\end{align}
where $u_{\sigma'}^{>\alpha} = u_{\sigma'} - u_{\sigma'}^{<\alpha}$. These five terms thus make up the error term $ \mathcal{E}_{S} $, and we will bound them in the following.


For $\mathrm{I}_a$ in  \eqref{eq: err S*a}, using the constraints on $r$ and $p$, i.e., $|r|\leq C\rho^{1/3}$ and $|p| \geq \rho^{1/3-\gamma}$, we can replace $\hat{u}_\sigma(r+p)$ by $\hat{u}^>_\sigma(r+p)$ defined in \eqref{eq: u< u> gamma}. Moreover, it is convenient to decompose $p$ as $p = (p-r^\prime)  + r^\prime$, and we denote the resulting terms by $\mathrm{I}_{a;1}$ and $\mathrm{I}_{a;2}$, respectively. We have
$$
 \mathrm{I}_{a;1} = 2\sum_{\sigma \neq \sigma^\prime}\sum_{\ell=1}^3\int dx \,  a^\ast_{\sigma^\prime}(u_y^{<\alpha}) a_{\sigma^\prime}(\partial_\ell u_y) b_\sigma( \varphi^>_y, \partial_\ell v_\sigma, u_\sigma^>) + \mathrm{h.c.},
 $$
where we used the notation introduced in \eqref{eq: def b general}. 
 Lemmas \ref{lem:2a} and \ref{lem: bounds phi} imply  
 that $\| b_\sigma( \varphi^>_y, \partial_\ell v_\sigma, u_\sigma^>)\|\leq C\rho^{2/3 + \gamma/2}$ uniformly in $y$. Therefore, in combination with the Cauchy--Schwarz inequality and  Lemma \ref{lem:1a}, we get  
\begin{align}\label{eq: est Ia1}
  |\langle \psi, \mathrm{I}_{a;1}\psi\rangle| &\leq C\rho^{\frac{2}{3} + \frac{\gamma}{2}}\sum_{\sigma^\prime}\sum_{\ell = 1}^3\int dx\, \|a_{\sigma^\prime}(\partial_\ell u_x)\|\|a_{\sigma^\prime}({u}^{<\alpha}_x)\psi\| \nn\\
  &\leq C\rho^{\frac{2}{3} + \frac{\gamma}{2}}\|\mathcal{N}^{\frac{1}{2}}\psi\|\|\mathbb{H}_0^{\frac{1}{2}}\psi\| + C\rho^{1+ \frac{\gamma}{2}}\langle \psi,\mathcal{N}\psi\rangle .
\end{align}
The term  
$$
 \mathrm{I}_{a;2} = - 2\sum_{\sigma \neq \sigma^\prime}\sum_{\ell=1}^3\int dx \,  a^\ast_{\sigma^\prime}(\partial_\ell u_y^{<\alpha}) a_{\sigma^\prime}(u_y) b_\sigma( \varphi^>_y, \partial_\ell v_\sigma, u_\sigma^>) + \mathrm{h.c.}
$$
can be estimated in the same way, and we conclude that 
\begin{align}\label{eq: est Ia}
  |\langle \psi, \mathrm{I}_{a}\psi\rangle| \leq C\rho^{\frac{2}{3} + \frac{\gamma}{2}}\|\mathcal{N}^{\frac{1}{2}}\psi\|\|\mathbb{H}_0^{\frac{1}{2}}\psi\| + C\rho^{1+ \frac{\gamma}{2}}\langle \psi,\mathcal{N}\psi\rangle .
\end{align}

Next, we consider $\mathrm{I}_b$ in \eqref{eq: err S*a}. Again we can replace $\hat{u}_\sigma(r+p)$ by $\hat{u}^>_\sigma(r+p)$ defined in \eqref{eq: u< u> gamma}. Moreover, in this case, it is convenient to split  $\mathrm{I}_b=\mathrm{I}_b^{<\alpha} + \mathrm{I}_{b}^{>\alpha}$,  by writing  $\hat{u}_{\sigma^\prime}(r^\prime -p) = \hat{{u}}^{<\alpha}_{\sigma^\prime}(r^\prime-p) +  \hat{{u}}^{>\alpha}_{\sigma^\prime}(r^\prime-p)$. 

We first estimate $\mathrm{I}_b^{>\alpha}$.  
Using $p = -(r^\prime - p) + r^\prime$, this term can be written  in configuration space as 
\[
  \mathrm{I}^{>\alpha}_b = -2\sum_{\sigma\neq \sigma^\prime}\sum_{\ell =1}^3 \int dy \, \left(a^\ast_{\sigma^\prime}(\partial_\ell u^{<\alpha}_y)a_{\sigma^\prime}(\partial_\ell u^{>\alpha}_y) - a_{\sigma^\prime}^\ast(\partial_\ell^2 u^{<\alpha}_y)a_{\sigma^\prime}(u^{>\alpha}_y)\right) b_{\sigma}(\varphi^>_y, v_\sigma, u^>_\sigma).
\]
 From Lemmas \ref{lem:2a} and~\ref{lem: bounds phi} we obtain
\begin{align*} 
 & |\langle \psi, \mathrm{I}_b^{>\alpha}\psi\rangle|  \nn
 \\
&\leq C \rho^{\frac{1}{3}+ \frac{\gamma}{2} } \sum_{\sigma^\prime}\sum_{\ell =1}^3\int dy \Big(\| a_{\sigma^\prime}(\partial_\ell{u}^{<\alpha}_y)\psi\|\|a_{\sigma^\prime}(\partial_\ell {u}^{>\alpha}_y) \psi\| + \|a_{\sigma^\prime}(\partial_\ell^2 u^{<\alpha}_y)\psi\|\|a_{\sigma^\prime}(u^{>\alpha}_y)\psi\|\Big). 
\end{align*}
The resulting terms can be bounded by Lemma 
 \ref{lem:1a} and the Cauchy--Schwarz inequality, using also the bound $\int dy\, 
\|a_\sigma(\partial_\ell^2{u}^{<\alpha}_y)\psi\|^2 \leq C\rho^{\frac{2}{3} - 2\alpha} \int dy\, \|a_\sigma(\partial_\ell{u}^{<\alpha}_y)\psi\|^2
$. The result is 
\begin{equation}\label{eq: Ib>alpha}
  |\langle \psi, \mathrm{I}_b^{>\alpha}\psi\rangle|   \leq C\rho^{\frac{1}{3} + \frac{\gamma}{2}}\langle \psi, \mathbb{H}_0\psi\rangle 
 + C\rho^{1+\frac{\gamma}{2}}\langle \psi, \mathcal{N}\psi\rangle.   
\end{equation}

Next, we consider $\mathrm{I}^{<\alpha}_b$. Using $ p = (r+ p) - r$ we find
$
\mathrm{I}^{<\alpha}_{b} =  \mathrm{I}^{<\alpha}_{b;1} + \mathrm{I}^{<\alpha}_{b;2}
$
with
\begin{align}\nn 
\mathrm{I}^{<\alpha}_{b;1} &= \frac{2}{L^3}\sum_{\sigma\neq \sigma^\prime}  \sum_{p,r,r^\prime} (r+p)\cdot r^\prime \hat{\varphi}^>(p) \hat{u}_{\sigma^\prime}^{<\alpha}(r^\prime) \hat{u}_{\sigma^\prime}^{<\alpha}(r^\prime-p)\hat{u}_\sigma(r+p)\hat{v}_\sigma(r)  \\
&\qquad \qquad \qquad \times \hat{a}^\ast_{r^\prime,\sigma^\prime}\hat{a}_{r^\prime-p,\sigma^\prime} \hat{a}_{r+p,\sigma}\hat{a}_{-r,\sigma} + \mathrm{h.c.},\label{def:Iab1}
\\
 \mathrm{I}^{ <\alpha}_{b;2} 
  &= \sum_{\sigma\neq \sigma^\prime}\sum_{\ell =1}^3 \int dy \, a^\ast_{\sigma^\prime}(\partial_\ell u^{<\alpha}_y)a_{\sigma^\prime}(u^{<\alpha}_y)b_\sigma(\varphi^>_y, \partial_\ell v_\sigma, u^>_\sigma).\nonumber
\end{align}
For the term $\mathrm{I}_{b;2}^{<\alpha}$, proceeding in the same way as above, using Lemmas \ref{lem:2a},~\ref{lem: bounds phi} and \ref{lem:1a}, we obtain
\begin{equation}\label{eq: est Ib2<alpha}
  |\langle \psi, \mathrm{I}^{<\alpha}_{b;2}\psi\rangle| \leq C\rho^{\frac{2}{3} + \frac{\gamma}{2}}\|\mathbb{H}_0^{\frac{1}{2}}\psi\|\|\mathcal{N}^{\frac{1}{2}}\psi\| + C\rho^{1+\frac{\gamma}{2}}\langle \psi, \mathcal{N}\psi\rangle.
\end{equation}

The analysis of $\mathrm{I}^{<\alpha}_{b;1}$
is a bit more delicate. First, we observe that $|r| < C\rho^{1/3}$ and $|r^\prime| , |r^\prime - p| < \rho^{1/3 - \alpha}$ imply $|p| \leq 2\rho^{1/3 - \alpha}$ and $|r+p| < 3\rho^{1/3 - \alpha}$. Hence in \eqref{def:Iab1} we can replace $\hat{u}_\sigma(r+p)$ with $\zeta^{<\alpha}_\sigma$ given by 
\begin{equation}\label{eq: def zeta alpha}
  \zeta^{<\alpha}_\sigma(k) = \begin{cases} 1, &\mbox{if}\,\,\, k_F^\sigma \leq |k| \leq 3\rho^{\frac{1}{3} - \alpha}, \\ 0, &\mbox{otherwise}.\end{cases}
\end{equation}
We shall further 
 split $\hat{u}^{<\alpha}_{\sigma^\prime}(r^\prime - p)$ into two parts,  $\hat{u}^{<\alpha}_{\sigma^\prime}(r^\prime - p) = \hat{u}^{<\eta}_{\sigma^\prime}(r^\prime - p) + \hat{u}^{\eta<\alpha}_{\sigma^\prime}(r^\prime - p)$ for some $0 <\eta < \alpha$, with
 $ \hat u^{<\eta}_{\sigma} (r^\prime -p) = \hat  u_\sigma (r^\prime - p) \1_{\{|r^\prime - p| < \rho^{1/3-\eta}\}}$ and correspondingly $\hat{u}^{\eta<\alpha}_{\sigma^\prime}(r^\prime - p) = \hat  u_\sigma (r^\prime - p) \1_{\{\rho^{1/3-\eta}\leq |r^\prime - p| < \rho^{1/3-\alpha}\}}$. Correspondingly, we write $\mathrm{I}_{b;1}^{<\alpha} = \mathrm{I}_{b;1}^{<\eta} + \mathrm{I}_{b;1}^{\eta    <\alpha}$. 
 For the first term, given by
 $$
 \mathrm{I}_{b;1}^{<\eta} = - 2 \sum_{\sigma\neq \sigma^\prime} \sum_{\ell=1}^3   \int dx dy \, \varphi^>(x-y) a^*_{\sigma'}(\partial_\ell u^{<\alpha}_y) a_{\sigma'}(u^{<\eta}_y) a_\sigma(\partial_\ell\zeta^{<\alpha}_x) a_\sigma(v_x) ,
 $$
  we can apply Lemmas \ref{lem:4a}, \ref{lem:1a} and \ref{lem: bounds phi} to obtain 
\begin{align*}
  |\langle \psi, \mathrm{I}_{b;1}^{<\eta} \psi\rangle| &\leq C \sum_{\sigma\neq \sigma'}  \|\varphi^>\|_1 \| u_{\sigma'}^{<\eta}\|_2 \| v_\sigma\|_2  \Big( \langle \psi,  \bH_0  \psi\rangle  +   \rho^{\frac 2 3} \langle \psi, \cN \psi\rangle  \Big)\\
 & \le  C\rho^{\frac{1}{3} + 2\gamma -\frac{3}{2}\eta} \left( \langle \psi, \mathbb{H}_0 \psi\rangle + \rho^{\frac 23}\langle \psi, \mathcal{N}\psi\rangle \right).
\end{align*}
For the second term $\mathrm{I}^{\eta < \alpha}_{b;1}$, which equals   
\begin{align*}
  \mathrm{I}^{\eta <\alpha}_{b;1} 
  =2 \sum_{\sigma\neq \sigma^\prime}\sum_{\ell =1}^3 \int dy \, a^\ast_{\sigma^\prime}(\partial_\ell u^{<\alpha}_y)a_{\sigma^\prime}(u^{\eta<\alpha}_y)b_\sigma(\varphi^>_y, v_\sigma, \partial_\ell \zeta_\sigma^{<\alpha}),
\end{align*}
we apply Lemmas \ref{lem:2a}, \ref{lem: bounds phi}, and \ref{lem:1a}, as well as $\| |\cdot|^{-1} \hat{\zeta}_\sigma^{<\alpha}\|_2 \leq C\rho^{1/6 - \alpha/2}$, with the result that  
\[
  |\langle \psi, \mathrm{I}^{\eta < \alpha}_{b;1}\psi\rangle| \leq C\rho^{\frac{1}{3} - \frac{\alpha}{2} + \eta}\|\mathbb{H}_0^{\frac{1}{2}}\psi\| \left( \|\mathbb{H}_0^{\frac{1}{2}}\psi\| + \rho^{\frac{1}{3}}\|\mathcal{N}^{\frac{1}{2}}\psi\|\right).
\]
Adding the above bounds and optimizing over $\eta$ (which gives $\eta = (4/5)\gamma + \alpha/5$), we conclude that 
\begin{align}\label{eq: est Ib1<alpha}
  |\langle \psi, \mathrm{I}^{< \alpha}_{b;1}\psi\rangle| \leq C\rho^{\frac{1}{3} + \frac{4}{5}\gamma -\frac{3}{10}\alpha} \left( \langle \psi, \mathbb{H}_0 \psi\rangle  + \rho^{\frac 23}\langle \psi, \mathcal{N}\psi\rangle \right).
\end{align}

Next we consider the term $\mathrm{I}_c$ in  \eqref{eq: err S*a}. Inserting the definition of $D^{<\alpha}_{p,\sigma'}$ (see \eqref{def:D}) 
gives
\begin{align*}
 \mathrm{I}_c = \frac{2}{L^3}\sum_{\sigma\neq \sigma^\prime}\sum_p|p|^2 \hat{\varphi}^<(p) \hat{u}^{<\alpha}_\sigma(r) \hat{u}_\sigma(r+p) \hat{u}_{\sigma^\prime}(r^\prime -p)\hat{v}_{\sigma^\prime}(r^\prime)\hat{a}_{r+p, \sigma}^\ast \hat{a}_{-r^\prime,\sigma^\prime}^\ast \hat{a}^\ast_{r^\prime -p, \sigma^\prime}\hat{a}_{r,\sigma} + \mathrm{h.c.}
\end{align*}
By the constraints  $|r^\prime| \leq C\rho^{1/3}$ and $|p|\leq (5/4)\rho^{1/3-\gamma}$, which follow from the definitions of $\hat{v}_\sigma$ and $\hat{\varphi}^<$ (see \eqref{eq: def u hat, v hat} and \eqref{eq: def phi><}), we can replace $\hat{u}_{\sigma^\prime}(r^\prime - p)$ by $\hat{u}^<_{\sigma^\prime}(r^\prime - p)$. 
Hence  $\mathrm{I}_c$ is given in  configuration space by 
\[
  \mathrm{I}_c = -2\sum_{\sigma\neq \sigma^\prime}\int dxdy\, \Delta\varphi^<(x-y) a^\ast_\sigma(u_x) a^\ast_{\sigma^\prime}(v_y)a^*_{\sigma^\prime}(u^<_y)a_{\sigma}(u^{<\alpha}_x) .
\]
With Lemma \ref{lem:4aaa} we obtain the bound 
\begin{align}\label{eq: est Ic<alpha}
 |\langle\psi, \mathrm{I}_c\psi\rangle| 
  \le C\rho^{1+\frac{\beta}{6}}\langle\psi,\mathcal{N}\psi\rangle + C\rho^{1-\frac{3}{2}\gamma}\|\mathcal{N}^{\frac{1}{2}}\psi\| \|\mathcal{N}^{\frac{1}{2}}_\beta\psi\|.
 \end{align}

The term $\mathrm{I}_d$ in  \eqref{eq: err S*a} can be written  as
 \[
 \mathrm{I}_d = 2\sum_{\sigma\neq \sigma^\prime}\int dxdy (-2\Delta \varphi - Vf + 8\pi a L^{-3})  (x-y) a^\ast_{\sigma^\prime}(u^{<\alpha}_{y})a_\sigma(v_x)a_{\sigma^\prime}(u_y) a_\sigma(u_x) + \mathrm{h.c.}
 \]
 where we inserted the constant $8\pi a L^{-3}$ by using the orthogonality $u_\sigma v_\sigma =0$ as in \eqref{eq: auav}. By using \eqref{eq: periodic-scattering-equation}, we can proceed similarly to \eqref{eq:proof-lem-5.1-aa} and obtain
 \begin{align}\label{eq: est Id<alpha}
\left| \left\langle \psi,   \mathrm{I}_d  \psi  \right\rangle \right| \le \mathfrak{e}_L   \int dxdy V (x-y) \| a_\uparrow(u_x) a_\downarrow(u_y)  \psi\| \le \mathfrak{e}_L   L^{3/2} \|\bQ_4^{1/2}\psi\|. 
\end{align}

It remains to bound $\mathrm{I}_e$ in  \eqref{eq: err S*a}.    Using $\|a_{\sigma}(v_x)\|\leq C\rho^{1/2}$ and \eqref{eq: est H0 sim}, we get 
 \begin{align}\label{eq: est Ie<alpha}
  |\langle \psi,  \mathrm{I}_e \psi\rangle| &\leq C\rho^{\frac{1}{2}}\sum_{\sigma\neq \sigma^\prime}\int dxdy\, V(x-y) \|a_{\sigma}(u_x)a_{\sigma^\prime}(u_y) \psi\| \|a_{\sigma^\prime}(u^{>\alpha}_y) \psi\| \nn\\
  &\leq  C\rho^{\frac{1}{6} + \alpha}\|\mathbb{Q}_4^{\frac{1}{2}} \psi\|\|\mathbb{H}_0^{\frac{1}{2}}  \psi\|. 
\end{align}
Collecting all the estimates in \eqref{eq: est Ia}--\eqref{eq: est Ib2<alpha} and  \eqref{eq: est Ib1<alpha}--\eqref{eq: est Ie<alpha}, 
this concludes the proof of Lemma \ref{lem: Sa}.
 \end{proof}

\subsection{Proof of Lemma \ref{lem: T} }

Now we prove Lemma \ref{lem: T}, which is among the most important technical results as it generates all relevant constants in the Huang--Yang formula. While the choice of the renormalization term $T^\ast_\sigma(r)$ is clearly motivated by the scattering cancellation in Lemma \ref{lem: Ta}, it is considerably more subtle to understand why the terms involving $\{T^\ast_\sigma(r), T_\sigma(r)  \}$ yield both the correct energy and the quadratic contribution in the first term on the right-hand side of \eqref{c1}.

\begin{proof} Inserting the definition \eqref{eq: def Tk} of $T_\sigma^\ast(r)$ in the anticommutator below,  we have
\begin{align}\nn
  &\sum_{\sigma}\sum_r||r|^2 - (k_F^\sigma)^2| \{T^\ast_\sigma(r), T_\sigma(r)  \} 
  \\ \nn 
  &= \frac{1}{L^6}\sum_{\sigma\neq \sigma^\prime}\sum_{r,p,r',q,s} ||r|^2 - (k_F^\sigma)^2| \Big\{\hat{a}_{p-r,\sigma}b_{-p,r^\prime, \sigma'}, b^\ast_{-q,s,\sigma^\prime}\hat{a}_{q-r,\sigma}^\ast\Big\}
 \\  \nn & \qquad \times \Big(  \widehat{\omega}^\varepsilon_{-r,r^\prime}(p)\widehat{\omega}^\varepsilon_{-r,s}(q)\hat{u}_{\sigma}(p-r) \hat{u}_\sigma(q- r)\hat v_\sigma(r)  \\ & \qquad \qquad +  
 					\widehat{\omega}^\varepsilon_{r-p,r^\prime}(p)\widehat{\omega}^\varepsilon_{r-q,s}(q)\hat{v}_{\sigma}(p-r) \hat{v}_\sigma(q-r)\hat u_\sigma(r)
  \Big ) .  \label{eq:T*T-computation}
\end{align}
Recalling the definition of $b_{p,k,\sigma} $ in \eqref{eq: def bp bpr}, 
we have to compute the 
 anticommutator 
 for $s,r' \in B_F^{\sigma'}$, $r'-p, s-q  \not\in B_F^{\sigma '}$ and $\sigma \neq \sigma'$, given in this case by 
\begin{align}\label{eq:TT*-r-in}
  & \Big\{\hat{a}_{p-r,\sigma} \hat a_{r^\prime-p, \sigma'} \hat a_{-r^\prime, \sigma'}
  , \hat a^\ast_{-s,\sigma^\prime} \hat a^\ast_{s-q,\sigma^\prime}  \hat{a}_{q-r,\sigma}^\ast\Big\} \nn\\ \nn
  & = \delta_{r^\prime, s}\delta_{p,q} \big( 1 - \hat{a}_{p-r,\sigma}^\ast \hat{a}_{p-r, \sigma} - \hat{a}_{r'-p, \sigma^\prime}^\ast \hat{a}_{r^\prime - p,\sigma^\prime} - \hat{a}_{-r',\sigma^\prime}^\ast \hat{a}_{-r^\prime , \sigma^\prime}\big) \\ 
  &\quad + \delta_{r^\prime,s}\hat{a}_{r'-q,\sigma^\prime}^\ast \hat{a}_{q-r,\sigma}^\ast\hat{a}_{p-r, \sigma}\hat{a}_{r^\prime-p, \sigma^\prime} 
   + \delta_{r^\prime-p, s-q}\hat{a}^\ast_{-s,\sigma^\prime}\hat{a}_{q-r,\sigma}^\ast\hat{a}_{p-r,\sigma}\hat{a}_{-r^\prime,\sigma^\prime} \nn \\ & \quad + \delta_{p, q}\hat{a}_{-s,\sigma^\prime}^\ast \hat{a}_{s-p,\sigma^\prime}^\ast \hat{a}_{r^\prime - p,\sigma^\prime}\hat{a}_{-r^\prime, \sigma^\prime}. 
\end{align}
After a  change of variables $r\notin B_F^\sigma\mapsto r+p$  {and $r\in B_F^\sigma\mapsto -r$}  and combining various terms, we obtain the identity  
\begin{align}\label{eq:RR-dec}
\sum_{\sigma,r} ||r|^2 - (k_F^\sigma)^2| \{T^\ast_{\sigma}(r), T_{\sigma}(r)\} = \sum_{j=1}^{10}\mathrm{I}_j
\end{align}
with
\begin{align*}
  \mathrm{I}_1 & = \frac{1}{L^6}\sum_{\sigma\neq \sigma^\prime}\sum_{p,r,r^\prime}\left(|r+p|^2 -|r|^2\right)(\widehat{\omega}^\varepsilon_{r,r^\prime}(p))^2 \hat{u}_{\sigma}(r+p)\hat{u}_{\sigma^\prime}(r^\prime - p)\hat{v}_\sigma(r)\hat{v}_{\sigma^\prime}(r^\prime),\\
  \mathrm{I}_2 &   = -\frac{1}{L^6}\sum_{\sigma\neq\sigma^\prime}\sum_{p,r,r^\prime} \left(|r+p|^2 -|r|^2 \right)  (\widehat{\omega}_{r,r^\prime}^\varepsilon (p))^2  \\
  &\qquad\qquad \times \hat{u}_\sigma(r+p) \hat{u}_{\sigma^\prime}(r^\prime - p)\hat{v}_\sigma(r)\hat{v}_{\sigma^\prime}(r^\prime) \hat{a}_{-r^\prime ,\sigma^\prime}^\ast \hat{a}_{-r^\prime, \sigma^\prime},\\
  \mathrm{I}_{3}  & = -\frac{1}{L^6}\sum_{\sigma\neq \sigma^\prime}\sum_{p,r,r^\prime}||r+p|^2 - (k_F^\sigma)^2| (\widehat{\omega}_{r,r^\prime}^\varepsilon (p))^2\hat{u}_\sigma(r+p) \hat{u}_{\sigma^\prime}(r^\prime - p)\hat{v}_\sigma(r)\hat{v}_{\sigma^\prime}(r^\prime)\hat{a}_{-r ,\sigma}^\ast \hat{a}_{-r,\sigma },\\
  \mathrm{I}_4 &= \frac{1}{L^6}\sum_{\sigma\neq \sigma^\prime}\sum_{p,q,r,r^\prime}||r+p|^2 - (k_F^\sigma)^2| \widehat{\omega}_{r,r^\prime}^\varepsilon (p)\widehat{\omega}_{r+p-q, r^\prime - p+q}^\varepsilon (q)\hat{u}_\sigma(r+p)\hat{v}_\sigma(r)\hat{u}_{\sigma^\prime}(r^\prime - p)  \\
&\qquad\qquad  \times  \hat{v}_{\sigma^\prime}(r^\prime)\hat{v}_\sigma(r+p-q)\hat{v}_{\sigma^\prime}(r^\prime - p+q) \hat{a}_{-r^\prime +p - q, \sigma^\prime}^\ast \hat{a}_{q-p-r,\sigma}^\ast \hat{a}_{-r,\sigma}\hat{a}_{-r^\prime, \sigma^\prime},\\
  \mathrm{I}_{5} & = \frac{1}{L^6}\sum_{\sigma\neq \sigma^\prime}\sum_{p,q, r,r^\prime}||r|^2 - (k_F^\sigma)^2| \widehat{\omega}_{r,r^\prime}^\varepsilon(p) \widehat{\omega}^\varepsilon_{r,r^\prime}(q) \hat{u}_\sigma(r+p)\hat{u}_{\sigma^\prime}(r^\prime - p)\hat{v}_\sigma(r)\hat{v}_{\sigma^\prime}(r^\prime) 
  \\
 &\qquad \qquad \times \hat{u}_{\sigma^\prime}(r^\prime - q)\hat{u}_\sigma(r+q) \hat{a}_{r^\prime -q, \sigma^\prime}^\ast \hat{a}_{r+q,\sigma}^\ast  \hat{a}_{r+p ,\sigma} \hat{a}_{r^\prime -p,\sigma^\prime}\\
  \mathrm{I}_6 &= \frac{1}{L^6}\sum_{\sigma\neq \sigma^\prime}\sum_{p,r,r^\prime,s} \left(|r+p|^2 -|r|^2\right) \widehat{\omega}_{r,r^\prime}^\varepsilon (p)\widehat{\omega}_{r,s}^\varepsilon (p)\hat{u}_{\sigma}(r+p)\hat{v}_\sigma(r)\hat{u}_{\sigma^\prime}(r^\prime - p) 
  \\
 &\qquad \qquad \times \hat{v}_{\sigma^\prime}(r^\prime)\hat{v}_{\sigma^\prime}(s)  \hat{u}_{\sigma^\prime}(s-p) \hat{a}_{-s,\sigma^\prime}^\ast\hat{a}_{s-p, \sigma^\prime}^\ast \hat{a}_{r^\prime - p, \sigma^\prime}\hat{a}_{-r^\prime, \sigma^\prime}\\
    \mathrm{I}_7 & = \frac{1}{L^6}\sum_{\sigma \neq \sigma^\prime}\sum_{p,q,r,r^\prime} ||r+p|^2 - (k_F^\sigma)^2| \widehat{\omega}_{r,r^\prime}^\varepsilon(p)\widehat{\omega}_{r+p-q, r^\prime}^{\varepsilon}(q)\hat{u}_{\sigma}(r+p)\hat{u}_{\sigma^\prime}(r^\prime - p)  
  \\
  &\qquad \qquad \times \hat{v}_\sigma(r)\hat{v}_{\sigma^\prime}(r^\prime)  \hat{u}_{\sigma^\prime}(r^\prime - q) \hat{v}_\sigma(r+p-q) \hat{a}_{r^\prime - q, \sigma^\prime}^\ast \hat{a}_{q-p-r, \sigma}^\ast \hat{a}_{-r,\sigma}\hat{a}_{r^\prime - p, \sigma^\prime}\\
  \mathrm{I}_{8} &= \frac{1}{L^6}\sum_{\sigma\neq \sigma^\prime}\sum_{p,q, r,r^\prime}||r|^2 - (k_F^\sigma)^2| \widehat{\omega}_{r,r^\prime}^\varepsilon(p) \widehat{\omega}^\varepsilon_{r,r^\prime -p+q}(q) \hat{u}_\sigma(r+p)\hat{u}_{\sigma^\prime}(r^\prime - p)\hat{v}_\sigma(r)\hat{v}_{\sigma^\prime}(r^\prime) 
  \\
  &\qquad\qquad \times\hat{v}_{\sigma^\prime}(r^\prime -p+q)\hat{u}_\sigma(r+q)\hat{a}_{-r^\prime +p-q, \sigma^\prime}^\ast \hat{a}_{r+q,\sigma}^\ast  \hat{a}_{r+p ,\sigma} \hat{a}_{-r^\prime,\sigma^\prime}\\
  \mathrm{I}_{9} &= -\frac{1}{L^6}\sum_{\sigma\neq \sigma^\prime} \sum_{p,r,r^\prime}\left(|r+p|^2 -  |r|^2\right) (\widehat{\omega}^\varepsilon_{r,r^\prime}(p))^2  \\
  &\qquad \qquad \times \hat{u}_\sigma(r+p)\hat{u}_{\sigma^\prime}(r^\prime - p)\hat{v}_\sigma(r)\hat{v}_{\sigma^\prime}(r^\prime)\hat{a}_{r^\prime - p, \sigma^\prime}^\ast\hat{a}_{r^\prime - p, \sigma^\prime}\\
  \mathrm{I}_{10} &= -\frac{1}{L^6}\sum_{\sigma\neq \sigma^\prime}\sum_{p,r,r^\prime}||r|^2 - (k_F^\sigma)^2| (\widehat{\omega}_{r,r^\prime}^\varepsilon(p))^2 \hat{u}_\sigma(r+p)\hat{u}_{\sigma^\prime}(r^\prime - p)\hat{v}_\sigma(r)\hat{v}_{\sigma^\prime}(r^\prime) \hat{a}_{r+p,\sigma}^\ast \hat{a}_{r+p,\sigma}. 
\end{align*}
The terms $\mathrm{I}_9$ and  $\mathrm{I}_{10}$ are non-positive and can hence be dropped for an upper bound. The main constant contribution comes  from $\mathrm{I}_1$, while $\mathrm{I}_2 + \mathrm{I}_3$ gives the desired (negative) quadratic contribution. 
All the $\mathrm{I}_j$ with $4 \leq j \leq 8$ are nonnegative, but they will only contribute to the error terms.

\bigskip
\noindent {\bf Analysis of $\mathrm{I}_{1}$.}  Using that $\widehat{\omega}^\varepsilon_{r,r^\prime}(p) = \widehat{\omega}^\varepsilon_{r',r}(-p)$ we can write
\begin{align}\label{eq:RR-I1-final}
  {\mathrm{I}_1} &= \frac{1}{L^6}\sum_{p,r,r^\prime}\left(|r+p|^2 -|r|^2  + |r^\prime -p|^2 - |r^\prime|^2 \right)(\widehat{\omega}^\varepsilon_{r,r^\prime}(p))^2 \hat{u}_{\uparrow}(r+p)\hat{u}_{\downarrow}(r^\prime - p)\hat{v}_\uparrow(r)\hat{v}_{\downarrow}(r^\prime)
  \nn\\
  &\le  \frac{1}{L^6}\sum_{p,r,r^\prime}\frac{(2|p|^2 \hat{\varphi}(p))^2}{\lambda_{p,r}+ \lambda_{-p,r^\prime} + 2\varepsilon}\hat{u}_{\uparrow}(r+p)\hat{u}_{\downarrow}(r^\prime - p)\hat{v}_\uparrow(r)\hat{v}_{\downarrow}(r^\prime) .
\end{align}

\bigskip
\noindent {\bf Analysis of $\mathrm{I}_{2}$ and $\mathrm{I}_{3}$.} Next we consider $\mathrm{I}_2$. Since each summand is non-negative,  we can restrict the sum in  $\mathrm{I}_2$ to  $\{|p|\ge \rho^{1/3-\gamma}\}$ for an upper bound. 
In this case, since  $|r|,|r'|\le C\rho^{1/3}$ and $\eps=\rho^{\frac 2 3+\delta}\le \rho^{\frac 2 3+\gamma}$,  we have  
\begin{align*}
(|r+p|^2-|r|^2) (\widehat{\omega}_{r,r^\prime}^\varepsilon (p))^2& = \frac{(|p|^2 + 2 p\cdot r)4|p|^4 |\widehat \varphi(p)|^2 }{(2|p|^2 + 2 p \cdot(r-r')+2\eps)^2}  \nn\\
&\ge  (|p|^2 + 2 p\cdot r) |\widehat \varphi(p)|^2 (1-C\rho^\gamma). 
\end{align*}
Here we used $(1+x)^{-2}\ge 1 -C|x|$ with $|x|=|p|^{-2}  |p \cdot(r-r')+\eps| \le C \rho^\gamma$. Using also $\sum_r p\cdot r\, v_\sigma(r)=0$, we obtain the bound
\begin{align*} 
  \mathrm{I}_2   &\leq  - \frac{1-C\rho^\gamma}{L^6}\sum_{\sigma\neq\sigma^\prime}\sum_{|p|\ge \rho^{1/3-\gamma}}\sum_{r,r^\prime} (|p|^2 + 2 p\cdot r) |\widehat \varphi(p)|^2 \hat{v}_\sigma(r)\hat{v}_{\sigma^\prime}(r^\prime) \hat{a}_{-r^\prime ,\sigma^\prime}^\ast \hat{a}_{-r^\prime, \sigma^\prime}
  \nn\\
  &=  - (1-C\rho^\gamma)\sum_{\sigma\neq\sigma^\prime}  \rho_{\sigma} \left( \frac{1}{L^3} \sum_{|p|\ge \rho^{1/3-\gamma}} |p|^2 |\widehat \varphi(p)|^2 \right) \sum_{r^\prime} \hat{v}_{\sigma^\prime}(r^\prime) \hat{a}_{-r^\prime ,\sigma^\prime}^\ast \hat{a}_{-r^\prime, \sigma^\prime} .
\end{align*} 
Since $|\widehat \varphi(p)|\le C |p|^{-2}$ the missing part in the sum over $p$ is bounded by  $L^{-3} \sum_{|p|\le \rho^{1/3-\gamma}} |p|^2 |\widehat \varphi(p)|^2 \le C \rho^{1/3-\gamma}$. Since $\gamma < 1/6$, this is negligible compared to $\rho^\gamma$, however. We conclude that  
\begin{equation}\label{eq:RR-I2-final}
  \mathrm{I}_2   \leq -  \sum_{\sigma\neq\sigma^\prime} \rho_{\sigma}   \| \nabla \varphi\|_2^2  \sum_{r'} \hat{v}_{\sigma^\prime}(r^\prime) \hat{a}_{-r^\prime ,\sigma^\prime}^\ast \hat{a}_{-r^\prime, \sigma^\prime} + C \rho^{1+\gamma}\cN. 
\end{equation} 

The term $\mathrm{I}_3$ can be treated in a similar manner, and satisfies the same bound. 

\medskip
\noindent {\bf Analysis of $\mathrm{I}_{4}$.} We now consider $\mathrm{I}_4$ in \eqref{eq:RR-dec}, which can be written as 
\begin{align*}
  \mathrm{I}_4 &= \frac{1}{L^6}\sum_{\sigma\neq \sigma^\prime}\sum_{r,r^\prime}||r|^2 - (k_F^\sigma)^2| \hat{u}_\sigma(r) \hat{u}_{\sigma^\prime}(r^\prime) 
  \\
  &\qquad\qquad \times \left|\sum_p \widehat{\omega}^\varepsilon_{r-p,r^\prime + p}(p) \hat{v}_\sigma(r-p)\hat{v}_{\sigma^\prime}(r^\prime + p)\hat{a}_{r-p,\sigma}\hat{a}_{r^\prime + p, \sigma^\prime}\right|^2
 \end{align*}
Using $1= \widehat{\chi}_<(p) + \widehat{\chi}_>(p)$, defined in \eqref{eq: def chi< chi>}, and the Cauchy--Schwarz inequality, we can bound $\mathrm{I}_4  \leq  2 (\mathrm{I}_4^>  + \mathrm{I}_4^<)$ with
\begin{align}\label{eq:R*R-I4-first-dec}
  \mathrm{I}_4^\gtrless &=  \frac{1}{L^6}\sum_{\sigma\neq \sigma^\prime}\sum_{r,r^\prime}||r|^2 - (k_F^\sigma)^2| \hat{u}_\sigma(r) \hat{u}_{\sigma^\prime}(r^\prime) 
  \nn \\
  &\qquad\qquad \times \left|\sum_p \widehat{\omega}^\varepsilon_{r-p,r^\prime + p}(p)  \widehat{\chi}_\gtrless(p)\hat{v}_\sigma(r-p)\hat{v}_{\sigma^\prime}(r^\prime + p)\hat{a}_{r-p,\sigma}\hat{a}_{r^\prime + p, \sigma^\prime}\right|^2 .
  \end{align}

We start with estimating $\mathrm{I}_4^>$. 
We observe that for $|r-p|, |r^\prime +p|^2 \leq C \rho^{1/3}\ll \rho^{1/3-\gamma} \leq C |p|, |r|, |r^\prime|$,   the coefficient $\widehat{\omega}^\varepsilon_{r-p,r^\prime+p}(p)$ is close to $\hat{\varphi}(p)$. With this in mind, we write
\[
  \widehat{\omega}^\varepsilon_{r-p,r^\prime+p}(p) \chi_>(p)= \hat{\varphi}^>(p)+ 2 \left( \frac{p\cdot((r^\prime + p) - (r-p)) -\varepsilon}{\lambda_{p,r-p} + \lambda_{-p,r^\prime+p} + 2\varepsilon} \right) \hat{\varphi}^>(p)
\]
and use the Cauchy--Schwarz inequality to bound  
\begin{align}\label{eq:RR-I4-dec}
  \mathrm{I}_4^> &\le \frac{2}{L^6}\sum_{\sigma\neq\sigma^\prime}\sum_{r,r^\prime}||r|^2 - (k_F^\sigma)^2| \left|\sum_p \hat{\varphi}^>(p) \hat{v}_\sigma(r-p)\hat{v}_{\sigma^\prime}(r^\prime + p)\hat{a}_{r-p,\sigma}\hat{a}_{r^\prime + p,\sigma^\prime}\right|^2 
  \nn\\
  &\quad +\frac{8}{L^6}\sum_{\sigma\neq\sigma^\prime}\sum_{r,r^\prime}||r|^2 - (k_F^\sigma)^2|   
  \nn\\
  &\qquad\times \left|\sum_p   \left(\frac{p\cdot((r^\prime + p) - (r-p)) -\varepsilon}{\lambda_{p,r-p} + \lambda_{-p,r^\prime+p} + 2\varepsilon}\right)\hat{\varphi}^>(p) \hat{v}_\sigma(r-p)\hat{v}_{\sigma^\prime}(r^\prime + p)\hat{a}_{r-p,\sigma}\hat{a}_{r^\prime + p,\sigma^\prime}\right|^2 
  \nn\\
  &=  \mathrm{I}_{4;a}^> + \mathrm{I}_{4;b}^>.
\end{align}
With a change of variables,  we can rewrite $\mathrm{I}_{4;a}^>$  as 
\begin{align*}
  \mathrm{I}_{4;a}^> & = \frac{2}{L^6}\sum_{\sigma\neq \sigma^\prime}\sum_{p,q,r,r^\prime}\Big(|p|^2 + 2p\cdot r +|r|^2 - (k_F^\sigma)^2\Big)   \hat{\varphi}(p) \hat{\varphi}(q) \widehat{\chi}_>(p)\widehat{\chi}_>(q) 
  \\
  &\qquad\qquad  \times \hat{v}_\sigma(r)\hat{v}_{\sigma^\prime}(r^\prime)\hat{v}_\sigma(r+p-q)\hat{v}_{\sigma^\prime}(r^\prime - p+q)  \hat{a}_{-r^\prime +p - q, \sigma^\prime}^\ast \hat{a}_{q-p-r,\sigma}^\ast \hat{a}_{-r,\sigma}\hat{a}_{-r^\prime, \sigma^\prime}\nn\\
  &=  \sum_{j=1}^4\mathrm{I}_{4;a;j}^>,
\end{align*}
where $I^{>}_{4;a;j}$ with $j=1,\ldots,4$, denote the contributions corresponding, respectively, to the four terms $|p|^2, 2p\cdot r, |r|^2, -(k_F^\sigma)^2$ in the parentheses above.
The first term can be written in configuration space as 
$$
 \mathrm{I}_{4;a;1}^> = 2 \sum_{\sigma\neq\sigma^\prime}\int dxdy \, \left(-\Delta\varphi^>\right)(x-y)\varphi^>(x-y)a^\ast_{\sigma^\prime}(v_y)a_\sigma^\ast(v_x)a_\sigma(v_x)a_{\sigma^\prime}(v_y)
$$
and can be bounded using Lemma \ref{lem:4aaa} and $\|(\Delta \varphi^>)\varphi^>\|_{1}\le \|(\Delta \varphi^>)\|_1 \|\varphi^>\|_{\infty}\le C$ from Lemma \ref{lem: bounds phi}, with the result that  
$$
| \langle \psi ,    \mathrm{I}_{4;a;1}^> \psi \rangle |  \le C \rho^{ 1 +\frac{\beta}{6}}  \langle \psi, \mathcal{N} \psi\rangle + C\rho  \|\mathcal{N}_\beta^{\frac{1}{2}}\psi\|  \|\mathcal{N}^{\frac{1}{2}}\psi\|. 
$$
The second term equals
$$
\mathrm{I}_{4;a;2}^> = 4 \sum_{\sigma\neq\sigma^\prime} \sum_{\ell = 1}^3 \int dxdy \, \left(-\partial_\ell\varphi^>\right)(x-y)\varphi^>(x-y)a^\ast_{\sigma^\prime}(v_y)a_\sigma^\ast(v_x)a_\sigma(\partial_\ell v_x)a_{\sigma^\prime}(v_y)
$$
and can be bounded with the aid of Lemmas~\ref{lem:1a} and~\ref{lem: bounds phi} as
$$
  \langle\psi, \mathrm{I}_{4;a;2}^>\psi\rangle    \leq C\rho^{\frac{4}{3}}\|\nabla{\varphi}^>\|_1 \|{\varphi}^>\|_\infty \langle \psi, \mathcal{N}\psi\rangle \leq C\rho^{1+\gamma}\langle \psi, \mathcal{N}\psi\rangle.
$$

Proceeding in the same way, one sees that $  |\langle\psi, \mathrm{I}_{4;a;3}^>\psi\rangle| \leq C\rho^{4/3+\gamma}\langle \psi, \mathcal{N}\psi\rangle$. 
Finally, $\mathrm{I}_{4;a;4}^> \leq 0$. Combining the estimates, we conclude that for all $0\le \beta<1$
\begin{equation}\label{eq: est Ia 1 R*R}
  \langle\psi, \mathrm{I}_{4;a}^> \psi \rangle \leq C\Big( \rho^{1+\gamma} +  
   \rho^{ 1 +\frac{\beta}{6}} \Big)  \langle \psi, \mathcal{N}\psi\rangle + C\rho  \|\mathcal{N}_\beta^{\frac{1}{2}}\psi\|  \|\mathcal{N}^{\frac{1}{2}}\psi\|. 
  \end{equation}

Next, we consider $\mathrm{I}_{4;b}^>$. Using \eqref{eq: int t conf space}, we can write 
\begin{align*}
  \mathrm{I}_{4;b}^> &= \frac{8}{L^6} \sum_{\sigma\neq \sigma^\prime}\sum_{r,r^\prime} ||r|^2- (k_F^\sigma)^2| \bigg|  \int_0^{\infty}dt\, e^{-2t\varepsilon} e^{-t ( |r|^2 + |r^\prime|^2 )}  \\
  &\quad \times \sum_{p} \left[p\cdot\big( (r^\prime + p) - (r-p)\big) -\varepsilon \right]  \hat{\varphi}^>(p)\hat{v}_{t,\sigma}(p-r)\hat{v}_{t,\sigma^\prime}(r^\prime +p)\hat{a}_{-r+p,\sigma}\hat{a}_{-r^\prime - p, \sigma^\prime}\bigg|^2.
\end{align*}
With the Cauchy--Schwarz inequality with respect to $t$, we can bound this as
\begin{align}\nn
  \mathrm{I}_{4;b}^> &\leq \frac{8}{L^6}\sum_{\sigma\neq \sigma^\prime}\sum_{r,r^\prime} ||r|^2 - (k_F^\sigma)^2| \left(\int_0^{\infty}dt\, e^{-2t(|r|^2 - (k_F^\sigma)^2)}\right) \int_0^{\infty}dt\, e^{-2t( |r^\prime|^2 + (k_F^{\sigma})^2 ) } 
  \\
 &\quad \times 
  \left|\sum_p \left[p\cdot ( (r^\prime +p) - (r-p) ) -\varepsilon\right]\hat{\varphi}^>(p) \hat{v}_{t,\sigma}(p-r)\hat{v}_{t,\sigma^\prime}(r^\prime + p) \hat{a}_{p-r,\sigma}\hat{a}_{-r^\prime - p, \sigma^\prime}\right|^2. 
  \label{tbrt}
\end{align}
The first integral on the right-hand side cancels the factor  $||r|^2 - (k_F^\sigma)^2|$.  
Moreover, on the support of $\varphi^>(p)$ and $\hat{v}_{\sigma^\prime}(r^\prime+p)$, we always have $|r^\prime| \geq 3\rho^{1/3 - \gamma}$.  
Applying also a Cauchy--Schwarz inequality for the  sum over the terms  in $p\cdot (r^\prime +p) - p\cdot (r-p) -\varepsilon$
we find 
\begin{align*}
  \mathrm{I}_{4;b}^>
  &\le   \frac{C}{L^6}  \sum_{\sigma\neq \sigma^\prime} \sum_{r,r^\prime} \int_0^{\infty}dt\, e^{-18t \rho^{\frac 2 3 -2\gamma } } \left[ \eps^2 \left| \sum_p \hat{\varphi}^>(p)\hat{v}_{t, \sigma}(r-p)\hat{v}_{t, \sigma^\prime}(r^\prime + p)\hat{a}_{p-r,\sigma}\hat{a}_{-r^\prime - p, \sigma^\prime} \right|^2 \right.  \\
  &\qquad\qquad + \sum_{\ell=1}^3 \left|\sum_p (p_\ell (r^\prime_\ell +p_\ell)  \hat{\varphi}^>(p)\hat{v}_{t, \sigma}(r-p)\hat{v}_{t, \sigma^\prime}(r^\prime + p)\hat{a}_{p-r,\sigma}\hat{a}_{-r^\prime - p, \sigma^\prime} \right|^2 \\
  &\qquad \qquad \left. + \sum_{\ell=1}^3  \left| \sum_p (p_\ell  (r_\ell -p_\ell )  \hat{\varphi}^>(p)\hat{v}_{t, \sigma}(r-p)\hat{v}_{t, \sigma^\prime}(r^\prime + p)\hat{a}_{p-r,\sigma}\hat{a}_{-r^\prime - p, \sigma^\prime} \right|^2\right].
  \end{align*}
  Writing these terms in configuration space and using \eqref{eq:Pauli} together with Lemma \ref{lem:1a} and Lemma \ref{lem: bounds phi}, we obtain
  \begin{align}\nn
 \mathrm{I}_{4;b}^> &\leq  C  \sum_{\sigma\neq \sigma^\prime}\int_0^{\infty}dt \, e^{-18t \rho^{\frac 2 3 -2\gamma } }  \int dx dy  \Big(  \eps^2   |\varphi^>  (x-y)|^2 a_\sigma^*(v_{t,x}) a_{\sigma'}^\ast(v_{t,y}) a_{\sigma'}(v_{t,y}) a_\sigma(v_{t,x})
  \\ \nn
 &\qquad\qquad\qquad\qquad\qquad +  \sum_{\ell=1}^3  |\partial_\ell \varphi^>_j  (x-y)|^2 a_\sigma^*(v_{t,x}) a_{\sigma'}^\ast(\partial_\ell v_{t,y}) a_{\sigma'}(\partial_\ell v_{t,y}) a_\sigma(v_{t,x}) \Big) \\ \nn
  &\le   C \sum_{\sigma\neq \sigma^\prime}\int_0^{\infty}dt\, e^{-18t \rho^{\frac 2 3 -2\gamma } }  \Big( \eps^2 e^{2t (k_F^{\sigma'})^2} \|\varphi^>\|_2^2 \|v_{t,\sigma}\|_2^2  + \|\nabla \varphi^>\|_2^2 e^{2t (k_F^\sigma)^2 } \|\nabla v_{t,\sigma'})\|_2^2\Big) \cN \\
&  \leq C(\rho^{\frac 4 3 + 3\gamma +2\delta}+ \rho^{1+2\gamma}) \cN.  \label{I4b}
\end{align}
In the final estimate we used the small factor $e^{-18 t \rho^{\frac 2 3 -2\gamma } }$ to dominate  the various terms $e^{t(k_F^\sigma)^2}$ resulting from the $v_{t,\sigma}$.

Next we shall consider $\mathrm{I}_{4}^<$ in \eqref{eq:R*R-I4-first-dec}.  Using again \eqref{eq: int t conf space} and the Cauchy--Schwarz inequality similarly as in \eqref{tbrt}, we obtain the bound  
\begin{multline*}
  \mathrm{I}_4^< \leq \frac{C}{L^6}\sum_{\sigma\neq \sigma^\prime}\int_0^{\infty}dt\, \sum_{r,r^\prime}\hat{u}_\sigma(r) \hat{u}_{\sigma^\prime}(r^\prime)e^{-2t\varepsilon}e^{-t|r|^2} e^{-t|r^\prime|^2} e^{-t(k_F^\sigma)^2} e^{-t(k_F^{\sigma^\prime})^2} 
  \\
  \times \left|\sum_p |p|^2 \hat{\varphi}^<(p) \hat{v}_{t,\sigma}(r-p)\hat{v}_{t,\sigma^\prime}(r^\prime + p)\hat{a}_{r-p,\sigma}\hat{a}_{r^\prime + p, \sigma^\prime}\right|^2.
\end{multline*}
Given the support of $\varphi^<$, the sum in the second line is non-zero only if both $r$ and $r'$ are in a ball of radius $6 \rho^{1/3-\gamma}$, 
hence we can replace $\hat{u}_\sigma(r)$ and $\hat{u}_{\sigma^\prime}(r^\prime)$ in the first line by $\hat{u}^<_\sigma(r)$ and $\hat{u}^<_{\sigma^\prime}(r^\prime)$ defined in \eqref{eq: u< u> gamma}. With the notation introduced in \eqref{eq: def u<>t}, we thus obtain with Cauchy--Schwarz  
\begin{align*}
  \langle\psi, \mathrm{I}^<_4\psi\rangle &\leq C\sum_{\sigma\neq \sigma^\prime}\int_0^{\infty}dt\, e^{-2t\varepsilon}e^{-t(k_F^\sigma)^2} e^{-t(k_F^{\sigma^\prime})^2}  \int dxdydzdz^\prime   
  \\
  &\qquad \times |\Delta\varphi^<(x-y)||\Delta\varphi^<(z-z^\prime)|| u_{t,\sigma}^<(x-z)||u_{t,\sigma^\prime}^<(y-z^\prime)|\\
  &\qquad \times |\langle\psi, a^\ast_\sigma(v_{t,x})a_{\sigma^\prime}^\ast(v_{t,y})a_{\sigma^\prime}(v_{t,y})a_{\sigma}(v_{t,x})\psi\rangle|.
\end{align*}
We can bound 
\begin{align*}
  \int dz dz' \, |\Delta\varphi^<(z-z')|| u^<_{t,\sigma}(x-z)||u^<_{t,\sigma^\prime}(y-z^\prime)|  \leq \|\Delta\varphi^<\|_1 \|{u}^<_{t,\sigma}\|_2 \|{u}^<_{t,\sigma^\prime}\|_2 
\end{align*}
uniformly in $x$ and $y$. With an application of Lemmas \ref{lem:4a}, \ref{lem:1a},  \ref{lem:t} and \ref{lem: bounds phi}, we thus obtain 
\begin{align*} 
\langle \psi, \mathrm{I}_4^< \psi \rangle &\leq C\rho \left( \int _{0}^\infty dt\, e^{-2t\varepsilon}e^{t(k_F^\sigma)^2} e^{t(k_F^{\sigma^\prime})^2} \|{u}^<_{t,\sigma}\|_2 \|{u}^<_{t,\sigma^\prime}\|_2\right)   \langle \psi,\mathcal{N}  \psi\rangle  \leq C\rho^{\frac{4}{3} - \gamma}\langle \psi,  \mathcal{N}\psi\rangle .
\end{align*}
Combining the latter bound with \eqref{eq: est Ia 1 R*R} and \eqref{I4b}, and using that $\gamma < 1/6$ by assumption, 
we conclude that for all $0\le \beta<1$
\begin{align}\label{eq:RR-I4-final}
  \langle\psi, \mathrm{I}_{4} \psi \rangle \leq C\Big( \rho^{1+\gamma} +\rho^{ 1 +\frac{\beta}{6}} \Big)  \langle \psi, \mathcal{N}\psi\rangle + C\rho  \|\mathcal{N}_\beta^{\frac{1}{2}}\psi\|  \|\mathcal{N}^{\frac{1}{2}}\psi\|. 
  \end{align}


\medskip
\noindent {\bf Analysis of $\mathrm{I}_{5}$.} The term $\mathrm{I}_{5}$ in \eqref{eq:RR-dec} can be written as
\begin{align}\label{eq:R*R-I5-first-dec}
  \mathrm{I}_5 &= \frac{1}{L^6}\sum_{\sigma \neq \sigma^\prime}\sum_{r,r^\prime}||r|^2 - (k_F^\sigma)^2|\hat{v}_\sigma(r)\hat{v}_{\sigma^\prime}(r^\prime) \nn\\
  &\qquad \qquad \times \left|\sum_p  \widehat{\omega}^\varepsilon_{r,r^\prime}(p) \hat{u}_\sigma(r+p)\hat{u}_{\sigma^\prime}(r^\prime - p)\hat{a}_{r+p,\sigma}\hat{a}_{r^\prime - p, \sigma^\prime}\right|^2. 
\end{align}
Using \eqref{eq: int t conf space} and the Cauchy--Schwarz inequality with respect to $t$, we cancel the factor $||r|^2 - (k_F^\sigma)^2|$ similarly as in the analysis of $\mathrm{I}_4$. We have
\begin{align*}
  \mathrm{I}_5  &= \frac{1}{L^6}\sum_{\sigma\neq \sigma^\prime}\sum_{r,r^\prime}||r|^2 - (k_F^\sigma)^2|\hat{v}_\sigma(r)\hat{v}_{\sigma^\prime}(r^\prime)\bigg| \int_0^{\infty} dt\,  e^{-2 t \eps} e^{t|r|^2 } e^{t|r^\prime|^2 } 
  \\
 &\qquad\qquad \times  \sum_p   2|p|^2 \hat{\varphi}(p)  \hat{u}_{t,\sigma}(r+p)\hat{u}_{t,\sigma^\prime}(r^\prime - p)\hat{a}_{r+p,\sigma}\hat{a}_{r^\prime - p, \sigma^\prime}\bigg|^2 
 \\
& \le   \frac{1}{L^6}\sum_{\sigma\neq \sigma^\prime}\int_0^{\infty} dt\, e^{-2t\varepsilon} e^{t(k_F^\sigma)^2}e^{t(k_F^{\sigma^\prime})^2} \sum_{r,r^\prime} \hat{v}_{t,\sigma}(r)\hat{v}_{t,\sigma^\prime}(r^\prime) 
  \\
&  \qquad\qquad \times \bigg|\sum_p   2|p|^2 \hat{\varphi}(p)   \hat{u}_{t,\sigma}(r+p)\hat{u}_{t,\sigma^\prime}(r^\prime - p)\hat{a}_{r+p,\sigma}\hat{a}_{r^\prime - p, \sigma^\prime}\bigg|^2.
\end{align*}
As for the term $\mathrm{I}_4$, we use $1 = \widehat{\chi}_>(p) + \widehat{\chi}_<(p)$ and the Cauchy--Schwarz inequality to bound $\mathrm{I}_5 \leq 2( \mathrm{I}_5^>+   \mathrm{I}_5^< )$, with
\begin{align}\label{eq:RR-I5-dec}
\mathrm{I}_5^\gtrless &=  \frac{1}{L^6}\sum_{\sigma\neq \sigma^\prime}\int_0^{\infty} dt\, e^{-2t\varepsilon}  e^{t(k_F^\sigma)^2}e^{t(k_F^{\sigma^\prime})^2} \sum_{r,r^\prime}\hat{v}_{t,\sigma}(r)\hat{v}_{t,\sigma^\prime}(r^\prime)  \nn\\
  &\qquad\quad \times \bigg|\sum_p   2|p|^2 \hat{\varphi}^\gtrless(p)  \hat{u}_{t,\sigma}^\gtrless(r+p)\hat{u}_{t,\sigma^\prime}^\gtrless(r^\prime - p)\hat{a}_{r+p,\sigma}\hat{a}_{r^\prime - p, \sigma^\prime}\bigg|^2 .
\end{align}
The support properties of $\hat\varphi^\gtrless$ allowed us to  replace  $\hat{u}_\sigma(r+p)$ and $\hat{u}_{\sigma^\prime}(r^\prime -p)$ by $\hat{u}_\sigma^\gtrless(r+p)$ and $\hat{u}_{\sigma^\prime}^\gtrless(r^\prime - p)$. 

We start with $\mathrm{I}_5^>$, which we write out as  
\begin{multline*}
  \mathrm{I}_5^> = \frac{4}{L^6}\sum_{\sigma\neq \sigma^\prime}\int_0^{\infty} \! dt\, \sum_{p,q,r,r^\prime} e^{-2t\varepsilon}e^{t(k_F^\sigma)^2}e^{t(k_F^{\sigma^\prime})^2}  |p|^2 \hat{\varphi}^>(p)  |q|^2 \hat{\varphi}^>(q)  \hat{v}_{t,\sigma}(r)\hat{v}_{t,\sigma^\prime}(r^\prime)  
  \\
  \times \hat{u}^>_{t,\sigma}(r+p)\hat{u}_{t,\sigma^\prime}^>(r^\prime - p)\hat{u}_{t,\sigma}^>(r+q)\hat{u}_{t,\sigma^\prime}^>(r^\prime - q)  \hat{a}_{r+p,\sigma}^\ast \hat{a}_{r^\prime - p, \sigma^\prime}^\ast\hat{a}_{r^\prime - q, \sigma^\prime}\hat{a}_{r+q,\sigma}.
\end{multline*}
We first multiply and divide by $|r+p|^2|r^\prime - q|^2$ and write
\begin{multline*}
  \mathrm{I}_5^> = \frac{4}{L^6}\sum_{\sigma\neq \sigma^\prime}\int_0^{\infty} \! dt\, \sum_{p,q,r,r^\prime} e^{-2t\varepsilon}e^{t(k_F^\sigma)^2}e^{t(k_F^{\sigma^\prime})^2}  |p|^2 \hat{\varphi}^>(p)  |q|^2 \hat{\varphi}^>(q)  \hat{v}_{t,\sigma}(r)\hat{v}_{t,\sigma^\prime}(r^\prime)
  \\
  \times |r+p|^2|r^\prime - q|^2 \frac{\hat{u}^>_{t,\sigma}(r+p)}{|r+p|^2}\frac{\hat{u}_{t,\sigma^\prime}^>(r^\prime - p)}{|r^\prime-p|^2}\hat{u}_{t,\sigma}^>(r+q)\hat{u}_{t,\sigma^\prime}^>(r^\prime - q)
  \\
  \times  \hat{a}_{r+p,\sigma}^\ast \hat{a}_{r^\prime - p, \sigma^\prime}^\ast\hat{a}_{r^\prime - q, \sigma^\prime}\hat{a}_{r+q,\sigma}.
\end{multline*}
In the numerator we use
\[
  |r+p|^2 = (r+p)\cdot (r + r^\prime) - (r+p)\cdot (r^\prime -p) ,
\quad |r^\prime - q|^2 = (r^\prime -q)\cdot (r+r^\prime)- (r^\prime - q) \cdot (r+q).
\]
This splits $\mathrm{I}_5^>$ into four terms $\mathrm{I}_{5;j}^>$ with $1\leq j\leq 4$, according to the presence of 
\begin{enumerate}
\item $(r+p)\cdot (r^\prime -p)(r^\prime - q)\cdot (r+q)$,
\item $- (r+p)\cdot (r + r^\prime)(r^\prime -q)\cdot (r+q)$, 
\item $- (r+p)\cdot (r^\prime -p)(r^\prime -q)\cdot (r+r^\prime)$,
\item $ (r+p)\cdot (r + r^\prime) (r^\prime -q)\cdot (r+r^\prime)$.
\end{enumerate}
The first term $\mathrm{I}_{5;1}^>$ can be written in configuration space as 
\begin{align*}
  \mathrm{I}_{5;1}^> &= 
4\sum_{\sigma\neq\sigma^\prime}\sum_{\ell, m =1}^3\int_0^{\infty}dt\,  e^{-2t\varepsilon}e^{t(k_F^\sigma)^2+t(k_F^{\sigma^\prime})^2} \int dxdydzdz^\prime \Delta{\varphi}^>(x-y)\Delta{\varphi}^>(z-z^\prime)  \\
&\qquad \times v_{t,\sigma}(x-z)v_{t,\sigma^\prime}(y-z^\prime) a^\ast_\sigma(\partial_\ell \widetilde{u}_{t,x}^>)a_{\sigma^\prime}^\ast(\partial_\ell u_{t,y}^>)a_{\sigma^\prime}(\partial_m \widetilde{u}_{t,z^\prime}^>)a_\sigma(\partial_m u_{t,z}^>),
\end{align*}
where we used the notation $\widetilde{u}^>_{t,\sigma}$ introduced in Lemma \ref{lem:t}.
Applying Lemmas \ref{lem:4aa}, \ref{lem:1a}, \ref{lem:t} and \ref{lem: bounds phi},  we can bound
\begin{align}\nn 
&  \langle \psi, \mathrm{I}_{5;1}^>\psi\rangle 
\\ &    \leq C\|\Delta{\varphi}^>\|_1^2  \left(\langle \psi, \mathbb{H}_0 \psi\rangle + \rho^{\frac{2}{3}}\langle \psi, \mathcal{N}\psi\rangle\right)\nonumber \\
&\quad \times 
\int_0^{\infty}dt\, e^{-2t\varepsilon}e^{t(k_F^\sigma)^2+t(k_F^{\sigma^\prime})^2} 
   \|{v}_{t,\sigma}\|_2\|{v}_{t,\sigma^\prime}\|_2 \||\cdot|^{-1}\hat{{u}}^>_{t,\sigma}\|_2 \||\cdot|^{-1}\hat{{u}}^>_{t,\sigma^\prime}\|_2 \| \hat u_{t,\sigma}^>\|_\infty  \| \hat u_{t,\sigma'}^>\|_\infty \nn \\
   & \leq C\rho^{\frac{2}{3} + \gamma}\langle \psi, \mathbb{H}_0 \psi\rangle + C\rho^{\frac{4}{3} + \gamma}\langle \psi, \mathcal{N}\psi\rangle. 
   \label{eq:RR-I51}
\end{align}
The second term is given by  
\begin{multline*}
\mathrm{I}^>_{5;2} = -4\sum_{\sigma\neq\sigma^\prime}\sum_{\ell, m =1}^3\int_0^{\infty}dt\, e^{-2t\varepsilon} e^{t(k_F^\sigma)^2} e^{t(k_F^{\sigma^\prime})^2} \int dxdydzdz^\prime \Delta{\varphi}^>(x-y)\Delta{\varphi}^>(z-z^\prime)  
\\
\times  \left(\partial_\ell v_{t,\sigma}(x-z)v_{t,\sigma^\prime}(y-z^\prime)  + v_{t,\sigma}(x-z)\partial_\ell v_{t,\sigma^\prime}(y-z^\prime)\right) \\
\times a^\ast_\sigma(\partial_\ell \widetilde{u}_{t,x}^>)a_{\sigma^\prime}^\ast(u_{t,y}^>)a_{\sigma^\prime}(\partial_m \widetilde{u}_{t,z^\prime}^>)a_\sigma(\partial_m u_{t,z}^>)
\end{multline*}
and can be bounded in the same way, yielding 
\begin{align}\label{eq:RR-I52}
&  \langle \psi, \mathrm{I}_{5;2}^>\psi\rangle
  \leq C\rho^{1+\gamma}\|\mathcal{N}^{\frac{1}{2}}\psi\|\|\mathbb{H}_0^{\frac{1}{2}}\psi\| + C\rho^{\frac{4}{3} + \gamma}\langle \psi, \mathcal{N}\psi\rangle. 
\end{align}
The same applies to the third term $\mathrm{I}_{5;3}^>$, and we omit the details. Finally,
\begin{multline*}
\mathrm{I}^>_{5;4} = 4\sum_{\sigma\neq\sigma^\prime}\sum_{\ell, m =1}^3\int_0^{\infty}dt\,  e^{-2t\varepsilon}e^{t(k_F^\sigma)^2} e^{t(k_F^{\sigma^\prime})^2} \int dxdydzdz^\prime \Delta{\varphi}^>(x-y)\Delta{\varphi}^>(z-z^\prime) 
\\
 \times \left[\sum_{j,k= 0}^1(\partial_\ell^j\partial_m^k v_{t,\sigma}(x-z))\, (\partial_\ell^{1-j}\partial_m^{1-k} v_{t,\sigma^\prime}(y-z^\prime))\right] \\ \times a^\ast_\sigma(\partial_\ell \widetilde{u}_{t,x}^>)a_{\sigma^\prime}^\ast(u_{t,y}^>)a_{\sigma^\prime}(\partial_m \widetilde{u}_{t,z^\prime}^>)a_\sigma( u_{t,z}^>).
\end{multline*}
Proceeding as for the other terms, we obtain  
\begin{align}\label{eq:RR-I53}
  \langle \psi, \mathrm{I}_{5;4}^>\psi\rangle
  &\leq C\rho^{\frac{4}{3} + \gamma}\langle \psi,\mathcal{N}\psi\rangle.
\end{align}
Collecting the  estimates in \eqref{eq:RR-I51}--\eqref{eq:RR-I53}, we find that $\mathrm{I}_5^>$ in \eqref{eq:RR-I5-dec} satisfies
\begin{align}\label{eq:RR-I5>}
  \langle \psi, \mathrm{I}_5^>\psi\rangle 
  &\leq C\rho^{\frac{2}{3}+\gamma}\langle \psi, \mathbb{H}_0 \psi\rangle + C\rho^{\frac{4}{3} +\gamma}\langle \psi,\mathcal{N}\psi\rangle.
\end{align}

We are left with analyzing the second term $\mathrm{I}_5^<$ in \eqref{eq:RR-I5-dec}, which can be written as 
\begin{align*}
  \mathrm{I}_{5}^< &= 
4\sum_{\sigma\neq\sigma^\prime}\int_0^{\infty}dt\,  e^{-2t\varepsilon}e^{t(k_F^\sigma)^2+t(k_F^{\sigma^\prime})^2} \int dxdydzdz^\prime \Delta{\varphi}^<(x-y)\Delta{\varphi}^<(z-z^\prime)  \\
&\qquad \times v_{t,\sigma}(x-z)v_{t,\sigma^\prime}(y-z^\prime) a^\ast_\sigma({u}_{t,x}^<)a_{\sigma^\prime}^\ast( u_{t,y}^<)a_{\sigma^\prime}({u}_{t,z^\prime}^<)a_\sigma(u_{t,z}^<).
\end{align*}
With Lemma \ref{lem:4aa} and the Cauchy--Schwarz inequality, we find
\begin{align}\label{eq:RR-I5<}
  \langle \psi, \mathrm{I}_5^<\psi\rangle
  &\leq C\|\Delta\varphi^<\|_1^2 \sum_{\sigma\neq \sigma'} \int_0^\infty dt\, e^{-2t\varepsilon}\|\hat{u}^<_{t,\sigma}\|_2 \|\hat{v}_{t,\sigma}\|_2 \|\hat{u}^<_{t,\sigma^\prime}\|_2 \|\hat{v}_{t,\sigma^\prime}\|_2 \langle \psi,\mathcal{N}\psi\rangle 
 \nn \\
  &\leq C \rho^{1-\frac{3}{2}\gamma}\left(\int_0^\infty dt\, e^{-2t\varepsilon}\|\hat{u}^<_{t,\sigma}\|_2 \|\hat{v}_{t,\sigma}\|_2 \right) \langle \psi,\mathcal{N}\psi\rangle \leq C\rho^{\frac{4}{3} -2\gamma - \kappa} \langle \psi,\mathcal{N}\psi\rangle
\end{align}
for any $\kappa>0$. Here we also used Lemma \ref{lem:1a} and  \eqref{eq: est int t vu< final}. Combining \eqref{eq:RR-I5>} and \eqref{eq:RR-I5<}, we conclude that for all $\kappa>0$, 
\begin{align}\label{eq:RR-I5-final}
  \langle \psi, \mathrm{I}_5 \psi\rangle \leq C\rho^{\frac{2}{3}+\gamma}\langle \psi, \mathbb{H}_0 \psi\rangle   + C\rho^{\frac{4}{3} -2\gamma -\kappa} \langle \psi,\mathcal{N}\psi\rangle.
\end{align}

\medskip
\noindent {\bf Analysis of $\mathrm{I}_{6}$.} Next, we consider the error term $\mathrm{I}_6$ in \eqref{eq:RR-dec}. We start by proceeding as with $\mathrm{I}_5$ above; i.e., we apply \eqref{eq: int t conf space}  and a Cauchy--Schwarz inequality to cancel the factor  $\left(|r+p|^2 -|r|^2\right)$ via \eqref{eq:int-t}. We further decompose $1 = \widehat{\chi}_<(p) + \widehat{\chi}_>(p)$ and obtain, in analogy to \eqref{eq:RR-I5-dec}, 
$\mathrm{I}_6 \leq 2( \mathrm{I}_6^>+   \mathrm{I}_6^< )$, with
\begin{align}\label{eq:RR-I6-dec}
\mathrm{I}_6^\gtrless& = \frac{4}{L^6}\sum_{\sigma\neq\sigma^\prime}\int_0^{\infty}dt\, e^{-2t\varepsilon}\sum_{p,r} \hat{u}^\gtrless_{t,\sigma}(r+p)\hat{v}_{t,\sigma}(r)  \nn \\
  &\qquad\qquad \times  \left|\sum_{r^\prime} |p|^2\hat{\varphi}^\gtrless(p) \hat{u}^\gtrless_{t,\sigma^\prime}(r^\prime - p)\hat{v}_{t,\sigma^\prime}(r^\prime)\hat{a}_{r^\prime -p,\sigma^\prime}\hat{a}_{-r^\prime, \sigma^\prime}\right|^2 .  
\end{align}
We first consider $\mathrm{I}_6^>$. It is convenient to multiply and divide  by $|r+p|^2$ and use that $|r+p|^2 = -(r+p)\cdot (r^\prime -p) + (r +p)\cdot (r+r^\prime)$ in the numerator, where $r^\prime$ denotes one of the two independent summation indices arising from expanding the square in \eqref{eq:RR-I6-dec}. 
Correspondingly, we find that $\mathrm{I}_6^>$ is given in configuration space as
\begin{align*}
  \mathrm{I}_6^>
  & = -4\sum_{\sigma\neq \sigma^\prime}\sum_{\ell=1}^3\int_0^{\infty} dt\, e^{-2t\varepsilon}\int dxdydz\, (\Delta{\varphi}^>\ast \Delta{\varphi}^>)(x-y-z)  \\
  &\qquad\qquad \times \Big( \partial_\ell \widetilde{u}^>_{t,\sigma}(x) v_{t,\sigma}(x) a^\ast_{\sigma^\prime}(\partial_\ell u^>_{t,y})a_{\sigma^\prime}^\ast(v_{t,y})a_{\sigma^\prime}(v_{t,z}) a_{\sigma^\prime}(u^>_{t,z})
  \\
 &\qquad\qquad\qquad -  \partial_\ell \widetilde{u}^>_{t,\sigma}(x) \partial_\ell v_{t,\sigma}(x) a^\ast_{\sigma^\prime}(u^>_{t,y})a_{\sigma^\prime}^\ast(v_{t,y})a_{\sigma^\prime}(v_{t,z}) a_{\sigma^\prime}(u^>_{t,z})
  \\
   &\qquad\qquad \qquad - \partial_\ell \widetilde{u}^>_{t,\sigma}(x) v_{t,\sigma}(x) a^\ast_{\sigma^\prime}( u^>_{t,y})a_{\sigma^\prime}^\ast(\partial_\ell v_{t,y})a_{\sigma^\prime}(v_{t,z}) a_{\sigma^\prime}(u^>_{t,z}) \Big). 
\end{align*}
To bound it, we can use \eqref{eq:Pauli}, Lemma \ref{lem:1a}  and \eqref{eq: est int t vu> final}  to bound the integral with respect to $t$. Note that $\|\Delta{\varphi}^>\ast \Delta{\varphi}^>\|_1\leq \|\Delta{\varphi}^>\|_1^2\leq C$ by Lemma \ref{lem: bounds phi} and 
\[
  \int |\partial_\ell \widetilde{u}^>_{t,\sigma}(x)||\partial_\ell^nv_{t,\sigma}(x)| \leq C\rho^{\frac{n}{3}}\||\cdot|^{-1}\hat{u}^>_{t,\sigma}\|_2 \|{v}_{t,\sigma}\|_2.
\]
We find, for any $\kappa>0$, 
\begin{align}\label{eq:RR-I6>}
  \langle \psi,\mathrm{I}_6^>\psi\rangle 
   &\leq C\rho^{1+\frac{\gamma}{2} - \kappa}\|\mathbb{H}_0^{\frac{1}{2}}\psi\|\|\mathcal{N}^{\frac{1}{2}}\psi\| + C\rho^{\frac{4}{3} +\frac{\gamma}{2} -\kappa}\langle \psi,\mathcal{N}\psi\rangle. 
\end{align}

The term $\mathrm{I}_6^<$ in \eqref{eq:RR-I6-dec} is given in configuration space as 
\begin{align*}
  \mathrm{I}_6^< &= 4\sum_{\sigma\neq\sigma^\prime}\int_0^{\infty}dt\,e^{-2t\varepsilon} \int dxdy dz\, (\Delta{\varphi}^<\ast \Delta{\varphi}^<) (x-y-z)  \\
  &\qquad \times u^<_{t,\sigma}(x) v_{t,\sigma}(x) a^\ast_{\sigma^\prime}(u^<_{t,y})a_{\sigma^\prime}^\ast(v_{t,y})a_{\sigma^\prime}(v_{t,z}) a_{\sigma^\prime}(u^<_{t,z}).
\end{align*}
Bounding it in the same way as for the terms in $\mathrm{I}_6^>$ above, this time using \eqref{eq: est int t vu< final},  we obtain 
\[
  \langle\psi, \mathrm{I}_6^<\psi\rangle 
  \leq C\rho^{\frac{4}{3} - \frac{\gamma}{2} - \kappa}\langle \psi, \mathcal{N}\psi\rangle
\]
for any $\kappa>0$.  We thus conclude that 
\begin{align}\label{eq:RR-I6-final}
  \langle \psi, \mathrm{I}_6 \psi\rangle \leq C\rho^{1+\frac{\gamma}{2} - \kappa}\|\mathbb{H}_0^{\frac{1}{2}}\psi\|\|\mathcal{N}^{\frac{1}{2}}\psi\| + C\rho^{\frac{4}{3} - \frac{\gamma}{2} -\kappa}\langle \psi,\mathcal{N}\psi\rangle.
\end{align}

\medskip
\noindent {\bf Analysis of $\mathrm{I}_{7}$ and $\mathrm{I}_{8}$.} The term  $\mathrm{I}_7$ in \eqref{eq:RR-dec}  can  be written after a change of variables as 
\begin{align}\nn
  \mathrm{I}_7 &= \frac{1}{L^6}\sum_{\sigma\neq \sigma^\prime}\sum_{r,r^\prime}||r|^2 - (k_F^\sigma)^2|\hat{u}_\sigma(r)\hat{v}_{\sigma^\prime}(r^\prime) \\
  &\qquad \times \left|\sum_p \hat{\omega}^\varepsilon_{r-p,r^\prime}(p) \hat{u}_{\sigma^\prime}(r^\prime - p)\hat{v}_\sigma(r-p) \hat{a}_{-r+p,\sigma}\hat{a}_{r^\prime -p, \sigma^\prime}\right|^2 .
  \label{eq:TT-I7-full}
\end{align}

Proceeding with the same steps as for the terms $\mathrm{I}_5$ and $\mathrm{I}_6$ above, we bound $  \mathrm{I}_7
\leq 2( \mathrm{I}_7^>+ \mathrm{I}_7^<)$ with 
 \begin{align}\label{eq:RR-I7-dec}
 \mathrm{I}_7^\gtrless &=  \frac{4}{L^6}\sum_{\sigma\neq \sigma^\prime}\sum_{r,r^\prime}\int_0^{\infty}dt\, e^{-2t\varepsilon} e^{-t(k_F^\sigma)^2}e^{t(k_F^{\sigma^\prime})^2}\hat{u}_{t,\sigma}^\gtrless(r)\hat{v}_{t,\sigma^\prime}(r^\prime)  \nn\\
  &\qquad \times \left|\sum_p |p|^2\hat{\varphi}^\gtrless(p)\hat{u}^\gtrless_{t,\sigma^\prime}(r^\prime -p)\hat{v}_{t,\sigma}(r-p)\hat{a}_{-r+p,\sigma}\hat{a}_{r^\prime -p, \sigma^\prime}\right|^2 .
\end{align}
To bound  $\mathrm{I}_7^>$,  it is useful to multiply and divide by $|r|^2$ and decompose the numerator as  
$$|r|^2 = r\cdot (r-p) - r\cdot (r^\prime - p) + r\cdot r^\prime.$$ 
This gives
\begin{align*}
  \mathrm{I}_7^> &= 4\sum_{\sigma\neq \sigma^\prime} \sum_{\ell=1}^3  \int_0^{\infty} dt\, e^{-2t\varepsilon} e^{t(k_F^{\sigma^\prime})^2-t(k_F^\sigma)^2}\int dxdydzdz^\prime\, \Delta{\varphi}^>(x-y)\Delta{\varphi}^>(z-z^\prime)   
  \\
  &\qquad \times \Big(\partial_\ell \widetilde{u}^>_{t,\sigma}(x-z)v_{t,\sigma^\prime}(y-z^\prime)a^\ast_{\sigma^\prime}(u_{t,z^\prime}^>)a_{\sigma}^\ast(v_{t,z})a_\sigma(\partial_\ell v_{t,x})a_{\sigma^\prime}(u_{t,y}^>)
  \\
  &\qquad \qquad -  \partial_\ell \widetilde{u}^>_{t,\sigma}(x-z)v_{t,\sigma}(y-z^\prime)a^\ast_{\sigma^\prime}(u^>_{t,z^\prime})a_{\sigma}^\ast(v_{t,z})a_\sigma(v_{t,x})a_{\sigma^\prime}(\partial_\ell u^>_{t,y})
  \\
  &\qquad \qquad + \partial_\ell \widetilde{u}^>_{t,\sigma}(x-z)\partial_\ell v_{t,\sigma^\prime}(y-z^\prime)a^\ast_{\sigma^\prime}(u^>_{t,z^\prime})a_{\sigma}^\ast(v_{t,z})a_\sigma( v_{t,x})a_{\sigma^\prime}(u^>_{t,y}) \Big)= \sum_{j=1}^3 \mathrm{I}_{7;j}^>. 
\end{align*}
With Lemmas \ref{lem:4aa}, \ref{lem:1a}, \ref{lem: bounds phi} and \eqref{eq: est int t vu> final}, 
\begin{align*}
  |\langle\psi, \mathrm{I}_{7;1}^>\psi\rangle| 
  &\leq C\rho^{\frac{4}{3}}\|\Delta{\varphi}^>\|_1^2\left(\int_0^{\infty} dt\, \| |\cdot|^{-1} \hat{u}^>_{t,\sigma}\|_2 \|{v}_{t,\sigma}\|_2\right)\langle \psi,\mathcal{N}\psi\rangle \leq C\rho^{\frac{4}{3} + \frac{\gamma}{2} - \kappa}\langle\psi,\mathcal{N}\psi\rangle
\end{align*}
for any $\kappa>0$. The estimates for $\mathrm{I}_{7;2}^>$ and $\mathrm{I}_{7;3}^>$ can be done similarly,  using also  \eqref{eq: est N u>t} from Lemma \ref{lem:1a} which gives an extra error term involving $\mathbb{H}_0$. We  obtain 
\[
  |\langle\psi, (\mathrm{I}_{7;2}^> + \mathrm{I}_{7;3}^>)\psi\rangle| \leq C\rho^{1+\frac{\gamma}{2} - \kappa}\|\mathcal{N}^{\frac{1}{2}}\psi\|\|\mathbb{H}_0^{\frac{1}{2}}\psi\| + C\rho^{\frac{4}{3} + \frac{\gamma}{2} - \kappa}\langle\psi,\mathcal{N}\psi\rangle
\]
for all $\kappa>0$. 
Thus the  term $\mathrm{I}_{7}^>$ in  \eqref{eq:RR-I7-dec} is bounded by 
\[
  \langle \psi, \mathrm{I}_{7}^>\psi\rangle \leq C\rho^{1+\frac{\gamma}{2} - \kappa}\|\mathcal{N}^{\frac{1}{2}}\psi\|\|\mathbb{H}_0^{\frac{1}{2}}\psi\| + C\rho^{\frac{4}{3} + \frac{\gamma}{2} - \kappa}\langle\psi,\mathcal{N}\psi\rangle.
\]

The term $\mathrm{I}_7^<$ in \eqref{eq:RR-I7-dec} can be written as 
\begin{multline*}
  \mathrm{I}_7^< = 4\sum_{\sigma\neq \sigma^\prime}\int_0^{\infty} dt\, e^{t(k_F^{\sigma^\prime})^2}e^{-t(k_F^\sigma)^2}\int dxdydzdz^\prime \, \Delta\varphi^< (x-y) \Delta\varphi^<(z-z^\prime)  
  \\
  \times u^<_{t,\sigma}(x-z)v_{t,\sigma^\prime}(y-z^\prime) a^\ast_{\sigma^\prime}(u^<_{t,z^\prime})a^\ast_\sigma(v_{t,z})a_\sigma(v_{t,x})a_{\sigma^\prime}(u^<_{t,y})
\end{multline*}
and estimated similarly as for $\mathrm{I}_7^>$, with the result that 
\[
  \langle \psi, \mathrm{I}_{7}^<\psi\rangle \leq C\rho \|\Delta\varphi^<\|_1^2 \left(\int_0^{\infty} dt\, \|\hat{u}^<_{t,\sigma}\|_2 \|\hat{v}_{t,\sigma}\|_2\right)\langle \psi, \mathcal{N}\psi\rangle \leq C\rho^{\frac{4}{3} - \frac{\gamma}{2}- \kappa}\langle \psi, \mathcal{N}\psi\rangle
\]
for any $\kappa>0$. Combining the bounds, we obtain 
\begin{align}\label{eq:RR-I7-final}
  \langle \psi, \mathrm{I}_7\psi\rangle \leq C\rho^{1+\frac{\gamma}{2} - \kappa}\|\mathcal{N}^{\frac{1}{2}}\psi\|\|\mathbb{H}_0^{\frac{1}{2}}\psi\| + C\rho^{\frac{4}{3} - \frac{\gamma}{2} - \kappa}\langle \psi, \mathcal{N}\psi\rangle.
\end{align}
The term $\mathrm{I}_8$ can be estimated in the same way as $\mathrm{I}_7$, with the same result; we omit the details. 

Combining  the bounds in \eqref{eq:RR-I1-final},  \eqref{eq:RR-I2-final},   
 \eqref{eq:RR-I4-final},  \eqref{eq:RR-I5-final},  \eqref{eq:RR-I6-final} and \eqref{eq:RR-I7-final}, and using also $\rho^{1+\frac{\gamma}{2} -\kappa}\|\mathbb{H}_0^{\frac{1}{2}}\psi\|\|\mathcal{N}^{\frac{1}{2}}\psi\| \le \rho^{ \frac 2 3 +\gamma} \|\mathbb{H}_0^{\frac{1}{2}}\psi\|^2 +\rho^{ \frac 4 3 -2\kappa}  \|\mathcal{N}^{\frac{1}{2}}\psi\|^2$,  we obtain the statement  of Lemma \ref{lem: T}. 
\end{proof}


\subsection{Proof of Lemma \ref{lem: S1} }\label{sec: lem: S1}
\begin{proof}
Recall the definitions \eqref{eq: def Sk} and \eqref{eq: def S1 S2}. 
We shall write 
\begin{align}\label{eq:S-S1-S2}
S^\ast_{1,\sigma}(r)=S^\ast_{1,1,\sigma}(r)+S^\ast_{1,2,\sigma}(r)
\end{align}
with 
\begin{align}\nn
  S^\ast_{1,1,\sigma}(r) &= \frac{1}{L^3}\sum_{p}\hat{\varphi}^>(p) (D^{<\alpha}_{p,\sigma^\prime})^\ast \hat{a}_{p-r,\sigma}\left(\hat{v}_\sigma(r-p)\hat{u}_\sigma(r) - \hat{u}_\sigma(p-r)\hat{v}_\sigma(r) \right),\\
  S^\ast_{1,2,\sigma}(r) &= - \frac{1}{L^3}\sum_{p}\hat{\varphi}^>(p)\hat{u}_\sigma(r)\hat{{u}}^{<\alpha}_\sigma(r+p)\hat{a}^\ast_{r+p,\sigma}b_{p,\sigma^\prime} \label{def:s1i}
\end{align}
where $\sigma' \neq \sigma$. 

The main quadratic contribution comes from the terms involving $\{S^\ast_{1,1,\sigma}(r), S_{1,1,\sigma}(r)\}$  {and $\{S^\ast_{1,2,\sigma}(r), S_{1,2,\sigma}(r)\}$}, while the others are error terms. In the following, we will often use Lemma \ref{lem:1a}, Lemma \ref{lem: bounds phi} and that 
\begin{equation}\label{eq: est u<alpha vinfty}
\|{u}^{<\alpha}_\sigma\|_2 \leq C\rho^{\frac{1}{2} - \frac{3}{2}\alpha}, \quad \|\partial^n v_\sigma\|_\infty\leq C\rho^{1+\frac{n}{3}}, \quad \|u^{<\alpha}_{\sigma}\|_\infty \leq C\rho^{1-3\alpha}.
\end{equation}

\medskip
\noindent {\bf Analysis of $\{S^\ast_{1,1,\sigma}(r), S_{1,1,\sigma}(r)\}$.} 
We need to prove an upper bound on 
\begin{align}\label{eq: S*11S11}
  &\sum_\sigma\sum_r||r|^2 - (k_F^\sigma)^2|\{S^\ast_{1,1,\sigma}(r), S_{1,1,\sigma}(r)\}\nonumber
  \\
  &= \frac{1}{L^6}\sum_{\sigma\neq \sigma^\prime}\sum_{p,q,r,r^\prime,s}||r|^2 - (k_F^\sigma)^2|\hat{\varphi}^>(p)\hat{\varphi}^>(q)\nonumber
  \\ & \qquad \times \Big(\hat{u}_\sigma(r) \hat{v}_\sigma(r-p)\hat{v}_\sigma(r-q) +  \hat{v}_\sigma(r)\hat{u}_\sigma(r-p)\hat{u}_\sigma(r-q) \Big)\nonumber
  \\
  &\qquad \times\hat{{u}}^{<\alpha}_{\sigma^\prime}(r^\prime)\hat{u}_{\sigma^\prime}(r^\prime - p)\hat{u}_{\sigma^\prime}(s -  q)\hat{{u}}^{<\alpha}_{\sigma^\prime}(s) \Big\{ \hat{a}_{r^\prime, \sigma^\prime}^\ast\hat{a}_{r^\prime - p, \sigma^\prime}\hat{a}_{p-r,\sigma}, \hat{a}_{q-r,\sigma}^\ast \hat{a}_{s-q,\sigma^\prime}^\ast\hat{a}_{s,\sigma^\prime} \Big\}  
\end{align}
Computing the anti-commutator
\begin{align}\label{eq: anticommutator S*11S11}
 & \Big\{ \hat{a}_{r^\prime, \sigma^\prime}^\ast\hat{a}_{r^\prime - p, \sigma^\prime}\hat{a}_{p-r,\sigma}, \hat{a}_{q-r,\sigma}^\ast \hat{a}_{s-q,\sigma^\prime}^\ast\hat{a}_{s,\sigma^\prime}\Big\} \nonumber\\ 
 &\quad= \delta_{q,p}\delta_{r^\prime , s} \hat{a}_{s,\sigma^\prime}^\ast\hat{a}_{s,\sigma^\prime} - \delta_{q, p}\hat{a}_{r^\prime, \sigma^\prime}^\ast \hat{a}_{s-p,\sigma^\prime}^\ast \hat{a}_{r^\prime -p,\sigma^\prime}\hat{a}_{s,\sigma^\prime}\nonumber
  \\
  &\qquad\qquad + \delta_{s-q,r^\prime - p}\hat{a}^\ast_{r^\prime, \sigma^\prime}\hat{a}_{q-r,\sigma}^\ast \hat{a}_{s,\sigma^\prime}\hat{a}_{p-r,\sigma} + \delta_{s,r^\prime}\hat{a}_{q-r,\sigma}^\ast \hat{a}_{s-q,\sigma^\prime}^\ast \hat{a}_{r^\prime - p,\sigma'}\hat{a}_{p-r,\sigma}
\end{align}
and suitably changing summation variables, we find
\begin{align}\label{eq:S11-S11-dec}
\sum_\sigma\sum_r ||r|^2 - (k_F^\sigma)^2|\{S^\ast_{1,1,\sigma}(r), S_{1,1,\sigma}(r)\}  
= \sum_{j=1}^6 \mathrm{I}_j
\end{align}
with
\begin{align*}
  \mathrm{I}_1 &= \frac{{1}}{L^6}\sum_{\sigma\neq \sigma^\prime}\sum_{p,r,r^\prime}(|r+ p|^2 - |r|^2)|\hat{\varphi}^>(p)|^2 \hat{u}_\sigma(r+p) \hat{v}_\sigma(r)\hat{{u}}^{<\alpha}_{\sigma^\prime}(r^\prime)\hat{u}_{\sigma^\prime}(r^\prime - p)\hat{a}_{r^\prime, \sigma^\prime}^\ast \hat{a}_{r^\prime, \sigma^\prime},
\\
  \mathrm{I}_2 & = -\frac{1}{L^6}\sum_{\sigma\neq \sigma^\prime}\sum_{p,r,r^\prime,s}(|r+p|^2 - |r|^2)|\hat{\varphi}^>(p)|^2 \hat{u}_\sigma(r+p) \hat{v}_\sigma(r)\hat{{u}}^{<\alpha}_{\sigma^\prime}(r^\prime) \\
&\qquad\qquad\qquad\qquad  \times \hat{u}_{\sigma^\prime}(r^\prime - p)\hat{u}_{\sigma^\prime}(s -  p)\hat{{u}}^{<\alpha}_{\sigma^\prime}(s) \hat{a}_{r^\prime, \sigma^\prime}^\ast \hat{a}_{s-p, \sigma^\prime}^\ast \hat{a}_{r^\prime -p, \sigma^\prime}\hat{a}_{s,\sigma^\prime},
\\
  \mathrm{I}_3 &= \frac{1}{L^6}\sum_{\sigma\neq \sigma^\prime}\sum_{p,q,r,r^\prime}||r+p|^2 - (k_F^\sigma)^2|\hat{\varphi}^>(p)\hat{\varphi}^>(q) \hat{u}_\sigma(r+p) \hat{v}_\sigma(r)\hat{v}_\sigma(r+p-q) 
  \\
  &\qquad\qquad \qquad\qquad
  \times \hat{{u}}^{<\alpha}_{\sigma^\prime}(r^\prime)\hat{u}_{\sigma^\prime}(r^\prime - p)\hat{{u}}^{<\alpha}_{\sigma^\prime}(r^\prime -p+q) \hat{a}_{r^\prime, \sigma^\prime}^\ast \hat{a}_{q-p-r,\sigma}^\ast \hat{a}_{r^\prime -p+q,\sigma^\prime}\hat{a}_{-r,\sigma},
 \\
  \mathrm{I}_4 &= \frac{1}{L^6}\sum_{\sigma\neq \sigma^\prime}\sum_{p,q,r,r^\prime}| {|r+p|^2} - (k_F^\sigma)^2|\hat{\varphi}^>(p)\hat{\varphi}^>(q) \hat{u}_\sigma(r+p) \hat{v}_\sigma(r)\hat{v}_\sigma(r+p-q) 
  \\
 &\qquad\qquad \qquad\qquad \times \hat{{u}}^{<\alpha}_{\sigma^\prime}(r^\prime)\hat{u}_{\sigma^\prime}(r^\prime - p)\hat{u}_{\sigma^\prime}(r^\prime -  q) \hat{a}_{q-p-r,\sigma}^\ast \hat{a}_{r^\prime-q,\sigma^\prime}^\ast \hat{a}_{r^\prime -p,\sigma'}\hat{a}_{-r,\sigma},
\\
  \mathrm{I}_5 &= -\frac{1}{L^6}\sum_{\sigma\neq \sigma^\prime}\sum_{p,q,r,r^\prime}||r|^2 - (k_F^\sigma)^2|\hat{\varphi}^>(p)\hat{\varphi}^>(q) \hat{u}_\sigma(r+p) \hat{v}_\sigma(r)\hat{u}_\sigma(r+q)  
  \\
  &\qquad \qquad \qquad\qquad \times\hat{{u}}^{<\alpha}_{\sigma^\prime}(r^\prime)\hat{u}_{\sigma^\prime}(r^\prime - p)\hat{{u}}^{<\alpha}_{\sigma^\prime}(r^\prime - p+q)\hat{a}_{r+q,\sigma}^\ast \hat{a}_{r^\prime, \sigma^\prime}^\ast \hat{a}_{r^\prime - p+q,\sigma^\prime}\hat{a}_{r+p,\sigma},
 \\
  \mathrm{I}_6 &= \frac{1}{L^6}\sum_{\sigma\neq \sigma^\prime}\sum_{p,q,r,r^\prime}||r|^2 - (k_F^\sigma)^2|\hat{\varphi}^>(p)\hat{\varphi}^>(q) \hat{u}_\sigma(r+p) \hat{v}_\sigma(r)\hat{u}_\sigma(r+q)  
  \\
  &\qquad \qquad \qquad\qquad \times\hat{{u}}^{<\alpha}_{\sigma^\prime}(r^\prime)\hat{u}_{\sigma^\prime}(r^\prime - p)\hat{u}_{\sigma^\prime}(r^\prime -  q) \hat{a}_{r+q,\sigma}^\ast \hat{a}_{r^\prime - q,\sigma^\prime}^\ast \hat{a}_{r^\prime -p,\sigma^\prime}\hat{a}_{r+p,\sigma}.
\end{align*}
More precisely, $\mathrm{I}_1$ is obtained by combining the two terms in \eqref{eq: S*11S11} arising from the quadratic term in the anticommutator \eqref{eq: anticommutator S*11S11}, after the substitution $r\mapsto r-p$ in one of them. Similarly, $\mathrm{I}_2$ combines the two contributions from the second term of the anticommutator, with the same substitution in one of them. The terms $\mathrm{I}_3$ and $\mathrm{I}_5$ arise from the third term of the anticommutator, while $\mathrm{I}_4$ and $\mathrm{I}_6$ arise from the last term. Here we use the substitution $r\mapsto r-p$ in $\mathrm{I}_3$ and $\mathrm{I}_4$, and $r\mapsto -r$ in $\mathrm{I}_5$ and $\mathrm{I}_6$.

We will extract the leading contribution from the term $\mathrm{I}_1$, while  all remaining terms contribute only to the error. 
Note that in the expression for  $\mathrm{I}_1$ in \eqref{eq:S11-S11-dec}, we have  $\hat{u}_\sigma(r+p) =1$ whenever $|r|\leq k_F^\sigma$ and $|p|\geq \rho^{1/3-\gamma}$.  Since $|p+r|^2-|r|^2= |p|^2 + 2 p \cdot r$ and $\sum_r p\cdot r \, \hat v_\sigma(r)=0$, we hence have 
\begin{align}\label{eq:S11-S11-I1}
  \mathrm{I}_1  
  &= \frac{1}{L^3}\sum_{\sigma\neq \sigma^\prime} \rho_\sigma \sum_{p,r^\prime} |p|^2  |\hat{\varphi}^>(p)|^2 \hat{u}^{<\alpha}_{\sigma^\prime}(r^\prime) \hat{u}_{\sigma^\prime}(r^\prime - p)\hat{a}_{r^\prime, \sigma^\prime}^\ast \hat{a}_{r^\prime, \sigma^\prime} \nn\\
  &\le \frac{1}{L^3}\sum_{\sigma\neq \sigma^\prime} \rho_\sigma \sum_{p,r^\prime} |p|^2  |\hat{\varphi}(p)|^2 \hat{u}_{\sigma^\prime}(r^\prime)\hat{a}_{r^\prime, \sigma^\prime}^\ast \hat{a}_{r^\prime, \sigma^\prime},
\end{align}
 where we have used  $|\hat{\varphi}^>(p)| \le |\hat{\varphi}(p)|$ in the last step.

Next, we consider $\mathrm{I}_2$ in \eqref{eq:S11-S11-dec}. Proceeding as for $\mathrm{I}_1$, using that $\hat{u}_\sigma(r+p) =1$ when $|r|\leq k_F^\sigma$ and $|p|\geq \rho^{1/3-\gamma}$,  and that $\sum_r p\cdot  r\, \hat v_\sigma(r)=0$, we find  
\begin{align*}
\mathrm{I}_2 
  &= - \frac{1}{L^3}\sum_{\sigma\neq \sigma^\prime} \rho_\sigma \sum_{p,r^\prime,s} |p|^2 |\hat{\varphi}^>(p)|^2 \hat{{u}}^{<\alpha}_{\sigma^\prime}(r^\prime)\hat{u}_{\sigma^\prime}(r^\prime - p)\hat{u}_{\sigma^\prime}(s -  p)\hat{{u}}^{<\alpha}_{\sigma^\prime}(s) 
  \\
&\qquad\qquad  \times \hat{a}_{r^\prime, \sigma^\prime}^\ast \hat{a}_{s-p, \sigma^\prime}^\ast \hat{a}_{r^\prime -p, \sigma^\prime}\hat{a}_{s,\sigma^\prime}\\
&= \sum_{\sigma\neq \sigma^\prime}\rho_\sigma \int dxdy \, (\Delta{\varphi}^>\ast {\varphi}^>) (x-y) {a}^\ast_{\sigma^\prime}({u}^{<\alpha}_x)a^\ast_{\sigma^\prime}(u_y)a_{\sigma^\prime}(u_x)a_{\sigma^\prime}(u^{<\alpha}_y).
\end{align*}
Using $\|\Delta\varphi^> \ast \varphi^>\|_1  \leq \|\Delta\varphi^>\|_1 \|\varphi^>\|_1\leq C \rho^{-\frac{2}{3} + 2\gamma}$ by Lemma \ref{lem: bounds phi}, the first bound in \eqref{eq: est u<alpha vinfty}, and the Cauchy--Schwarz inequality we can bound  
\begin{align}\label{eq:S11-S11-I2}
  |\langle \psi,\mathrm{I}_{2}\psi\rangle| &\leq C\sum_{\sigma^\prime}\rho^{2-3\alpha} \int dxdy \,|( \Delta{\varphi}^>\ast {\varphi}^>) (x-y)|\|a_{\sigma^\prime}(u_y)\psi\| \|a_{\sigma^\prime}(u_x)\psi\| \nn\\
  &\leq  C\rho^{\frac 4 3 + 2\gamma -3\alpha}\langle\psi,\mathcal{N}\psi\rangle. 
\end{align}

We now consider the term $\mathrm{I}_3$ in \eqref{eq:S11-S11-dec}. We use again $\hat{u}_\sigma(r+p) =1$ in all the summands, and expand
\begin{equation}\label{exp34} 
|r+p|^2 - (k_F^\sigma)^2 = |p|^2 + |r|^2 + 2p\cdot r - (k_F^\sigma)^2
\end{equation}
which leads to four contributions and the corresponding decomposition 
$\mathrm{I}_{3} = \mathrm{I}_{3;a} + \mathrm{I}_{3;b} + \mathrm{I}_{3;c} + \mathrm{I}_{3;d}.$ 
We start with analyzing  $\mathrm{I}_{3;a}$, which is given by 
\[
  \mathrm{I}_{3;a} = \sum_{\sigma\neq \sigma^\prime}\int dxdydz\, (-\Delta{\varphi}^>)(x-y) {\varphi}^>(y-z)u_{\sigma}(x-z)a^\ast_{\sigma^\prime}(\hat{{u}}^{<\alpha}_x){a}^\ast_{\sigma}(v_y){a}_{\sigma^\prime}({u}^{<\alpha}_z)a_{\sigma}(v_y).
  \]
  Using the identity 
\begin{equation}\label{eq: u = delta -v}
u_\sigma(x-z) = \delta(x-z) - v_{\sigma}(x-z),
\end{equation}
we correspondingly split  $\mathrm{I}_{3;a} = \mathrm{I}_{3;a;1} + \mathrm{I}_{3;a;2}$.

By using Lemma \ref{lem:4aaa} together with  $\|v_\sigma\|_2 \leq C\rho^{1/2}$ and $\|(\Delta{\varphi}^> ){\varphi}^>\|_{1} \le \|\Delta{\varphi}^>\|_{1} \|{\varphi}^>\|_{\infty} \le C$ by Lemma \ref{lem: bounds phi}, we can bound $\mathrm{I}_{3;a;1}$ as
\begin{align}\label{eq: I311 R*S error term}
   |\langle \psi, \mathrm{I}_{3;a;1} \psi\rangle| 
  \leq C\rho^{1 + \frac{\beta}{6}}  \langle \psi , \mathcal{N}\psi\rangle+ C\rho \|\mathcal{N}_{\beta}^{\frac{1}{2}}\psi\| \|\mathcal{N}^{\frac{1}{2}}\psi\|.
 \end{align}
 For the term $\mathrm{I}_{3;a;2}$, we use the bounds in \eqref{eq: est u<alpha vinfty},  Lemma \ref{lem:1a}, Lemma \ref{lem: bounds phi} and the Cauchy--Schwarz inequality to obtain
 \begin{align*}
  |\langle \psi, \mathrm{I}_{3;a;2}\psi\rangle| &\leq C\rho\sum_{\sigma\neq \sigma^\prime}\int dxdydz\, |\Delta{\varphi}^>(x-y)||{\varphi}^>(y-z)|  \\
  &\qquad \qquad \qquad \times \|a_\sigma^\ast(v_y)\|\|a_{\sigma}(v_x)\| \|a_{\sigma^\prime}(u^{<\alpha}_x)\psi\| \|a_{\sigma^\prime}(u^{<\alpha}_{z})\psi\|
  \\
 &\leq C\rho^2 \|{\varphi}^>\|_1 \|\Delta{\varphi}^>\|_1 \langle \psi, \mathcal{N}\psi\rangle\leq C\rho^{\frac{4}{3} + 2\gamma}\langle \psi, \mathcal{N}\psi\rangle. 
 \end{align*}
 The remaining terms $\mathrm{I}_{3;b}, \mathrm{I}_{3;c}, \mathrm{I}_{3;d}$ can be estimated similarly, using  
 \begin{equation}\label{eq: a partial ell v S*S}
  \|a_\sigma(\partial_\ell^n v_\cdot)\| \leq \left\||\cdot|^n \hat{v}_\sigma\right\|_2 \leq C\rho^{\frac{1}{2} + \frac{n}{3}},
 \end{equation}
  \eqref{eq: u = delta -v} and Lemma \ref{lem:1a}, with the result that 
\[
  |\langle \psi, (\mathrm{I}_{3;b} + \mathrm{I}_{3;c} + \mathrm{I}_{3;d}) \psi\rangle| 
  \leq C\rho^{1+\gamma}\langle \psi, \mathcal{N}\psi\rangle .  
\]
Combining the bounds, we conclude that
\begin{align}\label{eq:S11-S11-I3}
  |\langle \psi, \mathrm{I}_3 \psi\rangle| \leq C\left( \rho^{1 + \frac{\beta}{6}} + \rho^{1+\gamma}\right)  \langle \psi , \mathcal{N}\psi\rangle+ C\rho\|\mathcal{N}_\beta^{\frac{1}{2}}\psi\| \|\mathcal{N}^{\frac{1}{2}}\psi\| .
\end{align}

Next, we bound $\mathrm{I}_4$ in \eqref{eq:S11-S11-dec}. 
As above, we have $\hat{u}_\sigma(r+p)=1$ in all the summands  and we can expand using \eqref{exp34}, resulting in four terms 
$\mathrm{I}_4 = \mathrm{I}_{4;a} + \mathrm{I}_{4;b} + \mathrm{I}_{4;c}+ \mathrm{I}_{4;d}$.
The first one can be written as 
\[
  \mathrm{I}_{4;a} = -\int dxdydz\, \Delta{\varphi}^>(x-y){\varphi}^>(x-z){u}^{<\alpha}_{\sigma^\prime}(y;z)a^\ast_{\sigma}(v_x)a^\ast_{\sigma^\prime}(u_z)a_\sigma(u_y)a_\sigma(v_x),
\]
Using \eqref{eq: a partial ell v S*S}, the last bound in \eqref{eq: est u<alpha vinfty}, Lemma \ref{lem: bounds phi} and the Cauchy--Schwarz inequality, we obtain 
\begin{align*}
  |\langle \psi, \mathrm{I}_{4;a}\psi\rangle| &\leq C\rho^{2 - 3\alpha} \int dxdydz\, |\Delta{\varphi}^>(x-y)| |{\varphi}^>(x-z)| \|a_{\sigma^\prime}(u_z)\psi\| \|a_{\sigma}(u_y)\psi\| 
  \\
  &\leq  C\rho^{2 - 3\alpha}\|\Delta{\varphi}^>\|_1 \|{\varphi}^>\|_1 \langle \psi, \mathcal{N}\psi\rangle \leq C\rho^{\frac{4}{3} + 2\gamma - 3\alpha}\langle \psi, \mathcal{N}\psi\rangle.
\end{align*}
Similarly, we estimate the other terms in $\mathrm{I}_4$ as 
\begin{align*}
|\langle \psi, ( \mathrm{I}_{4;b}+ \mathrm{I}_{4;c} +\mathrm{I}_{4;d})\psi\rangle| &\leq C\rho^{2+\frac{2}{3} - 3\alpha}\|{\varphi}^>\|_1^2 \langle \psi, \mathcal{N}\psi\rangle + C\rho^{2+\frac{1}{3} - 3\alpha}\|\nabla {\varphi}^>\|_1 \|{\varphi}^>\|_1 \langle \psi, \mathcal{N}\psi\rangle \\
&\leq C\rho^{\frac{4}{3} +3\gamma - 3\alpha}\langle \psi, \mathcal{N}\psi\rangle.
\end{align*}
In combination, we find that 
\begin{align}\label{eq:S11-S11-I4}
  |\langle \psi, \mathrm{I}_4 \psi\rangle| \leq \rho^{\frac{4}{3} +2\gamma - 3\alpha}\langle \psi, \mathcal{N}\psi\rangle.
\end{align}

The term $\mathrm{I}_5$ in \eqref{eq:S11-S11-dec} can be written in configuration space as 
\begin{align*}
  \mathrm{I}_5 &= -\sum_{\sigma\neq \sigma^\prime}\int dxdydzdz^\prime\, {\varphi}^>(x-y){\varphi}^>(z-z^\prime) \big((k_F^\sigma)^2 v_\sigma(x-z) + \Delta v_\sigma(x-z)\big) \\
  &\qquad \qquad \times u_{\sigma^\prime}(y-z^\prime)  a^\ast_{\sigma}(u_z)a_{\sigma^\prime}^\ast ({u}^{<\alpha}_y)a_{\sigma^\prime}({u}^{<\alpha}_{z^\prime})a_{\sigma}(u_x).
\end{align*}
Using \eqref{eq: u = delta -v}, 
we decompose $\mathrm{I}_5$ into two parts, $\mathrm{I}_5 = \mathrm{I}_{5;1} + \mathrm{I}_{5;2}$.  
For $\mathrm{I}_{5;1}$, using   
\eqref{eq: est u<alpha vinfty}, Lemma \ref{lem:1a} and the  Cauchy--Schwarz inequality, we obtain
\begin{align}\nn
  |\langle \psi, \mathrm{I}_{5;1}\psi\rangle| &\leq C\rho^{\frac{8}{3} -3\alpha} \sum_{\sigma\neq \sigma^\prime}\int dxdydz\, |{\varphi}^>(x-y)||{\varphi}^>(z-y)|\|a_{\sigma}(u_z)\psi\| \|a_{\sigma}(u_x)\psi\|
  \\
  &\leq C\rho^{\frac{8}{3} -3\alpha}\|{\varphi}^>\|_1^2 \langle \psi, \mathcal{N}\psi\rangle \leq C\rho^{\frac{4}{3} + 4\gamma - 3\alpha}\langle \psi, \mathcal{N}\psi\rangle.  \label{I51}
\end{align}

By Lemma \ref{lem:4aa} and the bound $\| {\varphi}^>\|_1\le C \rho^{-\frac 2 3 + 2 \gamma}$ from Lemma \ref{lem: bounds phi},  we further have
\begin{align}\label{eq:S11-S11-I5}
  |\langle \psi, \mathrm{I}_{5;2} \psi\rangle| \leq C \rho^{\frac 5 3} \|\varphi^>\|_1^2  \max_{\sigma'}\|{u}^{<\alpha}_{\sigma^\prime}\|_2^2 \langle \psi, \mathcal{N}\psi\rangle \le    \rho^{\frac{4}{3} +4\gamma - 3\alpha}\langle \psi, \mathcal{N}\psi\rangle.
\end{align}

Finally, the term $\mathrm{I}_6$  in \eqref{eq:S11-S11-dec} is given by 
\begin{align*}
  \mathrm{I}_6 &= \sum_{\sigma \neq \sigma^\prime}\int dxdy dzdz^\prime\, {\varphi}^>(x-y){\varphi}^>(z-z^\prime) \big((k_F^\sigma)^2 v_\sigma(x-z) + \Delta v_\sigma(x-z) \big) \\
  &\qquad \qquad \times ({u}^{<\alpha}_{\sigma^\prime})^2(y-z^\prime) a_{\sigma}^\ast(u_z)a^\ast_{\sigma^\prime}(u_{z^\prime}) a_{\sigma^\prime}(u_y)a_\sigma(u_x) .  
\end{align*}
Using  $0\leq \hat{u}_\sigma\leq 1$, for  fixed $y,z$ we have the bound on the operator norms 
\begin{align*}
  \left\| \int dz^\prime\, {\varphi}^>(z-z^\prime) (u^{<\alpha}_{\sigma^\prime})^2(y-z^\prime) a_{\sigma^\prime}(u_{z^\prime})\right\| &\leq \|\varphi^>_z (u^{<\alpha}_{\sigma^\prime})^2_y\|_2,\\
  \left\|\int dx\, {\varphi}^>(x-y) ((k_F^\sigma)^2+\Delta) v_\sigma(x-z)a_\sigma(u_x)\right\| & \leq \|\varphi^>_y ((k_F^\sigma)^2+\Delta) v_{z,\sigma}\|_2. 
\end{align*}
Therefore, using the Cauchy--Schwarz inequality, the first bound in \eqref{eq: est u<alpha vinfty} and $\| ((k_F^\sigma)^2+\Delta) v_\sigma\|_2 \leq  C\rho^{7/6}$, we obtain 
\begin{align}\label{eq:S11-S11-I6}
  |\langle \psi, \mathrm{I}_{6}\psi\rangle|  &\le \sum_{\sigma\neq \sigma^\prime}\int dzdy\,   \|\varphi^>_z (u^{<\alpha}_{\sigma^\prime})^2_y\|_2  \|\varphi^>_y ((k_F^\sigma)^2+\Delta) v_{z,\sigma}\|_2 
  \left\| a_\sigma(u_z)\psi \right\|  \left\|  a_{\sigma^\prime}(u_y)\psi\right\| \nn\\
 &  \leq  \sum_{\sigma\neq \sigma^\prime} \sqrt{\int dydzdw\, | {\varphi}^>(z-w)|^2|{u}^{<\alpha}_{\sigma^\prime}(y-w)|^2  \|a_\sigma(u_z)\psi\|^2}  \nn\\
  &\qquad\qquad \times \sqrt{\int dydzdw\,  |{\varphi}^>(w-y)|^2 |v_\sigma(w-z)|^2 \| a_{\sigma^\prime}(u_{y})\psi\|^2} \nn\\
  &\leq C\rho^{\frac{2}{3}}\|{\varphi}^>\|_2^2 \|\hat{{u}}^{<\alpha}_{\sigma^\prime}\|_2\|\hat{v}_\sigma\|_2 \langle \psi, \mathcal{N}\psi\rangle \leq C\rho^{\frac{4}{3} + \gamma -\frac{3}{2}\alpha }\langle \psi, \mathcal{N}\psi\rangle.
\end{align}

Inserting the bounds in  \eqref{eq:S11-S11-I1}--\eqref{eq:S11-S11-I2} and \eqref{eq:S11-S11-I3}--\eqref{eq:S11-S11-I6} in 
\eqref{eq:S11-S11-dec} we conclude that
\begin{align*} 
&  \sum_{\sigma\in \{\uparrow, \downarrow\}} \sum_{r\in \Lambda^*} ||r|^2 - (k_F^\sigma)^2| \{S^\ast_{1,1,\sigma}(r), S_{1,1,\sigma}(r)\} \\   & \quad \leq \frac{1}{L^3}\sum_{\sigma\neq \sigma^\prime}\rho_\sigma\sum_{p,r^\prime}|p|^2 |\hat{\varphi}^>(p)|^2 \hat{u}_{\sigma^\prime}(r^\prime)\hat{a}_{r^\prime, \sigma^\prime}^\ast \hat{a}_{r^\prime, \sigma^\prime}  + \mathcal{E}_{S^\ast_{1,1} S_{1,1}}
\end{align*}
with
\begin{equation}\label{eq: S11*S11}
|\langle \psi ,  \mathcal{E}_{S^\ast_{1,1} S_{1,1}} \psi \rangle| \leq 
 C\big(\rho^{1+\gamma}+\rho^{1+\frac{\beta}{6}}+\rho^{\frac{4}{3} + 2\gamma - 3\alpha}\big) \langle \psi, \mathcal{N}\psi\rangle + C\rho \|\mathcal{N}^{\frac{1}{2}}_\beta\psi\|\|\mathcal{N}^{\frac{1}{2}}\psi\|
\end{equation}
for any $0\le \beta<1$ and $\psi\in\mathcal{F}_{\mathrm{f}}$.

\medskip
\noindent{\bf Analysis of $\{S^\ast_{1,2,\sigma}(r), S_{1,2,\sigma}(r)\}$.} Next, we prove an upper bound on 

\begin{align*}
  &  \sum_{\sigma\in \{\uparrow, \downarrow\}}\sum_{r\in \Lambda^*} ||r|^2 - (k_F^\sigma)^2| \{S^\ast_{1,2,\sigma}(r), S_{1,2,\sigma}(r)\} 
  \\
&= \frac{1}{L^6}\sum_{\sigma\neq\sigma^\prime}\sum_{p,q,r,r^\prime,s} ||r|^2 - (k_F^\sigma)^2| \hat{\varphi}^>(p)\hat{\varphi}^>(q) \hat{u}_\sigma(r) \hat{u}^{<\alpha}_\sigma(r-p)\hat{u}^{<\alpha}_\sigma(r-q)  
\\
&\qquad\times\hat{u}_{\sigma^\prime}(r^\prime - p)\hat{v}_{\sigma^\prime}(r^\prime)\hat{u}_{\sigma^\prime}(s-q)\hat{v}_{\sigma^\prime}(s) \{\hat{a}_{r-p,\sigma}^\ast \hat{a}_{r^\prime - p, \sigma^\prime}\hat{a}_{-r^\prime, \sigma^\prime}, \hat{a}^\ast_{-s, \sigma^\prime}\hat{a}_{s-q, \sigma^\prime}^\ast \hat{a}_{r-q, \sigma}\}.
\end{align*}
Computing the anti-commutator 
\begin{align*}
  &\{\hat{a}_{r-p,\sigma}^\ast \hat{a}_{r^\prime - p, \sigma^\prime}\hat{a}_{-r^\prime, \sigma^\prime}, \hat{a}^\ast_{-s, \sigma^\prime}\hat{a}_{s-q, \sigma^\prime}^\ast \hat{a}_{r-q, \sigma}\}\\
  &= \delta_{r^\prime, s}\delta_{p, q}\hat{a}_{r-p,\sigma}^\ast \hat{a}_{r-q,\sigma} - \delta_{r^\prime, s }\hat{a}_{r-p,\sigma}^\ast \hat{a}_{s-q, \sigma^\prime}^\ast \hat{a}_{r^\prime -p, \sigma^\prime}\hat{a}_{r-q, \sigma} 
  \\
&\qquad  - \delta_{s-q, r^\prime - p}\hat{a}_{-s,\sigma^\prime}^\ast \hat{a}_{r-p,\sigma}^\ast \hat{a}_{r-q,\sigma}\hat{a}_{-r^\prime, \sigma^\prime} + \delta_{p,q}\hat{a}_{-s,\sigma^\prime}^\ast \hat{a}_{s-q, \sigma^\prime}^\ast \hat{a}_{r^\prime - p, \sigma^\prime}\hat{a}_{-r^\prime, \sigma^\prime}
\end{align*}  
and suitably changing summation variables, we find 
\begin{align}\label{eq:S12-S12-dec}
\sum_\sigma\sum_r ||r|^2 - (k_F^\sigma)^2|\{S^\ast_{1,2,\sigma}(r), S_{1,2,\sigma}(r)\}  
= \sum_{j=1}^4 \mathrm{J}_j
\end{align}
with
\begin{align*}  
  \mathrm{J}_1 &= \frac{1}{L^6}\sum_{\sigma\neq \sigma^\prime}\sum_{p,r,r^\prime}||r+p|^2 - (k_F^\sigma)^2||\hat{\varphi}^>(p)|^2 \hat{u}_\sigma(r+p)\hat{u}^{<\alpha}_{\sigma}(r) \hat{u}_{\sigma^\prime}(r^\prime - p) \hat{v}_{\sigma^\prime}(r^\prime) \hat{a}_{r,\sigma}^\ast \hat{a}_{r,\sigma},\\
  \mathrm{J}_2 &= - \frac{1}{L^6}\sum_{\sigma\neq \sigma^\prime} \sum_{p,q,r,r^\prime} ||r+p|^2 - (k_F^\sigma)^2| \hat{\varphi}^>(p)\hat{\varphi}^>(q) \hat{u}_\sigma(r+p)\hat{u}^{<\alpha}_\sigma(r)\hat{u}^{<\alpha}_\sigma(r+p-q) 
  \\
  &\qquad \qquad \times \hat{u}_{\sigma^\prime}(r^\prime-p)\hat{v}_{\sigma^\prime}(r^\prime)\hat{u}_{\sigma^\prime}(r^\prime -q)\hat{a}_{r,\sigma}^\ast \hat{a}_{r^\prime - q, \sigma^\prime}^\ast \hat{a}_{r^\prime -p, \sigma^\prime}\hat{a}_{r+p-q,\sigma},\\
  \mathrm{J}_3 &= -\frac{1}{L^6}\sum_{\sigma \neq \sigma^\prime}\sum_{p,q,r,r^\prime}||r|^2 - (k_F^\sigma)^2| \hat{\varphi}^>(p)\hat{\varphi}^>(q)\hat{u}_\sigma(r)\hat{u}_\sigma^{<\alpha}(r-p)\hat{u}_\sigma^{<\alpha}(r-q) 
  \\
  &\qquad \qquad\times \hat{u}_{\sigma^\prime}(r^\prime - p)\hat{v}_{\sigma^\prime}(r^\prime) \hat{v}_{\sigma^\prime}(r^\prime - p + q)\hat{a}_{-r^\prime +p-q, \sigma^\prime}^\ast \hat{a}^\ast_{r-p,\sigma}\hat{a}_{r-q, \sigma}\hat{a}_{-r^\prime, \sigma^\prime},\\
\\
  \mathrm{J}_4 &= \frac{1}{L^6}\sum_{\sigma\neq \sigma^\prime}\sum_{p,r,r^\prime, s}||r+p|^2 - (k_F^\sigma)^2| |\hat{\varphi}^>(p)|^2 \hat{u}_{\sigma}(r+p)\hat{u}^{<\alpha}_{\sigma}(r) \hat{u}_{\sigma^\prime}(r^\prime - p) 
  \\
 &\qquad \qquad \times \hat{v}_{\sigma^\prime}(r^\prime)\hat{u}_{\sigma^\prime}(s-p)\hat{v}_{\sigma^\prime}(s)\hat{a}_{-s, \sigma^\prime}^\ast \hat{a}_{s-p, \sigma^\prime}^\ast \hat{a}_{r^\prime - p, \sigma^\prime}\hat{a}_{-r^\prime, \sigma^\prime}.
\end{align*}  
Again the main quadratic term comes from $\mathrm{J}_1$,  all the other terms only contribute to the error. For $\mathrm{J}_1$, using the bound $||r+p|^2 - (k_F^\sigma)^2| \hat{u}_\sigma(r+p) \hat{u}_{\sigma^\prime}(r^\prime - p) \le |r+p|^2$, $\sum_p p \cdot r |\hat{\varphi}^>(p)|^2 =0$, and $\|\varphi^>\|^2_2 \le C \rho^{-\frac 1 3 +\gamma}$ from Lemma \ref{lem: bounds phi},  we obtain 
\begin{align}\label{eq:S12-S12-J1}
  \mathrm{J}_1 &\le \frac{1}{L^3}\sum_{\sigma\neq \sigma^\prime}\rho_{\sigma^\prime}\sum_{p,r}(|p|^2 + |r|^2 )|\hat{\varphi}^>(p)|^2 \hat{u}^{<\alpha}_\sigma(r) \hat{a}_{r,\sigma}^\ast \hat{a}_{r,\sigma} \nn \\
&\le  \frac{1}{L^3}\sum_{\sigma\neq \sigma^\prime}\rho_{\sigma^\prime}\sum_{p,r} |p|^2 |\hat{\varphi}(p)|^2 \hat{u}_\sigma(r) \hat{a}_{r,\sigma}^\ast \hat{a}_{r,\sigma} + C \rho^{\frac 4 3 + \gamma -2 \alpha} \cN.  
\end{align}  

Next we bound $\mathrm{J}_2$ in \eqref{eq:S12-S12-dec}. Expanding $|r+p|^2 - (k_F^\sigma)^2 = [|p|^2- (k_F^\sigma)^2] + 2p\cdot r + |r|^2,$ 
we write it as the sum of three terms, $\mathrm{J}_2 = \mathrm{J}_{2;a} + \mathrm{J}_{2;b} + \mathrm{J}_{2;c}$. 
 The first term $\mathrm{J}_{2;a}$ can be written in configuration space as 
\begin{align*}
  \mathrm{J}_{2;a} &= \sum_{\sigma\neq \sigma^\prime}\int dxdydzdz^\prime\, \left(\Delta\varphi^>(x-y) + (k_F^\sigma)^2 \varphi^>(x-y)\right)\varphi^>(z-z^\prime)\\
  &\qquad\qquad\times u_{\sigma}(x-z)v_{\sigma^\prime}(y-z^\prime) a^\ast_{\sigma}(u^{<\alpha}_x)a^\ast_{\sigma^\prime}(u_{z^\prime})a_{\sigma^\prime}(u_y)a_{\sigma}(u^{<\alpha}_z). 
\end{align*}
With \eqref{eq: u = delta -v}, 
we can further split it as $\mathrm{J}_{2;a} = \mathrm{J}_{2;a;1} + \mathrm{J}_{2;a;2}$. The term $\mathrm{J}_{2;a;1}$ can be estimated similarly as in Lemma \ref{lem:4a}. To be precise, using \eqref{eq:Pauli}, \eqref{eq: est u<alpha vinfty}, the bounds $\|\Delta\varphi^>\|_1 \|\varphi^>\|_1\le C \rho^{-\frac 2 3 + 2\gamma}$ and $\|(k_F^\sigma)^2\varphi^>\|_1 \|\varphi^>\|_1\le C \rho^{-\frac 2 3 + 4\gamma}$ from Lemma \ref{lem: bounds phi}, and the Cauchy--Schwarz  inequality, we obtain
\begin{align*}
  |\langle \psi, \mathrm{J}_{2;a;1}\psi\rangle| &\leq  \sum_{\sigma\neq \sigma^\prime}\int dxdydz^\prime\, \left(|\Delta\varphi^>(x-y)| + (k_F^\sigma)^2 |\varphi^>(x-y)| \right)|\varphi^>(x-z^\prime)| \\
  &\qquad\qquad \qquad \times   |v_{\sigma^\prime}(y-z^\prime)|   \|\hat{u}^{<\alpha}_\sigma\|^2\|a_{\sigma^\prime}(u_{z^\prime})\psi\|\|a_{\sigma^\prime}(u_y)\psi\|
  \\
  &\leq C\rho^{2-3\alpha}(\|\Delta\varphi^>\|_1 + \rho^{\frac{2}{3}}\|\varphi^>\|_1) \|\varphi^>\|_1 \langle \psi, \mathcal{N}\psi\rangle\leq C\rho^{\frac{4}{3}+ 2\gamma -3\alpha}\langle \psi, \mathcal{N}\psi\rangle,
\end{align*}
where we also used that  
$\|v_{\sigma^\prime}\|_\infty \leq C\rho$.
From Lemma \ref{lem:4aa} we also get
\begin{align*}
  |\langle \psi, \mathrm{J}_{2;a;2}\psi\rangle| & \leq C \|\hat{u}^{<\alpha}_\sigma\|_2^2 \left( \|\Delta\varphi^>\|_1 + \rho^{\frac 23} \| \varphi^>\|_1\right) \|\varphi^>\|_1 \|{v}_\sigma\|_2 \|{v}_{\sigma^\prime}\|_2 \langle \psi, \mathcal{N}\psi\rangle
  \\ & \leq C\rho^{\frac{4}{3} + 2\gamma - 3\alpha}\langle \psi, \mathcal{N}\psi\rangle. 
\end{align*}
For $\mathrm{J}_{2;b}$, it is  convenient to multiply and divide by $|r^\prime - q|^2$, in order to obtain the identity
\begin{align*}
  \mathrm{J}_{2;b} &= \sum_{\sigma \neq \sigma^\prime}\sum_{\ell=1,m}^3 \int dxdydzdz^\prime \, \partial_\ell \varphi^>(x-y) u_\sigma(x-z) a_{\sigma}^\ast(\partial_\ell u^{<\alpha}_x)a^\ast_{\sigma^\prime}(\widetilde{u}^>_{z^\prime})a_{\sigma^\prime}(u^>_y)a_{\sigma}(u^{<\alpha}_z) 
  \\
  &\quad  \times \big(\partial_m^2\varphi^>(z-z^\prime) v_{\sigma^\prime}(y-z^\prime) - 2 \partial_m \varphi^>(z-z^\prime)\partial_m v_{\sigma^\prime}(y-z^\prime) + \varphi^>(z-z^\prime)\partial_m^2 v_{\sigma'}(y-z^\prime) \big) 
\end{align*}
We use again \eqref{eq: u = delta -v}, and  denote by $J_{2;b;1}$ the term containing $\delta(x-z)$, and by $J_{2;b;2}$ the remaining one with $v_\sigma(x-z)$. For $J_{2;b;1}$, we have 
\begin{align*}
  &|\langle \psi, \mathrm{J}_{2;b;1}\psi\rangle| 
  \\ & \leq C\sum_{\sigma\neq \sigma^\prime}\sum_{\ell, m =1}^3\int dxdydz^\prime\, |\partial_\ell\varphi^>(x-y)| \|\widetilde{u}^>_{\sigma'}\|_2 \|u^{<\alpha}_\sigma\|_2 \|a_\sigma(\partial_\ell u^{<\alpha}_x)\psi\| \|a_{\sigma^\prime}(u^>_y)\psi\| 
  \\
  &\quad  \times \big|\partial_m^2\varphi^>(z-z^\prime) v_{\sigma^\prime}(y-z^\prime) - 2 \partial_m \varphi^>(z-z^\prime)\partial_m v_{\sigma^\prime}(y-z^\prime) + \varphi^>(z-z^\prime)\partial_m^2 v_{\sigma'}(y-z^\prime) \big|
\end{align*}
By Cauchy-Schwarz, Lemma \ref{lem: bounds phi} and Lemma \ref{lem:1a}, we obtain
$$
  |\langle \psi, \mathrm{J}_{2;b;1}\psi\rangle|   \leq C\rho^{\frac 23+\frac{5}{2}\gamma -\frac{3}{2}\alpha} \left( \langle \psi, \mathbb{H}_0 \psi\rangle + \rho^{\frac 13} \|\mathbb{H}_0^{\frac{1}{2}}\psi\|\|\mathcal{N}^{\frac{1}{2}}\psi\|\right).
$$ 

For $\mathrm{J}_{2;b;2}$, we can use Lemma \ref{lem:4aa} similarly as for $\mathrm{J}_{2;a;2}$, obtaining the same bound.
Similarly, we can proceed with  $\mathrm{J}_{2;c}$, given by  
\begin{align*}
  \mathrm{J}_{2;c} &= - \sum_{\sigma \neq \sigma^\prime}\sum_{\ell=1,m}^3 \int dxdydzdz^\prime \,  \varphi^>(x-y) u_\sigma(x-z) a_{\sigma}^\ast(\partial_\ell^2 u^{<\alpha}_x)a^\ast_{\sigma^\prime}(\widetilde{u}^>_{z^\prime})a_{\sigma^\prime}(u^>_y)a_{\sigma}(u^{<\alpha}_z)
  \\
  &\quad  \times \big(\partial_m^2\varphi^>(z-z^\prime) v_{\sigma^\prime}(y-z^\prime) -2 \partial_m \varphi^>(z-z^\prime)\partial_m v_{\sigma^\prime}(y-z^\prime) + \varphi^>(z-z^\prime)\partial_m^2 v(y-z^\prime)\big) 
\end{align*}
In a similar way as for $\mathrm{J}_{2;b}$, using also that 
\begin{align*}
  \int dx\, \|a_\sigma(\partial^2_\ell u^{<\alpha}_x)\psi\|^2 &\leq C\rho^{\frac{2}{3} - 2\alpha}\int dx\, \|a_\sigma(\partial_\ell u^{<\alpha}_x)\psi\|^2
  \\
  & \leq C\rho^{\frac{2}{3} - 2\alpha}\langle \psi, \mathbb{H}_0 \psi\rangle + C\rho^{\frac 4 3 -2\alpha}\langle \psi, \mathcal{N}\psi\rangle,
\end{align*}
we obtain 
$$
 |\langle \psi, \mathrm{J}_{2;c}\psi\rangle| \leq C\rho^{\frac 23+\frac{7}{2}\gamma -\frac{5}{2}\alpha} \left( \langle \psi, \mathbb{H}_0 \psi\rangle + \rho^{\frac 13} \|\mathbb{H}_0^{\frac{1}{2}}\psi\|\|\mathcal{N}^{\frac{1}{2}}\psi\|\right).
$$
 Thus, since $\alpha > \gamma$ by assumption, we obtain the bound 
\begin{align}\label{eq:S12-S12-J2}
  |\langle \psi, \mathrm{J}_2 \psi\rangle| &\leq C\rho^{\frac{4}{3} + 2\gamma - 3\alpha}\langle \psi, \mathcal{N}\psi\rangle \nn
   \\
  &\quad + C\rho^{\frac{2}{3}+\frac{7}{2}\gamma -\frac{5}{2}\alpha} \left(\langle \psi, \mathbb{H}_0 \psi\rangle + \rho^{\frac{1}{3}}\|\mathbb{H}_0^{\frac{1}{2}}\psi\|\|\mathcal{N}^{\frac{1}{2}}\psi\|\right).
\end{align}

The term $\mathrm{J}_3$ in \eqref{eq:S12-S12-dec}
can be rewritten with a  change of summation variables $r^\prime \mapsto r^\prime +p$ as 
\begin{align*}
  \mathrm{J}_3 &= -\frac{1}{L^6}\sum_{\sigma\neq \sigma^\prime}\sum_{r,r^\prime} (|r|^2 - (k_F^\sigma)^2) \hat{u}_\sigma(r)\hat{u}_{\sigma^\prime}(r^\prime) 
  \\
  &\quad\times \left| \sum_p \hat\varphi^>(p) \hat{v}_{\sigma^\prime}(r^\prime +p)\hat{u}^{<\alpha}_{\sigma}(r-p)\hat{a}^\ast_{r-p, \sigma}\hat{a}_{-r-p,\sigma^\prime}\right|^2. 
\end{align*}
Thus $\mathrm{J}_3 \leq 0$ and it can be dropped for an upper bound. 

It remains to bound the last term $\mathrm{J}_4$ in \eqref{eq:S12-S12-dec}, given by
\begin{align*}
  \mathrm{J}_4 &= \frac{1}{L^6}\sum_{\sigma\neq \sigma^\prime}\sum_{p,r} ||r+p|^2 - (k_F^\sigma)^2|    |\hat{\varphi}^>(p)|^2 \hat{u}_{\sigma}(r+p)\hat{u}^{<\alpha}_{\sigma}(r)^2 
\\
  &\qquad \qquad \times \left|\sum_{r^\prime}\hat{u}_{\sigma^\prime}^>(r^\prime - p) \hat{v}_{\sigma^\prime}(r^\prime) \hat{a}_{r^\prime - p, \sigma^\prime}\hat{a}_{-r^\prime, \sigma^\prime}\right|^2.
 \end{align*}
Using   
$$
\sum_r ||r+p|^2 - (k_F^\sigma)^2| \hat{u}_{\sigma}(r+p)\hat{u}^{<\alpha}_{\sigma}(r) \le \sum_r  |r+p|^2 \hat{u}^{<\alpha}_{\sigma}(r) \le C L^3 (\rho^{1-3\alpha} |p|^2 + \rho^{\frac 5 3 -5 \alpha}) 
$$
and also replacing $\hat{u}_{\sigma^\prime}(r^\prime -p)$ by $\hat{u}^>_{\sigma^\prime}(r^\prime -p)$ defined in \eqref{eq: u< u> gamma} due to the constraints  $|p|\geq 4\rho^{1/3-\gamma}$ and $|r'|\leq C\rho^{1/3}$, we have 
\begin{align*}
  \mathrm{J}_4 &\le  \frac{C}{L^3}\sum_{\sigma\neq \sigma^\prime}\sum_{p} (\rho^{1-3\alpha} |p|^2 + \rho^{\frac 5 3 -5 \alpha})  |\hat{\varphi}^>(p)|^2\left|\sum_{r^\prime}\hat{u}_{\sigma^\prime}^>(r^\prime - p) \hat{v}_{\sigma^\prime}(r^\prime) \hat{a}_{r^\prime - p, \sigma^\prime}\hat{a}_{-r^\prime, \sigma^\prime}\right|^2\\
  &= \mathrm{J}_{4;a} + \mathrm{J}_{4;b} 
  \end{align*}
where $\mathrm{J}_{4;a}$ and $\mathrm{J}_{4;b}$ are the terms involving $\rho^{1-3\alpha} |p|^2$ and $\rho^{\frac 5 3 -5 \alpha}$, respectively. 

With the aid of Lemmas \ref{lem:4a}, \ref{lem:1a} and \ref{lem: bounds phi}, 
we can bound $\mathrm{J}_{4;a}$  as 
\begin{align*}
 |\langle \psi,  \mathrm{J}_{4;a} \psi\rangle| 
  & \leq C\rho^{2-3\alpha}\int dxdy\, |(\Delta\varphi^> \ast \varphi^>)(x-y)| \|a_{\sigma^\prime}(u^>_x)\psi\|\|a_{\sigma^\prime}(u^>_y)\psi\|\\
  & \leq C\rho^{2-3\alpha} \|\Delta\varphi^>\|_1 \|\varphi^>\|_1 \rho^{-\frac 2 3 +2\gamma}\langle \psi, \mathbb{H}_0 \psi\rangle \leq C\rho^{\frac{2}{3} + 4\gamma - 3\alpha}\langle \psi, \mathbb{H}_0 \psi\rangle.
  \end{align*}
Similarly, 
\begin{align*}
  |\langle \psi, \mathrm{J}_{4;b}\psi\rangle| &\leq C\rho^{\frac{8}{3} - 5\alpha}\int dxdy\, |\varphi^> \ast \varphi^>(x-y)| \|a_{\sigma^\prime}(u^>_x)\psi\|\|a_{\sigma^\prime}(u^>_y)\psi\|\\
  &\le C\rho^{\frac{8}{3} - 5\alpha} \|\varphi^>\|_1^2 \rho^{-\frac 2 3 +2\gamma}\langle \psi, \mathbb{H}_0 \psi\rangle \leq C\rho^{\frac{2}{3} + 6\gamma - 5\alpha}\langle \psi, \mathbb{H}_0 \psi\rangle.
\end{align*}
Combining these bounds, we obtain 
\begin{align}\label{eq:S12-S12-J4}
  |\langle \psi, \mathrm{J}_4\psi\rangle| \leq C\rho^{\frac{2}{3} + 6\gamma - 5\alpha} \langle \psi, \mathbb{H}_0 \psi\rangle .
\end{align}
From \eqref{eq:S12-S12-J1}, \eqref{eq:S12-S12-J2}, \eqref{eq:S12-S12-J4}, and recalling that $\mathrm{J}_3\le 0$, we conclude that 
\begin{multline*}
\sum_{\sigma}\sum_{r} ||r|^2 - (k_F^\sigma)^2| \{S^\ast_{1,2,\sigma}(r), S_{1,2,\sigma}(r)\} 
\\
\leq \frac{1}{L^3}\sum_{\sigma\neq \sigma^\prime}\rho_\sigma \sum_{p,r^\prime}|p|^2 |\hat{\varphi}^>(p)|^2 \hat{u}_{\sigma^\prime}(r^\prime)\hat{a}_{r^\prime, \sigma^\prime}^\ast \hat{a}_{r^\prime, \sigma^\prime} + \mathcal{E}_{S^\ast_{1,2,\sigma} S_{1,2,\sigma}}
\end{multline*}
with 
\begin{align}\label{eq: err S12*S12}
  |\langle \psi, \mathcal{E}_{S^\ast_{1,2,\sigma} S_{1,2,\sigma}}\psi\rangle | &\leq   C\rho^{\frac{4}{3} + 2\gamma - 3\alpha}  \langle \psi, \mathcal{N}\psi\rangle + C\rho^{\frac{2}{3} + 6\gamma - 5\alpha}\langle \psi, \mathbb{H}_0 \psi\rangle\nonumber
 \\
 &\quad  + C\rho^{1+\frac{7}{2}\gamma - \frac{5}{2}\alpha}\|\mathbb{H}_0^{\frac{1}{2}}\psi\|\|\mathcal{N}^{\frac{1}{2}}\psi\|,
\end{align}
for all $\psi\in\mathcal{F}_{\rm f}$, 
where we used again that $\gamma < \alpha$ by assumption.

\medskip
\noindent{\bf Analysis of $\left(\{S^\ast_{1,2,\sigma}(r), S_{1,1,\sigma}(r)\} + \mathrm{h.c.}\right)$.} 
Finally we consider 
\begin{align}\label{eq:S12-S11-dec}
  &\sum_{\sigma} \sum_r||r|^2 - (k_F^\sigma)^2| \{S^\ast_{1,2,\sigma}(r), S_{1,1,\sigma}(r)\} \nn\\
&= - \frac{1}{L^6}\sum_{p,q,r,r^\prime}(|r|^2 - (k_F^\sigma)^2) \hat{\varphi}^>(p)\hat{\varphi}^>(q) \hat{{u}}^{<\alpha}_{\sigma}(r-p) \hat{v}_\sigma(r-q)  
\nn\\
&\qquad\qquad \times \hat{v}_{\sigma^\prime}(r^\prime) \hat{{u}}^{<\alpha}_{\sigma^\prime}(r^\prime -p+q)\hat{a}_{r-p,\sigma}^\ast \hat{a}_{q-r,\sigma}^\ast \hat{a}_{r^\prime - p+q, \sigma^\prime}\hat{a}_{-r^\prime, \sigma^\prime}
\end{align}
where we computed the anti-commutator
\[
  \{\hat{a}_{r-p,\sigma}^\ast \hat{a}_{r^\prime - p, \sigma^\prime}\hat{a}_{-r^\prime,\sigma^\prime}, \hat{a}_{q-r,\sigma}^\ast \hat{a}^\ast_{s-q,\sigma^\prime}\hat{a}_{s,\sigma^\prime}\} =  \delta_{s-q, r^\prime-p}\hat{a}_{r-p,\sigma}^\ast \hat{a}^\ast_{q-r,\sigma}\hat{a}_{s,\sigma^\prime}\hat{a}_{-r^\prime, \sigma^\prime}
\]
and used $\hat u_\sigma(r)=1=\hat u_{\sigma'}(r'-p)$ in the relevant domain. 
We decompose  
$$|r|^2-(k_F^\sigma)^2 = ( |r-q|^2 - (k_F^\sigma)^2) + 2(r-q)\cdot q + |q|^2$$ 
and correspondingly write the right-hand side of \eqref{eq:S12-S11-dec} as $\sum_{j=1}^3 \mathrm{K}_j$, where $\mathrm{K}_1,\mathrm{K}_2,\mathrm{K}_3$ are the terms involving $|r-q|^2- (k_F^\sigma)^2$, $2(r-q)\cdot q$, $|q|^2$, respectively. 
We can write $\mathrm{K}_{1}$ in configuration space 
as
\[
  \mathrm{K}_{1} = \sum_{\sigma\neq \sigma^\prime}\int dxdy\, |{\varphi}^>(x-y)|^2 a^\ast_\sigma({u}^{<\alpha}_x)a_\sigma^\ast((\Delta + (k_F^\sigma)^2) v_x) a_{\sigma^\prime}({u}^{<\alpha}_y)a_{\sigma^\prime}(v_y).
\] 
and bound it using Lemmas \ref{lem:4a}, \ref{lem: bounds phi} and \ref{lem:1a}  as 
\begin{align*}
  |\langle\psi, \mathrm{K}_{1}\psi\rangle| 
  &\leq C\| ( |\cdot|^2 - (k_F^\sigma)^2) \hat{v}_\sigma\|_2 \|{v}_{\sigma^\prime}\|_2 \|{\varphi}^>\|_2^2 \langle \psi, \mathcal{N}\psi\rangle \leq C\rho^{\frac{4}{3} + \gamma}\langle \psi, \mathcal{N}\psi\rangle. 
\end{align*}
For the term $\mathrm{K}_{2}$, we proceed similarly and find that 
\begin{align*}
  |\langle \psi, \mathrm{K}_{2}\psi\rangle| 
  &\leq C\|\nabla{\varphi}^>\|_1\|{\varphi}^>\|_\infty \| |\cdot| \hat{v}_{\sigma}\|_2 \|{v}_{\sigma^\prime}\|_2 \langle \psi, \mathcal{N}\psi\rangle \leq C\rho^{1+\gamma}\langle \psi,\mathcal{N}\psi\rangle.
\end{align*}
The term 
\[
  \mathrm{K}_3 = \sum_{\sigma\neq \sigma^\prime}\int dxdy\, \Delta{\varphi}^>(x-y){\varphi}^>(x-y)a^\ast_\sigma({u}^{<\alpha}_x)a_\sigma^\ast(v_x)a_{\sigma^\prime}({u}^{<\alpha}_y)a_{\sigma^\prime}(v_y)
\] 
can be bounded with Lemma \ref{lem:4aaa} and Lemma \ref{lem: bounds phi}. 
We obtain for any $0\le \beta<1$
\begin{align*}
  |\langle \psi, \mathrm{K}_{3}\psi\rangle | 
&  \leq C\rho^{1+\frac{\beta}{6}}\langle \psi, \mathcal{N}\psi\rangle + C\rho\|\mathcal{N}_\beta^{\frac{1}{2}}\psi\|\|\mathcal{N}^{\frac{1}{2}}\psi\|. 
\end{align*}
In combination, the bounds on $ \mathrm{K}_{j}$ for $1\leq j\leq 3$
imply that
\begin{align}\label{eq: est S11*S12}
  &\sum_\sigma\sum_r ||r|^2 - (k_F^\sigma)^2||\langle \psi, \{S^\ast_{1,2,\sigma}, S_{1,1,\sigma}(r)\}\psi\rangle| \nn\\
  & \leq C \big( \rho^{1+\gamma} + \rho^{1+\frac{\beta}{6}}\big)\langle \psi, \mathcal{N}\psi\rangle + C\rho\|\mathcal{N}_\beta^{\frac{1}{2}}\psi\|\|\mathcal{N}^{\frac{1}{2}}\psi\|. 
\end{align}

\medskip
\noindent
{\bf Conclusion.} Combining the estimates in \eqref{eq: S11*S11}, \eqref{eq: err S12*S12} and \eqref{eq: est S11*S12} and using also $\rho^{1 + \frac{7}{2}\gamma - \frac{5}{2}\alpha}\|\mathbb{H}_0^{\frac{1}{2}}\psi\|\|\mathcal{N}^{\frac{1}{2}}\psi\| \le \rho^{1 +6 \gamma - 5\alpha}\|\mathbb{H}_0^{\frac{1}{2}}\psi\|^2 + \rho^{1+\gamma} \|\mathcal{N}^{\frac{1}{2}}\psi\|^2$, we complete the proof of Lemma \ref{lem: S1}.
\end{proof}

\subsection{Proof of Lemma \ref{lem: TS1} }
\begin{proof}
As in the proof of Lemma \ref{lem: S1} in the previous subsection, we decompose $S_{1,\sigma}(r) = S_{1,1,\sigma}(r) + S_{1,2,\sigma}(r)$ as in \eqref{eq:S-S1-S2}, and estimate the terms involving $\{T_\sigma^\ast(r), S_{1,1,\sigma}(r) \}$ and $\{T_\sigma^\ast(r), S_{1,2,\sigma}(r) \}$ separately.

\medskip
\noindent{\bf Analysis of $\{T_\sigma^\ast(r), S_{1,1,\sigma}(r) \}$.} From the definitions \eqref{eq: def Tk}, \eqref{def:s1i}, \eqref{def:D} and \eqref{eq: def bp bpr}, we find
\begin{align*}
  &\sum_\sigma\sum_r ||r|^2 - (k_F^\sigma)^2| \{T^\ast_{\sigma}(r), S_{1,1,\sigma}(r)\}
  \\
  &= -  \frac{1}{L^6}\sum_{\sigma\neq\sigma^\prime} \sum_{p,q,r,s,r^\prime} ||r|^2 - (k_F^\sigma)^2| \hat \varphi^>(q) \hat{{u}}^{<\alpha}_{\sigma^\prime}(s)\hat{u}_{\sigma^\prime}(s-q) \hat{u}_{\sigma'}(r'-p)\hat{v}_{\sigma'}(r')
 \\
  &\qquad  \times  \Big( \widehat{\omega}^\varepsilon_{-r,r^\prime}(p)\hat{u}_\sigma(p-r)\hat{u}_\sigma(q-r)\hat{v}_\sigma(r) - \widehat{\omega}^\varepsilon_{r-p,r^\prime}(p)\hat{v}_\sigma(r-p)\hat{v}_\sigma(r-q) \hat{u}_\sigma(r)\Big) 
  \\ & \qquad  \times \Big\{ \hat{a}_{p-r,\sigma} \hat{a}_{r'-p,\sigma'}\hat{a}_{-r',\sigma'} , \hat{a}^\ast_{q-r,\sigma}\hat{a}^\ast_{s - q, \sigma^\prime}\hat{a}_{s, \sigma^\prime}\Big\}.
\end{align*}
For $\sigma \neq \sigma'$, $r' \in B_F^{\sigma'}$ and $r'-p, s-q \not\in B_F^{\sigma'}$, the anticommutator equals
\begin{align}\label{eq: anticommutator T*S11}
& \Big\{ \hat{a}_{p-r,\sigma} \hat{a}_{r'-p,\sigma'}\hat{a}_{-r',\sigma'} , \hat{a}^\ast_{q-r,\sigma}\hat{a}^\ast_{s - q, \sigma^\prime}\hat{a}_{s, \sigma^\prime}\Big\} = \delta_{p,q}\delta_{r^\prime, s}\hat{a}_{s,\sigma^\prime}\hat{a}_{-r^\prime,\sigma^\prime} \nn\\
&\qquad \qquad - \delta_{p,q}\hat{a}^\ast_{s-q,\sigma^\prime}\hat{a}_{r^\prime -p, \sigma^\prime}\hat{a}_{s,\sigma^\prime}\hat{a}_{-r^\prime,\sigma^\prime}
 - \delta_{r^\prime - p, s-q}\hat{a}^\ast_{q-r,\sigma} \hat{a}_{p-r,\sigma}\hat{a}_{s,\sigma^\prime}\hat{a}_{-r^\prime,\sigma^\prime}.
\end{align}
The first term gives a vanishing contribution since $\hat{{u}}^{<\alpha}_{\sigma^\prime}(r^\prime)\hat{v}_{\sigma^\prime}(r^\prime) = 0$. For the remaining ones, we find after a suitable change of summation variables that 
\begin{align}\label{eq:TS11-dec}
\sum_\sigma\sum_r ||r|^2 - (k_F^\sigma)^2| \{T^\ast_{\sigma}(r), S_{1,1,\sigma}(r)\} = \sum_{j=1}^3 \mathrm{I}_j
\end{align} 
with
\begin{align*}
  \mathrm{I}_1 &= - \frac{1}{L^6}\sum_{\sigma\neq\sigma^\prime} \sum_{p,r,s,r^\prime} (|r+p|^2 -|r|^2)\widehat{\omega}^\varepsilon_{r,r^\prime}(p)\hat{\varphi}^>(p) \hat{u}_\sigma(r+p)\hat{v}_{\sigma}(r) \hat{{u}}^{<\alpha}_{\sigma^\prime}(s) \\
  &\qquad \qquad \times  
  \hat{u}_{\sigma^\prime}(s-p)\hat{u}_{\sigma^\prime}(r^\prime - p)\hat{v}_{\sigma^\prime}(r^\prime) \hat{a}^\ast_{s-p,\sigma^\prime}\hat{a}_{r^\prime - p, \sigma^\prime}\hat{a}_{s,\sigma^\prime}\hat{a}_{-r^\prime,\sigma^\prime},
\\
  \mathrm{I}_2 &= \frac{1}{L^6}\sum_{\sigma\neq\sigma^\prime} \sum_{p,q,r,r^\prime} ||r+p|^2 - (k_F^\sigma)^2|\widehat{\omega}^\varepsilon_{r,r^\prime}(p)\hat{\varphi}^>(q) \hat{u}_\sigma(r+p)\hat{v}_{\sigma}(r) \hat{{u}}^{<\alpha}_{\sigma^\prime}(r^\prime - p+q)  \\
  &\qquad \qquad \times \hat{v}_\sigma(r+p-q)\hat{u}_{\sigma^\prime}(r^\prime -p)\hat{v}_{\sigma^\prime}(r^\prime) \hat{a}_{q-p-r,\sigma}^\ast \hat{a}_{-r,\sigma}\hat{a}_{r^\prime -p+q,\sigma^\prime}\hat{a}_{-r^\prime,\sigma^\prime},
\\
  \mathrm{I}_3 & = -\frac{1}{L^6}\sum_{\sigma\neq\sigma^\prime} \sum_{p,q,r,r^\prime} ||r|^2 - (k_F^\sigma)^2|\widehat{\omega}^\varepsilon_{r,r^\prime}(p)\hat{\varphi}^>(q)\hat{u}_\sigma(r+p)\hat{v}_{\sigma}(r) \hat{{u}}^{<\alpha}_{\sigma^\prime}(r^\prime - p+q)  \\
  &\qquad \qquad \times \hat{u}_\sigma(r+q)\hat{u}_{\sigma^\prime}(r^\prime -p)\hat{v}_{\sigma^\prime}(r^\prime) \hat{a}_{r+q,\sigma}^\ast \hat{a}_{r+p,\sigma}\hat{a}_{r^\prime -p+q,\sigma^\prime}\hat{a}_{-r^\prime,\sigma^\prime}.
\end{align*}
More precisely, $\mathrm{I}_1$ combines the two contributions from the second term in \eqref{eq: anticommutator T*S11}, after the substitutions $r\mapsto -r$ in one and $r\mapsto r-p$ in the other. The terms $\mathrm{I}_2$ and $\mathrm{I}_3$ arise from the last term in \eqref{eq: anticommutator T*S11}, after the substitutions $r\mapsto r-p$ and $r\mapsto -r$, respectively.

\medskip
\noindent{\bf Analysis of $\mathrm{I}_1$.}  
With the Cauchy--Schwarz inequality, we can bound
\begin{align}\label{eq:TS-I2-CS}
\pm \big(  \mathrm{I}_1 +\mathrm{I}_1^* \big) &\le \frac{\rho^{-\eta}}{L^6}\sum_{\sigma\neq\sigma^\prime} \sum_{|p| \ge \rho^{1/3-\gamma}}\sum_r  \frac{(|r+p|^2 -|r|^2)^2}{|p|^2} \hat{u}_{\sigma}(r+p)\hat{v}_\sigma(r)  \nn\\
&\quad \times \left| \sum_{r^\prime}\widehat{\omega}_{r,r^\prime}^\varepsilon (p)\hat{u}_{\sigma^\prime}(r^\prime - p)\hat{v}_{\sigma^\prime}(r^\prime)   \hat{a}_{r^\prime - p, \sigma^\prime}\hat{a}_{-r^\prime, \sigma^\prime}\right|^2 \nn\\
&\quad + \frac{\rho^\eta}{L^6}\sum_{\sigma\neq\sigma^\prime} \sum_{p,r}  \hat{v}_\sigma(r)   |p|^2 (\hat{\varphi}^>(p))^2  \left | \sum_{s} \hat{{u}}^{<\alpha}_{\sigma^\prime}(s)  \hat{u}_{\sigma^\prime}(s-p) \hat{a}^*_{s,\sigma^\prime} \hat{a}_{s-p,\sigma^\prime} \right|^2\nn\\
& =:  \mathrm{I}_{1;a} + \mathrm{I}_{1;b},
\end{align}
for any $\eta>0$, to be chosen later.
The first term $\mathrm{I}_{1;a}$ on the right-hand side of \eqref{eq:TS-I2-CS} can be estimated using the analysis of the term $\mathrm{I}_6$ in the proof of Lemma~\ref{lem: T}, see \eqref{eq:RR-I6-dec}.
 To be precise, using $0\le |r+p|^2 -|r|^2 \le C|p|^2$ for $|r|\leq C \rho^{1/3} \ll  \rho^{1/3-\gamma} \le |p|$, we have
\begin{align} 
\mathrm{I}_{1;a} & \le \frac{C\rho^{-\eta}}{L^6}\sum_{\sigma\neq\sigma^\prime} \sum_{p,r} ||r+p|^2 -|r|^2| \hat{u}_{\sigma}(r+p)\hat{v}_\sigma(r)   \nn
\\
&\qquad \times \left| \sum_{r^\prime}\widehat{\omega}_{r,r^\prime}^\varepsilon (p)\hat{u}_{\sigma^\prime}(r^\prime - p)\hat{v}_{\sigma^\prime}(r^\prime)   \hat{a}_{r^\prime - p, \sigma^\prime}\hat{a}_{-r^\prime, \sigma^\prime}\right|^2 \,.
\label{eq: est I2a Lem 55}
\end{align}
 This is of the same form as $\mathrm{I}_6$ in \eqref{eq:RR-dec}, and we can  deduce from  \eqref{eq:RR-I6-dec}--\eqref{eq:RR-I6-final}  that for any $\kappa>0$
$$
\langle \psi, \mathrm{I}_{1;a} \psi \rangle \le C\rho^{-\eta}\rho^{1+\frac{\gamma}{2} - \kappa}\|\mathbb{H}_0^{\frac{1}{2}}\psi\|\|\mathcal{N}^{\frac{1}{2}}\psi\| + C\rho^{-\eta}\rho^{\frac{4}{3} -\frac{\gamma}{2} -\kappa}\langle \psi,\mathcal{N}\psi\rangle. 
$$
The second term is given by 
\[
\mathrm{I}_{1;b} =  \frac{\rho^\eta}{L^3}\sum_{\sigma\neq\sigma^\prime} \rho_\sigma \sum_{p}   |p|^2 (\hat{\varphi}^>(p))^2  \left | \sum_{s} \hat{{u}}^{<\alpha}_{\sigma^\prime}(s)  \hat{u}_{\sigma^\prime}(s-p) \hat{a}^*_{s,\sigma^\prime} \hat{a}_{s-p,\sigma^\prime} \right|^2 .
\]
Writing it in configuration space and applying Lemma \ref{lem:4a}, we obtain 
\begin{align*}
|\langle \psi, \mathrm{I}_{1;b}\psi\rangle| 
&\le C \rho^{1+\eta} \sum_{ \sigma^\prime} \| u^{<\alpha}_{\sigma'}\|_2^2  \|\Delta \hat{\varphi}^> * \hat{\varphi}^>\|_1 \langle \psi, \cN\psi\rangle 
\leq C\rho^{\frac 4 3 + 2\gamma -3 \alpha + \eta} \langle \psi, \cN\psi\rangle.
\end{align*}
The optimal choice of $\eta$ is thus $\eta = -(5/4)\gamma -\kappa/2 +(3/2)\alpha$, 
and we arrive at  
\begin{equation}\label{i1n}
|\langle \psi, I_1\psi\rangle|  \leq C\rho^{\frac{4}{3} + \frac{3}{4}\gamma - \frac{3}{2}\alpha - \kappa} \langle \psi, \mathcal{N}\psi\rangle + C\rho^{1+ \frac{7}{4}\gamma - \frac{3}{2}\alpha - \kappa}\|\mathbb{H}_0^{\frac{1}{2}}\psi\|\|\mathcal{N}^{\frac{1}{2}}\psi\|. 
\end{equation}

\medskip
\noindent{\bf Analysis of $\mathrm{I}_2$.} We now turn to the estimate of the term $\mathrm{I}_2$ in \eqref{eq:TS11-dec}. 
On the support  of $\hat v_\sigma(r)$, $\hat v_{\sigma'}(r')$, $v_\sigma(r+p-q)$ and $\hat{\varphi}^>(q)$, we have $\hat{u}_\sigma(r+p) =1= \hat{u}_{\sigma^\prime}(r^\prime - p)$. Using the identity in \eqref{eq: comparison w eps with phi} we can write 
 \begin{equation}\label{eq: R*S-I3-dec}
  \mathrm{I}_2 =  \mathrm{I}_{2;a} + \mathrm{I}_{2;b} + \mathrm{I}_{2;c},
\end{equation}
 with 
\begin{align*}  
  \mathrm{I}_{2;a} &= \frac{2}{L^6}\sum_{\sigma\neq\sigma^\prime} \sum_{p,q,r,r^\prime} ||r+p|^2 - (k_F^\sigma)^2|\hat{\varphi}(p)\hat{\varphi}^>(q) \hat{v}_{\sigma}(r) \hat{{u}}_{\sigma^\prime}^{<\alpha}(r^\prime - p+q) \nn\\
  &\qquad\qquad \times \hat{v}_\sigma(r+p-q)\hat{v}_{\sigma^\prime}(r^\prime) \hat{a}_{q-p-r,\sigma}^\ast \hat{a}_{-r,\sigma}\hat{a}_{r^\prime -p+q,\sigma^\prime}\hat{a}_{-r^\prime,\sigma^\prime},\\
  \mathrm{I}_{2;b} & = -\frac{4}{L^6}\sum_{\sigma\neq\sigma^\prime} \sum_{p,q,r,r^\prime} ||r+p|^2 - (k_F^\sigma)^2|\frac{p\cdot (r-r^\prime)\hat{\varphi}(p)\hat{\varphi}^>(q)}{|r+p|^2 - |r|^2 + |r^\prime - p|^2 - |r^\prime|^2+2\eps} \nn\\
  &\qquad\qquad \times   \hat{v}_{\sigma}(r) \hat{{u}}_{\sigma^\prime}^{<\alpha}(r^\prime - p+q) \hat{v}_\sigma(r+p-q)\hat{v}_{\sigma^\prime}(r^\prime)  \hat{a}_{q-p-r,\sigma}^\ast \hat{a}_{-r,\sigma}\hat{a}_{r^\prime -p+q,\sigma^\prime}\hat{a}_{-r^\prime,\sigma^\prime},\\
  \mathrm{I}_{2;c} &= - \frac{4\varepsilon}{L^6}\sum_{\sigma\neq\sigma^\prime} \sum_{p,q,r,r^\prime} ||r+p|^2 - (k_F^\sigma)^2|\frac{\hat{\varphi}(p)\hat{\varphi}^>(q)}{|r+p|^2 - |r|^2 + |r^\prime - p|^2 - |r^\prime|^2 + 2\eps}  
 \nn \\
 &\qquad \qquad  \times \hat{v}_{\sigma}(r) \hat{{u}}^{<\alpha}_{\sigma^\prime}(r^\prime - p+q) \hat{v}_\sigma(r+p-q)\hat{v}_{\sigma^\prime}(r^\prime)  \hat{a}_{q-p-r,\sigma}^\ast \hat{a}_{-r,\sigma}\hat{a}_{r^\prime -p+q,\sigma^\prime}\hat{a}_{-r^\prime,\sigma^\prime}.
\end{align*}
 In all the three terms, due to the  constraints on $r,r+p-q$ and $q$, we observe that $\hat{\varphi}(p)$ can be replaced by $\hat{\varphi}_1^>(p)$ with 
  \begin{equation}\label{eq: def phi1>}
    \supp \hat{\varphi}_1^> \subset \{ |p|\geq 2\rho^{1/3 - \gamma} \}
  \end{equation}
  where we define $\hat{\varphi}_1^>$ using a smooth cut-off such that $\varphi_1^>$ satisfies the same bounds as $\varphi^>$ in Lemma \ref{lem: bounds phi}.   
We first estimate $\mathrm{I}_{2;a}$ in \eqref{eq: R*S-I3-dec}. Using  
\begin{align}\label{eq:TS-split-I3ac}
|r+p|^2 - (k_F^\sigma)^2 = |p|^2  + 2p\cdot r  + |r|^2 - (k_F^\sigma)^2,
\end{align}
we split $\mathrm{I}_{2;a} = \sum_{j=1}^4 \mathrm{I}_{2;a;j}$ accordingly. Using $\|\Delta \varphi_1^>\|_1\le C$ and $\|\varphi^>\|_\infty\le C$ from Lemma \ref{lem: bounds phi}, we can bound the term 
\begin{equation}\label{eq: I3a1 T*R}
  \mathrm{I}_{2;a;1} = -2\sum_{\sigma\neq \sigma^\prime}\int dxdy\, \Delta{\varphi}^>_1 (x-y){\varphi}^>(x-y) a_\sigma^\ast(v_x)a_{\sigma^\prime}({u}^{<\alpha}_y)a_\sigma(v_x)a_{\sigma^\prime}(v_y)
\end{equation}
with Lemma \ref{lem:4aaa} and obtain 
\[
  |\langle \psi, \mathrm{I}_{2;a;1}\psi\rangle| \leq C\rho^{1+\frac{\beta}{6}}\langle \psi, \mathcal{N}\psi\rangle + C\rho \|\mathcal{N}^{\frac{1}{2}}\psi\|\|\mathcal{N}_\beta^{\frac{1}{2}}\psi\|
\]
for any $0\le \beta <1$. For the other terms, we can use Lemma \ref{lem:1a} and Lemma \ref{lem: bounds phi} to get 
\begin{align*}
  |\langle \psi, (\mathrm{I}_{2;a;2} + \mathrm{I}_{2;a;3} + \mathrm{I}_{2;a;4})\psi\rangle| &\leq C\rho\|\varphi^>_1\|_\infty\left(\rho^{\frac{1}{3}}\|\nabla\varphi^>_1\|_1 + \rho^{\frac{2}{3}}\|\varphi^>_1\|_1\right) \langle \psi, \mathcal{N}\psi\rangle \\
  & \leq C\rho^{1+\gamma}\langle \psi, \mathcal{N}\psi\rangle.
\end{align*}
Thus, 
\[
|\langle \psi, \mathrm{I}_{2;a}\psi\rangle| \leq C \big( \rho^{1+\frac{\beta}{6}} + \rho^{1+\gamma}\big)\langle \psi, \mathcal{N}\psi\rangle + C\rho \|\mathcal{N}^{\frac{1}{2}}\psi\|\|\mathcal{N}_\beta^{\frac{1}{2}}\psi\|.
\]
We now turn to the terms $\mathrm{I}_{2;b}$ and $\mathrm{I}_{2;c}$ in \eqref{eq: R*S-I3-dec}.
  For $\mathrm{I}_{2;b}$, note that on the support  of $\hat v_\sigma(r)$, $\hat v_{\sigma'}(r')$, $v_\sigma(r+p-q)$ and $\hat{\varphi}^>(q)$, we have $|r+p|\ge 3 \rho^{\frac 1 3 -\gamma}$.  
We then use 
\begin{equation}\label{eq: r+p wrt q}
  ||r+p|^2 - (k_F^\sigma)^2| = |q|^2 +2q\cdot(r+p-q) + |r+p-q|^2- (k_F^\sigma)^2
\end{equation}
and split $\mathrm{I}_{2;b}=\sum_{j=1}^4 \mathrm{I}_{2;b;j}$ accordingly. Using \eqref{eq: int t conf space}, we can write the first term in configuration space as 
\begin{align*}  
  \mathrm{I}_{2;b;1} & = 
  4\sum_{\sigma\neq\sigma^\prime}  \int_0^{\infty} dt\,  e^{-2t\varepsilon} \int dx dy dz dz'  \sum_{\ell=1}^3 (\partial_\ell \varphi_1^>)(x-y) (-\Delta \varphi^>)(z-z')   \nn\\
  &\qquad  {\times \zeta^t_>(x-z) \zeta^t_1(y-z^\prime) a_\sigma^\ast(v_{z})a_{\sigma^\prime}({u}_{z^\prime}^{<\alpha}) \big(a_{\sigma^\prime}(v_{t,y})a_{\sigma}(\partial_\ell v_{t,x}) - a_{\sigma^\prime}(\partial_\ell v_{t,y})a_{\sigma}( v_{t,x})\big)}
\end{align*}
with $\zeta_1^t$, $\zeta_>^t$  defined in Lemma \ref{lem: zeta t L1}. With the Cauchy--Schwarz inequality
 {\begin{align*} 
&\int dx dy dz dz' |\partial_\ell \varphi_1^>(x-y)| |\Delta \varphi^>(z-z')|  |\zeta^t_>(x;z)| |\zeta^t_1(y;z^\prime)| \| a_\sigma(v_x) \psi\| \| a_{\sigma^\prime}({u}_y^{<\alpha}) \psi\|  \\
&\le \| \nabla\varphi_1^>\|_2 \|\Delta \varphi^>\|_1  \|\zeta^t_>\|_2 \|\zeta^t_1\|_1 \langle \psi, \cN \psi\rangle
\end{align*}}
and the bounds in Lemmas~\ref{lem: zeta t L1} and~\ref{lem: bounds phi}, we find that 
\begin{align*}  
|\langle\psi,   \mathrm{I}_{2;b;1}\psi\rangle| 
&\le  C   \sum_{\sigma\neq \sigma'} \int_0^{\infty} dt\,  \| \nabla \varphi_1^>\|_2 \|\Delta \varphi^>\|_1 \| \zeta^t_>\|_2 \| \zeta^t_1 \|_1     \|v_{t,\sigma}\|_2  \|  \nabla v_{t,\sigma'} \|_2 \langle \psi, \cN \psi\rangle\\
 &\le C \rho^{4/3} \, \int_0^{\infty} dt\,  t^{-3/4}  {e^{-  \frac{9}2 t \rho^{\frac 2 3 -2\gamma }}}  {e^{t(k_F^\sigma)^2} e^{t(k_F^{\sigma^\prime})^2}}\,   \langle\psi, \cN \psi\rangle 
  \leq  C \rho^{\frac 7 6 + \frac \gamma 2}\langle \psi, \cN\psi\rangle.
 \end{align*}
 Similarly we obtain
\begin{align*}  
& \sum_{j=2}^4 |\langle \psi, \mathrm{I}_{2;b;j}\psi\rangle| 
\\&  \le  C\int_0^{\infty} dt\,  \| \nabla \varphi_1^>\|_2 \Big( \rho^{\frac 1 3 }\|\nabla \varphi^>\|_1 + \rho^{\frac 2 3 }\|\varphi^>\|_1 \Big) e^{t(k_F^\sigma)^2} e^{t(k_F^{\sigma^\prime})^2} \| \zeta^t_>\|_2 \| \zeta^t_1\|_1 \rho^{\frac 4 3} \langle \psi, \cN \psi\rangle\\
 &
  \le C \rho^{\frac 7 6 + \frac {3\gamma} 2}\langle \psi, \cN \psi\rangle. 
 \end{align*}
 Thus, in summary,
$$
  |\langle \psi, \mathrm{I}_{2;b}\psi\rangle| \le C \rho^{\frac 7 6 + \frac \gamma 2}\langle \psi, \cN\psi\rangle. 
$$

Now we consider the term $\mathrm{I}_{2;c}$. We use \eqref{eq:TS-split-I3ac} and split $\mathrm{I}_{2;c}=\sum_{j=1}^4 \mathrm{I}_{2;c;j}$ accordingly. 
The first term $\mathrm{I}_{2;c;1}$ equals
\begin{multline*} 
\mathrm{I}_{2;c;1} = 2\eps \sum_{\sigma\neq\sigma^\prime}  \int_0^{\infty} dt\,  e^{-2t\varepsilon} \int dx dy dz dz' \Delta \varphi^>_1(x-y) \varphi^>(z-z')   \nn\\
   {\times \zeta^t_>(x-z) \zeta^t_1(y-z^\prime) a_\sigma^\ast(v_{z})a_{\sigma^\prime}({u}_{z^\prime}^{<\alpha}) a_{\sigma^\prime}(v_{t,y})a_{\sigma} (v_{t,x})}.
   \end{multline*}
 By using the Cauchy--Schwarz inequality and the bounds in Lemmas~\ref{lem: zeta t L1} and~\ref{lem: bounds phi}, similarly to the analysis of $\mathrm{I}_{2;b;1}$, we obtain 
   \begin{align*} 
  |\langle \psi, \mathrm{I}_{2;c;1}\psi\rangle|&\le C \eps \sum_{\sigma\neq\sigma^\prime} \int_0^{\infty} dt\,  \| \Delta \varphi^>_1\|_1 \|\varphi^>\|_2 \| \zeta^t_>\|_2 \| \zeta^t_1 \|_1     \|v_{t,\sigma}\|_2 \|v_{t,\sigma'}\|_2   \langle \psi,  \cN \psi\rangle\\
  &\le C \eps\rho^{\frac{5}{6} + \frac{\gamma}{2}}\left(\int_0^{\infty} dt\,  t^{-\frac 3 4}  {e^{-  \frac{9}{2} t\rho^{\frac 2 3 -2\gamma }}}  {e^{t(k_F^\sigma)^2} e^{t(k_F^{\sigma^\prime})^2}}\right)  \langle \psi, \cN\psi\rangle
  \\
  & 
  \le  C \rho^{\frac 4 3  + \gamma +\delta}\langle \psi, \cN\psi\rangle. 
\end{align*}
For the other terms $\mathrm{I}_{2;c;j}$ with $j=2,3,4$, we have  {$p_\ell\hat \varphi_1(p)$ (with $\ell =1,2,3)$} or $\hat{\varphi}_1(p)$ instead of $|p|^2 \hat\varphi_1(p)$ in the momentum space. Proceeding as above, 
we find
\begin{align*}  
 \sum_{j=2}^4 |\langle \psi, \mathrm{I}_{2;c;j}\psi\rangle| 
   \le C \rho^{\frac 4 3  + 2\gamma +\delta}\langle \psi, \cN\psi\rangle. 
\end{align*}
Thus, in summary,
$$
  |\langle \psi,\mathrm{I}_{2;c}\psi\rangle| \le C \rho^{\frac 4 3  + \gamma +\delta}\langle \psi, \cN\psi\rangle.
$$
Since $0<\gamma<1/3$ by assumption, we thus conclude that 
 \eqref{eq: R*S-I3-dec} is bounded as 
 \begin{equation}\label{eq: R*S-I3}
|\langle \psi,  \mathrm{I}_2\psi\rangle| \le C\left(\rho^{1+\frac{\beta}{6}} + \rho^{1+\gamma} \right) \langle \psi, \mathcal{N}\psi\rangle + \rho \|\mathcal{N}^{\frac{1}{2}}\psi\|\|\mathcal{N}^{\frac{1}{2}}_\beta\psi\|.
\end{equation}

\medskip
\noindent{\bf Analysis of $\mathrm{I}_3$.} We insert the decomposition $1 = \widehat{\chi}_>(p) + \widehat{\chi}_<(p)$ with $\widehat{\chi}_<$ and $\widehat{\chi}_>$ defined in \eqref{eq: def chi< chi>}, and correspondingly write  $\mathrm{I}_3$  in \eqref{eq:TS11-dec} as 
$$\mathrm{I}_3 = \mathrm{I}_3^> + \mathrm{I}_3^<.$$
We first estimate $\mathrm{I}_3^<$. For $|p| \leq 5\rho^{1/3 - \gamma}$ and  $|r^\prime| \leq C\rho^{1/3}$, we have $|r^\prime - p|\leq 6\rho^{1/3 -\gamma}$, which allows to replace $\hat{u}_{\sigma^\prime}(r^\prime - p)$ by $\hat{u}^<_{\sigma^\prime}(r^\prime - p)$ defined in \eqref{eq: u< u> gamma} and the same for $\hat u_\sigma(r+p)$. With the aid of \eqref{eq: int t conf space} we can thus write in configuration space 
\begin{multline*}
  \mathrm{I}_{3}^{<} = 2\sum_{\sigma\neq \sigma^\prime} \int_0^{\infty} dt\, e^{-2t\varepsilon} \int dxdydzdz^\prime\, \Delta\varphi^<(x-y) {\varphi}^>(z-z^\prime)  
  \\
  \times ( \Delta + (k_F^\sigma)^2) v_{t,\sigma}) (x-z) u_{t,\sigma^\prime}^<(y-z^\prime) a^\ast_{\sigma}(u_z) a_\sigma(u_{t,x}^<) a_{\sigma^\prime}({u}^{<\alpha}_{z^\prime})a_{\sigma^\prime}(v_{t,y})
\end{multline*}
and then estimate $\mathrm{I}_{3}^{<}$ using Lemma \ref{lem:4aa}.  
Using also Lemma \ref{lem: bounds phi} and \eqref{eq: est int t vu< final}, this yields 
\begin{align*}
  |\langle \psi, \mathrm{I}_3^<\psi\rangle| 
 &\leq C\rho^{\frac{4}{3} +\frac{3}{2}\gamma -\frac{3}{2}\alpha -\kappa}\langle \psi, \mathcal{N}\psi\rangle
\end{align*}
for any $\kappa>0$ and $\psi \in\mathcal{F}_{\mathrm{f}}$. 

Next we consider $\mathrm{I}_3^>$. Using  the constraints $|p|\geq 4\rho^{1/3-\gamma}$, $|r|\leq C\rho^{1/3}$, $|r^\prime|\leq C\rho^{1/3}$, we can replace $\hat{u}_{\sigma^\prime}(r^\prime - p)$ and $\hat{u}_{\sigma}(r + p)$ by $\hat{u}^>_{\sigma^\prime}(r^\prime - p)$ and $\hat{u}^>_\sigma(r+ p)$, respectively, with $\hat{u}^>_\sigma$ defined in \eqref{eq: u< u> gamma}. Moreover, we multiply and divide by $|r^\prime-p|^2$ and redistribute the momenta by writing 
$$|r^\prime - p |^2 = |r^\prime - p+q|^2 + |q|^2 - 2q\cdot (r^\prime - p+q).$$
This way, we obtain the identity 
\begin{align*}
  \mathrm{I}_{3}^> &= 2 \sum_{\sigma \neq \sigma^\prime}\int_0^{\infty} dt\, e^{-2t\varepsilon} \int dxdydzdz^\prime\, \Delta\varphi^>(x-y) (\Delta + (k_F^\sigma)^2) v_{t,\sigma}(x-z)  \widetilde{u}^>_{t,\sigma^\prime}(y-z^\prime)  
  \\
 & \times  \sum_{\ell=1}^3 \big( \partial_\ell^2 {\varphi}^>(z-z^\prime) a_{\sigma^\prime}({u}^{<\alpha}_{z^\prime}) + 2 \partial_\ell{\varphi}^>(z-z^\prime) a_{\sigma^\prime}(\partial_\ell {u}^{<\alpha}_{z^\prime}) + {\varphi}^>(z-z^\prime) a_{\sigma^\prime}(\partial_\ell^2{u}^{<\alpha}_{z^\prime}) \big)  
  \\
&\times a^\ast_{\sigma}(u_z) a_\sigma(u_{t,x}^>) a_{\sigma^\prime}(v_{t,y}).
\end{align*}
Again it can be estimated by Lemma \ref{lem:4aa}, with the result that 
\begin{align*}
  |\langle \psi, \mathrm{I}_{3}^>\psi\rangle 
   &\leq C\rho^{1+\frac{9}{2}\gamma - \frac{7}{2}\alpha - \kappa}\|\mathbb{H}_0^{\frac{1}{2}}\psi\|\|\mathcal{N}^{\frac{1}{2}}\psi\|
\end{align*}
for any $\kappa>0$. Here we used Lemmas~\ref{lem:1a} and~\ref{lem: bounds phi},  \eqref{eq: est int t vu> final-b} and the support properties of $\hat{u}^{<\alpha}_\sigma$. As a result,  
\begin{align}\label{i3n}
|\langle \psi, \mathrm{I}_{3}\psi\rangle| \le C\rho^{\frac{4}{3} + \frac{3}{2}\gamma - \frac{3}{2}\alpha - \kappa}\langle \psi, \mathcal{N}\psi\rangle + C\rho^{1+\frac{9}{2}\gamma - \frac{7}{2}\alpha - \kappa}\|\mathbb{H}_0^{\frac{1}{2}}\psi\|\|\mathcal{N}^{\frac{1}{2}}\psi\|.
\end{align}

{ {
\medskip
\noindent{\bf Analysis of $\{ S_{1,2,\sigma}^\ast(r), T_\sigma(r) \}$.} The term to consider is given by   
\begin{align*}
  &  \sum_{\sigma}\sum_r ||r|^2 - (k_F^\sigma)^2| \{ S_{1,2,\sigma}^\ast(r), T_\sigma(r) \} \\
  &= \frac{1}{L^6}\sum_{\sigma\neq \sigma^\prime}\sum_{q,p,r,r^\prime, s} ||r|^2 - (k_F^\sigma)^2| \hat{\varphi}^>(q)\widehat{\omega}^\varepsilon_{r-p,s}(p)\hat{u}_\sigma(r)\hat{{u}}^{<\alpha}_\sigma(r-q)  \hat{v}_{\sigma^\prime}(r^\prime)\hat{u}_{\sigma^\prime}(r^\prime -q) 
    \\
   &\qquad\qquad \times  \hat{v}_\sigma(r-p)\hat{u}_{\sigma^\prime}(s-p)\hat{v}_{\sigma^\prime}(s) \Big\{\hat{a}_{r-q,\sigma}^\ast \hat{a}_{r^\prime - q, \sigma^\prime} \hat{a}_{-r^\prime, \sigma^\prime}, \hat{a}_{-s,\sigma^\prime}^\ast \hat{a}^\ast_{s-p, \sigma^\prime }\hat{a}_{p-r,\sigma}^\ast \Big\}.
  \end{align*}
  Computing the anti-commutator explicitly, we obtain
  \begin{align}\label{eq: anticommutator S12*T}
&\Big\{\hat{a}_{r-q,\sigma}^\ast \hat{a}_{r^\prime - q, \sigma^\prime} \hat{a}_{-r^\prime, \sigma^\prime}, \hat{a}_{-s,\sigma^\prime}^\ast \hat{a}^\ast_{s-p, \sigma^\prime }\hat{a}_{p-r,\sigma}^\ast \Big\} = \delta_{s-p, r^\prime -q}\hat{a}^\ast_{r-q,\sigma}\hat{a}_{-s,\sigma^\prime}^\ast \hat{a}_{p-r,\sigma}^\ast\hat{a}_{-r^\prime, \sigma^\prime}\nn
\\
&\qquad -\delta_{s,r^\prime}\hat{a}_{r-q, \sigma}^\ast \hat{a}_{s-p, \sigma^\prime}^\ast \hat{a}_{r^\prime - q, \sigma^\prime}\hat{a}_{p-r,\sigma}^\ast + \delta_{s,r^\prime}\delta_{p,q}\hat{a}_{r-q,\sigma}^\ast \hat{a}_{p-r,\sigma}^\ast.
  \end{align}
With a change of variables and using $\hat{{u}}^{<\alpha}_\sigma(r-p) \hat{v}_\sigma(r-p)=0$, we can write 
  \begin{align}\label{eq:S12T-dec}
\sum_{\sigma}\sum_r ||r|^2 - (k_F^\sigma)^2| \{ S_{1,2,\sigma}^\ast(r), T_\sigma(r) \}  = \sum_{j=1}^2 \mathrm{J}_j
  \end{align}
  with 
    \begin{align*}
  \mathrm{J}_1 &= \frac{1}{L^6}\sum_{\sigma\neq \sigma^\prime}\sum_{q,p,r,r^\prime}||r+p|^2 - (k_F^\sigma)^2| \widehat{\omega}^\varepsilon_{r,r^\prime}(p) \hat{\varphi}^>(q) \hat{u}_\sigma(r+p)\hat{{u}}^{<\alpha}_{\sigma}(r+p-q)
  \\
  &\qquad \qquad \times \hat{v}_{\sigma^\prime}(r^\prime+q-p)\hat{u}_{\sigma^\prime}(r^\prime -p)\hat{v}_\sigma(r)\hat{v}_{\sigma^\prime}(r^\prime)\hat{a}_{r-q+p,\sigma}^\ast \hat{a}_{-r^\prime,\sigma^\prime}^\ast \hat{a}_{-r,\sigma}^\ast\hat{a}_{-r^\prime -q+p, \sigma^\prime},\\
  \mathrm{J}_2 &= \frac{1}{L^6}\sum_{\sigma\neq \sigma^\prime}\sum_{q,p,r,r^\prime}||r+p|^2 - (k_F^\sigma)^2| \widehat{\omega}^{\varepsilon}_{r, r^\prime}(p)\hat{\varphi}^>(q) \hat{u}_\sigma(r+p)\hat{{u}}^{<\alpha}_{\sigma}(r+p-q)  \\
&\qquad \qquad  \times \hat{v}_{\sigma^\prime}(r^\prime) \hat{u}_{\sigma^\prime}(r^\prime - q)\hat{v}_\sigma(r) 
\hat{u}_{\sigma^\prime}(r^\prime - p)\hat{a}_{r+p-q,\sigma}^\ast \hat{a}_{r^\prime -p, \sigma^\prime}^\ast\hat{a}_{-r,\sigma}^\ast \hat{a}_{r^\prime -q, \sigma^\prime}.
\end{align*}
More precisely, $\mathrm{J}_1$ arises from the second term in \eqref{eq: anticommutator S12*T} after the change of variables $r\mapsto r-p$ and $r^\prime\mapsto r^\prime-q+p$, while $\mathrm{J}_2$ arises from the last term in the same equation after the change of variable $r\mapsto r-p$. 
The term $\mathrm{J}_1$ is similar to the term $\mathrm{I}_2$ in \eqref{eq:TS11-dec}, and can be bounded in the same way as 
\begin{equation}\label{eq: est J1 lem 55}
  |\langle \psi,  \mathrm{J}_1\psi\rangle| \le C\left(\rho^{1+\frac{\beta}{6}} + C\rho^{1+\gamma}\right) \langle \psi, \mathcal{N}\psi\rangle + \rho \|\mathcal{N}^{\frac{1}{2}}\psi\|\|\mathcal{N}^{\frac{1}{2}}_\beta\psi\|
\end{equation}
The term  $\mathrm{J}_2$ we split as 
$
\mathrm{J}_2 = \mathrm{J}_2^> + \mathrm{J}_2^<
$
using $1=\widehat{\chi}_>(p) + \widehat{\chi}_<(p)$.  
For $\mathrm{J}_2^<$, we observe that 
 the support properties of $p,r,r^\prime$ imply $|r+p|, |r^\prime -p| \leq 6\rho^{1/3 - \gamma}$. As a consequence, we may replace $\hat{u}_\sigma(r+p)$ and $\hat{u}_{\sigma^\prime}(r^\prime - p)$ by $\hat{u}^<_\sigma(r+p)$ and $u^<_{\sigma^\prime}(r^\prime - p)$, respectively, with $\hat{u}^<_\sigma$ defined in \eqref{eq: u< u> gamma}. Moreover, using \eqref{eq: r+p wrt q} we further split 
$
 \mathrm{J}_2^<= \mathrm{J}_{2;a}^< + \mathrm{J}_{2;b}^< + \mathrm{J}_{2;c}^< + \mathrm{J}_{2;d}^<
$
accordingly. The first term, which involves $\widehat{\omega}^{\varepsilon}_{r, r^\prime}(p)\widehat{\chi}_<(p) |q|^2\hat{\varphi}^>(q)$, can be written in configuration space using  \eqref{eq: int t conf space}  as  
\begin{align*}
  \mathrm{J}_{2;a}^< &
=  2 \sum_{\sigma\neq \sigma^\prime}\int_0^{\infty}dt\, e^{-2t\varepsilon}\int dxdydzdz^\prime\, \Delta{\varphi}^>(x-y)\Delta\varphi^<(z-z^\prime)u_{t,\sigma}^<(x-z)v_{t,\sigma^\prime}(y-z^\prime) \\
   &\qquad   \times a^\ast_\sigma({u}^{<\alpha}_x)a^*_{\sigma^\prime}(u_{t,z^\prime}^<)a_{\sigma}^\ast(v_{t,z}) a_{\sigma^\prime}(u_y).
\end{align*}
Using Lemma \ref{lem:4aa} and  $\|\hat{v}_{t,\sigma^\prime}\|_2 \|\hat{u}^<_{t,\sigma^\prime}\|_2 \leq C\rho^{1 -(3/2)\gamma}$, we can bound
\begin{align*}
|\langle \psi, \mathrm{J}_{2;a}^<\psi\rangle| &\leq C \rho^{1 - \frac {3 \gamma} 2}  \left(\int_0^{\infty} dt\, e^{-2t\varepsilon}  \|\hat{u}_{t,\sigma}^<\|_2 \|\hat{v}_{t,\sigma}\|_2\right) \|\Delta\varphi^< \|_1 \|\Delta\varphi^> \|_1 \langle \psi, \cN\psi\rangle   \\
  &\leq 
  C \rho^{\frac{4}{3} - 2\gamma - \kappa}\langle \psi, \mathcal{N}\psi\rangle
\end{align*}
for any $\kappa>0$. 
Here we also applied Lemma \ref{lem: bounds phi} and  \eqref{eq: est int t vu< final}. The other terms in $\mathrm{J}_{2}^<$ can be estimated in a similar way, using also Lemma \ref{lem:1a}. 
We omit the details, and only state the final result
\begin{align*}
 &  |\langle \psi, (\mathrm{J}_{2;b}^<+\mathrm{J}_{2;c}^<+\mathrm{J}_{2;d}^<) \psi\rangle|\leq 
  C\rho^{1 - \kappa}(\rho^{-\alpha} + \rho^{-\gamma})\left(\|\mathbb{H}_0^{\frac{1}{2}}\psi\|\|\mathcal{N}^{\frac{1}{2}}\psi\| + \rho^{\frac{1}{3}}\langle \psi, \mathcal{N}\psi\rangle\right).
\end{align*}
Since $\gamma < \alpha$ by assumption, we thus obtain,  for any $\kappa>0$, 
\begin{align}\label{eq: estimate I2< T1*R}
  |\langle\psi, \mathrm{J}_2^< \psi\rangle| \leq C\left(\rho^{\frac{4}{3}  - 2\gamma - \kappa} + \rho^{\frac{4}{3} - \alpha - \kappa}\right) \langle \psi, \mathcal{N}\psi\rangle + C \rho^{1-\alpha -\kappa}\|\mathbb{H}_0^{\frac{1}{2}}\psi\| \|\mathcal{N}^{\frac{1}{2}}\psi\|.
\end{align}

For the term $\mathrm{J}^>_2$, we observe that on the support  of $\hat v_\sigma(r)$, $\hat v_{\sigma'}(r')$, $\hat{\varphi}^>(p)$ and $\hat{\varphi}^>(q)$, we have $\hat u_\sigma(r'-p)=\hat u^>_\sigma(r'-p)$ and $\hat u_\sigma(r'-q)=\hat u^>_\sigma(r'-q)$ with $u_\sigma^>$ defined in \eqref{eq: u< u> gamma}.  
Using the identity in \eqref{eq: comparison w eps with phi} we can write 
 $ \mathrm{J}^>_2 =  \mathrm{J}^>_{2;a} + \mathrm{J}^>_{2;b} + \mathrm{J}^>_{2;c} $ 
 with 
\begin{align*}  
  \mathrm{J}^>_{2;a} &= \frac{2}{L^6}\sum_{\sigma\neq \sigma^\prime}\sum_{q,p,r,r^\prime}||r+p|^2 - (k_F^\sigma)^2| \hat{\varphi}^>(p)\hat{\varphi}^>(q)\hat{v}_{\sigma^\prime}(r^\prime) \hat{u}_\sigma(r+p) \\
&\qquad \qquad  \times 
\hat{u}^{<\alpha}_{\sigma}(r+p-q)  \hat{u}^>_{\sigma^\prime}(r^\prime - p)  \hat{u}^>_{\sigma^\prime}(r^\prime - q)  \hat{v}_\sigma(r)  \hat{a}_{r+p-q,\sigma}^\ast \hat{a}_{r^\prime -p, \sigma^\prime}^\ast\hat{a}^\ast_{-r,\sigma}\hat{a}_{r^\prime -q, \sigma^\prime},\\
  \mathrm{J}^>_{2;b} &= -\frac{4}{L^6}\sum_{\sigma\neq \sigma^\prime}\sum_{q,p,r,r^\prime}||r+p|^2 - (k_F^\sigma)^2| \frac{p\cdot (r-r^\prime)  \hat{\varphi}^>(p) \hat{\varphi}^>(q)}{|r+p|^2 - |r|^2 + |r^\prime - p|^2 - |r^\prime|^2+2\eps}  \hat{u}_\sigma(r+p) \hat{v}_{\sigma^\prime}(r^\prime)  \\
&\qquad \qquad  \times 
\hat{u}^{<\alpha}_{\sigma}(r+p-q)  \hat{u}^>_{\sigma^\prime}(r^\prime - p)  \hat{u}^>_{\sigma^\prime}(r^\prime - q)\hat{v}_\sigma(r)  \hat{a}_{r+p-q,\sigma}^\ast \hat{a}_{r^\prime -p, \sigma^\prime}^\ast \hat{a}^\ast_{-r,\sigma}\hat{a}_{r^\prime -q, \sigma^\prime},\\
  \mathrm{J}^>_{2;c} &= -\frac{4\eps}{L^6}\sum_{\sigma\neq \sigma^\prime}\sum_{q,p,r,r^\prime}||r+p|^2 - (k_F^\sigma)^2| \frac{\hat{\varphi}^>(p)\hat{\varphi}^>(q)}{|r+p|^2 - |r|^2 + |r^\prime - p|^2 - |r^\prime|^2 + 2\eps} \hat{u}_\sigma(r+p)\hat{v}_{\sigma^\prime}(r^\prime) \\
&\qquad \qquad  \times 
\hat{u}^{<\alpha}_{\sigma}(r+p-q)  \hat{u}^>_{\sigma^\prime}(r^\prime - p)  \hat{u}^>_{\sigma^\prime}(r^\prime - q)  \hat{v}_\sigma(r)  \hat{a}_{r+p-q,\sigma}^\ast \hat{a}_{r^\prime -p, \sigma^\prime}^\ast \hat{a}^\ast_{-r,\sigma}\hat{a}_{r^\prime -q, \sigma^\prime}.
\end{align*}
In the term $\mathrm{J}^>_{2;a}$, we observe that $\hat{u}(r+p)= 1$. We use  {\eqref{eq:TS-split-I3ac}} and further split $\mathrm{J}^>_{2;a}=\sum_{j=1}^4 \mathrm{J}^>_{2;a;j}$ accordingly. The first term, involving $|p|^2 \hat{\varphi}^>(p)\varphi^>(q)$, can be written in configuration space as
 {\begin{align*}  
\mathrm{J}^>_{2;a;1}= -2\sum_{\sigma\neq \sigma^\prime} \int dx dy dz\, \Delta \varphi^>(x-y) \varphi^>(x-z) v_{\sigma'}(y-z) a_\sigma^*(u_x^{<\alpha})  a_{\sigma'}^*(u^>_y)a_{\sigma}^\ast(v_{x}) a_{\sigma'} (u^>_z)  
\end{align*}  
Applying the Cauchy--Schwarz inequality and Lemmas~\ref{lem:1a} and~\ref{lem: bounds phi}, we have
\begin{align*}  
|\langle \psi, \mathrm{J}^>_{2;a;1}\psi\rangle| & \le \sum_{\sigma\neq \sigma'} \|\Delta \varphi^>\|_1 \|\varphi^>\|_1 \| v_{\sigma'}\|_\infty  \|v_{\sigma}\|_2 \|u^{<\alpha}_\sigma\|_2  
\\
&\qquad\times   \sqrt{\int dy \|a_{\sigma'} (u_y^>)\psi\|^2 }\sqrt{\int dz \|a_{\sigma'} (u_z^>)\psi\|^2 }\\
&\le  
C\rho^{\frac 2 3 + 4 \gamma - \frac{3}{2}\alpha} \langle \psi, \bH_0\psi\rangle. 
\end{align*}}
Proceeding in a similar way for the other terms, we find
\begin{align*}  
\sum_{j=2}^4 |\langle \psi, \mathrm{J}^>_{2;a;j}\psi\rangle| &\le
C\rho^{\frac 2 3 + 5 \gamma - \frac{3}{2}\alpha} \langle \psi, \bH_0\psi\rangle. 
\end{align*}  
Thus 
$$
|\langle \psi, \mathrm{J}^>_{2;a} \psi\rangle|\le C\rho^{\frac 2 3 + 4\gamma - \frac{3}{2}\alpha} \langle\psi, \mathbb{H}_0\psi\rangle.
$$

For $\mathrm{J}^>_{2;b}$, we use \eqref{eq: r+p wrt q} and split $\mathrm{J}^>_{2;b}=\sum_{j=1}^4 \mathrm{J}^>_{2;b;j}$ accordingly.  For the first term involving $p \cdot(r-r') \hat{\varphi}^>(p) |q|^2 \hat\varphi^>(q)$, we use that  on the supports of $\hat v_\sigma(r)$ and $\hat{\varphi}^>(p)$ we can replace $\hat{u}_\sigma(r+p)$ by $\hat{u}^>_\sigma(r+p)$ defined in \eqref{eq: u< u> gamma}, and then use \eqref{eq: int t conf space} to write
\begin{align*}  
\mathrm{J}^>_{2;b;1} 
&= -4 \sum_{\sigma\neq \sigma^\prime} \int_0^\infty dt\, e^{-2t\eps} \sum_{\ell=1}^3\int dx dy dz dz' (\partial_\ell \varphi^>)(x-y) (\Delta \varphi^>)(z-z')  \zeta^t_>(x-z)  \\
&\qquad \times a_\sigma^*(u^{<\alpha}_{z})  a_{\sigma'}^*(u^>_{t,y}) \Big( (\partial_\ell v_{t,\sigma'})(y-z') a_{\sigma}^\ast(v_{t,x}) -  v_{t,\sigma'}(y-z') a_{\sigma}^\ast(\partial_\ell v_{t,x}) \Big)a_{\sigma'} (u^>_{z^\prime}) . 
\end{align*} 
{}Applying the Cauchy--Schwarz inequality, as well as Lemmas~\ref{lem: bounds phi}, \ref{lem: zeta t L1} and \ref{lem:1a}, we obtain
\begin{align*}   
|\langle\psi, \mathrm{J}^>_{2;b;1}\psi\rangle| &\le 8\sum_{\sigma\neq \sigma^\prime} \int_0^\infty dt \|\nabla\varphi^>\|_1 \|\Delta \varphi^>\|_1 \|\zeta^t_>\|_2  \|v_{t,\sigma}\|_2 \|\nabla v_{t,\sigma^\prime}\|_2\|u^{<\alpha}_\sigma\|_2   \\
&\qquad\times  \sqrt{\int dy \|a_{\sigma'}(u^>_{t,y})\psi\|^2}\sqrt{\int dz^\prime \, \|a_{\sigma'}(u^>_{z^\prime})\psi\|^2 } \\
&\le C  \int_0^\infty dt\, t^{-\frac 3 4} e^{- \frac{9}{2} t  \rho^{\frac 2 3 -2\gamma}}e^{t(k_F^\sigma)^2} \,  {\rho^{\frac{5}{6} + 3\gamma - \frac{3}{2}\alpha} \langle \psi, \bH_0\psi\rangle \leq C  \rho^{\frac 2 3 + \frac{7\gamma}{2} - \frac{3}{2}\alpha} \langle \psi, \bH_0\psi\rangle.} 
\end{align*}  
In a similar way, we obtain for the remaining terms  
\begin{align*}  
&\sum_{j=2}^4 |\langle \psi, \mathrm{J}^>_{2;b;j} \psi\rangle|\le 
C\rho^{\frac 2 3 + \frac{7}{2}\gamma  -\frac{3}{2}\alpha}\left(\rho^{\gamma- \alpha} + \rho^{2\gamma  - 2\alpha}\right)\langle \psi, \bH_0\psi\rangle.
\end{align*}  
Thus, for $\gamma<\alpha$, we get
\begin{align*}  
|\langle \psi,  \mathrm{J}^>_{2;b}\psi\rangle \leq C\rho^{\frac 2 3 + \frac{11}{2}\gamma  -\frac{7}{2}\alpha}\langle \psi, \bH_0\psi\rangle.
\end{align*}  

The term $\mathrm{J}^>_{2;c}$ can be treated similarly to $\mathrm{J}^>_{2;b}$. Using \eqref{eq: r+p wrt q}, we decompose $\mathrm{J}^>_{2;c}=\sum_{j=1}^4 \mathrm{J}^>_{2;c;j}$ accordingly. The first term, $J_{2;c;1}$ can be estimated similarly as $J_{2;b;1}$,  effectively replacing the bound  $\|\nabla \varphi^>\|_1  \|v_{t,\sigma}\|_2  \|\nabla v_{t,\sigma^\prime}\|_2 \le C\rho^{1 +\gamma}$  by  
$\eps \|\varphi^>\|_1  \|v_{t,\sigma}\|_2\|v_{t,\sigma^\prime}\|_2 \le C\rho^{1+2\gamma+\delta}$.
This yields the improved estimate
\begin{align*}   
|\langle\psi, \mathrm{J}^>_{2;c;1}\psi\rangle| &\le C\varepsilon\sum_{\sigma\neq \sigma^\prime} \int_0^\infty dt \|\varphi^>\|_1 \|\Delta \varphi^>\|_1 \|\zeta^t_>\|_2  \|v_{t,\sigma}\|_2 \| v_{t,\sigma^\prime}\|_2\|u^{<\alpha}_\sigma\|_2   \\
&\qquad\times  \sqrt{\int dy \|a_{\sigma'}(u^>_{t,y})\psi\|^2}\sqrt{\int dz^\prime \, \|a_{\sigma'}(u^>_{z^\prime})\psi\|^2 } \\
&\le C  \int_0^\infty dt\, t^{-\frac 3 4} e^{- \frac{9}{2} t  \rho^{\frac 2 3 -2\gamma}}e^{t(k_F^\sigma)^2} \, \|\varphi^>\|_1 \|\Delta \varphi^>\|_1 \ \rho^{\frac{3}{2} + 2\gamma - \frac{3}{2}\alpha+\delta} \langle \psi, \bH_0\psi\rangle \\
&\le C 
\rho^{\frac 2 3 + \frac{9\gamma}{2} - \frac{3}{2}\alpha +\delta} \langle \psi, \bH_0\psi\rangle.
\end{align*}
The remaining terms are treated analogously. Combining these estimates, and using that $\gamma <\alpha$, we obtain 
\[
  \sum_{j=1}^4 |\langle\psi, J^>_{2;c;j}\psi\rangle| \leq C\rho^{\frac{2}{3} +\frac{9}{2}\gamma - \frac{3}{2}\alpha + \delta}\left(1+ \rho^{\gamma - \alpha} + \rho^{2\gamma - 2\alpha}\right)\langle \psi, \mathbb{H}_0 \psi\rangle.
\] 
Thus,
\begin{align*}  
|\langle \psi, \mathrm{J}^>_{2;c} \psi\rangle| \le C\rho^{\frac 2 3 + \frac{9}{2} \gamma  -\frac{3}{2}\alpha +\delta}\langle \psi, \bH_0\psi\rangle. 
\end{align*}  
Combining all the bounds, we conclude that 
\begin{align}\nonumber
  |\langle \psi, \mathrm{J}_{2} \psi\rangle| &\le  C\left(\rho^{\frac{4}{3} - 2\gamma - \kappa} + \rho^{\frac{4}{3} - \alpha -\kappa}\right)\langle \psi,\mathcal{N}\psi\rangle
  \\
  &\quad + C\rho^{1-\alpha - \kappa}\|\mathbb{H}_0^{\frac{1}{2}}\psi\|\|\mathcal{N}^{\frac{1}{2}}\psi\| +  C\rho^{\frac 2 3 + \frac{11}{2}\gamma  -\frac{7}{2}\alpha}\langle \psi, \bH_0\psi\rangle.
  \label{final:j2}
\end{align}

\medskip
\noindent
{\bf Conclusion.} The statement of Lemma~\ref{lem: TS1} follows from a combination of the bounds in \eqref{i1n}, \eqref{eq: R*S-I3}, \eqref{i3n}, \eqref{eq: est J1 lem 55}  and \eqref{final:j2}.
}}
\end{proof}

\subsection{Proof of Lemma \ref{lem: S2}}

By splitting $S_\sigma(r)=S_{1,\sigma}(r)+ S_{2,\sigma}(r)$ as in  \eqref{eq: def S1 S2}, we have  
\begin{align*}  
S^\ast_{2,\sigma}(r)(T_\sigma(r) + S_\sigma(r)) = S^\ast_{2,\sigma}(r)S_{2,\sigma}(r) + S^\ast_{2,\sigma}(r)S_{1,\sigma}(r) + S^\ast_{2,\sigma}(r)T_\sigma(r). 
  \end{align*}
In the following, we estimate the contribution of the three operators on the right-hand side  separately.

\medskip

\noindent
{\bf Analysis of $S^\ast_{2,\sigma}(r)S_{2,\sigma}(r) $.} By a direct computation we have 
\begin{align*}
 & \sum_{\sigma}\sum_{r}||r|^2 -(k_F^\sigma)^2| S^\ast_{2,\sigma}(r)S_{2,\sigma}(r) \\
&= \sum_{\sigma\neq \sigma^\prime}\frac{1}{L^3}\sum_r ||r|^2 - (k_F^\sigma)^2| \hat{u}^{<\alpha}(r)\int dxdydzdz^\prime\, e^{ir\cdot(x-z)}{\varphi}^>(x-y){\varphi}^>(z-z^\prime) 
\\
&\qquad\qquad\times 
 a^\ast_{\sigma}(u_x)a_{\sigma^\prime}^\ast(v_y)a^\ast_{\sigma^\prime}(u_y^>)a_{\sigma^\prime}(u_{z^\prime}^>) a_{\sigma^\prime}(v_{z^\prime})a_{\sigma}(u_z),
\end{align*}
with $\hat{u}^>_{\sigma^\prime}$ defined in \eqref{eq: u< u> gamma}.
Thus, using the notation introduced in \eqref{eq: def b general}, we have  for any $\psi\in\mathcal{F}_{\mathrm{f}}$ by the Cauchy-Schwarz inequality (in $r$)  
\begin{align}\label{eq: est S2*S2}
 & \sum_{\sigma}\sum_{r}||r|^2 -(k_F^\sigma)^2| \langle \psi, S^\ast_{2,\sigma}(r)S_{2,\sigma}(r)\psi\rangle  = \frac{1}{L^3}\sum_{\sigma\neq\sigma^\prime}\sum_r (|r|^2 - (k_F^\sigma)^2) \hat{{u}}^{<\alpha}_\sigma(r) 
  \nn\\
&\qquad  \times \left\langle \left(\int dx\, e^{-ir\cdot x} b_{\sigma^\prime}({\varphi}_x^>, v, u^>)a_\sigma(u_x)\psi\right), \left(\int dz\,e^{-ir\cdot z} b_{\sigma^\prime}({\varphi}_z^>, v, u^>) a_\sigma(u_z)\psi\right)\right\rangle \nn\\
&\leq \sum_{\sigma \neq \sigma^\prime}\| (|\cdot|^2 - (k_F^\sigma)^2)\hat{u}^{<\alpha}_\sigma\|_\infty \sup_x \|b_{\sigma^\prime}({\varphi}_x^>, v, u^>)\|^2\int dx\, \|a_\sigma(u_x)\psi\|^2 .
  \end{align}
  Here we applied Lemma \ref{lem:6a} together with $||k|^2 - (k_F^\sigma)^2| \hat{{u}}^{<\alpha}_\sigma(k) \leq C\rho^{2/3 - 2\alpha}$. By Lemmas \ref{lem: bounds phi} and \ref{lem:2a} (see \eqref{eq: b op phi vu}), we hence get 
  \begin{equation}\label{56.1}
  \sum_{\sigma}\sum_{r}||r|^2 -(k_F^\sigma)^2| |\langle \psi, S^\ast_{2,\sigma}(r)S_{2,\sigma}(r)\psi\rangle| \leq C\rho^{\frac{4}{3} +\gamma - 2\alpha} \langle \psi, \mathcal{N}\psi\rangle.
  \end{equation}

\medskip

\noindent
{\bf Analysis of $S^\ast_{2,\sigma}(r)S_{1,\sigma}(r) $.} 
By the Cauchy--Schwarz inequality, we deduce from  \eqref{56.1} and Lemma \ref{lem: S1} that 
\begin{align}\label{eq: est S2*S1}
 & \sum_{\sigma}\sum_{r}||r|^2 -(k_F^\sigma)^2| \langle \psi, S^\ast_{2,\sigma}(r)S_{1,\sigma}(r)\psi\rangle \nn\\
 &\le \rho^{-\eta} \sum_{\sigma}\sum_{r}||r|^2 -(k_F^\sigma)^2| \|S_{2,\sigma}(r)\psi\|^2+  \rho^{\eta} \sum_{\sigma}\sum_{r}||r|^2 -(k_F^\sigma)^2| \|S_{1,\sigma}(r)\psi\|^2 \nn\\
 &\le C\rho^{\frac{4}{3} + \gamma - 2\alpha -\eta} \langle \psi, \mathcal{N}\psi\rangle + C\rho^{1+\eta}\langle \psi, \mathcal{N}\psi\rangle + C\rho^\eta|\langle \psi, \mathcal{E}_{S_1^\ast S_1}\psi\rangle| 
  \end{align}
  for any $\eta>0$.
Using Lemma \ref{lem: S1}, under the conditions $0 <\gamma <1/6<\alpha<1/3$ and \eqref{eq: mix cond alpha gamma}, 
we have 
\[
  |\langle \psi, \mathcal{E}_{S_1^\ast S_1}\psi\rangle|\leq C\rho\langle \psi, \mathcal{N}\psi\rangle + C\rho^{\frac{2}{3} + 6\gamma - 5\alpha}\langle \psi, \mathbb{H}_0 \psi\rangle.
\]
Thus, after an optimization over $\eta$, which yields $\eta= 1/6+\gamma/2-\alpha$, we find 
\begin{multline}\label{56.2}
  \sum_{\sigma}\sum_{r}||r|^2 -(k_F^\sigma)^2| |\langle \psi, S^\ast_{2,\sigma}(r)S_{1,\sigma}(r)\psi\rangle | \\ \leq C\rho^{\frac{7}{6} + \frac{\gamma}{2} - \alpha}\langle \psi, \mathcal{N}\psi\rangle + C\rho^{\frac{5}{6} + \frac{13}{2}\gamma - 6\alpha}\langle \psi, \mathbb{H}_0 \psi\rangle.
\end{multline}

\medskip

\noindent
{\bf Analysis of $S^\ast_{2,\sigma}(r) T_\sigma(r)$.} 
The final term to consider is 
\begin{align}\nonumber
&\sum_\sigma \sum_r  ||r|^2 - (k_F^\sigma)^2| S^\ast_{2,\sigma}(r) T_{\sigma}(r)  
\\ \nonumber
& = -\frac{1}{L^6}\sum_{\sigma\neq \sigma^\prime}\sum_{p,q,r,r^\prime,s} ||r|^2 - (k_F^\sigma)^2|\hat{\varphi}^>(p)\widehat{\omega}^{\varepsilon}_{r-q, s}(q)\\
  &\qquad \qquad \qquad\qquad \times \hat{{u}}^{<\alpha}_\sigma(r) \hat{u}_\sigma(r+p) \hat{v}_\sigma(r-q)  \hat{a}_{r+p,\sigma}^\ast b_{-p,r^\prime,\sigma^\prime}^\ast b^\ast_{-q,s,\sigma^\prime}\hat{a}_{q-r,\sigma}^\ast .\label{eq: def I and J S2*T}
\end{align}
We insert the decomposition $1=\widehat{\chi}_<(q)+ \widehat{\chi}_>(q)$ and hence write it as a sum of two terms, $ \mathrm{I}+ \mathrm{J}$, with $\mathrm{I}$ containing an additional factor  $\widehat{\chi}_<(q)$ in \eqref{eq: def I and J S2*T}, while $\mathrm{J}$ contains  $\widehat{\chi}_>(q)$.
We start with estimating $\mathrm{I}$. We shall split $\widehat{\omega}^\varepsilon_{r-q,s}(q)$ into three terms using 
 \eqref{eq: comparison w eps with phi}, and write 
$\mathrm{I} = \mathrm{I}_{1} + \mathrm{I}_{2} + \mathrm{I}_{3}$ accordingly. The first term equals
\begin{align*}
  \mathrm{I}_{1}&= -\frac{1}{L^6}\sum_{\sigma\neq \sigma^\prime}\sum_{p,q,r,r^\prime,s} ||r|^2 - (k_F^\sigma)^2| \hat{\varphi}^>(p)\hat{\varphi}^>(q)\hat{{u}}^{<\alpha}_\sigma(r)\hat{u}_\sigma(r+p) \hat{v}_\sigma(r-q)  \hat{u}_{\sigma^\prime}(r^\prime - p) 
  \\
  &\qquad \qquad \qquad \qquad \times \hat{u}_{\sigma^\prime}(s-q)\hat{v}_{\sigma^\prime}(r^\prime)\hat{v}_{\sigma^\prime}(s) \hat{a}_{r+p,\sigma}^\ast \hat{a}_{-r^\prime, \sigma^\prime}^\ast\hat{a}_{r^\prime-p,\sigma^\prime}^\ast \hat{a}_{-s,\sigma^\prime}^\ast \hat{a}^\ast_{s-q,\sigma^\prime}\hat{a}_{q-r,\sigma}^\ast.
\end{align*}
Using the constraints $|p|\geq 4\rho^{1/3-\gamma}$, $|r^\prime|\leq C\rho^{1/3}$ and $|q|\geq 4\rho^{1/3-\gamma}$, $|s| \leq C\rho^{1/3}$, we replace $\hat{u}_{\sigma^\prime}(r^\prime - p)$ and $\hat{u}_{\sigma^\prime}(s-q)$ with $\hat{u}^>_{\sigma^\prime}(r^\prime - p)$ and $\hat{u}^>_{\sigma^\prime}(s-q)$, where $\hat{u}^>_\sigma$ is defined  in \eqref{eq: u< u> gamma}. Moreover, we split 
\[
  \hat{u}_\sigma(r+p) = \hat{\zeta}_{\sigma}^{<\alpha}(r+p) + \hat{u}_\sigma(r+p)(1-\hat{\zeta}_{\sigma}^{<\alpha}(r+p)) =:   \hat{\zeta}_{\sigma}^{<\alpha}(r+p) + \hat{\zeta}_\sigma^{>\alpha}(r+p),
\]
with $\hat{\zeta}_{\sigma}^{<\alpha}$ defined  in \eqref{eq: def zeta alpha}, and write $\mathrm{I}_{1} = \mathrm{I}_{1}^{<} + \mathrm{I}_{1}^{>}$ accordingly.  We have
\begin{align*}
  &\mathrm{I}_{1}^{>} = \sum_{\sigma\neq\sigma^\prime}\frac{1}{L^3}\sum_r (-|r|^2 + (k_F^\sigma)^2)\hat{u}^{<\alpha}_\sigma(r) 
  \\
  &\qquad \qquad \times\int dxdz\,e^{ir\cdot(x-z)} a^\ast_\sigma(\zeta^{>\alpha}_x)b_{\sigma^\prime}^\ast({\varphi}_x^>, v,u^>)b^\ast_{\sigma^\prime}({\varphi}_{z}^>, v, u^>)a_\sigma^\ast(v_z).
\end{align*}
Applying Lemma \ref{lem:6a}, and using that $| |r|^2 - (k_F^{\sigma})^2 |\hat{u}^{<\alpha}_\sigma(r) \leq C\rho^{2/3 - 2\alpha}$, we get 
\begin{align}\label{eq: I>alpha 1a}
|\langle \psi, \mathrm{I}_{1}^{>} \psi\rangle| &\leq C\rho^{\frac{2}{3} -2\alpha} \sum_{\sigma\neq \sigma^\prime}\sup_x \| b_{\sigma^\prime}({\varphi}^>_x, v, u^>)\|^2  
  \nn\\
   &\qquad \times \left(\int dx\| a_{\sigma}(\zeta^{>\alpha}_x)\psi\|^2\right)^{\frac{1}{2}}\left(\int dz\, \| a^\ast_\sigma(v_z)\psi\|^2\right)^{\frac{1}{2}}\nn\\
   &\le CL^{\frac{3}{2}}\rho^{\frac{3}{2} +\gamma - \alpha}\|\mathbb{H}_0^{\frac{1}{2}}\psi\|
\end{align}
for all $\psi\in\mathcal{F}_{\mathrm{f}}$. In the last estimate we used Lemmas \ref{lem:1a} and \ref{lem: bounds phi} and the bound \eqref{eq: b op phi vu}. 

For $\mathrm{I}_{1}^{<}$, we write 
\[
  - |r|^2 + (k_F^{\sigma})^2= \left(r\cdot r^\prime -  r\cdot (r+p) + (k_F^{\sigma})^2 \right) - r\cdot (r^\prime -p) ,
\]
and split $\mathrm{I}_{1}^{<}=\mathrm{I}_{1;a}^{<}+\mathrm{I}_{1;b}^{<}$ with $\mathrm{I}_{1;b}^{<}$ the term involving $ r\cdot (r^\prime -p)$. We first consider $\mathrm{I}_{1;a}^{<}$. Using the constraints $|p|\geq 4\rho^{1/3-\gamma}$, $|r^\prime|\leq C\rho^{1/3}$ and $|q|\geq 4\rho^{1/3-\gamma}$, $|s| \leq C\rho^{1/3}$, we can replace $\hat{u}_{\sigma^\prime}(r^\prime - p)$ and  $\hat{u}_{\sigma^\prime}(s-q)$ by $\hat{u}_{\sigma^\prime}^>(r^\prime - p)$ and $\hat{u}_{\sigma^\prime}^>(s-q)$, with $\hat{u}^>_{\sigma^\prime}$ defined in \eqref{eq: u< u> gamma}.
Thus, we have the identity
\begin{align*}
&\mathrm{I}_{1;a}^{<} = \\
 & -\frac{1}{L^3}\sum_{\sigma\neq \sigma^\prime}\sum_{\ell = 1}^3\sum_r r_\ell \hat{u}^{<\alpha}_\sigma(r) \int dxdz\,e^{ir\cdot (x-z)} a^\ast_\sigma(\zeta^{<\alpha}_x)b^\ast_{\sigma^\prime}(\varphi^>_x, \partial_\ell v, u^>)b^\ast_{\sigma^\prime}(\varphi^>_z, v, u^>)a^\ast_{\sigma}(v_z)
 \\
 &  + \frac{1}{L^3}\sum_{\sigma\neq\sigma^\prime}\sum_{\ell = 1}^3\sum_r r_\ell \hat{u}^{<\alpha}_\sigma(r)  \int dxdz\, e^{ir\cdot (x-z)}a^\ast_{\sigma}(\partial_\ell \zeta^{<\alpha}_x) b^\ast_{\sigma^\prime}(\varphi^>_x, v, u^>)b^\ast(\varphi^>_z, v, u^>)a^\ast_{\sigma}(v_z)
 \\
 &  +\frac{(k_F^\sigma)^2}{L^3}\sum_{\sigma\neq \sigma^\prime}\sum_r \hat{u}^{<\alpha}_\sigma(r) \int dxdz\, e^{ir\cdot (x-z)} a^\ast_\sigma(\zeta^{<\alpha}_x)b^\ast_{\sigma^\prime}(\varphi^>_x,  v, u^>)b^\ast_{\sigma^\prime}(\varphi^>_z, v, u^>)a^\ast_{\sigma}(v_z)
\end{align*}
where we used the notation introduced in \eqref{eq: def b general}. The three terms above can be estimated by applying Lemma \ref{lem:6a}, using also that $|r_\ell| \hat{u}^{<\alpha}_\sigma\leq C\rho^{1/3 - \alpha}$, with the result that  
\begin{align*}
|\langle\psi, \mathrm{I}_{1;a}^{<}  \psi\rangle| &\leq CL^{\frac{3}{2}}\sum_{\sigma \neq \sigma^\prime}\rho^{\frac{5}{6} - \alpha}\sup_z \|b_{\sigma^\prime}^\ast(\varphi^>_z, v, u^>) \|  
\\
&\quad \times \bigg(\sup_x \|b_{\sigma^\prime}(\varphi^>_x, \partial_\ell  v, u^>)\|\left(\int dx\,  \|a_\sigma(\zeta^{<\alpha}_x)\psi\|^2\right)^{\frac{1}{2}}
\\
&\quad \qquad\qquad\qquad  + \sup_x \|b_{\sigma^\prime}(\varphi^>_x,  v, u^>)\|\left(\int dx\,  \|a_\sigma(\partial_\ell \zeta^{<\alpha}_x)\psi\|^2\right)^{\frac{1}{2}}  \bigg)
\\
 &\quad + CL^{\frac{3}{2}}\rho^{\frac{7}{6}}\left(\sup_x \|b_{\sigma^\prime}^\ast(\varphi^>_x, v, u^>)\|\right)^2  
 \left(\int dx\,  \|a_\sigma(\zeta^{<\alpha}_x)\psi\|^2\right)^{\frac{1}{2}}.
\end{align*}
From \eqref{eq:Pauli-1} and Lemmas \ref{lem: bounds phi}, \ref{lem:2a} and \ref{lem:1a} we then conclude that
\begin{equation}\label{eq: I<alpha 1a1}
  |\langle\psi, \mathrm{I}^{<}_{1;a}\psi\rangle|\leq CL^{\frac{3}{2}}\rho^{\frac{11}{6} + \gamma - \alpha}\|\mathcal{N}^{\frac{1}{2}}\psi\| + CL^{\frac{3}{2}}\rho^{\frac{3}{2} + \gamma - \alpha}\|\mathbb{H}_0^{\frac{1}{2}}\psi\|. 
\end{equation}

We now consider $\mathrm{I}_{1;b}^{<}$. Due to the support properties of $\hat{v}_\sigma(r-q)$ and $\hat{{u}}^{<\alpha}_\sigma(r)$, we find that $|q|\leq 2 \rho^{1/3 -\alpha}$, which implies also that $\hat{u}_{\sigma^\prime}(s-q)$ is supported only for $|s-q| \leq 3\rho^{1/3-\alpha}$ since $|s| \leq k_F^{\sigma^\prime}$. We thus replace $\hat{u}_{\sigma^\prime}(s-q)$ with $\zeta^{<\alpha}_\sigma(s-q)$ defined in \eqref{eq: def zeta alpha}, and find
\begin{align*}
\mathrm{I}_{1;b}^{<} 
& =\frac{1}{L^3}\sum_{\sigma\neq \sigma^\prime}\sum_{\ell = 1}^3\sum_r \hat{u}^{<\alpha}_\sigma(r) \int dxdydzdz^\prime\,  e^{ir\cdot(x-z)}{\varphi}^>(x-y)  {\varphi}^>(z-z^\prime)  
  \\
  &\quad \qquad \times a^\ast_\sigma(\zeta^{<\alpha}_x)a^\ast_{\sigma^\prime}(v_y)a_{\sigma^\prime}^\ast(\partial_\ell u_y) \bigg(a_{\sigma^\prime}^\ast(v_{z^\prime})a_{\sigma^\prime}^\ast (\partial_\ell \zeta^{<\alpha}_{z^\prime})a_{\sigma}^\ast(v_z) + 
  \\
  &\qquad \qquad\qquad  + a_{\sigma^\prime}^\ast(\partial_\ell v_{z^\prime})a_{\sigma^\prime}^\ast ( \zeta^{<\alpha}_{z^\prime})a_{\sigma}^\ast(v_z) + a_{\sigma^\prime}^\ast( v_{z^\prime})a_{\sigma^\prime}^\ast (\zeta^{<\alpha}_{z^\prime})a_{\sigma}^\ast(\partial_\ell v_z)\bigg). 
\end{align*}

Proceeding similarly as in the proof of Lemma \ref{lem:6a}, by using the Cauchy-Schwarz inequality in $r$ and $|\hat{{u}}^{<\alpha}_\sigma(k)|\leq 1$, we get 
\begin{align*}
  &|\langle \psi, \mathrm{I}_{1;b}^{<}\psi\rangle| 
  \\
  &\leq CL^{\frac{3}{2}}\sum_{\sigma\neq \sigma^\prime}\sum_{\ell = 1}^3\left(\frac{1}{L^3}\sum_r \left\|   \int dxdy\, e^{-ir\cdot x}{\varphi}^>(x-y)a_{\sigma^\prime}(\partial_\ell u_y)a_{\sigma^\prime}(v_y)a_{\sigma}(\zeta^{<\alpha}_x)\psi  \right\|^2\right)^{\frac{1}{2}} 
  \\
  &\qquad \times \bigg( \sup_z \|b^\ast_{\sigma^\prime}(\varphi^>_z, v, \partial_\ell \zeta^{<\alpha})a^\ast_\sigma(v_z)\|+  \sup_z \|b^\ast_{\sigma^\prime}(\varphi^>_z, \partial_\ell v, \zeta^{<\alpha})a^\ast_\sigma(v_z)\|
  \\
  &\qquad \qquad\qquad\qquad \qquad\qquad+ \sup_z \|b^\ast_{\sigma^\prime}(\varphi^>_z, v, \zeta^{<\alpha})a^\ast_\sigma(\partial_\ell v_z)\|\bigg). 
\end{align*}
The first factor on the right-hand side above can be bounded as  
\begin{align*}
&\frac{1}{L^3}\sum_r \left\|   \int dxdy\, e^{-ir\cdot x}{\varphi}^>(x-y)a_{\sigma^\prime}(\partial_\ell u_y)a_{\sigma^\prime}(v_y)a_{\sigma}(\zeta^{<\alpha}_x)\psi  \right\|^2
\\
&\leq \|\zeta^{<\alpha}_\sigma\|_2^2  \|v_{\sigma'}\|_2^2 \int dxdydy^\prime |{\varphi}^>(x-y)||{\varphi}^>(x-y^\prime) \|a_{\sigma^\prime}(\partial_\ell u_y)\psi\|\|a_{\sigma^\prime}(\partial_\ell u_{y^\prime})\psi\|
\\
&\leq  
C\rho^{\frac{2}{3}-3\alpha + 4\gamma}\left(\langle \psi, \mathbb{H}_0 \psi\rangle + \rho^{\frac{2}{3}}\langle \psi, \mathcal{N}\psi\rangle\right),
\end{align*}
where we used also \eqref{eq:Pauli-1}, Lemma \ref{lem: bounds phi} and Lemma \ref{lem:1a}. The other terms can be estimated using \eqref{eq:Pauli} together with Lemmas \ref{lem:2a} and \ref{lem: bounds phi}, resulting in
\begin{multline*} 
\sup_z\| b^\ast_{\sigma^\prime}(\varphi^>, v, \partial_\ell \zeta^{<\alpha})a^\ast_\sigma(v_z)\|^2 +\sup_z \|b^\ast_{\sigma^\prime}(\varphi^>, \partial_\ell v, \zeta^{<\alpha})a^\ast_\sigma(v_z)\| ^2 
\\
+ \sup_z\|b^\ast_{\sigma^\prime}(\varphi^>, v, \zeta^{<\alpha})a^\ast_\sigma(\partial_\ell v_z)\|^2 
\leq C \rho^{\frac{7}{3} - \alpha}. 
\end{multline*}
Thus, combining the above estimates, we find 
\begin{equation}\label{56.3}
  |\langle \psi, \mathrm{I}^{<}_{1;b}\psi\rangle| \leq CL^{\frac{3}{2}}\rho^{\frac{3}{2} + 2\gamma - 2\alpha}\|\mathbb{H}_0^{\frac{1}{2}}\psi\| + CL^{\frac{3}{2}}\rho^{\frac{11}{6} + 2\gamma - 2\alpha}\|\mathcal{N}^{\frac{1}{2}}\psi\|. 
\end{equation}

Next we consider $I_2$, given by 
\begin{align*}
  \mathrm{I}_{2}  &= \frac{2}{L^6}\sum_{\sigma\neq \sigma^\prime}\sum_{p,q,r,r^\prime,s} (|r|^2 -(k_F^\sigma)^2) \hat{\varphi}^>(p)\frac{q\cdot ( (r-q) -s)}{|r|^2 - |r-q|^2 + |s-q|^2 - |s|^2 + 2\varepsilon}\hat{\varphi}^>(q) 
  \\
  &\qquad \times \hat{{u}}^{<\alpha}_\sigma(r)\hat{u}_\sigma(r+p) \hat{v}_\sigma(r-q)
  \hat{u}_{\sigma^\prime}(r^\prime - p)\hat{u}_{\sigma^\prime}(s-q)\hat{v}_{\sigma^\prime}(r^\prime)\hat{v}_{\sigma^\prime}(s) 
  \\
  &\qquad \times  \hat{a}_{r+p,\sigma}^\ast \hat{a}_{-r^\prime, \sigma^\prime}^\ast \hat{a}_{r^\prime -p,\sigma^\prime}^\ast \hat{a}_{-s, \sigma^\prime}^\ast \hat{a}^\ast_{s-q,\sigma^\prime}\hat{a}_{q-r,\sigma}^\ast .
\end{align*}
Using $|r|\leq \rho^{1/3-\alpha}$ and $|r-q|\leq C\rho^{1/3}$, we obtain $|q|\leq 2\rho^{1/3 -\alpha}$. We therefore replace $\hat{u}_{\sigma^\prime}(s-q)$ by $\zeta^{<\alpha}_{\sigma^\prime}(s-q)$ defined in \eqref{eq: def zeta alpha}. Moreover, since $|p|\geq 4\rho^{1/3-\gamma}$ and $|r^\prime| \leq C\rho^{1/3}$, we can  replace $\hat{u}_{\sigma^\prime}(r^\prime - p)$ by $\hat{u}^>_{\sigma^\prime}(r^\prime - p)$ defined in \eqref{eq: u< u> gamma}. We split $\mathrm{I}_{2}=\mathrm{I}_{2;a}+\mathrm{I}_{2;b}$ using 
\[
  \hat{u}_\sigma(r+p) =  {\hat{u}^{>\eta}_\sigma(r+p) + \hat{u}^{<\eta}_\sigma(r+p)},
\]
with $\hat{u}^{<\eta}_\sigma$ supported for $k_F^\sigma < |k| < \rho^{1/3 - \eta}$ and $\hat{u}^{>\eta}_\sigma = \hat{u}_\sigma(1 -\hat{u}^{<\eta}_\sigma) $, for some $\eta>0$ to be optimized. 
We first estimate $\mathrm{I}_{2;a}$, given in configuration space by
\begin{align*}
  &\mathrm{I}_{2;a}= -  \sum_{\ell = 1}^3 \int_0^{\infty}dt\, e^{-2t\varepsilon} \frac{1}{L^3}\sum_{\sigma\neq \sigma^\prime} \sum_r (|r|^2 - (k_F^\sigma)^2) \hat{u}^{<\alpha}_\sigma(r)e^{-t|r|^2}  
  \\
  &\qquad  \times\left(\int dx \, e^{ir\cdot x} a^\ast_\sigma(u_x^{>\eta})b^\ast_{\sigma^\prime}(\varphi_x, v, u^>)\right) 
\\
  &\qquad \times \left(\int dzdz^\prime e^{-ir\cdot z} \partial_\ell{\varphi}^>(z-z^\prime) a_{\sigma^\prime}^\ast (\zeta^{<\alpha}_{t,z^\prime}) (a_{\sigma^\prime}^\ast(v_{t,z^\prime})a_{\sigma}^\ast(\partial_\ell v_{t,z}) -  a_{\sigma^\prime}^\ast(\partial_\ell v_{t,z^\prime})a_{\sigma}^\ast(v_{t,z}) )\right)
\end{align*}
where 
$\hat{\zeta}^{<\alpha}_{t,\sigma^\prime}(k) = \hat{\zeta}^{<\alpha}_{\sigma^\prime}(k)e^{-t|k|^2}$. 
We shall apply Lemma \ref{lem:6a}, use
 $||r|^2 - (k_F^\sigma)^2|{u}^{<\alpha}_\sigma(r)e^{-t|r|^2} \leq C\rho^{2/3 - 2\alpha}e^{-t(k_F^\sigma)^2}$, the bound
\begin{equation}\label{eq: int z' a*a*}
\left\|\int dz^\prime \partial_\ell {\varphi}^>(z-z^\prime) a_{\sigma^\prime}^\ast (\zeta^{<\alpha}_{t,z^\prime}) a_{\sigma^\prime}^\ast(\partial_\ell^n v_{t,z^\prime})\right\|\leq C\rho^{\frac{n}{3}}\|\partial_\ell \varphi^>\|_1\|{\zeta}^{<\alpha}_{t,\sigma^\prime}\|_2 \|{v}_{t,\sigma^\prime}\|_2,
\end{equation}
together with Lemma \ref{lem:1a}  to get
\begin{align*}  
  |\langle \psi, \mathrm{I}_{2;a}\psi\rangle| &\leq CL^{\frac{3}{2}}\rho^{\frac{2}{3} - 2\alpha} \sum_{\sigma\neq \sigma^\prime}\sum_{\ell = 1}^3\|\partial_\ell v_\sigma\|_2  \left(\int_0^\infty dt\, e^{-2t\varepsilon} e^{-t(k_F^\sigma)^2} \|\zeta_{t,\sigma^\prime}^{<\alpha}\|_2 \|v_{t,\sigma^\prime}\|_2\right)
  \\
  &\quad \times \left(\int dx \, \sup_x \|b_{\sigma^\prime}(\varphi_x, v, u^>)\|^2 \|a_\sigma(u^{>\eta}_x)\psi\|^2\right)^{\frac{1}{2}}\|\partial_\ell \varphi\|_1
  \\
  &\leq C\rho^{\frac{2}{3} + \frac{1}{2} - 2\alpha +\frac{3}{2}\gamma + \eta}\left(\int_0^{\infty} dt\, e^{-2t\varepsilon} e^{-t(k_F^\sigma)^2} \|\hat{\zeta}^{<\alpha}_{t,\sigma^\prime} \|_2 \|v_{t,\sigma^\prime}\|_2\right)\|\mathbb{H}_0^{\frac{1}{2}}\psi\|.
\end{align*}
Using \eqref{eq: est int t vu< final} with $\gamma$ replaced by $\alpha$, i.e., 
\begin{equation}\label{eq: int t <alpha}
  \int_0^{\infty} dt\, e^{-2t\varepsilon} e^{-t(k_F^\sigma)^2} \|\hat{\zeta}^{<\alpha}_{t,\sigma^\prime} \|_2 \|v_{t,\sigma^\prime}\|_2 \leq C\rho^{\frac{1}{3} -\frac{\alpha}{2} -\kappa}, 
\end{equation}
which can be proved in the same way as \eqref{eq: est int t vu< final} using  $\hat{\zeta}^{<\alpha}_{t,\sigma^\prime} (k) = \hat{\zeta}_{\sigma^\prime}^{<\alpha}(k)e^{-t|k|^2}$ and the definition of $\hat{\zeta}_{\sigma^\prime}^{<\alpha}(k)$ in \eqref{eq: def zeta alpha}, we find 
\begin{align}\label{eq: est I1a2 <}
  |\langle \psi,\mathrm{I}_{2;a}\psi\rangle|&\leq 
  CL^{\frac{3}{2}}\rho^{\frac{3}{2} + \frac{3}{2}\gamma +\eta - \frac{5}{2}\alpha - \kappa}\|\mathbb{H}_0^{\frac{1}{2}}\psi\|.
\end{align}
 Next we consider $\mathrm{I}_{2;b}$. By writing 
\[
|r|^2 -(k_F^\sigma)^2 =  - r\cdot r^\prime + r\cdot (r+p) + r\cdot (r^\prime - p)  - (k_F^\sigma)^2,
\]
we split it into four terms, 
$\mathrm{I}_{2;b} = \mathrm{I}_{2;b;1} + \mathrm{I}_{2;b;2}+ \mathrm{I}_{2;b;3} + \mathrm{I}_{2;b;4} $.
 The first one is 
\begin{align*}
&\mathrm{I}_{2;b;1} =
\frac{2}{L^3}\sum_{\sigma\neq \sigma^\prime}\sum_{\ell,m = 1}^3 \sum_r \int_0^{\infty} dt\, e^{-2t\varepsilon} (-ir_\ell) \hat{u}^{<\alpha}_{\sigma}(r) e^{-t|r|^2} 
\\
&\quad \times  \int dx\,e^{ir\cdot x} a^\ast_\sigma(u_x^{<\eta})b^\ast_{\sigma^\prime}(\varphi^>_x,\partial_\ell v, u^>) 
\\
&\quad \times \int dzdz^\prime \, e^{-ir \cdot z} \partial_m {\varphi}^>(z-z^\prime)a_{\sigma^\prime}^\ast (\zeta^{<\alpha}_{t,z^\prime}) \big(a_{\sigma^\prime}^\ast(\partial_m v_{t,z^\prime})a_{\sigma}^\ast( v_{t,z}) - a_{\sigma^\prime}^\ast(v_{t,z^\prime})a_{\sigma}^\ast(\partial_m v_{t,z})\big) .
\end{align*} 
Using again Lemma \ref{lem:6a} with \eqref{eq: int z' a*a*}, \eqref{eq:Pauli} and $|r| \hat{u}^{<\alpha}_{\sigma}(r) e^{-t|r|^2}\leq C\rho^{1/3 - \alpha}e^{-t(k_F^{\sigma})^2}$,  we obtain
\begin{align*}
  |\langle \psi, \mathrm{I}_{2;b;1} \psi\rangle| &\leq CL^{\frac{3}{2}}\rho^{\frac{2}{3} -\alpha}\sum_{\sigma\neq\sigma^\prime}\sum_{\ell, m = 1}^3\|\partial_m {\varphi}^>\|_1 \sup_x\| b_{\sigma^\prime}({\varphi}_x^>, \partial_\ell v, u^>)\| 
  \\
  &\quad \times \int_0^{\infty}dt\, e^{-2t\varepsilon}e^{-t(k_F^\sigma)^2}\|{\zeta}^{<\alpha}_{t,\sigma^\prime}\|_2 \|{v}_{t,\sigma^\prime}\|_2 \|{v}_{t,\sigma}\|_2\,  \sqrt{\int dx\, \| a_{\sigma}(u_x^{<\eta})\psi\|^2}\\
  &\le CL^{\frac{3}{2}}\rho^{\frac{11}{6} +\frac{3}{2}\gamma - \frac{3}{2}\alpha -\kappa}\|\mathcal{N}^{\frac{1}{2}}\psi\|.
\end{align*}
Here we also used Lemmas \ref{lem:1a},  \ref{lem: bounds phi}, and  \ref{lem:2a},  as well as \eqref{eq: int t <alpha}. 

The term $\mathrm{I}_{2;b;2}$ can be treated in a similar way, with the result 
\begin{align*}
  &|\langle \psi,  \mathrm{I}_{2;b;2} \psi\rangle|
 \le CL^{\frac{3}{2}}\rho^{\frac{3}{2} + \frac{3}{2}\gamma - \frac{3}{2}\alpha -\kappa}\left(\|\mathbb{H}_0^{\frac{1}{2}}\psi\| + \rho^{\frac{1}{3}}\|\mathcal{N}^{\frac{1}{2}}\psi\|\right).
\end{align*}

The term  
\begin{align*}
 \mathrm{I}_{2;b;3} &=  2\sum_{\sigma\neq \sigma^\prime}\sum_{\ell,m = 1}^3\sum_r (-ir_\ell) \hat{u}^{<\alpha}_{\sigma}(r)\int_0^{\infty} dt\, e^{-2t\varepsilon}e^{-t|r|^2} 
\\
& \qquad \times \int dxdydzdz^\prime\, {\varphi}^>(x-y)\partial_m {\varphi}^>(z-z^\prime)   
\\
&\qquad \times a^\ast_\sigma(u_x^{<\eta})a^\ast_{\sigma^\prime}(v_y)a_{\sigma^\prime}^\ast(\partial_\ell u_y^>) a_{\sigma^\prime}^\ast (\zeta^{<\alpha}_{t,z^\prime}) \big(a_{\sigma^\prime}^\ast(\partial_m v_{t,z^\prime})a_{\sigma}^\ast( v_{t,z})  - a_{\sigma^\prime}^\ast(v_{t,z^\prime})a_{\sigma}^\ast(\partial_m v_{t,z}) \big)
\end{align*} 
can also be handled  similarly as in the proof of Lemma \ref{lem:6a}. Using $|r_\ell| \hat{u}^{<\alpha}_\sigma\leq C\rho^{1/3 - \alpha}$ and the Cauchy-Schwarz inequality, we get
\begin{align*}
|\langle \psi, \mathrm{I}_{2;b;3}\psi\rangle| &\leq CL^{\frac{3}{2}}\rho^{\frac{2}{3} -\alpha}\sum_{\sigma\neq\sigma^\prime}\sum_{\ell = 1}^3\|\partial_\ell {\varphi}^>\|_1  \int_0^{\infty} dt\, e^{-2t\varepsilon}e^{-t(k_F^\sigma)^2}\|{\zeta}^{<\alpha}_{t,\sigma^\prime}\|_2 \|{v}_{t,\sigma^\prime}\|_2 \|{v}_{t,\sigma}\|_2
\\
&\quad \times \left(\frac{1}{L^3}\sum_r \left\|\int dxdy\, e^{-ir\cdot x}\varphi^>(x-y)a_\sigma(u^{<\eta}_x)a_{\sigma^\prime}(v_y)a_{\sigma^\prime}(\partial_\ell u^>_y)\psi\right\|^2\right)^{\frac{1}{2}},
\end{align*}
where we used \eqref{eq: int z' a*a*} and \eqref{eq:Pauli}. By the Cauchy-Schwarz inequality, 
\begin{align*}
&\frac{1}{L^3}\sum_r \left\|\int dxdy\, e^{-ir\cdot x}\varphi^>(x-y)a_\sigma(u^{<\eta}_x)a_{\sigma^\prime}(v_y)a_{\sigma^\prime}(\partial_\ell u^>_y)\psi\right\|^2
\\
&\qquad\leq  \|\varphi^>\|_1^2 \|{u}^{<\eta}_\sigma\|_2^2 \|{v}_{\sigma^\prime}\|_2^2 \int dy\, \|a_{\sigma^\prime}(\partial_\ell u_y)\psi\|^2 .
\end{align*}
Using also Lemmas \ref{lem: bounds phi} and \ref{lem:1a} as well as \eqref{eq:Pauli-1}, we find 
\begin{align*}
|\langle \psi, \mathrm{I}_{2;b;3} \psi\rangle| &\leq CL^{\frac{3}{2}}\rho^{\frac{7}{6} + 3\gamma - \frac{3}{2}\eta - \alpha} \left(\int_0^{\infty}dt\, e^{-2t\varepsilon} \|\hat{v}_{t,\sigma^\prime}\|_2 \|\hat{\zeta}^{<\alpha}_{t,\sigma^\prime}\|_2\right)\left(\|\mathbb{H}_0^{\frac{1}{2}}\psi\| + \rho^{\frac{1}{3}}\|\mathcal{N}^{\frac{1}{2}}\psi\|\right)
 \\
 &\leq CL^{\frac{3}{2}} \rho^{\frac{3}{2} + 3\gamma - \frac{3}{2}\eta - \frac{3}{2}\alpha - \kappa} \left( \|\mathbb{H}_0^{\frac{1}{2}}\psi\| + \rho^{\frac{1}{3}}\|\mathcal{N}^{\frac{1}{2}}\psi\|\right).
\end{align*}
The estimate for $\mathrm{I}_{2;b;4}$ can be done in the same way, yielding 
\begin{align*}
  &|\langle \psi, \mathrm{I}_{2;b;4} \psi\rangle| 
 \leq CL^{\frac{3}{2}}\rho^{\frac{11}{6} +\frac{3}{2}\gamma -\frac{\alpha}{2} -\kappa}\|\mathcal{N}^{\frac{1}{2}}\psi\|.
\end{align*}
Combining all the above estimates,  
we get 
\begin{align*}
  |\langle\psi, \mathrm{I}_{2} \psi\rangle| &\leq  CL^{\frac{3}{2}}\rho^{\frac{3}{2} + \frac{3}{2}\gamma +\eta - \frac{5}{2}\alpha - \kappa}\|\mathbb{H}_0^{\frac{1}{2}}\psi\| + CL^{\frac{3}{2}} \rho^{\frac{3}{2} + \frac{3}{2}\gamma - \frac{3}{2}\alpha - \kappa}(\|\mathbb{H}_0^{\frac{1}{2}}\psi\| + \rho^{\frac{1}{3}} \|\mathcal{N}^{\frac{1}{2}}\psi\|)
  \\
  &\quad +CL^{\frac{3}{2}} \rho^{\frac{3}{2} + 3\gamma -\frac{3}{2}\eta - \frac{3}{2}\alpha - \kappa} (\|\mathbb{H}_0^{\frac{1}{2}}\psi\| + \rho^{\frac{1}{3}}\|\mathcal{N}^{\frac{1}{2}}\psi\|).
\end{align*}
An optimization over $\eta$ yields $\eta = (2/5)\alpha + (3/5)\gamma$. Thus, since $\alpha > \gamma$, we obtain
\begin{equation}\label{I2f}
  |\langle\psi, \mathrm{I}_{2} \psi\rangle| \leq  CL^{\frac{3}{2}}\rho^{\frac{3}{2} + \frac{21}{10}\gamma- \frac{21}{10}\alpha - \kappa} (\|\mathbb{H}_0^{\frac{1}{2}}\psi\| + \rho^{\frac{1}{3}}\|\mathcal{N}^{\frac{1}{2}}\psi\|).
\end{equation}

To conclude the analysis of $\mathrm{I}$, it remains to estimate 
\begin{align}\label{eq: I1a3>}
\mathrm{I}_{3} &= -2\varepsilon \sum_{\sigma\neq \sigma^\prime} \frac{1}{L^3}\sum_r  ||r|^2 - (k_F^\sigma)^2|\hat{u}^{<\alpha}_\sigma(r)\int_0^{\infty}dt\, e^{-2t\varepsilon}e^{-t|r|^2}  \\
&\qquad \times \int dxdydzdz^\prime \, e^{ir\cdot(x-z)}b^\ast_{\sigma^\prime}(\varphi_x^>, v, u^>) a^\ast_\sigma(u_x)b^\ast_{\sigma^\prime}(\varphi^>_z, v_t, \zeta^{<\alpha}_t)a_{\sigma}^\ast( v_{t,z}). \nn
\end{align}
Using Lemma \ref{lem:6a} and $|r|^2 \hat{{u}}^{<\alpha}_\sigma(r)e^{-t|r|^2} \leq C\rho^{2/3 - 2\alpha}e^{-t(k_F^\sigma)^2}$, we get the bound 
\begin{align}\label{eq: I1a3> 3}
  &|\langle \psi, \mathrm{I}_{3}\psi\rangle| \leq C\rho^{\frac{4}{3} -2\alpha +\delta}\sum_{\sigma\neq\sigma^\prime}\int_0^{\infty}dt\, e^{-2t\varepsilon}e^{-t(k_F^\sigma)^2}\sup_x \|b_{\sigma^\prime}(\varphi^>_x, v, u^>)  \| \nn
\\
&\quad \times \left(\int dx\, \|a_\sigma(u_x) \psi\|^2\right)^{\frac{1}{2}}  \sup_z \|b^\ast_{\sigma^\prime}(\varphi^>_z, v_t, \zeta^{<\alpha}_t)\|\left(\int dz\, \|a_{\sigma}^\ast(v_{t,z}) \psi\|^2\right)^{\frac{1}{2}}.
\end{align}
From Lemmas \ref{lem:2a}, \ref{lem: bounds phi} and \ref{lem:1a}, \eqref{eq:Pauli-1} and \eqref{eq: int t <alpha}, we get 
\begin{align*}
|\langle \psi, \mathrm{I}_{3}\psi\rangle|
 &\leq C\rho^{\frac{3}{2} + \frac{5}{2}\gamma -2\alpha + \delta} \left(\int_0^{\infty}dt\, e^{-2t\varepsilon} \|{v}_{t,\sigma^\prime}\|_2 \|{\zeta}^{<\alpha}_{t,\sigma^\prime}\|_2\right)L^{\frac{3}{2}} \|\mathcal{N}^{\frac{1}{2}}\psi\|
 \\
 &\leq C \rho^{\frac{11}{6} + \frac{5}{2}\gamma - \frac{5}{2}\alpha +\delta-\kappa} L^{\frac{3}{2}} \|\mathcal{N}^{\frac{1}{2}}\psi\|.
\end{align*}
Combining the above estimates for $\mathrm{I}_{1}$, $\mathrm{I}_{2}$ and $\mathrm{I}_{3}$, we find that  
\begin{align}\label{eq: est T2*R I>}
  |\langle \psi, \mathrm{I}\psi\rangle| &\leq C  L^{\frac{3}{2}}\left(\rho^{\frac{3}{2} +2\gamma - 2\alpha} + \rho^{\frac{3}{2} +\frac{21}{10}\gamma - \frac{21}{10}\alpha - \kappa}\right) \left(\|\mathbb{H}_0^{\frac{1}{2}}\psi\| + \rho^{\frac{1}{3}}\|\mathcal{N}^{\frac{1}{2}}\psi\|\right) \nn
  \\
  &\quad + C L^{\frac{3}{2}}\rho^{\frac{11}{6} + \frac{5}{2}\gamma - \frac{5}{2}\alpha +\delta-\kappa} \|\mathcal{N}^{\frac{1}{2}}\psi\|.
\end{align}

It remains to consider  the term $\mathrm{J}$ defined by \eqref{eq: def I and J S2*T} with an additional factor $\widehat{\chi}_<(q)$, i.e., 
\begin{align*}
  \mathrm{J} &= -\frac{1}{L^6}\sum_{\sigma\neq \sigma^\prime}\sum_{p,q,r,r^\prime,s} (|r|^2 - (k_F^\sigma)^2)\hat{\varphi}^>(p)\widehat{\omega}^{\varepsilon}_{r-q, s}(q)\widehat{\chi}_<(q)
  \\
  &\qquad \qquad\times \hat{{u}}^{<\alpha}_\sigma(r)\hat{u}_\sigma(r+p) \hat{v}_\sigma(r-q)\hat{u}_{\sigma^\prime}(r^\prime - p)\hat{v}_{\sigma^\prime}(r^\prime)\hat{u}_{\sigma^\prime}(s-q)\hat{v}_{\sigma^\prime}(s)
  \\
  &\qquad \qquad \times \hat{a}_{r+p,\sigma}^\ast \hat{a}_{-r^\prime,\sigma^\prime}^\ast\hat{a}^\ast_{r^\prime - p,\sigma^\prime} \hat{a}^\ast_{-s,\sigma^\prime}\hat{a}_{s-q,\sigma^\prime}^\ast\hat{a}_{q-r,\sigma}^\ast. 
\end{align*}
In the relevant domain, we can replace $\hat{u}_{\sigma^\prime}(s-q)$, $\hat{{u}}^{<\alpha}_\sigma(r)$ by $\hat{u}^<_{\sigma^\prime}(s-q)$, $\hat{u}^<_\sigma(r)$, respectively, and replace  $\hat{u}_{\sigma^\prime}(r^\prime - p)$ by  $\hat{u}^>_{\sigma^\prime}(r^\prime - p)$, with $\hat{u}^<_\sigma$ and $\hat{u}^>_{\sigma'}$ defined in \eqref{eq: u< u> gamma}. 
Further  splitting 
\[
  \hat{u}_\sigma(r+p) =  \hat{u}^{>}_\sigma(r+p) + \hat{u}^{<}_\sigma(r+p), 
\] 
we write $\mathrm{J} =  \mathrm{J}^{>} + \mathrm{J}^{<}$ accordingly.
We first consider $\mathrm{J}^{>}$, which  
can be written as 
\begin{align*}
  \mathrm{J}^{>} &= \frac{2}{L^3}\sum_{\sigma\neq \sigma^\prime}\sum_r ||r|^2 - (k_F^\sigma)^2|\hat{u}^{<}_\sigma(r)\int_0^{\infty} dt\, e^{-2t\varepsilon}e^{-t|r|^2}  
  \\
  &\qquad \qquad \qquad \quad\times \int dxdzdz^\prime  e^{ir\cdot (x-z)} a^\ast_\sigma(u^{>}_x)b^\ast_{\sigma^\prime}({\varphi}_x^>, v, u^>)b^\ast_{\sigma^\prime}(\Delta\varphi_z^<, v_{t}, u^<_t)a^\ast_\sigma(v_{t,z}).
\end{align*}
Using Lemmas \ref{lem:6a}, \ref{lem:2a},  \ref{lem: bounds phi} and \ref{lem:1a}, we find that 
\begin{align*}
  |\langle \psi, \mathrm{J}^{>}\psi\rangle| 
  &\leq CL^{\frac{3}{2}}\rho^{\frac{2}{3} - 2\gamma}\sum_{\sigma\neq\sigma^\prime}\int_0^{\infty}dt\, e^{-2t\varepsilon} e^{-t(k_F^\sigma)^2}\sup_x \| b_{\sigma^\prime}({\varphi}_x^>, v, u^>)\| 
  \\
  &\qquad\qquad\qquad \times \left( \int dx\,\|a_\sigma(u^{>}_x)\psi\|^2\right)^{\frac{1}{2}} \sup_z \|b_{\sigma^\prime}(\Delta\varphi^<_z, v_{t}, u^<_t)\| \| v_{t,\sigma}\|_2 \\
  &
  \leq CL^{\frac{3}{2}}\rho^{\frac{3}{2} -\gamma -\kappa}\|\mathbb{H}_0^{\frac{1}{2}}\psi\|,
\end{align*}
where in the last estimate we also used  \eqref{eq: est int t vu< final}. Next, we consider $\mathrm{J}^{<}$. In this case, it is convenient to write $|r|^2 - (k_F^\sigma)^2 = r\cdot (r+p) - r\cdot p- (k_F^\sigma)^2$ and accordingly split
 $\mathrm{J}^<$ into three terms $\mathrm{J}^<_1$, $\mathrm{J}^<_2$ and $\mathrm{J}^<_3$   
containing $r\cdot (r+p)$, $r\cdot p$  and $(k_F^\sigma)^2$, respectively. 
The first term $\mathrm{J}^{<}_{1}$ is given by 
\begin{align*}
 \mathrm{J}^{<}_{1}&= \frac{2i}{L^3} \sum_{\sigma\neq \sigma^\prime}\sum_{\ell = 1}^3\sum_r r_\ell \hat{u}^{<}(r) \int_0^{\infty}dt\, e^{-2t\varepsilon} e^{- t|r|^2} 
  \\
  &\qquad \times  \int dxdz\, e^{-ir\cdot (x-z)} b_{\sigma^\prime}^\ast({\varphi}_x^>, v, u^>)a_\sigma^\ast(\partial_\ell u^{<}_x)b_{\sigma^\prime}^\ast(\Delta\varphi_z^<, v_t, u^<_t)a_\sigma^\ast(v_{t,z}).
\end{align*}
Proceeding similarly as for $\mathrm{J}^{>}$, using Lemmas \ref{lem:6a}, \ref{lem:2a}, \ref{lem:1a}, \ref{lem: bounds phi} and \eqref{eq:Pauli-1}, we obtain
\begin{align*}
  |\langle \psi, \mathrm{J}^{<}_{1}\psi\rangle| &\leq CL^{\frac{3}{2}}\rho^{\frac{7}{6} - \frac{\gamma}{2}}\sum_{\sigma\neq\sigma^\prime}\sum_{\ell=1}^3 \int_0^{\infty} dt\, e^{-2t\varepsilon} \|{u}^<_{t,\sigma^\prime}\|_2 \|{v}_{t,\sigma^\prime}\|_2\, \sqrt{\int dx \, \|a_\sigma(\partial_\ell u^{<}_x)\psi\|^2}
  \\
  &\leq CL^{\frac{3}{2}} \rho^{\frac{3}{2} - \gamma -\kappa}\left(\|\mathbb{H}^{\frac{1}{2}}_0\psi\| + \rho^{\frac{1}{3}}\|\mathcal{N}^{\frac{1}{2}}\psi\|\right).
\end{align*}
Next, we estimate  
\begin{multline*}
  \mathrm{J}^{<}_{2} = \frac{-2i}{L^3}\sum_{\sigma\neq \sigma^\prime}\sum_{\ell = 1}^3\sum_r r_\ell \hat{u}^{<}_\sigma(r)\int_0^{\infty} dt\, e^{-2t\varepsilon} e^{-t|r|^2}\int dxdz\, e^{-ir\cdot(x-z)} 
  \\
  \times a^\ast_{\sigma}(u^{<}_x)b^\ast_{\sigma^\prime}(\partial_\ell\varphi^>_x, v, u^>)b^\ast_{\sigma^\prime}(\Delta\varphi_z^<, v_{t}, u^<_t)a_\sigma^\ast(v_{t,z}).
\end{multline*}
Similarly as in the proof of Lemma \ref{lem:6a}, we can use the Cauchy--Schwarz inequality in $r$ and get 
\begin{align*}
|\langle \psi, J_{2}^{<}\psi\rangle| &\leq CL^{\frac{3}{2}}\rho^{\frac{1}{3} - \gamma} \sum_{\sigma\neq \sigma^\prime}\int_0^{\infty}dt\, e^{-2t\varepsilon}e^{-t(k_F^\sigma)^2} \sup_z \|b^\ast_{\sigma^\prime}(\Delta\varphi_z^<, v_t, u^<_t)a^\ast_{\sigma}(v_{t,z})\| 
\\
&\quad \times \left(\frac{1}{L^3}\sum_r \left\|\int dxdy\, e^{-ir\cdot x} \partial_\ell {\varphi}^>(x-y) a_\sigma(u^{<}_x)a_{\sigma^\prime}(u^>_y)a_{\sigma^\prime}(v_y)\psi\right\|^2\right)^{\frac{1}{2}}.
\end{align*}
Using now that 
\begin{multline*}
 \frac{1}{L^3}\sum_r \left\|\int dxdy\, e^{-ir\cdot x} \partial_\ell {\varphi}^>(x-y) a_\sigma(u^{<}_x)a_{\sigma^\prime}(u^>_y)a_{\sigma^\prime}(v_y)\psi\right\|^2
  \\
   \leq C \| u^<_\sigma\|_2^2 \|v_{\sigma'}\|_2^2\|\nabla{\varphi}^>\|_1^2 \int dy\ \|a_{\sigma^\prime}(u^>_y)\psi\|^2 
\end{multline*}
together with \eqref{eq:Pauli-1}, Lemma \ref{lem: bounds phi}  and \eqref{eq:1a-3}, we get 
\[
  |\langle\psi, \mathrm{J}^{<}_{2}\psi\rangle| \leq CL^{\frac{3}{2}}\rho^{\frac{3}{2}-\gamma - \kappa}\|\mathbb{H}_0^{\frac{1}{2}}\psi\|.
\]
Finally, the term 
\begin{align*}
 \mathrm{J}^{<}_{3}&=  \sum_{\sigma\neq \sigma^\prime}\frac{(k_F^\sigma)^2}{L^3}\sum_r  \hat{u}^{<}(r) \int_0^{\infty}dt\, e^{-2t\varepsilon} e^{-t|r|^2} 
  \\
 & \qquad \times \int dxdz\, e^{-ir\cdot (x-z)} b_{\sigma^\prime}^\ast({\varphi}_x^>, v, u^>)a_\sigma^\ast(u^{<}_x)b_{\sigma^\prime}^\ast(\Delta\varphi_z^<, v_t, u^<_t)a_\sigma^\ast(v_{t,z})
\end{align*}
can be treated similarly to $\mathrm{J}^{<}_1$ and $\mathrm{J}^{<}_2$. Using Lemmas \ref{lem:6a}, \ref{lem:2a}, \ref{lem:1a}, \ref{lem: bounds phi} and \eqref{eq:Pauli-1},   
we obtain
\begin{align*}
  |\langle \psi,  \mathrm{J}^{<}_{3} \psi\rangle| &\leq C\rho^{\frac{2}{3}}\sum_{\sigma\neq \sigma^\prime}\int_0^\infty dt\, e^{-2t\varepsilon}e^{-t(k_F^\sigma)^2} \sup_x \|b_{\sigma^\prime}(\varphi^>_x, v, u^>)\|\sup_z \|b_{\sigma^\prime}^\ast (\Delta\varphi_z^<, v_t,u_t^<)\| 
  \\
  &\quad \times \left(\int dx \|a_\sigma(u^<_x)\psi\|^2\right)^{\frac{1}{2}}\left(\int dz \|a_\sigma^\ast(v_{t,z})\psi\|^2\right)^{\frac{1}{2}}
  \\
  &\leq CL^{\frac{3}{2}}\rho^{\frac{3}{2} + \frac{\gamma}{2}}\left(\int_0^{\infty} dt\, e^{-2t\varepsilon} \|{u}^<_{t,\sigma^\prime}\|_2 \|{v}_{t,\sigma^\prime}\|_2\right)\|\mathcal{N}^{\frac{1}{2}}\psi\| \leq CL^{\frac{3}{2}}\rho^{\frac{11}{6} - \kappa}\|\mathcal{N}^{\frac{1}{2}}\psi\|. 
\end{align*}
Combining all the above estimates, we find  that 
\begin{equation}\label{Jf}
  |\langle \psi, \mathrm{J} \psi\rangle| \leq CL^{\frac{3}{2}}\rho^{\frac{3}{2} - \gamma - \kappa} \left(\|\mathbb{H}_0^{\frac{1}{2}}\psi\| + \rho^{\frac{1}{3}}\|\mathcal{N}^{\frac{1}{2}}\psi\|\right).
\end{equation}
In combination with \eqref{eq: est T2*R I>} this gives
\begin{align}\nn
&\sum_{\sigma}\sum_{r\in \Lambda^\ast} ||r|^2 - (k_F^\sigma)^2| |\langle \psi, S^\ast_{2,\sigma}(r)T_\sigma(r)\psi\rangle|
  \\
 &\leq C ( \rho^{\frac{3}{2} -\gamma - \kappa} + \rho^{\frac{3}{2} + \frac{21}{10}\gamma- \frac{21}{10}\alpha - \kappa} ) (L^{\frac 3 2}\|\mathbb{H}_0^{\frac{1}{2}}\psi\| + \rho^{\frac{1}{3}} L^{\frac{3}{2}} \|\mathcal{N}^{\frac{1}{2}}\psi\|).
 \label{eq: est  S2*T}
 \end{align}

The statement of the Lemma follows from \eqref{56.1}, \eqref{56.2} and \eqref{eq: est  S2*T}.

\section{Improved a priori bounds}\label{sec:apriori}

In this section we prove Proposition \ref{pro: a priori}. 

\begin{proof}[Proof of Proposition \ref{pro: a priori}]
Let $\Psi=R\psi$ be a ground state of $H_N$ or, more generally, any state with an energy expectation exceeding the ground state energy by at most $O(L^3 \rho^{7/3})$. From the upper bound on the ground state energy in \cite[Theorem 1.2]{GHNS} and Proposition \ref{pro: fermionic transf} (with the choice $\beta =0$), we know that 
\begin{align}
&E_{\mathrm{FFG}}  + \left\langle  \psi, \Big[ \mathbb{H}_0 +   \mathbb{Q}_2 + \mathbb{Q}_3 + \mathbb{Q}_4 \Big]  \psi\right\rangle \nn\\
&\le   \frac{3}{5}(6\pi^2)^{\frac{2}{3}} \left(\rho_\uparrow^{\frac{5}{3}}  + \rho_\downarrow^{\frac{5}{3}}\right) L^3 + 8\pi a \rho_\uparrow\rho_\downarrow  L^3+ C L^3 \rho^{7/3} + C\rho \langle \psi, \cN \psi\rangle.  
 \label{eq:apriori-0}
\end{align} 
We shall derive a lower bound for the left-hand side of \eqref{eq:apriori-0} by using simplified versions of Propositions   \ref{pro: completion square Vphi} and \ref{pro: completion square Vf}. By Lemma \ref{lem: Ta}, Lemma \ref{lem: T} and the a-priori estimates for $\mathbb{H}_0$ and $\mathbb{Q}_4$ in Lemma \ref{lem: a priori}, we have 
\begin{align}\label{eq:apriori-1a}
& \left\langle \psi, \Big[  \mathbb{H}_0 + \mathbb{Q}_2\big\vert_{Vf} \Big] \psi\right\rangle \ge   \left\langle \psi,  \big[ \widetilde{\mathbb{H}}_0 - C\rho \cN\big] \psi \right\rangle - C\rho^{7/3} L^3-  \mathfrak{e}_L   L^{3} \nn\\
&\qquad \qquad  - \frac{1}{L^6}\sum_{p,r,r^\prime}\frac{(2|p|^2 \hat{\varphi}(p))^2}{\lambda_{p,r}+ \lambda_{-p,r^\prime} + 2\varepsilon}\hat{u}_{\uparrow}(r+p)\hat{u}_{\downarrow}(r^\prime - p)\hat{v}_\uparrow(r)\hat{v}_{\downarrow}(r^\prime) 
 \end{align}
with $\widetilde{\mathbb{H}} _0$ introduced in  \eqref{eq:tH0}. The bound \eqref{eq:apriori-1a} serves as a simplified version of Proposition \ref{pro: completion square Vf}, since it does not take into account any contribution of $\bQ_3$. The full term $\bQ_3$ can in fact be handled using 
\begin{align}  \label{eq:apriori-1b}
&  \mathbb{Q}_4 +  \mathbb{Q}_3 + \mathbb{Q}_2\big\vert_{V\varphi} \ge -L^3\rho_\uparrow\rho_\downarrow \int_{\Lambda}\, V\varphi^2 - C\rho \cN 
\\
  &\qquad +\frac{1}{2}\sum_{\sigma\neq\sigma^\prime}\int_{\Lambda^2} dxdy\, V(x-y)\left| a_{\sigma^\prime}(u_y)a_\sigma(u_x) + \mathcal{T}_{\sigma,\sigma^\prime}(x,y) + \widetilde{\mathcal{S}}_{\sigma,\sigma^\prime}(x,y)\right|^2  \nn
  \end{align}
with $\mathcal{T}_{\sigma,\sigma^\prime}(x,y)$ defined in Definition \ref{def:ren-Q2} and $\widetilde{\mathcal{S}}^\ast_{\sigma,\sigma^\prime}(x,y) =    a_{\sigma}^\ast(u_x)a_{\sigma^\prime}(v_y) - a_{\sigma^\prime}^\ast(u_y)a_\sigma(v_x)$,  a simpler form of $\mathcal{S}^\ast_{\sigma,\sigma^\prime}(x,y)$ in Definition \ref{def:ren-Q3}.  The proof of \eqref{eq:apriori-1b} can be obtained in a straightforward way following the one of Proposition \ref{pro: completion square Vphi}, but is in fact simpler since we only aim at getting an error of order $\rho \cN$.  For completeness, we repeat the main steps of the argument. Using \eqref{eq:Q4-square-dec-1} and that 
\[
  \sum_{\sigma\neq \sigma^\prime}\int \, dxdy\, V(x-y)\widetilde{\mathcal{S}}^\ast_{\sigma,\sigma^\prime}(x,y)a_{\sigma^\prime}(u_y)a_\sigma(u_x) + \mathrm{h.c.} = \mathbb{Q}_3,
\]
which follows directly from the definition of $\widetilde{\mathcal{S}}^\ast_{\sigma,\sigma^\prime}(x,y)$, we obtain
\begin{align*}
& \mathbb{Q}_4 + \mathbb{Q}_3 + \mathbb{Q}_2\big\vert_{V\varphi}  - \frac{1}{2}\sum_{\sigma\neq \sigma^\prime}\int dxdy\, V(x-y)  \left| a_{\sigma^\prime}(u_y)a_\sigma(u_x) + \mathcal{T}_{\sigma, \sigma^\prime}(x,y) + \widetilde{\mathcal{S}}_{\sigma, \sigma^\prime}(x,y)\right|^2 
  \\
  & =- \frac{1}{2}\sum_{\sigma\neq\sigma^\prime}\int dxdy\, V(x-y)\left( |\mathcal{T}_{\sigma, \sigma'}(x,y)|^2 +|\widetilde{\mathcal{S}}_{\sigma, \sigma'}(x,y)|^2 + (\mathcal{T}^\ast_{\sigma, \sigma^\prime}(x, y)\widetilde{\mathcal{S}}_{\sigma, \sigma^\prime}(x,y) + \mathrm{h.c.}) \right).\nn 
\end{align*}
We first consider the term involving $|\mathcal{T}_{\sigma,\sigma^\prime} (x,y)|^2$. Writing the corresponding contribution in normal order, as in Proposition \ref{pro: completion square Vphi}, and isolating the constant term, we obtain
\begin{align}\label{eq:Q4-square-dec-3-rhoN}
  & -\frac{1}{2}\sum_{\sigma\neq \sigma^\prime}\int dxdy\, V(x-y) |\mathcal{T}_{\sigma,\sigma^\prime}(x,y)|^2  \\
  &\qquad = -\rho_\uparrow\rho_\downarrow L^3\int_{\Lambda} V\varphi^2 - \mathbb{Q}_0\big\vert_{V\varphi^2}  + \sum_{\sigma\neq\sigma^\prime} \rho_\sigma \int V\varphi^2 \int dx |a_{\sigma^\prime}(v_x)|^2\nn
  \\
  &\qquad\geq-\rho_\uparrow\rho_\downarrow L^3\int_{\Lambda} V\varphi^2 - \mathbb{Q}_0\big\vert_{V\varphi^2}\nn
\end{align}
where $\mathbb{Q}_0\big\vert_{V\varphi^2}$ is defined  in \eqref{eq:Q0-Vphi}.
Proceeding as in Proposition \ref{pro: completion square Vphi} but replacing ${S}_{\sigma,\sigma^\prime}(x,y)$ with $\widetilde{\mathcal{S}}_{\sigma,\sigma^\prime}(x,y)$, we have
\begin{align}\label{eq:Q4-square-dec-4-rhoN}
  \frac{1}{2}\sum_{\sigma\neq \sigma^\prime}\int dxdy\, V(x-y)|\widetilde{\mathcal{S}}_{\sigma, \sigma^\prime}(x,y)|^2 =
 \sum_{\sigma\neq\sigma^\prime}\rho_{\sigma^\prime}\int  V \int dx |a_\sigma(u_x)|^2 +\mathcal{E}_{\widetilde{\mathcal{S}}} 
\end{align}
with
\begin{align*}
\mathcal{E}_{\widetilde{\mathcal{S}}}    
&= - \sum_{\sigma\neq \sigma^\prime}\int dxdy\, V (x-y)a_\sigma^\ast(u_x)a_{\sigma^\prime}^\ast(v_y)a_{\sigma^\prime}(v_y)a_\sigma(u_x) \nn
  \\
  &\quad + \sum_{\sigma\neq \sigma^\prime}\int dxdy\, V(x-y) a^\ast_\sigma(u_x)a_\sigma^\ast(v_x)a_{\sigma^\prime}(v_y)a_{\sigma^\prime}(u_y)
\end{align*}
and also, using $u_\sigma v_\sigma=0$, 
\begin{align} \label{eq:Q4-square-dec-5-rep}
 & \frac{1}{2}\sum_{\sigma\neq \sigma^\prime} \int dxdy\, V(x-y) \mathcal{T}_{\sigma, \sigma^\prime}^\ast(x,y)\widetilde{\mathcal{S}}_{\sigma, \sigma^\prime}(x,y) \nn\\
&= \sum_{\sigma\neq \sigma^\prime}\int dxdy\, (V\varphi)(x-y) a^\ast_{\sigma^\prime}(v_y)a_\sigma(v_x)a_{\sigma^\prime}(v_y)a_{\sigma}(u_x)=: \mathcal{E}_{\mathcal{T}\widetilde{\mathcal{S}}} 
   \end{align}
   We thus conclude that 
\begin{align*}
  &\mathbb{Q}_4 + \mathbb{Q}_3 + \mathbb{Q}_2\big\vert_{V\varphi} - \frac{1}{2}\sum_{\sigma\neq\sigma^\prime}\int_{\Lambda^2} dxdy\, V(x-y)\left| a_{\sigma^\prime}(u_y)a_\sigma(u_x) + \mathcal{T}_{\sigma,\sigma^\prime}(x,y) + \widetilde{\mathcal{S}}_{\sigma,\sigma^\prime}(x,y)\right|^2 
  \\
  & \geq - L^3\rho_\uparrow\rho_\downarrow \int_\Lambda V\varphi^2 - \mathbb{Q}_0\big\vert_{V\varphi^2} - \sum_{\sigma\neq \sigma^\prime}\rho_{\sigma^\prime}\int V \int dx\, |a_\sigma(u_x)|^2 - \mathcal{E}_{\widetilde{\mathcal{S}}} - \mathcal{E}_{\mathcal{T}\widetilde{\mathcal{S}}} \nn
\end{align*}
All terms on the right hand side above, but the constant term, can be bounded proceeding by proceeding as in  Lemma \ref{lem:4aaa} and using that $\|V\varphi^2\|_1 \le \|V\|_1 \|\varphi\|_\infty^2\le C$ (and similarly for $\|V\varphi\|_1$) from Lemma \ref{lem: bounds phi}. In the present setting, however, the estimates are simpler, since no cancellations among the various terms are needed. More precisely, we have 
\begin{align*}
  |\langle\psi, \mathbb{Q}_0\big\vert_{V\varphi^2}\psi\rangle| &\leq C \sum_{\sigma\neq \sigma^\prime}\int_{\Lambda^2}\,dxdy\, |V\varphi^2(x-y)|(\sup_y \|a_{\sigma^\prime}^\ast(v_y)a_{\sigma^\prime}(v_y)\|)\|a_\sigma(v_x)\psi\|^2
  \\
  &\leq C\rho\|V\varphi^2\|_1 \langle\psi\mathcal{N}\psi\rangle \leq C\rho\langle \psi,\mathcal{N}\psi\rangle,
\end{align*}
using also \eqref{eq:Pauli-1} and Lemma \ref{lem:1a}. Similarly,
\[
    \sum_{\sigma\neq \sigma^\prime}\rho_{\sigma^\prime}\int V \int dx\, |a_\sigma(u_x)|^2 \leq C\rho\|V\|_1 \mathcal{N}\leq C\rho\mathcal{N},
\]
using again Lemma \ref{lem:1a}. The terms in $\mathcal{E}_{\widetilde{\mathcal{S}}}$ and $\mathcal{E}_{\mathcal{T}\widetilde{\mathcal{S}}}$ are estimated in the same way, and are each bounded by $C\rho\mathcal{N}$. This proves \eqref{eq:apriori-1b}.

Combining \eqref{eq:apriori-1a} and \eqref{eq:apriori-1b}, dropping the last positive term of \eqref{eq:apriori-1b} for a lower bound and using that $V = V\varphi + Vf$,  we arrive at 
\begin{align}
&E_{\mathrm{FFG}} + \left\langle \psi, \Big[ \mathbb{H}_0 +  \mathbb{Q}_2 + \mathbb{Q}_3 + \mathbb{Q}_4 \Big] \psi\right\rangle \nonumber\\
&  \ge  \left\langle \psi,  \widetilde{\mathbb{H}} _0  \psi \right\rangle -  C \rho \left\langle \psi,  \cN \psi \right\rangle -  C\rho^{7/3} L^3 - \mathfrak{e}_L   L^{3}  \nn\\
&+E_{\mathrm{FFG}}  -L^3\rho_\uparrow\rho_\downarrow \int_{\Lambda} V\varphi^2  - \frac{1}{L^6}\sum_{p,r,r^\prime}\frac{(2|p|^2 \hat{\varphi}(p))^2}{\lambda_{p,r}+ \lambda_{-p,r^\prime} + 2\varepsilon}\hat{u}_{\uparrow}(r+p)\hat{u}_{\downarrow}(r^\prime - p)\hat{v}_\uparrow(r)\hat{v}_{\downarrow}(r^\prime) 
 \nn\\
&\ge  \frac{3}{5}(6\pi^2)^{\frac{2}{3}} \left(\rho_\uparrow^{\frac{5}{3}}  + \rho_\downarrow^{\frac{5}{3}}\right) L^3 + 8\pi a \rho_\uparrow\rho_\downarrow  L^3 + \left\langle \psi,  \widetilde{\mathbb{H}} _0  \psi \right\rangle -  C \rho \left\langle \psi,  \cN \psi \right\rangle -  C\rho^{7/3} L^3 - \mathfrak{e}_L   L^{3},  
 \label{eq:apriori-1}
\end{align}
where in the last step we applied the bounds \eqref{eq: FFG energy} and  \eqref{eq:corr-energy-1}. 
In combination, it follows from  \eqref{eq:apriori-0} and \eqref{eq:apriori-1} that
\begin{align}\label{eq:apriori-2}
\langle \psi, \widetilde{\mathbb{H}}_0  \psi\rangle \le C \rho \langle \psi, \cN \psi\rangle+C L^3 \rho^{7/3}.
\end{align}

We shall now show that for all $1\le \eta \le C\rho^{-1/3}$,  
\begin{align}\label{eq:apriori-3}
 \langle \psi,  \cN  \psi\rangle \le C L^3 \eta \rho^{4/3} + C\eta^{-1} \langle \psi, ( \rho^{-1} \widetilde{\mathbb{H}}_0 + \cN) \psi\rangle. 
\end{align}
Using \eqref{eq:Pauli} we get the simple bound  
\begin{align}\label{eq:apriori-4}
\cN = \sum_{\sigma}  \sum_{r\in \Lambda^*} \hat a_{r,\sigma}^* \hat a_{r,\sigma} &\le \sum_{\sigma} \left( \sum_{| |r|-k_F^\sigma| \le \eta \rho^{2/3} }  1 +   \sum_{| |r|-k_F^\sigma|>\eta \rho^{2/3} }  \hat a_{r,\sigma}^* \hat a_{r,\sigma} \right) \nn \\
& \le C L^3 \eta \rho^{4/3} +\sum_{\sigma} \sum_{| |r|-k_F^\sigma|>\eta \rho^{2/3} }  \hat a_{r,\sigma}^* \hat a_{r,\sigma}. 
\end{align}
In the last term in \eqref{eq:apriori-4} we apply the Cauchy--Schwarz inequality  as $\hat a_{r,\sigma}^* \hat a_{r,\sigma} \le  2 ( \left|  \hat a_{r,\sigma} + T_{\sigma}(r) \right|^2  +  |T_{\sigma}(r) |^2 )$ with $T_\sigma(r)$ defined in \eqref{eq: def Tk}. By the definition of $\widetilde{\mathbb{H}}_0$ in \eqref{eq:tH0} we have the obvious inequality 
\begin{align}\label{eq:apriori-6}
\sum_\sigma \sum_{{| |r|-k_F^\sigma|>\eta \rho^{2/3} } }  \left|  \hat a_{r,\sigma} + T_{\sigma}(r)\right|^2 
&\le C \sum_\sigma \sum_{{| |r|-k_F^\sigma|>\eta \rho^{2/3} } } \frac{||r|^2- (k_F^\sigma)^2|}{\eta \rho }  \left|  \hat a_{r,\sigma} + T_{\sigma}(r) \right|^2 \nn\\
&\le C\eta^{-1}\rho^{-1} \widetilde{\mathbb{H}}_0.  
\end{align}
 
For the sum involving $\left|   T_{\sigma}(r) \right|^2$, we adapt the exact computation in \eqref{eq:RR-dec} in the proof of Lemma \ref{lem: T} to decompose  
\begin{align}\label{eq:apriori-8}
\sum_\sigma \sum_{| |r|-k_F^\sigma|>\eta \rho^{2/3} } \left\{   T_{\sigma}(r)^\ast, T_{\sigma}(r)\right\} =\sum_{j=1}^{10}\widetilde{\mathrm{I}}_j
 \end{align}
where the terms $\widetilde{\mathrm{I}}_j$ are identical to the ${\mathrm{I}}_j$ in \eqref{eq:RR-dec}, except for the replacement 
of the factors  $||r|^2-k_F^2|$ by $\1_{\{  | |r| - k_F^\sigma| > \eta \rho^{2/3}\}}$, and correspondingly also of factors $||r+p|^2-|r|^2|$ (where $r$ is inside and $p+r$ outside the Fermi ball) by  $\1_{\{  | |r+p| - k_F^\sigma| > \eta \rho^{2/3}\}} + \1_{\{  | |r| - k_F^\sigma| > \eta \rho^{2/3}\}}$. We shall only write down explicitly the first term, given by 
\begin{align}\nn
 \widetilde{\mathrm{I}}_1 & = \frac{1}{L^6}\sum_{\sigma\neq \sigma^\prime}\sum_{p,r,r^\prime}\left(\1_{\{  | |r+p| - k_F^\sigma| > \eta \rho^{2/3}\}} + \1_{\{  | |r| - k_F^\sigma| > \eta \rho^{2/3}\}} \right)  \\
 &\qquad\qquad\times (\widehat{\omega}^\varepsilon_{r,r^\prime}(p))^2 \hat{u}_{\sigma}(r+p)\hat{u}_{\sigma^\prime}(r^\prime - p)\hat{v}_\sigma(r)\hat{v}_{\sigma^\prime}(r^\prime). \label{def:ti1}
\end{align}
The terms $\widetilde{\mathrm{I}}_{2}, \widetilde{\mathrm{I}}_{3}, \widetilde{\mathrm{I}}_{9}, \widetilde{\mathrm{I}}_{10}$ are non-positive, and can be dropped for an upper bound. The terms $\widetilde{\mathrm{I}}_{j}$ with $4\leq j \leq 8$ are non-negative, but they can be estimated  simply using, as in \eqref{eq:apriori-6}, 
 $$
 1_{\{ | |r| - k_F^\sigma|> \eta \rho^{2/3}\}}  \le C  (\eta \rho)^{-1} | |r|^2- (k_F^\sigma)^2|
 $$
and the bounds on ${\mathrm{I}}_{j}$ in  the proof of Lemma \ref{lem: T},  yielding  
\begin{align}\label{eq:apriori-8a}
\langle \psi, \widetilde{\mathrm{I}}_{j} \psi \rangle \le C(\eta \rho)^{-1}  \langle \psi, {\mathrm{I}}_{j}  \psi \rangle \le C \eta^{-1} \big(\rho^{4/3} L^3 + \langle \psi ,  \cN \psi\rangle \big) , \quad 4\leq j \leq 8.
\end{align}
To obtain this bound, we also used the a-priori estimates in Lemma \ref{lem: a priori}, which actually yield an exponent larger than $4/3$ in the first term, but \eqref{eq:apriori-8a} is sufficient for our purposes.

Finally, we consider the constant term $\widetilde{\mathrm{I}}_1$ in  \eqref{def:ti1}. On the set $\{  | |r+p| - k_F^\sigma| > \eta \rho^{2/3}\}$, we bound $\omega^\eps_{r,r'}(p) \leq C/ (|p+r|^2 - (k_F^\sigma)^2)$, using that  $|p|^2|\hat \varphi(p)| \leq C$. On the set $\{  | |r+p| - k_F^\sigma| \leq \eta \rho^{2/3}\}$, on the other hand, we use $\omega^\eps_{r,r'}(p) \leq C/ (  (k_F^\sigma)^2 - |r|^2)$; the remaining sum over $p$ then gives  a factor bounded by $L^3 \eta \rho^{4/3} \leq C L^3 \rho$. Combining both cases, we thus find
\begin{align}\label{eq:constant-widetilde-I1}
\widetilde{\mathrm{I}}_1 & \leq \frac{C}{L^6}\sum_{\sigma\neq \sigma^\prime} \sum_{p,r,r^\prime} \frac{ \1_{\{  | |r+p| - k_F^\sigma| > \eta \rho^{2/3}\}} }{ | |r+p|^2 -( k_F^{\sigma})^2|^2} \hat{v}_\sigma(r)\hat{v}_{\sigma^\prime}(r^\prime)\nonumber
\\
&\quad +\frac{C}{L^6}\sum_{\sigma\neq \sigma^\prime} \sum_{p,r,r^\prime} \frac{ \1_{\{  | |r+p| - k_F^\sigma| \leq \eta \rho^{2/3}\}}\1_{\{  | |r| - k_F^\sigma| \geq \eta \rho^{2/3}\}} }{ | |r|^2 -( k_F^{\sigma})^2|^2} \hat{v}_\sigma(r)\hat{v}_{\sigma^\prime}(r^\prime)\nonumber
\\
&\le C  \rho^2 \sum_{r} \frac{ \1_{\{  | |r| - k_F^\sigma| > \eta \rho^{2/3}\}} }{ | |r|^2 -( k_F^{\sigma})^2|^2} \le C \rho^2  L^3  \int_{\R^3} dx \frac{\1_{\{  | |x| - k_F^\sigma| > \eta \rho^{2/3}\}} }{||x|^2-(k_F^{\sigma})^2|^2}\nn \\
&\le  C\rho^2 L^3 \int_{0}^\infty \frac{  \1_{\{  | t - k_F^\sigma| > \eta \rho^{2/3}\}} }{|t-k_F^{\sigma}|^2} \le C  \eta^{-1} \rho^{4/3} L^3,
 \end{align}
where in the penultimate estimate we approximated $\sum_r$ with an integral over $\mathbb{R}^3$.

From \eqref{eq:constant-widetilde-I1},  \eqref{eq:apriori-8a} and \eqref{eq:apriori-8}, we conclude that 
\begin{align}\label{eq:apriori-8c}
\sum_\sigma \sum_{| |r|- k_F^\sigma|> \eta \rho^{2/3} } \langle \psi ,   \left|   T_{\sigma}(r) \right|^2 \psi \rangle   \le C \eta^{-1} \rho^{4/3}L^3 + C\eta^{-1}  \langle \psi , \cN \psi\rangle. 
\end{align}
Moreover, together with 
 \eqref{eq:apriori-6} this gives 
\begin{align}\label{eq:apriori-9}
\sum_\sigma \sum_{| |r|- k_F^\sigma|> \eta \rho^{2/3} }  \langle \psi,  \hat a^*_{r,\sigma}  \hat a_{r,\sigma}   \psi\rangle \le  C\eta^{-1} \langle \psi, ( \rho^{-1} \widetilde{\mathbb{H}}_0 + \cN + \rho^{4/3}L^3) \psi\rangle
\end{align}
for all $1\le \eta \le C\rho^{-1/3}$. 
In combination with \eqref{eq:apriori-4}, this yields the desired bound \eqref{eq:apriori-3}.

  \medskip
 {\bf Bound for $\cN$.} Inserting \eqref{eq:apriori-2} into the right-hand side of \eqref{eq:apriori-3}, we get
\begin{align}\label{eq:apriori-9a}
\langle \psi, \cN  \psi\rangle \le C  \eta \rho^{4/3} L^3   + C \eta^{-1}   \langle \psi, \cN \psi\rangle
\end{align}
for all $1\le \eta \le C\rho^{-1/3}$. In \eqref{eq:apriori-9a}, we can choose $\eta$ sufficiently large, but independent of $\rho$, to conclude that $\langle \psi, \cN \psi\rangle\le C\rho^{4/3}L^3$. 

 \medskip
 {\bf Bound for $\cN_\beta$.}
From the  bound $\langle \psi, \cN \psi\rangle\le C\rho^{4/3}L^3$ and \eqref{eq:apriori-2}, we have
\begin{align}\label{eq:apriori-10}
\langle \psi, \widetilde{\mathbb{H}}_0 \psi\rangle \le C L^3 \rho^{7/3}.  
\end{align}
Inserting \eqref{eq:apriori-10} and the  bound $\langle \psi, \cN \psi\rangle\le C\rho^{4/3}L^3$ into \eqref{eq:apriori-9}, we deduce that 
\begin{align}\label{eq:apriori-9b}
\sum_\sigma \sum_{| |r| - k_F^\sigma| > \eta \rho^{2/3}} \langle \psi,  \hat a^*_{r,\sigma}  \hat a_{r,\sigma}   \psi\rangle  \le C\eta^{-1} \rho^{4/3}L^3
\end{align}
for all $1\le \eta \le C\rho^{-1/3}$. For every $0 \le \beta < 1$, we can apply \eqref{eq:apriori-9b} with $\eta = \min_\sigma (k_F^{\sigma})^{1+\beta} \rho^{-2/3} \le C \rho^{(\beta - 1)/3}$ to conclude that 
$$\langle \psi, \mathcal{N}_\beta \psi \rangle \le C L^3 \rho^{\frac {5-\beta} 3}.$$
The proof of Proposition \ref{pro: a priori} is complete. 
\end{proof}


\end{document}